\documentclass[runningheads,clevref]{llncs}
\usepackage[T1]{fontenc}
\usepackage[hyphens]{url}  
\usepackage{graphicx}
\usepackage{fullpage}
\usepackage[utf8]{inputenc}

\usepackage{graphicx}
\usepackage{color}
\usepackage{amsmath}
\usepackage{amssymb}
\usepackage{nicefrac}
\usepackage{booktabs}
\usepackage{paralist}
\usepackage{thm-restate}

\usepackage[inline]{enumitem}
\usepackage{multirow}

\usepackage{cleveref}

\usepackage{subcaption}
\usepackage{todonotes}
\usepackage{pgfplots}
\pgfplotsset{compat=newest}
\usepackage{tikz}

\newcommand{\naturals}{\mathbb{N}}

\newcommand{\sharpp}{{{\mathrm{\#P}}}}
\newcommand{\np}{{{\mathrm{NP}}}}
\newcommand{\fpt}{{{\mathrm{FPT}}}}
\newcommand{\p}{{{\mathrm{P}}}}

\definecolor{darkgreen}{rgb}{0,0.5,0}
\definecolor{darkpink}{rgb}{0.75,0.25,0.25}
\definecolor{RED}{rgb}{1,0,0}

\newcommand{\figurecut}[1]{}
\newcommand{\cutfornow}[1]{}
\newcommand{\xp}{{{\mathrm{XP}}}}
\newcommand{\wone}{{{\mathrm{W[1]}}}}
\newcommand{\sharpwone}{{{\#\mathrm{W[1]}}}}

\usepackage{environ}

\newcommand{\repeatproposition}[1]{  \begingroup
  \renewcommand{\theproposition}{\ref{#1}}  \expandafter\expandafter\expandafter\proposition
  \csname repproposition@#1\endcsname
  \endproposition
  \endgroup
  \setcounter{theorem}{\value{theorem}-1}
}

\NewEnviron{repproposition}[1]{  \global\expandafter\xdef\csname repproposition@#1\endcsname{    \unexpanded\expandafter{\BODY}  }  \expandafter\proposition\BODY\unskip\label{#1}\endproposition
}

\newcommand{\repeattheorem}[1]{  \begingroup
  \renewcommand{\thetheorem}{\ref{#1}}  \expandafter\expandafter\expandafter\theorem
  \csname reptheorem@#1\endcsname
  \endtheorem
  \endgroup
  \setcounter{theorem}{\value{theorem}-1}
}

\NewEnviron{reptheorem}[1]{  \global\expandafter\xdef\csname reptheorem@#1\endcsname{    \unexpanded\expandafter{\BODY}  }  \expandafter\theorem\BODY\unskip\label{#1}\endtheorem
}

\newcommand{\rxthreec}{\textsc{RX3C}}

\newcommand{\greedyAVRule}{\textsc{GreedyAV}}
\newcommand{\greedyCostRule}{\textsc{GreedyCost}}
\newcommand{\MESPB}{\textsc{MES-PB}}
\newcommand{\MESAprUtilRule}{\textsc{MES-Apr}}
\newcommand{\MESCostUtilRule}{\textsc{MES-Cost}}
\newcommand{\phragmenRule}{\textsc{Phragm{\'e}n}}

\usepackage{xspace}
\newcommand{\Greedy}{\greedyAVRule\xspace}
\newcommand{\GreedyCost}{\greedyCostRule\xspace}
\newcommand{\MES}{\MESCostUtilRule\xspace}
\newcommand{\Phragmen}{\phragmenRule\xspace}

\newcommand{\AddRemProbName}{\textsc{Flip-Bribery}}    
    
\newcommand{\greedyAVAddRemPB}{\textsc{GreedyAV-\AddRemProbName}}
\newcommand{\greedyCostAddRemPB}{\textsc{GreedyCost-\AddRemProbName}}

\newcommand{\MESAprUtilAddRemPB}{\textsc{MES-Apr-\AddRemProbName}}
\newcommand{\MESCostUtilAddRemPB}{\textsc{MES-Cost-\AddRemProbName}}
\newcommand{\phragmenAddRemPB}{\textsc{Phragm{\'e}n-\AddRemProbName}}

\newcommand{\pre}{{\mathrm{pre}}}

\DeclareMathOperator{\cost}{cost}

\begin{document}

\title{Robustness of Participatory Budgeting Outcomes:\\Complexity and Experiments}

\author{Niclas Boehmer\inst{1} \and
Piotr Faliszewski\inst{2} \and
Łukasz Janeczko\inst{2} \and 
Andrzej Kaczmarczyk\inst{2}}
\authorrunning{N. Boehmer et al.}
\institute{Technische Universit\"at Berlin, Berlin, \email{niclas.boehmer@tu-berlin.de}  \and
AGH University, Krakow, \email{\{faliszew,ljaneczk,andrzej.kaczmarczyk\}@agh.edu.pl}}

\maketitle

\begin{abstract}
  We study the robustness of approval-based participatory budgeting
  (PB) rules to random noise in the votes. Our contributions are
  twofold.  First, we study the computational complexity of the
  \#\AddRemProbName{} problem, where given a PB instance we ask for
  the number of ways in which we can flip a given number of approvals
  in the votes, so that a specific project is selected.  The idea is
  that \#\AddRemProbName{} captures the problem of computing the
  funding probabilities of projects in case random noise is added.
  Unfortunately, the problem is intractable even for the simplest PB
  rules. Second, we analyze the robustness of several prominent PB
  rules (including the basic greedy rule and the Method of Equal
  Shares) on real-world instances from Pabulib. Since
  \#\AddRemProbName{} is intractable, we resort to sampling to obtain
  our results. We quantify the extent to which simple, greedy PB rules
  are more robust than proportional ones, and we identify three types
  of (very) non-robust projects in real-world PB instances.
\end{abstract}

\section{Introduction}\label{sec:intro}

We study the robustness of approval-based participatory budgeting (PB)
rules to random noise in the votes. To this end, we first analyze the complexity of the computational problems that capture our approach,
and then we perform robustness experiments on real-world PB instances.

In a PB scenario, we are given a set of projects, each with its cost,
a collection of voters, and a budget. Each voter indicates which
projects he or she approves, and the task is to choose a set of projects
that receive sufficient support and whose total cost is within the
budget. In particular, participatory budgeting is a common way for
cities to let the residents influence their spending: First, the most
active citizens submit projects they would like to see implemented,
next, the whole electorate votes on these projects, and finally the selected projects
are carried out. However, while PB exercises are very
appealing in general, they do face several issues, of which we mention two:
First, the number of submitted projects is often quite high (in
Pabulib~\cite{pabulib}, a collection of real-world PB instances, over
$30\%$ of the instances involve at least $25$ projects, and over
$10\%$ involve at least $50$). Second, usually only a small fraction
of the electorate actually votes (as seen in Pabulib~\cite{pabulib},
typically this is between 4\% and 10\% of the cities'
residents). Hence, PB results may be seen as somewhat shaky.  Indeed,
it is possible that some citizens did not inform themselves about the full set of
submitted projects and would have changed their votes had they done
so, while some other citizens might have chosen to vote (or to refrain
from voting) based on somewhat random criteria (e.g., they saw a PB
advertisement at a more or less convenient moment).  We study the
potential instability of PB results by analyzing their robustness
to random noise in the votes.

More specifically, we take the following approach: Given a PB
instance, we modify each vote by randomly flipping some of the
approvals (i.e., by randomly adding approvals for some of the
unsupported projects and randomly removing approvals from some of the
supported ones) and, then, we evaluate the probability that the
outcome changes. As part of our analysis, we also examine the
probability that specific projects are funded if we add noise to the
votes like this.  On a computational level, we capture this approach
through the \#\AddRemProbName{} problem, where we ask for the number
of ways in which a given number of approval flips can be performed, so
that a specific project is selected. Unfortunately, our complexity
results show that \#\AddRemProbName{} is intractable even for the
simplest PB rules (often already in the problem's decision variant,
where we ask if the problem's solution is non-zero).
Consequently, in our numerical experiments we use sampling to estimate
the relevant probabilities: Given a PB instance, we randomly flip
approvals in the votes, following the resampling noise model of Szufa
et al.~\cite{DBLP:conf/ijcai/SzufaFJLSST22}, parameterized by a noise
level parameter~$\phi \in [0,1]$.
In our experimental analysis of real-world PB instances, our main
robustness measure is the $50\%$-winner threshold: the smallest value
of the noise level for which the outcome changes with probability at
least
$50\%$~\cite{boe-bre-fal-nie:c:counting-bribery,boe-bre-fal-nie:c:robustness-single-winner}.

We focus on five prominent PB rules. The first two, $\greedyAVRule$
and $\greedyCostRule$, consider the projects one-by-one, in the
non-increasing order of either the number of approvals they receive,
or their approval-to-cost ratios (a project's
approval-to-cost ratio is the number of approvals this project
received divided by its cost); a project is selected if there still
are enough unallocated funds at the time when it is considered. We
study these two rules because they are among the simplest ones and,
more importantly, because $\greedyAVRule$ is one of the most commonly
used rules in practice. Yet, these two rules are often criticized as
providing results that do not represent the voters proportionally.
Hence, we also consider a PB variant of the $\phragmenRule$
rule~\cite{bri-fre-jan-lac:c:phragmen,los-chr-gro:c:phragmen-pb}, and
two rules from the \emph{Method of Equal Shares} (\textsc{MES})
family, proposed by Peters, Pierczyński, and
Skowron~\cite{pet-sko:c:welfarism-mes,pet-pie-sko:c:pb-mes}. The
\textsc{MES} rules are particularly attractive both due to their good
theoretical properties and the fact that one of them was recently used
in a real-world PB exercise~\cite{mes}. All our rules are computable in
polynomial time (at least for appropriate tie-breaking; see the work
of Janeczko and Faliszewski for a deeper discussion of this
issue~\cite{jan-fal:c:ties-multiwinner}).

\newcommand{\tentry}[1]{
  \begin{minipage}[c]{3cm}
    \centering    #1
  \end{minipage}
}

\newcommand{\binaryenc}{\\ {\scriptsize (binary encoding)}}
\newcommand{\unaryenc}{\\ {\scriptsize (unary encoding)}}
\newcommand{\unitenc}{\\ {\scriptsize (unit costs)}}

\begin{table*}[t]
    \caption{\label{tab:results}Computational complexity of
    $\AddRemProbName$ for our rules when project costs are encoded in binary. Unless stated otherwise, the
    results are for order-based tie-breaking (see \Cref{sec:prelims} for definitions). For the counting variant, all
    intractability results hold already for a single voter and all
    tractability results hold for arbitrary numbers of voters. The
    parentheses after $\fpt$, $\wone$, and $\xp$ results indicate for
    which parameters they hold ($m$ for the number of projects, $n$
    for the number of voters, and $r$ for the number of approval
    flips).}

  \centering
  \scalebox{1}{
    \begin{tabular}{c|cc|c|c}
      \toprule
                   & \multicolumn{2}{c|}{decision variant}   & \multicolumn{1}{c|}{counting variant} & \\
      rule         & $1$ voter & $n$ voters &  & reference \\
      \midrule
        \tentry{\greedyAVRule} & 
        \tentry{$\np$-com.\\ $\wone$-hard($r$) \\ $\fpt(m)$} &
        \tentry{$\np$-com.\\ $\wone$-hard($r$) \\ $\fpt(m)$} &
        \tentry{$\sharpp$-com. \\ $\sharpwone$-hard($r$)  \\ $\fpt(m)$} &
        \tentry{Thms.~\ref{thm:greedy-av-pb-npc}, \ref{thm:greedy-av-pb-ftp-m}, \ref{thm:greedyav-sharpphard-with-one-voter}} \\[6mm]

        \tentry{\greedyAVRule \\ \scriptsize{(cheaper-first)} } & 
        \tentry{$\p$} &
        \tentry{$\np$-com. \\ $\wone$-hard$(n)$ \\ $\fpt(m)$ \\ $\xp(n)$} & \tentry{$\sharpp$-com. \\ $\sharpwone$-hard($r$)  \\ $\fpt(m)$} &
        \tentry{Thms.~\ref{thm:greedy-av-pb-ftp-m}, \ref{thm:greedy-av-pb-voters-xp},  \ref{thm:greedyav-sharpphard-with-one-voter}} \\[6mm]
  
        \midrule
  
        \tentry{\greedyCostRule } & 
  			\tentry{$\np$-com.\\ $\wone$-hard($r$)} &
  			\tentry{$\np$-com.\\ $\wone$-hard($r$)} &
  			\multicolumn{1}{c|}{\tentry{$\sharpp$-com. \\ $\sharpwone$-hard($r$)}} &
		\tentry{Thm.~\ref{thm:greedycost-pb-npc}, Cor.~\ref{cor:other-rules-counting}} \\[4mm]
  
        \tentry{\phragmenRule, \\ \MESAprUtilRule} & 
  			\tentry{$\p$} &
  			\tentry{$\np$-com \unitenc} &
  			\multicolumn{1}{c|}{\tentry{$\sharpp$-com. \\ $\sharpwone$-hard($r$)}} &
       \tentry{Thm.~\ref{thm:mes-phragmen}, Cors.~\ref{cor:mes-phr-poly},~\ref{cor:other-rules-counting}} \\[4mm]
  
        \tentry{\MESCostUtilRule} & 
  			\tentry{$\np$-com.\\ $\wone$-hard($r$)} &
  			\tentry{$\np$-com \unitenc} &
  			\multicolumn{1}{c|}{\tentry{$\sharpp$-com. \\ $\sharpwone$-hard($r$)}} &
        \tentry{Thm.~\ref{thm:mes-phragmen}, Cors.~\ref{cor:mes-phr-poly},~\ref{cor:other-rules-counting}} \\[4mm]
      \bottomrule
      \end{tabular}%
  }
\end{table*}

\paragraph{Contributions.} Below we summarize our main contributions:
\begin{enumerate}
\item For all our rules, we show that both \#\AddRemProbName{} and its
  decision variant are, in general, intractable. Somewhat
  surprisingly, for \greedyAVRule, this holds already for a
  single voter. Nonetheless, we do show some ways to circumvent these intractability
  results (e.g., an $\fpt$ algorithm parameterized by the number
  of projects, or an $\xp$ algorithm parameterized by the number of
  voters, for a particular tie-breaking scheme). However, generally,
  these approaches are insufficient for large PB instances, including
  many of those from Pabulib~\cite{pabulib}. We summarize our
  complexity results in Table~\ref{tab:results}.

\item We perform a robustness analysis of all our PB rules (except for
  one of the variants of \textsc{MES}) on PB instances from
  Pabulib~\cite{pabulib}. We find that the greedy
  rules (i.e., $\greedyAVRule$ and $\greedyCostRule$) are the most
  robust ones (i.e., on average they have the highest
  $50\%$-winner thresholds), whereas the proportional rules (i.e.,
  \Phragmen and a variant of \textsc{MES}) are less robust. The extent
  of the difference between the two kinds of rules is interesting, but
  the effect itself is expected: Proportional rules follow much more
  nuanced algorithms and, hence, react to noise in less predictable and potentially more extreme
  ways.

\item We find many real-world PB instances where random noise quickly
  changes project's funding probabilities. For example, in a PB
  instance from Wrzeciono Mlociny from 2019 flipping $0.24\%$ of
  approvals randomly changes the outcome in a majority of cases.  We
  identify three main types of non-robust projects: Originally
  selected projects of the first type have low funding probability
  even for a fairly small noise level, and, as we increase it, their
  funding probability decreases even further. Such projects can be
  seen as selected ``by luck.''  (Analogously, projects that
  originally are not selected, but whose funding probability quickly
  increases even for small noise levels, can be seen as rejected ``by
  misfortune.''). Non-robust projects of the second type quickly reach
  a funding probability of $50\%$ and stay at this level irrespective
  of whether further noise is added to the votes. Such projects can be
  seen as tied with each other. Finally, for projects of the third
  type the funding probability behaves non-monotonically as the level
  of noise increases (such projects deserve individual analysis).
\end{enumerate}
All in all, we believe that a robustness
analysis should always be performed after each participatory budgeting
exercise. Its results can either be used by the city to fund some
additional projects (when they are identified as ``effectively tied''
or ``unselected due to misfortune''), or by the voters, to understand
how the projects they care about performed and how ``far away'' they
were from (not) getting funded.

\paragraph{Related Work.}
Approval-based participatory budgeting is an extension of the
approval-based (multiwinner) voting framework, where each project
(typically referred to as a candidate) has a unit cost and the goal is
to select a fixed-size committee. We point to the works of Rey and
Maly~\cite{rey-mal:t:pb-survey} and Lackner and
Skowron~\cite{lac-sko:b:approval-survey} for overviews of the PB and
multiwinner voting frameworks, respectively.

Analyses of the probability that election outcomes change under noise
were recently pursued by Boehmer et
al.~\cite{boe-bre-fal-nie:c:counting-bribery} and Baumaister and
Hogrebe~\cite{bau-hog:c:robust-winner}. Both teams analyzed counting
variants of various bribery problems, i.e., problems where we ask for
the number of ways to modify the votes so that a particular candidate
wins, and they both focused on the single-winner, ordinal case, where
the voters rank the candidates and the voting rule outputs a single
winner. Relevant to our work, Boehmer et
al.~\cite{boe-bre-fal-nie:c:counting-bribery} introduced the notion of
the $50\%$-winner threshold, which they later applied to analyzing
real-world elections~\cite{boe-bre-fal-nie:c:robustness-single-winner}. 

Another view of election robustness was considered by Shiryaev et
al.~\cite{shi-yu-elk:c:robustness}, who instead of asking for the
probability that the result changes under a given noise model, asked
for the smallest number of modifications in the votes needed for the
change (this is also related to computing the margin of victory of
particular
candidates~\cite{mag-riv-she-wag:c:stv-bribery,xia:margin-of-victory,car:c:stv-margin-of-victory}). This
idea was later pursued by Faliszewski et
al.~\cite{bre-fal-kac-nie-sko-tal:j:multiwinner-robustness} in the
ordinal, multiwinner setting, and by Gawron and
Faliszewski~\cite{gaw-fal:c:robustness-approval} and Faliszewski et
al.~\cite{fal-gaw-kus:c:greedy-robustness} in the approval-based one.
For multiwinner approval elections, similar types of bribery problem
were studied by Faliszewski et
al.~\cite{fal-sko-tal:c:bribery-measure-success} and while their
motivation was somewhat different, their technical contributions are
in a similar spirit.

More broadly, the family of bribery problems was first studied by
Faliszewski et al.~\cite{fal-hem-hem:j:bribery}. For an overview of
this line of work, we point to the survey of Faliszewski and
Rothe~\cite{fal-rot:b:control-bribery}.

\section{Preliminaries}\label{sec:prelims}

A \emph{participatory budgeting} instance (a PB instance) is a 5-tuple
$E = (C, V, A, B, \cost)$, where $C = \{c_1, c_2, \ldots, c_m\}$ is a
set of \emph{projects} (also called \emph{candidates} in the
literature), $V = (v_1, v_2, \ldots, v_n)$ is a collection of
\emph{voters}, $A \colon V \rightarrow 2^C$ is a function that
associates each voter $v_i$ with the set of projects that $v_i$
\emph{approves}, $B \in \naturals$ is the available \emph{budget}, and
$\cost \colon C \rightarrow \naturals$ is a function that associates
each project with its cost.  For convenience, we overload the function
$A$ so that for each project~$c_i$, $A(c_i)$ is the set of voters that
approve~$c_i$ (i.e.,\ the set of voters $v_j\in V$ with
$c_i \in A(v_j)$).  We refer to $|A(c_i)|$ as the \emph{approval
  score} of $c_i$.  Given a set of projects $S\subseteq C$, by
$\cost(S)$ we mean $\sum_{c_i \in S} \cost(c_i)$, i.e., the total cost
of the members of $S$.  Given a PB instance $E = (C,V, A, B, \cost)$,
a set $S \subseteq C$ is a \emph{feasible outcome} if
$\cost(S) \leq B$, i.e., if its cost is within the available budget.

\subsection{Budgeting Rules}

A \emph{budgeting rule} is a function $f$ that given a PB instance
outputs the winning feasible outcome (note that we assume that our rules are
resolute; we will discuss the used tie-breaking mechanisms later).  We
focus on sequential budgeting rules, which compute a feasible
outcome~$W$ starting from an empty set and including projects
one-by-one, until a certain condition is met. We consider the
following rules:
\begin{enumerate}
\item The \emph{$\greedyAVRule$} rule considers the projects in the
  order of their non-increasing approval scores (breaking ties between
  projects according to a given, prespecified rule) and it includes a
  project into the outcome~$W$ if doing so would not exceed the
  available budget.

\item The \emph{$\greedyCostRule$} rule proceeds analogously
  to~$\greedyAVRule$, but it considers the projects in the order of
  non-increasing ratios of their approval score to their cost.

\item The \emph{\phragmenRule} rule constructs the outcome in a
  continuous process, in which voters virtually buy the projects to be
  funded. The voters start with an empty (virtual) bank account and
  are given (virtual) money at a constant rate as the process runs
  (the exact value of the constant is irrelevant); slightly abusing
  the notation, we refer to the balance of voter~$v_i \in V$ as~$b(v_i)$
  (the balance fluctuates over time but we use this notation only at
  well-defined points of the procedure). As soon as there is a
  project~$c$ who is not yet included in $W$, such that:
  \[
    \sum_{v_i \in A(c)} b(v_i) = \cost(c),
    \quad\quad \text{and}    \quad\quad
    \cost(W) + \cost(c) \leq B,
  \]
  (i.e., the voters who approve $c$ have enough money to pay for it
  and doing so would not exceed the overall budget $B$), the voters
  who approve $c$ buy this project (i.e., $c$ is included in $W$ and
  the bank accounts of the voters who approve $c$ are reset to be
  empty). The process runs until no more projects can be added to~$W$
  (i.e., there is no project that is approved by at least one voter,
  such that adding this project to $W$ would not exceed the overall
  budget $B$).
  
\item The \emph{\MESAprUtilRule} and \emph{\MESCostUtilRule} rules are two
  incarnations of the \emph{Method of Equal Shares (MES)}.  Just like
  $\phragmenRule$, these rules are also based on the concept of voters
  buying the projects, but this time each voter receives~$\nicefrac{B}{|V|}$
  units of virtual money upfront. Then, the voters sequentially buy the
  projects with the best price-per-utility value (the difference between
  the two variants is in how they define this utility). Formally,
  let~$b_i(v)$ be the (virtual) account balance of voter~$v \in V$
  \emph{just before} stage~$i$ of the procedure, and
  let~$u_j(c) \in \naturals$ be the utility that voter~$v_j \in V$
  assigns to some project~$c \in C$. At each stage~$i$, a
  project~$c \not\in W$ is called \emph{$q$-affordable} if it holds
  that:
  \[
    \sum_{v_j \in A(c)} \min\{b_i(v_j), u_j(c) \cdot q \} = \cost(c).
  \]
  In other words, project $c$ is $q$-affordable if the voters who
  approve it can collect $\cost(c)$ amount of money by spending $q$
  units of budget for each unit of utility that they derive from the
  project, where the voters approving the project who do not have enough money in their
  accounts give all the money they have
  left.  If there are any $q$-affordable projects, then the rule
  includes in $W$ the one which is~$q$-affordable for the smallest~$q$
  (breaking ties according to a prespecified rule). Otherwise, the rule terminates. Both \MESAprUtilRule{} and
  \MESCostUtilRule{} follow the above-described scheme, but under the
  first one the utility~$u_j(c)$ that voter~$v_j$ gives to a
  project~$c$ is $1$ if the voter approves~$c$ and is~$0$ otherwise;
  under the second rule, this utility is~$\cost(c)$ if the voter
  approves $c$ and it is $0$ otherwise.
\end{enumerate}

Regarding tie-breaking, we assume that each PB instance also comes
with a fixed order of projects and whenever we need to resolve a
tie between projects, we choose the one that is ranked
highest in this order. We refer to this approach as
\emph{order-based} tie-breaking. For the case of $\greedyAVRule$ we
also consider the \emph{cheaper-first} tie-breaking: Whenever there is
a tie among a set of projects, we choose the one with the lowest cost
(and if there are several projects with the same lowest cost, we
choose one of them using order-based tie-breaking). We observe that
cheaper-first tie-breaking is a special case of the order-based
one, so all algorithms that work for the latter also work for the
former, and all hardness results for the former also hold for the
latter.

Unlike the other rules we study, \MESAprUtilRule{}
and~\MESCostUtilRule{} output only projects which they consider
``sufficiently supported,'' even for the case where each project
received at least one approval. This implies in some instances that
the set of funded projects (returned by these rules) can be extended
by some additional projects and still fit in the budget. Formally
speaking, the MES rules (as defined above) fail exhaustiveness. In our
theoretical part, we stick to these variants to focus on fundamental
properties of the MES rules.  In the experimental part of the paper we
discuss exhaustive variants of MES.

\subsection{Computational Complexity}

We assume familiarity with both classic and parameterized
computational complexity theory, including the $\p$, $\np$, $\fpt$,
$\wone$, and $\xp$ classes, as well as the notions of reductions,
hardness, and completeness.  The $O^*(\cdot)$ notation is
analogous to the classic $O(\cdot)$ one, except that it also omits
polynomial terms.
Some of our intractability results follow
by reductions from the following variant of the \textsc{Subset Sum}
problem.
\begin{definition}
  The input of the \textsc{Sized Subset Sum} problem consists of a set
  $U = \{u_1, \ldots, u_n\}$ of positive integers, a target integer
  $t$, and a number $k \in \naturals$. We ask if there are $k$~numbers
  from $U$ that sum up to~$t$.
\end{definition}
The problem is well-known to be $\np$-complete (see, e.g., the
textbook of Garey and Johnson~\cite{gar-joh1979:b:comp-and-intract})
and $\wone$-hard with respect to parameterization by the solution
size~$k$ (see the work of Downey and
Fellows~\cite{dow-fel:c:fpt-intractability}).

Given a decision problem $X$, where we ask if a mathematical object
with some properties exists, we write $\#X$ to denote its counting
variant, i.e., the problem where we ask for the number of such objects.
A counting problem belongs to the class $\sharpp$ if its decision
variant belongs to $\np$. Similarly, $\sharpwone$ is the counting
analog of $\wone$~\cite{flu-gro:j:parameterized-counting}.  For the
case of counting problems, instead of the classic many-one reductions
it is typical to use Turing reductions. That is, to show that a
counting problem $\#A$ reduces to a counting problem $\#B$ in
polynomial time, it suffices to show a polynomial time algorithm for
$\#A$ that is free to invoke a $\#B$ oracle (i.e., it can obtain
answers to the instances of $\#B$ in unit time).  We mention that
\textsc{\#Sized Subset Sum} is both $\sharpp$-complete and
$\sharpwone$-hard for the parameterization by the solution size
$k$~\cite{mcc:j:parameterized-counting}.

\subsection{Flip-Bribery Problem}\label{sec:flipB}

We are interested in computing the probability that a given project
wins in a given PB instance provided we can flip a given number of
approvals. Formally, we capture this problem in the following
definition.

\begin{definition}
  Let $f$ be a budgeting rule.  In the $f$-\AddRemProbName{} decision
  problem, given a PB instance~$E = (C,V,A,B,\cost)$, a preferred
  project $p \in C$, and a radius $r \in \naturals$, we ask if there
  is a PB instance~$E'$ such that $p$ belongs to $f(E')$, where $E'$
  can be obtained from $E$ by performing at most $r$~approval flips.
  In the counting variant of the problem, denoted
  $\#f$-\AddRemProbName{}, we ask for the number of such PB
  instances~$E'$ (obtained by making \emph{exactly}
  $r$~approval~flips\footnote{Note that the difference between
    \emph{at most $r$} and \emph{exactly $r$} is immaterial here from
    the point of view of computational complexity: Both variants are
    Turing-reducible to each other.}).
\end{definition}

The interpretation of the counting version of the problem is as
follows: If our original PB instance has $m$ projects and $n$ voters,
and we make exactly $r$ approval flips, then there are
$y = {mn \choose r}$ different PB instances that we can obtain. If the
solution to the counting problem is $x$ (i.e., in $x$ of these $y$
instances $p$ is selected) then the probability of selecting $p$ after
making $r$ random approval flips is $x/y$. This is why we are
interested in the complexity of the counting variant of the problem
for our PB rules. However, as we will see, even this
problem---modeling a rather basic form of noise---is
intractable. Consequently, in our experiments we will evaluate the
probability that a given project is selected (under a given noise
model) using a sampling approach. Taking advantage of that, we will
use a more involved noise model that captures reality somewhat better
than the one implied by $\#\AddRemProbName$ (however, as our
intractability proofs for $\#\AddRemProbName$ work already for a
single voter, they imply hardness for any natural model of noise).

We also consider the decision variant of the problem because, on the
one hand, if already the decision problem is intractable then so is
the counting one, and, on the other hand, because the solution to the decision problem already gives us a first idea
regarding how well a particular project performed: The fewer approval
flips we need to ensure a project's victory, the stronger it is.

\section{Complexity Results for $\greedyAVRule$}\label{sec:greedyav}
We start our complexity analysis of the $(\#)\AddRemProbName$ problems
by focusing on the $\greedyAVRule$ rule.  Intuitively, it seems that
$\greedyAVAddRemPB$ is rather simple: If we want some project $p$ to
be selected, then we should give it as many approvals as possible and,
if we still have some approval flips left (i.e., if $p$ ends up being
approved by all the voters before we run out of the available flips),
then we should remove approvals from the most expensive,
approved-by-all projects that precede $p$ in the tie-breaking order (if such projects exist).
However, it turns out that the situation is more involved and, instead of giving $p$ as many
approvals as possible, it may be more beneficial to promote some other
project.  For example, $p$ may not be selected because some expensive
project~$c$ is considered and selected earlier, not leaving enough
budget for~$p$.  However, if we promote some project~$d$ to be
considered (and selected) before~$c$, then $\greedyAVRule$ may first select $d$ and then
reject~$c$ due to insufficient funds left, and may later
select~$p$.
Consequently, $\greedyAVAddRemPB$ can be quite tricky.  Indeed, we
find that
$\greedyAVAddRemPB$ is $\np$-complete (and, in fact,
para-$\np$-complete for the parameterization by the number of voters).

\begin{theorem} \label{thm:greedy-av-pb-npc}
$\greedyAVAddRemPB$ is $\np$-complete, even for a single
voter, as well as $\wone$-hard for the parameterization by the
number~$r$ of approval flips.
\end{theorem}
\begin{proof}
  We see that $\greedyAVAddRemPB$ is in $\np$: It suffices to guess the approvals to flip and verify that the preferred
  project is selected in the resulting modified PB instance. To show hardness,
  we give a reduction from \textsc{Sized Subset Sum}.

  Consider an instance $I$ of \textsc{Sized Subset Sum} with a set
  $U = \{u_1, \ldots, u_n\}$ of positive integers, a target integer
  $t$, and a number $k \in \naturals$, where we ask if there are
  $k$~numbers from $U$ that sum up to~$t$. W.l.o.g., we assume that
  there is no collection of fewer than~$k$~elements of~$U$ that sum up
  to~$t$ (otherwise, we could add $S = 1 + t + \sum_{u \in U} u_i$ to
  each element of $U$ and increase~$t$ by~$k \cdot S$, so the
  updated~$t$ could only be achieved by summing exactly~$k$ elements
  from the updated~$U$).  Further, we assume that $t > 0$, and we
  define $T = 3 \sum_{i=1}^n u_i$. Clearly, $T$ is larger than the sum
  of all the elements from $U$. We assume that $t \leq \frac{1}{3}T$
  (otherwise $I$ would certainly be a \emph{no}-instance).

  We form a PB instance with project set
  $C = \{x_1, \ldots, x_n\} \cup \{y_1, \ldots, y_{2k+1}\} \cup \{d_1,
  \ldots, d_{k+1}\} \cup \{p\}$ and a single voter who approves all
  the projects except for $x_1, \ldots, x_n$. The project costs are as
  follows:
  \begin{enumerate}
  \item For each $i \in [n]$, $\cost(x_i) = T + u_i$;
  \item For each $i \in [2k+1]$, $\cost(y_i) = T$;
  \item For each $i \in [k+1]$, $\cost(d_i) = T-t+1$;
  \item $\cost(p) = T-t$.
  \end{enumerate}
  The internal tie-breaking order is
  $x_1\succ \ldots\succ x_n\succ y_1\succ \ldots\succ y_{2k+1}\succ
  d_1\succ \ldots\succ d_{k+1}\succ p$, i.e., $x_1$ is ranked first in
  the tie-breaking order and $p$ is ranked last. We set the budget to
  be $B = kT+T$ and the number of operations to be $r = k$.

  Assume that $I$ is a \emph{yes}-instance of \textsc{Sized
    Subset Sum} and let $S \subseteq [n]$ be such that $|S| = k$ and
  $\sum_{i \in S}u_i = t$. In this case it is possible to ensure that
  $p$ is among the winning projects by adding a single approval for
  each project $x_i$ with $i \in S$. Indeed, after doing so
  $\greedyAVRule$ proceeds as follows: First, due to the specified tie-breaking, it selects the $k$
  projects $\{x_i\mid i\in S\}$ that got the approvals; their total cost is $kT+t$. Next,
  $\greedyAVRule$ considers projects $y_1, \ldots, y_{2k+1}$, but each
  of them costs $T$ and is thus above the available budget. Then,
  $\greedyAVRule$ considers projects $d_1, \ldots, d_{k+1}$ which,
  again, are too expensive (each of them costs $T-t+1$, while there is
  only $T-t$ units of budget left). Eventually, $p$, which costs
  exactly $T-t$, is selected.

  Considering the opposite direction, let us assume that it is
  possible to flip up to $r=k$ approvals so that $p$ is selected.  The
  approvals can only be added to projects~$x_1$, $x_2$, $\ldots$,
  $x_n$ and removed from the remaining ones, excluding~$p$
  (otherwise~$p$ would not be in the winning committee; this follows
  from our assumptions regarding instance $I$).  Let
  $X' \subseteq \{x_1, \ldots, x_n\}$ be the set of projects that got
  the extra approvals.
  On the modified instance, $\greedyAVRule$ proceeds as follows:
  First, it selects all the projects from $X'$. Their total cost is of
  the form $k'T+t'$, where $0 \leq k' \leq k$ and $0 \leq t' \leq
  T$. Then, $\greedyAVRule$ considers those projects among
  $y_1, \ldots, y_{2k+1}$ that did not lose their approval. Since $r = k$,
  there are at least $k+1$ of them, each with cost equal to $T$.  If
  $t' = 0$, which is possible only if no approvals were added, then
  $\greedyAVRule$ chooses $k+1$ still approved projects among
  $y_1, \ldots, y_{2k+1}$, with total cost $kT+T$.  Consequently, the
  whole budget is used and $p$ is not selected.  Hence, we know that
  $t' > 0$. So, while considering still-approved projects among
  $y_1, \ldots, y_{2k+1}$, $\greedyAVRule$ selects $k-k'$ of them. The
  total cost of the selected projects is then $kT+t'$ (where
  $t' > 0$). Next there are three possible cases:
  \begin{enumerate}
  \item If $t' > t$, then both $p$ and all the projects from $D$ are too
    expensive and $\greedyAVRule$ does not select any of them.
  \item If $t' = t$, then projects from $D$ are too expensive (each of
    them costs $T-t+1$), but after considering them (and the remaining
    ones among $y_1, \ldots, y_{2k+1}$), $\greedyAVRule$ chooses~$p$
    (which costs $T-t$).
  \item If $t' < t$, then $\greedyAVRule$ selects the first
    still-approved project from $D$ (since $|D| = k+1$, there is at
    least one such project). After this, the total cost of the so-far
    selected projects is~$kT+t' + T-t+1 \geq kT+1 + T-t+1 = kT+T-t+2$.
    Due to our choice of $T$ and the assumptions about $t$, we see
    that the remaining budget is smaller than $T-t$, so
    $\greedyAVRule$ does not choose $p$.
  \end{enumerate}
  Consequently, if by flipping at most $k$ approvals it is possible to
  ensure that $\greedyAVRule$ selects $p$, then it must be the case
  that there are exactly $k$ numbers in $U$ that sum up to $t$ (the
  reason why there are exactly $k$ of them is because we have assumed
  that only size-$k$ subsets of $U$ may sum up to~$t$).  The reduction
  clearly runs in polynomial time, which proves that $\greedyAVRule$
  is $\np$-hard and, overall, $\np$-complete. Moreover, the number~$r$
  of approval flips in the constructed instance is the same as the
  solution size~$k$ in the input instance, which yields the requested
  $\wone$-hardness.  \qed\end{proof}
In the following, we explore ways to circumvent
this hardness result.

\subsection{Parameterization by the Number of Projects} \label{sec:para_m}

One way to get around the above hardness result is to seek $\fpt$
algorithms.  However, we have already established in
\Cref{thm:greedy-av-pb-npc} that such algorithms are (most likely) out
of reach for the parameterizations by the number of voters and the
number of available approval flips.  This leaves us with the number of
projects as the remaining most natural parameter and, indeed, we will
present a fixed-parameter tractable algorithm for this parameter that
even works for the counting variant of the problem.  The reader may
wonder if there is any value in considering $\fpt$ algorithms for this
parameterization, given that in typical participatory budgeting
scenarios one would consider fairly large project sets. However, small
PB instances are also common in practice. For example, at the time of
writing this paper, the Pabulib library of real-world PB
instances~\cite{pabulib} contained $730$ instances, of which $224$
instances included at most $10$ projects.

Interestingly, previous attempts to find such $\fpt$ algorithms for
counting the number of successful briberies in ordinal elections
remained unsuccessful~\cite{boe-bre-fal-nie:c:counting-bribery}, even
though such algorithms for the decision variants of the problem are
common
\cite{kno-kou-mni:j:fpt-m-bribery,fal-rot:b:control-bribery}. The
reason for this is that in the decision case for ordinal elections it
is common to use ILP formulations of the problem and employ one of
several $\fpt$ algorithms for ILPs.
However, for the counting variants,
the number of solutions of these ILPs does not directly translate to
the number of solutions of the original problem.
In our case, we obtain our algorithm because the structure of
$\greedyAVAddRemPB$ is sufficiently simple to apply dynamic
programming.

\newcommand{\sig}{{\mathit{sig}}}
\newcommand{\sigs}{{\mathit{Sig}}}
\begin{theorem} \label{thm:greedy-av-pb-ftp-m}
  There is an $\fpt$ algorithm for
  $\#\greedyAVAddRemPB$ parameterized by the number
  of projects with running time $O^*(3^m)$.
\end{theorem}
\begin{proof}
  Let us consider a PB instance $E = (C,V,A,B,\cost)$, where
  $C = \{c_1, \ldots, c_m\}$ and $V = \{v_1, \ldots, v_n\}$. The
  preferred project is $p$ (so $p = c_\ell$ for some $\ell \in [m]$)
  and we should perform $r$ approval flips.  Our algorithm proceeds by
  dynamic programming, based on partitioning the space of possible
  situations that $\greedyAVRule$ may run into after the bribery.  To
  be more concrete, we make use of the observation that to make a
  funding decision on a project in a $\greedyAVRule$ run, it suffices
  to know which projects the rule has already examined and which of
  them got funded.

  To capture this idea, we introduce the notion of a \emph{signature
    function}, which associates each project from~$C$ with a number
  from $\{0,1,2\}$. For a signature function~$\sig$, we
  write~$C(\sig)$ to denote the set of those projects~$c \in C$ for
  which $\sig(c) > 0$. We interpret $\sig$ as describing the
  following situation (after the bribery): $\greedyAVRule$ was ran
  exactly on the projects from the set $C(\sig)$ and selected all
  those for which the signature function gives value $1$ (the projects
  with value $2$ were considered, but not selected due to insufficient
  funds).  For each subset $D$ of projects, by $\sigs(D)$ we mean the
  set of all possible signature functions that assign nonzero values
  exactly to the projects from $D$.  The idea of our dynamic
  programming is that to capture a run of $\greedyAVRule$ on a subset
  of projects $D$ after the bribery, it is sufficient to know the
  respective signature function from $\sigs(D)$, the last project among
  those in $D$ that $\greedyAVRule$ considered, and the number of its
  approvals. Using this, it is then possible to extend this run with
  another project with fewer approvals (or, at most as many
  approvals, depending on the tie-breaking order).

  To be more concrete, for each signature function $\sig$, project $c' \in C(\sig)$, and
  nonnegative integers~$\ell'$ and~$r'$, we define the following
  function:
  \begin{enumerate}
  \item[] $f(\sig, c', \ell', r')$ is the number of ways to perform
    $r'$ approval flips affecting only projects from $C(\sig)$ so that
    (a) $\greedyAVRule$, restricted to the projects from $C(\sig)$,
    selects exactly those projects $c \in C(\sig)$ for which
    $\sig(c) = 1$, and (b) the last considered project, i.e., the
    project with the minimum number of approvals appearing last in the
    tie-breaking order, is $c'$, and it has exactly $\ell'$ approvals.
  \end{enumerate}
  Our algorithm should output:
  \[
    \sum_{\sig \in \sigs(C)} \, \sum_{c \in C} \, \sum_{\ell \in [n] \cup \{0\}} f(\sig,c,\ell,r) \cdot [\sig(p) = 1],
  \]  
  where by $[sig(p) = 1]$ we mean $1$ if $\sig(p) = 1$ and $0$
  otherwise. Intuitively, this sum takes into account all projects $c$
  that $\greedyAVRule$ may consider last, all their possible scores
  $\ell$, and requires that $p$ is selected (upon performing $r$
  approval flips).  It remains to show how to compute the values of
  $f$ in time $O^*(3^m)$.

  First, we observe that for each project $c'$, each signature
  function $\sig \in \sigs(\{c'\})$, and each two nonnegative integers
  $\ell' \leq n$ and $r' \leq r$, computing $f(\sig, c', \ell', r')$
  is straightforward.
  Next, we show how to compute $f(\sig, c', \ell', r')$ for the case
  where $C(\sig)$ includes at least two projects; i.e., $c'$ and some
  further one(s).  Let $\sig'$ be the signature function identical to
  $\sig$, except that $\sig'(c') = 0$ (capturing the situation before
  $c'$ is considered).  We make a case distinction: If $\sig(c')=1$
  and the cost of all projects that have value $1$ under $\sig'$
  exceeds $B-\cost(c')$ then we set $f(\sig, c', \ell', r')$ to be
  zero (because $\sig$ describes an infeasible situation).  Similarly,
  we set $f(\sig, c', \ell', r')$ to zero if $\sig(c')=2$ and
  the cost of all projects that have value $1$ under $\sig'$ is at
  most $B-\cost(c')$.  Otherwise, we proceed as follows.  For each
  nonnegative integer $r''$, let $h(r'',\ell')$ be the number of ways
  to perform $r''$ approval flips regarding project $c'$ so that it
  ends up with $\ell'$ approvals.  Further, for each
  $d \in C(\sig) \setminus \{c'\}$ let $t_d = 0$ if $d$ precedes $c'$ in
  the tie-breaking order, and let $t_d = 1$ otherwise.  Then, for each
  $d \in C(\sig) \setminus \{c'\}$ we define:
  \[
    g(d) = \sum_{\ell'' = \ell'+t_d}^n \, \sum_{r'' = 0}^{r'} f(\sig', d,
    \ell'', r'-r'') \cdot h(r'', \ell').
  \]
  Intuitively, $g(d)$ gives the number of ways in which we can reach
  the situation described by function $f(\sig,c',\ell',r')$, provided
  that the project that $\greedyAVRule$ considers just before $c'$
  is $d$. 
  Based on this equation, using standard dynamic programming
  techniques, and observing that there are at most $3^m$ signature
  functions, we conclude that it is possible to compute the values of
  function~$f$ in time $O^*(3^m)$.  \qed\end{proof}
Additionally, we provide a slower FPT algorithm, which uses polynomial
space (see \Cref{app:greedy})

\begin{restatable}{theorem}{thmfptmpspace}
  There is an $\fpt$ algorithm for
  $\#\greedyAVAddRemPB$ parameterized by the number
  of projects, with running time $O^*(m!)$, which uses
  polynomial space.
\end{restatable}

\subsection{Cheaper-First Tie-Breaking} \label{sec:cheap-first}

Interestingly, it turns out that another way toward obtaining
tractability results is to consider a more restricted tie-breaking
scheme: If instead of allowing for an arbitrary tie-breaking order, we
break ties according to project costs, with the cheaper-first policy,
the complexity changes: For the parameterization by the number of
voters, $\greedyAVAddRemPB$ now lies in $\xp$, yet still is
$\wone$-hard.  The $\wone$-hardness proof follows by adapting the
arguments from the proof of \Cref{thm:greedy-av-pb-npc} (using
sufficiently many voters, we remove the dependence on the very
specific tie-breaking order used there), and the $\xp$ algorithm
considers all possible bribery flips if their number is bounded by the
number of voters, and follows the simple greedy strategy outlined just
before \Cref{thm:greedy-av-pb-npc} otherwise (in this case it turns
out to be correct).
While the parameterization by the number of voters is not particularly
practical---PB instances in Pabulib typically include at least
hundreds of votes~\cite{pabulib}---we find it interesting to observe
how the internal tie-breaking scheme affects the complexity of the
problem.

\begin{restatable}{theorem}{cheapFirst}
\label{thm:greedy-av-pb-voters-xp}
For cheaper-first tie-breaking and the parameterization by the
  number of voters, $\greedyAVAddRemPB$ is $\wone$-hard and belongs to
  $\xp$.
\end{restatable}
\begin{proof}[Algorithm]
  We first give the $\xp$ algorithm.  Let us consider a PB instance
  $E = (C,V,A,B,\cost)$, a preferred project $p \in C$, and a
  nonnegative integer $r$. Our goal is to decide if it is possible to
  flip at most $r$ approvals so that $\greedyAVRule$ selects $p$ (we
  assume that $\greedyAVRule$ uses the cheaper-first tie-breaking
  scheme, with a fixed, known tie-breaking order for projects with the
  same cost and the same number of approvals). Let $n = |V|$ be the
  number of voters and let $m$ be the number of projects.

  Our algorithm proceeds as follows. If $r < n$ then we brute-force
  over all possible solutions. This has running time $O^*((2m)^r)$,
  because for each of the at most $r$ approval flips we need to choose
  the project to which we would apply this flip, and whether it would
  lead to increasing or decreasing the number of this project's
  approvals. Since $r < n$, the running time is bounded by
  $O^*((2m)^n)$.  If $r \geq n$ then our algorithm first ensures that
  $p$ is approved by all the voters (by adding an approval for $p$ to
  every voter who initially did not approve $p$), and then performs
  the following operation until it  either runs out of allowed approval
  flips, or the projects to consider: It finds the most costly project
  that is also approved by all the voters and whose cost is smaller
  than $\cost(p)$ (or is equal to $\cost(p)$, but this project is
  preferred to $p$ in the tie-breaking order); then it removes a
  single approval from this project.  This part of the algorithm can
  be implemented to run in polynomial time. All in all, it is clear
  that if we assume that $n$ is a constant, then the full algorithm
  runs in polynomial time. This means that if the algorithm is correct
  then $\greedyAVAddRemPB$ is in $\xp$ for the parameterization by the
  number of voters. We show the correctness of this algorithm in
  \Cref{app:greedy}. Therein, we also show how to adapt the proof
  of~\Cref{thm:greedy-av-pb-npc} to obtain the hardness result.  \qed
  \end{proof}

On the other hand, cheaper-first tie-breaking seems not to make the   counting variants of our problems much easier, as they turn out to be
$\sharpp$-hard even for a single voter. 

\begin{restatable}{theorem}{greedyavsharpone}
\label{thm:greedyav-sharpphard-with-one-voter}
$\#\greedyAVAddRemPB$ with cheaper-first tie-breaking is
  $\sharpp$-complete and $\sharpwone$-hard for the parameterization by
  the number of approval flips, even for $1$ voter.
\end{restatable}

Finally, a very natural way to circumvent the hardness from
\Cref{thm:greedy-av-pb-npc} would be to consider a variant of the
problem where project costs are encoded in unary. Unfortunately, in
this case the complexity has, so far, been elusive. On the one hand,
one might hope to design a polynomial-time algorithm based on dynamic
programming (as is the case for the classic \textsc{Knapsack}
problem). However, the problem is that dynamic programming typically
needs to work on some easily-identifiable order over the projects,
such as the order in which $\greedyAVRule$ considers them. However,
this order changes due to approval flips and keeping track of it seems
to require exponential time (see, e.g., the proof of
\Cref{thm:greedy-av-pb-ftp-m}). On the other hand, for a hardness
reduction the problem seems to have relatively little to work
with. Nonetheless, we do have a polynomial-time algorithm for the case
of unit prices (while this is a very restricted setting, in the next
section we will see that for some other rules even this is too much to
hope for).

\begin{restatable}{theorem}{greedyavunitx}
  \label{thm:greedyav-unit}
  There is a polynomial-time algorithm for $\#\greedyAVAddRemPB$ where
  each project has the same unit cost.
\end{restatable}

\section{Complexity Results For Rules Beyond \greedyAVRule}
\label{sec:beyond}

Next let us move to the analysis of the complexity of
$\AddRemProbName$ for our remaining rules, namely $\greedyCostRule$,
$\phragmenRule$, and the two variants of \textsc{MES}. Unfortunately,
we mostly obtain fairly strong hardness results. Indeed, for the
latter three rules the problem is $\np$-complete even for the case
where all projects have the same unit cost. These results follow by
careful analysis (and, for the case of \textsc{MES}, some adaptation)
of proofs already available in the
literature~\cite{fal-gaw-kus:c:greedy-robustness,jan-fal:c:ties-multiwinner};
for $\phragmenRule$ the proof given by Faliszewski et
al.~\cite{fal-gaw-kus:c:greedy-robustness} already deals with a very
related bribery problem, whereas for variants of \textsc{MES}, the
work of Janeczko and Faliszewski~\cite{jan-fal:c:ties-multiwinner}
deals with the complexity of detecting ties (their proof uses the
parallel-universes tie-breaking model and our modification ensures
that we can simulate this approach under order-based tie-breaking by
flipping approvals). We provide details in \Cref{app:mes-phragmen}.
$\greedyCostRule$ is polynomial-time computable in this case, as it is
then equivalent to $\greedyAVRule$ (and so \Cref{thm:greedyav-unit}
applies).

\begin{theorem}\label{thm:mes-phragmen}
  $\AddRemProbName$ is $\np$-complete for each of $\phragmenRule$,
$\MESCostUtilRule$ and $\MESAprUtilRule$, even if all projects have
the same unit cost. $\greedyCostRule$ is in $\p$ in this case.
\end{theorem}
For $\greedyCostRule$, our hardness result requires binary encoding,
but works already for a single voter.

\begin{restatable}{theorem}{greedycostnpc}
\label{thm:greedycost-pb-npc}
$\greedyCostAddRemPB$ is $\np$-complete, as well as $\wone$-hard
  for the parameterization by the number of approval flips, already
  for a single voter. The counting variant is $\sharpp$-hard and
  $\sharpwone$-hard (for the same parameterization).
\end{restatable}

Results in the single-voter case for $\phragmenRule$ and $\MESAprUtilRule$
follow by observing that in this specific case these rules become equivalent to
$\greedyAVRule$ with cheaper-first tie-breaking (with the exception that
$\greedyAVRule$ can select projects that do not receive any approvals, whereas
$\phragmenRule$ and $\MESAprUtilRule$ disregard such projects). 
Moreover, for any fixed tie-breaking order $\MESCostUtilRule$ becomes equivalent
to $\greedyAVRule$ with the same tie-breaking order (under the same caveat
regarding projects without approvals). Consequently,
\Cref{thm:greedy-av-pb-npc} and \Cref{thm:greedy-av-pb-voters-xp} yield the
following corollary about, respectively, hardness and polynomial-time
solvability of the discussed rules in the single-voter case.

\begin{corollary}\label{cor:mes-phr-poly}
  For a single voter, $\MESAprUtilAddRemPB$ and $\phragmenAddRemPB$ are
  poly\-nomial-time solvable, whereas $\MESCostUtilRule$ is $\np$-complete and
  $\wone$-hard for the parameterization by the number of approval flips.
\end{corollary}

For the counting variant,
\Cref{thm:greedyav-sharpphard-with-one-voter} implies
$\sharpp$-completeness and $\sharpwone$-hardness for all three rules
(recall that intractability for cheaper-first tie-breaking clearly implies
intractability for the order-based one).

\begin{corollary}\label{cor:other-rules-counting}
  For $\phragmenRule$, $\MESAprUtilRule$, and
$\MESCostUtilRule$, $\#\AddRemProbName$ is $\sharpp$-complete and $\sharpwone$-hard for
  the parameterization by the number of approval flips even for a
  single voter.
\end{corollary}
Notably, as $\phragmenAddRemPB$ and $\MESAprUtilAddRemPB$ are polynomial-time solvable for a single voter, our results are not sufficient to rule out the existence of $\fpt$ (or XP) algorithms for the number of voters for these problems in the decision and counting variants.
Moreover, to maintain focus, we have not considered the parameter number of projects in this section, leaving this as a direction for future work.

\section{Experiments}

As already discussed in \Cref{sec:intro,sec:flipB}, counting and decision
variants of \textsc{Flip-Bribery} can be used to estimate the robustness of
outcomes of budgeting rules.  Unfortunately, we have shown above that
in most cases it is computationally intractable to solve these
problems. Our only algorithms for the counting variant depend
exponentially on the number of projects of which there are up to $150$
in our considered real-world instances, making these algorithms
unsuitable for our
purposes. 
Thus, in this section, we resort to a sampling-based approach to
assess the robustness of funding decisions in participatory budgeting
to random noise in the votes.  In particular, we use real-world
participatory budgeting instances from Pabulib \cite{pabulib}
\begin{enumerate*}[label=(\roman*)]
\item to analyze examples of non-robust outcomes,   
\item to evaluate how robust outcomes typically are in practice and how this depends on the used budgeting rule, and
\item to form insights regarding what types of projects and instances have particularly (non)-robust funding decisions.
\end{enumerate*}
We will describe the setup of our experiments in
\Cref{sec:setup}, analyze the general robustness of outcomes produced
by different rules in \Cref{sub:aggreg}, show and analyze non-robust
outcomes in \Cref{sec:non-robust}, and conclude by summarizing
additional experiments from \Cref{app:experiments} in
\Cref{sub:additional}.

\subsection{Setup} \label{sec:setup}

\paragraph{Rules.}
In line with our theoretical investigations, in our experiments, we focus on the following four rules:
\begin{enumerate*}[label=(\roman*)]
\item \Greedy,   
\item \GreedyCost, 
\item \Phragmen, and 
\item \MES.\footnote{To maintain focus, we do not consider \MESAprUtilRule{} here. The reason why we examine \MES is that this is the variant that has been used in practice \cite{mes}.}
\end{enumerate*}
As \MES regularly does not output exhaustive outcomes, i.e., the
leftover funds would be sufficient to include additional projects
\cite{mes}, a completion method is needed in practical applications.
In our experiments, we use a variant of \MES similar to the rule used
in real-world applications \cite{mes}:\footnote{We describe, analyze,
and compare other existing completion variants in
\Cref{sub:completion}.} In case the produced outcome is not
exhaustive, the initial endowments of the voters are increased (by
increasing the total budget by $1\%$).  We continue to do so until the
cost of the projects selected by the rule exceeds the original budget.
Unfortunately, the produced outcome is still not guaranteed to be
exhaustive.  In this case, we apply \Greedy to spend the remaining
budget.  We use the implementation provided by the Pabutools library
\cite{pabutools} for \Phragmen and \MES (and our own implementations for the other rules).  In the execution of the
rules, we break ties according to the order in which the projects
appear in the file on Pabulib.

\paragraph{Noise Model and 50\%-Winner Threshold.} 
We use the resampling method introduced by Szufa et al.\
\cite{DBLP:conf/ijcai/SzufaFJLSST22} as our noise model, which is
parameterized by a \emph{resampling probability} $\phi\in [0,1]$ (we
will also sometimes refer to~$\phi$ as the noise level).  To define
the model, for each voter $v_i\in V$, let $p_i:=1/|A(v_i)|$ be the
inverse of the number of projects approved by voter $v_i$. For each
voter $v_i\in V$, we execute the following procedure: For every
project $c_j\in C$, with probability $1-\phi$, we do not change
whether $v_i$ approves $c_j$; however, with probability~$\phi$ we
``resample'' the approval. This means that---independent of whether
$v_i$ initially approved~$c_j$---with probability $p_i$ we let $v_i$
approve $c_j$, and with probability $1-p_i$ we let $v_i$ disapprove
$c_j$.  One advantage of this model is that it does not change the
expected number of projects approved by a voter, which also implies
that the expected number of new approvals is the same as the expected
number of new disapprovals.\footnote{Note that there are naturally
  also other possible noise models. For instance, it would also be
  possible to flip each approval with probability $\phi$. However, as
  initially voters typically only approve a small fraction of
  projects, a growing flipping probability leads to an increased
  expected number of approvals, which makes the instances less and
  less realistic. We also conducted some preliminary experiments using
  this noise model and found that funding decisions are less robust
  under it than under the resampling model.}

In our experiments, to assess the robustness of an outcome produced by
rule $\mathcal{R}$ on some par\-tici\-patory budgeting instance
$E=(C,V,A,B,\cost)$, we execute the following procedure.  For each
$\phi\in \{0,1\%,2\%, \dots, 24\%, 25\%\}$, we modify $E$ as described
above using the resampling model with resampling parameter $\phi$ and
record which projects are funded when rule $\mathcal{R}$ is applied to
the modified instance.  For each considered value of the resampling
probability $\phi$, we repeat this process $100$ times, and, then,
compute for each project its \emph{funding probability}, i.e., the
fraction of sampled instances in which the project is included in the
computed outcome.\footnote{Our choice of focusing on resampling
  probabilities up to $25\%$ is in some sense arbitrary. However, we
  believe that this range captures practically relevant cases. Indeed,
  if we need to introduce more than $25\%$ of noise to affect the
  results, then it is natural to consider the original results to be
  robust.}  Informally speaking, we say that the funding decision on a
project is robust if its funding probability does not quickly change
when the resampling probability is increased; in contrast, we speak of
non-robust funding decisions if already in the case where we perturb
votes using a small resampling probability, the project's funding
probability changes substantially.

Given a PB instance and a budgeting rule, we quantify this instance's
robustness using the notion of \emph{$50\%$-winner threshold}: Its
value is the smallest (considered) value of $\phi$ for which a
majority of sampled instances have different outcome than the initial
one. The idea of the 50\%-winner threshold was introduced by Boehmer
et al.~\cite{boe-bre-fal-nie:c:counting-bribery}, and was later used
by the same team for practical analysis of the robustness in
single-winner
elections~\cite{boe-bre-fal-nie:c:robustness-single-winner}.

\paragraph{Data.}
In our experiments, we analyze instances from the Pabulib platform \cite{pabulib}. 
Specifically, we consider all instances available on Pabulib as of April 2023 in which approval ballots are used (these can be downloaded at \url{http://pabulib.org/?hash=643d7a7937f76}). 
However, in order to ensure a feasible computation time of our experiments, we discarded all instances for which \MES took more than 2 minutes to compute on the initial instance (on a single thread of an Intel(R) Xeon(R) Gold 6338 CPU @ 2.00GHz core). 
The resulting dataset contains $460$ instances and we refer to it as the \emph{full dataset}. 
In some cases, it will be interesting to focus on instances where some funding decisions are close. 
For this, we put together a second \emph{selected dataset}, which contains all instances where the $50\%$-winner threshold for at least one of our rules is smaller or equal to $25\%$.
This selected dataset consists of $257$ out of our $460$ instances.

\paragraph{Visualization of Results.}
Examining the robustness of the
outcome of a single instance in more detail, we will
visualize the funding probabilities of projects as line plots for
varying values of the resampling probability (see, e.g.,
\Cref{fig:non-robust}; we always indicate the place where the PB exercise took place in the caption).  In these plots, each line corresponds to one
project and shows the funding probability of this project ($y$-axis),
depending on the used value of the resampling probability
($x$-axis).\footnote{We sometimes refer to projects by the color of
  their lines. Doing so for the first time, we often specify the
  number of approvers of the project and its cost in brackets. }
Thus, when moving from left to right on these plots, we move further
and further away from the initial instance and increase the level of
noise.  To ensure the readability of these plots, we only include
projects which meet the following criterion: For the initially funded
projects, we include those whose funding probability drops below
$90\%$ for some resampling probability, and for the initially
non-funded ones, we include those whose funding probability reaches at
least $10\%$ for some resampling probability.  In some cases, these
line plots will have legends.  In such legends, the entries are of the
form $x/y$, where $x$ is the number of approvals the project got in
the initial instance and $y$ is the project's cost.
For all instances visualized as line plots, we increased the resolution and sample size of our robustness experiment. 
Specifically, on these instances, we examine all values $\phi\in \{0,0.01\%, 0.02\%, \dots, 24.99\%, 25\%\}$ of the resampling probability and for each of them generate $1000$ samples to obtain the average funding probabilities.

\subsection{An Aggregate View on Outcome's
  Robustness} \label{sub:aggreg} In this section, we take an aggregate
look at the robustness of outcomes of real-world participatory
budgeting instances, to check how fragile are the outcomes to random
noise, and how this depends on the budgeting rule used.  Generally
speaking, in our analysis we distinguish two types of instances: Those
with a $50\%$-winner threshold above $25\%$, i.e., in which even for a
resampling probability of $25\%$ the probability of a change in the
outcome is below $50\%$, and those whose $50\%$-winner threshold is
smaller than $25\%$.  The former type of instances can be viewed as
very robust, whereas the latter ones may exhibit a more complex
behavior in terms of robustness and are, thus, of higher interest to
our analysis.  To get a first overview of the robustness of PB
outcomes, in \Cref{ta:50-winner}, we provide some statistical
quantities regarding the 50\%-winner threshold.  We can see here that
for all the rules a clear majority of instances has a very robust
outcome (their $50\%$-winner threshold is above $25\%$).  However,
there are also numerous instances where funding decisions were closer;
in particular, there is a non-negligible number of instances for which
already a resampling probability of $5\%$ or even smaller is
sufficient to change the outcome in a majority of cases (we will
explore those instances in more detail in \Cref{sec:non-robust}).
This contrast between many very robust outcomes and some less robust
ones underlines that it might be worthwhile to check for the
robustness of outcomes in practice.
\begin{table}[t!]
    \caption{Statistical quantities regarding the $50\%$-winner
      threshold for different budgeting rules for the full dataset, consisting
      of~$460$~instances. The first three columns contain the number of
      instances with a 50\%-winner threshold smaller or equal to
      $25\%/10\%/5\%$. The last two columns give the mean/median $50\%$-winner
      threshold among all the instances for which it is smaller or equal to
      $25\%$. }\label{ta:50-winner} 
    \centering
    \begin{tabular}{c|c|c|c|c|c}
         budgeting rule & \# instances $\leq 25\%$ & \# instances $\leq 10\%$ & \# instances $\leq 5\%$ & mean (for $\leq 25\%$) & median (for $\leq 25\%$)   \\\hline
        \Greedy & 77  & 34 & 17  & $12\%$  & $12\%$ \\
        \GreedyCost & 128  & 67 & 41 & $11\%$ & $10\%$\\
        \Phragmen & 151 & 86 & 47 &  $10\%$ & $9\%$ \\
        \MES & 187 & 121 & 75 & $9\%$ & $7\%$ \\
    \end{tabular}
\end{table}

Notably, we see in \Cref{ta:50-winner} that the robustness of outcomes substantially depends on the
used budgeting rule. In terms of the 50\%-winner threshold, in
aggregate, \Greedy produces the most robust results, then \GreedyCost,
followed by \Phragmen and \MES, which has the lowest robustness.
This observation holds irrespective of whether we focus on instances with $50\%$-winner thresholds up to $25\%$, $10\%$, or $5\%$.
Moreover, it also applies to the distribution of the $50\%$-winner
thresholds among the less robust instances (i.e., those with a
threshold of at most $25\%$): For \Greedy, the distribution is roughly
uniform, and all possible threshold values (below $25\%$) appear
roughly the same number of times, which results both in the mean and
the median being around $12\%$.  In contrast, for \MES the distribution
is significantly skewed towards lower thresholds, which results in a
mean of $9\%$ and a median of $7\%$ (in particular, there are $24$
instances with a $50\%$-winner threshold of $1\%$).  Overall, the
general gap between \Greedy and \MES is quite significant: For \MES,
more than twice as many instances have a $50\%$-winner threshold
smaller or equal to $25\%$ than for \Greedy, and among these
instances, the fraction of very non-robust ones is substantially
higher.

As a rationale for the different behavior of the rules, recall that
the $50\%$-winner threshold is only concerned with whether the outcome
changes, and not by how much.  For the greedy rules, if the ordering
of the projects (according to the number of approvals or the
approval-to-cost ratio) does not change, then the outcome also stays
fixed.  On real-world instances, there is typically only a small
number of swaps in this ordering that are both likely to happen and are
significant enough to change the outcome of the greedy rules.  Thus,
practically speaking, whether an outcome is robust often comes down to whether the
difference in the approval scores (or the approval-to-cost ratios) of
several selected projects is small.  In contrast, for \Phragmen and
\MES, changes in the votes can have a much more subtle influence on the
execution of the rules and, thus, on the produced outcome: Here, an
additional approval can not only cause a different project to be
funded, but also a different distribution of costs to voters,
potentially leading to additional follow-up changes.  Thus, on a high
level, for \Phragmen and \MES there are typically more ``fragile''
moments in the execution of the rules that can have a decisive
influence on the outcome.

\begin{figure}[t!]
	\centering 
	\begin{subfigure}{0.23\textwidth}
		\resizebox{1.05\textwidth}{!}{\begin{tikzpicture}[every plot/.append style={line width=2.5pt}]

\definecolor{color0}{rgb}{1,0.549019607843137,0}
\definecolor{color1}{rgb}{0.133333333333333,0.545098039215686,0.133333333333333}
\definecolor{color2}{rgb}{0.117647058823529,0.564705882352941,1}

\begin{axis}[
legend columns=2, 
legend cell align={left},
legend style={
  fill opacity=0.8,
  draw opacity=1,
  draw=none,
  text opacity=1,
  at={(0.42,1.46)},
  line width=1.5pt,
  anchor=north,
   /tikz/column 2/.style={
  	column sep=10pt,
  }, font=\Large
},
legend image post style={line width =4.5pt},
legend entries={\Greedy,
	full dataset,
	\GreedyCost,
	selected dataset, 
	\Phragmen, {\phantom{a}},
	\MES},
tick align=outside,
tick pos=left,
x grid style={white!69.0196078431373!black},
xlabel={resampling probability},
xmin=0, xmax=0.25,
xtick style={color=black},
ytick={0.4,0.6,0.8,1},
yticklabels={40\%,60\%,80\%,100\%},
xtick={0,0.05,0.1,0.15,0.2,0.25},
xticklabels={0\%,5\%,10\%,15\%,20\%,25\%},
y grid style={white!69.0196078431373!black},
ylabel={probability of same outcome},
ymin=0.3, ymax=1,
ytick style={color=black},every tick label/.append style={font=\Large}, 
label style={font=\Large}
]
\addlegendimage{red!54.5098039215686!black}
\addlegendimage{gray}
\addlegendimage{color0}
\addlegendimage{gray,dotted}
\addlegendimage{color1}
\addlegendimage{white,dashed}
\addlegendimage{color2}
\addplot [semithick, red!54.5098039215686!black]
table {%
0 0.999999999999995
0.01 0.963239130434778
0.02 0.952391304347822
0.03 0.939847826086953
0.04 0.933282608695649
0.05 0.925391304347824
0.06 0.917326086956518
0.07 0.913413043478258
0.08 0.908369565217389
0.09 0.902999999999997
0.1 0.894913043478258
0.11 0.889608695652171
0.12 0.885326086956519
0.13 0.880413043478259
0.14 0.874673913043476
0.15 0.869304347826084
0.16 0.864978260869564
0.17 0.858456521739128
0.18 0.854217391304346
0.19 0.851565217391303
0.2 0.847934782608693
0.21 0.840586956521737
0.22 0.838847826086954
0.23 0.83382608695652
0.24 0.828021739130433
0.25 0.820521739130433
};
\addplot [semithick, red!54.5098039215686!black, dotted]
table {%
0 0.999999999999998
0.01 0.935291828793773
0.02 0.917042801556419
0.03 0.897354085603112
0.04 0.887003891050583
0.05 0.874824902723734
0.06 0.862217898832684
0.07 0.856264591439688
0.08 0.849260700389104
0.09 0.840739299610894
0.1 0.828365758754863
0.11 0.820233463035019
0.12 0.814941634241245
0.13 0.809105058365759
0.14 0.79875486381323
0.15 0.790817120622568
0.16 0.784941634241245
0.17 0.774785992217899
0.18 0.770350194552529
0.19 0.766731517509727
0.2 0.761050583657587
0.21 0.750739299610895
0.22 0.748132295719844
0.23 0.742023346303502
0.24 0.736186770428016
0.25 0.722996108949416
};
\addplot [semithick, color0]
table {%
0 0.999999999999995
0.01 0.968347826086952
0.02 0.947717391304344
0.03 0.929630434782605
0.04 0.914413043478257
0.05 0.903913043478257
0.06 0.893869565217388
0.07 0.882239130434779
0.08 0.87069565217391
0.09 0.857282608695649
0.1 0.844565217391301
0.11 0.832999999999997
0.12 0.822304347826084
0.13 0.812760869565215
0.14 0.804152173913041
0.15 0.795347826086954
0.16 0.786499999999997
0.17 0.776934782608693
0.18 0.767826086956519
0.19 0.76069565217391
0.2 0.752347826086954
0.21 0.74428260869565
0.22 0.738630434782606
0.23 0.731630434782606
0.24 0.727021739130432
0.25 0.718804347826084
};
\addplot [semithick, color0, dotted]
table {%
0 0.999999999999998
0.01 0.944085603112839
0.02 0.90793774319066
0.03 0.875719844357976
0.04 0.849221789883268
0.05 0.830583657587548
0.06 0.812918287937742
0.07 0.792723735408559
0.08 0.772334630350194
0.09 0.748871595330739
0.1 0.726225680933851
0.11 0.706264591439688
0.12 0.686575875486381
0.13 0.670389105058365
0.14 0.654747081712062
0.15 0.639066147859922
0.16 0.624785992217898
0.17 0.607976653696498
0.18 0.591439688715953
0.19 0.579416342412451
0.2 0.565797665369649
0.21 0.551906614785992
0.22 0.543774319066148
0.23 0.531206225680934
0.24 0.523657587548638
0.25 0.510155642023346
};
\addplot [semithick, color1]
table {%
0 0.999999999999995
0.01 0.959739130434779
0.02 0.936869565217387
0.03 0.920456521739127
0.04 0.902456521739127
0.05 0.887978260869562
0.06 0.872391304347823
0.07 0.855934782608693
0.08 0.841456521739128
0.09 0.827347826086953
0.1 0.814173913043476
0.11 0.800499999999997
0.12 0.787543478260867
0.13 0.774934782608693
0.14 0.763369565217389
0.15 0.752630434782606
0.16 0.740217391304345
0.17 0.729543478260868
0.18 0.719521739130432
0.19 0.713999999999998
0.2 0.705760869565215
0.21 0.698260869565215
0.22 0.691217391304345
0.23 0.684456521739128
0.24 0.677630434782607
0.25 0.671413043478259
};
\addplot [semithick, color1, dotted]
table {%
0 0.999999999999998
0.01 0.928404669260699
0.02 0.88828793774319
0.03 0.859105058365758
0.04 0.827898832684824
0.05 0.80260700389105
0.06 0.774474708171206
0.07 0.745175097276264
0.08 0.719922178988326
0.09 0.695175097276264
0.1 0.671945525291828
0.11 0.647626459143969
0.12 0.625914396887159
0.13 0.603346303501945
0.14 0.583268482490272
0.15 0.564513618677043
0.16 0.542762645914397
0.17 0.52408560311284
0.18 0.506731517509728
0.19 0.497587548638132
0.2 0.482762645914397
0.21 0.470272373540856
0.22 0.459688715953307
0.23 0.449299610894942
0.24 0.437859922178988
0.25 0.429066147859922
};
\addplot [semithick, color2]
table {%
0 0.999999999999995
0.01 0.916391304347823
0.02 0.883065217391301
0.03 0.858086956521736
0.04 0.834456521739128
0.05 0.81423913043478
0.06 0.797717391304345
0.07 0.778260869565215
0.08 0.764760869565216
0.09 0.749021739130433
0.1 0.736543478260867
0.11 0.724652173913042
0.12 0.713586956521737
0.13 0.701521739130433
0.14 0.692978260869564
0.15 0.682521739130433
0.16 0.672043478260868
0.17 0.66386956521739
0.18 0.657086956521738
0.19 0.649717391304346
0.2 0.641521739130433
0.21 0.632543478260868
0.22 0.625673913043477
0.23 0.618695652173912
0.24 0.614304347826086
0.25 0.605913043478259
};
\addplot [semithick, color2, dotted]
table {%
0 0.999999999999998
0.01 0.852607003891049
0.02 0.795719844357976
0.03 0.75124513618677
0.04 0.710661478599221
0.05 0.675214007782101
0.06 0.646731517509728
0.07 0.613307392996109
0.08 0.590194552529182
0.09 0.562879377431907
0.1 0.540544747081712
0.11 0.520311284046693
0.12 0.502295719844358
0.13 0.482529182879378
0.14 0.467898832684825
0.15 0.451206225680934
0.16 0.432918287937744
0.17 0.42
0.18 0.408949416342413
0.19 0.395992217898833
0.2 0.383813229571985
0.21 0.371556420233463
0.22 0.360272373540856
0.23 0.349416342412451
0.24 0.343774319066148
0.25 0.332879377431907
};

\end{axis}

\end{tikzpicture}}
		\caption{Probability that the outcome is identical to the initial one.}\label{fig:ov_robust1}
	\end{subfigure}\hfill
\begin{subfigure}{0.23\textwidth}
	\resizebox{1.05\textwidth}{!}{\begin{tikzpicture}[every plot/.append style={line width=2.5pt}]

\definecolor{color0}{rgb}{1,0.549019607843137,0}
\definecolor{color1}{rgb}{0.133333333333333,0.545098039215686,0.133333333333333}
\definecolor{color2}{rgb}{0.117647058823529,0.564705882352941,1}

\begin{axis}[
legend columns=2, 
legend cell align={left},
legend style={
  fill opacity=0.8,
  draw opacity=1,
  draw=none,
  text opacity=1,
  at={(0.42,1.46)},
  line width=1.5pt,
  anchor=north,
   /tikz/column 2/.style={
  	column sep=10pt,
  }, font=\Large
},
legend image post style={line width =4.5pt},
legend entries={\Greedy,
	full dataset,
	\GreedyCost,
	selected dataset,
	\Phragmen, {\phantom{a}},
	\MES},
tick align=outside,
tick pos=left,
x grid style={white!69.0196078431373!black},
xlabel={resampling probability},
xmin=0, xmax=0.25,
xtick style={color=black},
ytick={0.4,0.6,0.8,1},
yticklabels={40\%,60\%,80\%,100\%},
xtick={0,0.05,0.1,0.15,0.2,0.25},
xticklabels={0\%,5\%,10\%,15\%,20\%,25\%},
y grid style={white!69.0196078431373!black},
ylabel={funding probability},
ymin=0.3, ymax=1,
ytick style={color=black},every tick label/.append style={font=\Large}, 
label style={font=\Large}
]
\addlegendimage{red!54.5098039215686!black}
\addlegendimage{gray}
\addlegendimage{color0}
\addlegendimage{gray,dotted}
\addlegendimage{color1}
\addlegendimage{white,dashed}
\addlegendimage{color2}
\addplot [semithick, red!54.5098039215686!black]
table {%
0 0.999999999999995
0.01 0.966152173913039
0.02 0.956434782608692
0.03 0.944934782608692
0.04 0.939347826086953
0.05 0.932021739130432
0.06 0.926086956521736
0.07 0.922347826086953
0.08 0.918043478260867
0.09 0.913282608695649
0.1 0.907130434782606
0.11 0.90269565217391
0.12 0.899086956521737
0.13 0.893891304347824
0.14 0.890369565217389
0.15 0.885326086956519
0.16 0.882326086956519
0.17 0.877130434782606
0.18 0.874043478260867
0.19 0.871108695652172
0.2 0.868456521739129
0.21 0.863869565217389
0.22 0.862608695652171
0.23 0.858478260869563
0.24 0.853652173913041
0.25 0.844260869565216
};
\addplot [semithick, red!54.5098039215686!black, dotted]
table {%
0 0.999999999999998
0.01 0.940505836575874
0.02 0.924280155642022
0.03 0.906459143968871
0.04 0.897859922178987
0.05 0.88669260700389
0.06 0.877898832684824
0.07 0.872217898832684
0.08 0.86657587548638
0.09 0.859027237354085
0.1 0.85
0.11 0.843307392996109
0.12 0.839105058365758
0.13 0.83260700389105
0.14 0.826108949416343
0.15 0.818677042801556
0.16 0.815252918287938
0.17 0.80692607003891
0.18 0.804630350194552
0.19 0.799844357976653
0.2 0.796498054474708
0.21 0.790583657587549
0.22 0.788560311284047
0.23 0.78295719844358
0.24 0.777976653696498
0.25 0.76284046692607
};

\addplot [semithick, color0]
table {%
0 0.999999999999995
0.01 0.984999999999996
0.02 0.971739130434778
0.03 0.958717391304343
0.04 0.945695652173909
0.05 0.935239130434778
0.06 0.925456521739127
0.07 0.915195652173909
0.08 0.90315217391304
0.09 0.889630434782605
0.1 0.878369565217388
0.11 0.866804347826084
0.12 0.855065217391301
0.13 0.843673913043475
0.14 0.833565217391301
0.15 0.82323913043478
0.16 0.811456521739128
0.17 0.799826086956519
0.18 0.789652173913041
0.19 0.780652173913041
0.2 0.771369565217389
0.21 0.761304347826085
0.22 0.752999999999997
0.23 0.745652173913041
0.24 0.740304347826084
0.25 0.73223913043478
};
\addplot [semithick, color0, dotted]
table {%
0 0.999999999999998
0.01 0.973891050583656
0.02 0.950933852140076
0.03 0.927743190661477
0.04 0.9052140077821
0.05 0.886614785992217
0.06 0.869455252918287
0.07 0.851595330739298
0.08 0.830428015564202
0.09 0.80669260700389
0.1 0.786731517509727
0.11 0.766653696498054
0.12 0.745175097276263
0.13 0.725603112840466
0.14 0.707315175097276
0.15 0.688793774319065
0.16 0.669299610894941
0.17 0.648910505836576
0.18 0.630428015564201
0.19 0.614980544747081
0.2 0.599610894941634
0.21 0.582217898832684
0.22 0.569416342412451
0.23 0.555992217898833
0.24 0.547120622568093
0.25 0.534007782101167
};

\addplot [semithick, color1]
table {%
0 0.999999999999995
0.01 0.984456521739126
0.02 0.970413043478257
0.03 0.959760869565213
0.04 0.946369565217387
0.05 0.934847826086952
0.06 0.920869565217388
0.07 0.906456521739127
0.08 0.891760869565214
0.09 0.878826086956518
0.1 0.866760869565214
0.11 0.852130434782606
0.12 0.840891304347823
0.13 0.828499999999997
0.14 0.81473913043478
0.15 0.804043478260867
0.16 0.789543478260867
0.17 0.778739130434781
0.18 0.764543478260867
0.19 0.756499999999998
0.2 0.745652173913041
0.21 0.736760869565215
0.22 0.724847826086954
0.23 0.717586956521737
0.24 0.707913043478259
0.25 0.700543478260867
};
\addplot [semithick, color1, dotted]
table {%
0 0.999999999999998
0.01 0.972645914396886
0.02 0.948210116731516
0.03 0.929260700389104
0.04 0.906147859922178
0.05 0.886108949416341
0.06 0.860856031128404
0.07 0.835291828793773
0.08 0.809494163424124
0.09 0.786770428015563
0.1 0.76556420233463
0.11 0.739377431906614
0.12 0.720661478599221
0.13 0.698677042801556
0.14 0.674396887159532
0.15 0.65556420233463
0.16 0.630389105058365
0.17 0.61136186770428
0.18 0.586731517509727
0.19 0.573190661478599
0.2 0.553346303501945
0.21 0.538677042801556
0.22 0.519338521400778
0.23 0.507898832684825
0.24 0.491128404669261
0.25 0.480544747081712
};

\addplot [semithick, color2]
table {%
0 0.999999999999995
0.01 0.945239130434779
0.02 0.920239130434779
0.03 0.898326086956519
0.04 0.877347826086953
0.05 0.860543478260866
0.06 0.843652173913041
0.07 0.825586956521737
0.08 0.814130434782607
0.09 0.799999999999998
0.1 0.786543478260868
0.11 0.774065217391303
0.12 0.762260869565215
0.13 0.747478260869564
0.14 0.738043478260868
0.15 0.727130434782607
0.16 0.716891304347825
0.17 0.707630434782607
0.18 0.700413043478259
0.19 0.695782608695651
0.2 0.685673913043476
0.21 0.678608695652173
0.22 0.671108695652172
0.23 0.666543478260868
0.24 0.664260869565216
0.25 0.653021739130434
};
\addplot [semithick, color2, dotted]
table {%
0 0.999999999999998
0.01 0.904124513618676
0.02 0.861478599221789
0.03 0.822373540856031
0.04 0.786459143968871
0.05 0.75692607003891
0.06 0.728054474708171
0.07 0.697081712062257
0.08 0.67704280155642
0.09 0.652996108949416
0.1 0.628715953307393
0.11 0.607198443579766
0.12 0.587704280155642
0.13 0.563463035019455
0.14 0.547198443579766
0.15 0.529455252918288
0.16 0.511750972762646
0.17 0.496809338521401
0.18 0.484708171206226
0.19 0.476575875486381
0.2 0.460389105058366
0.21 0.45136186770428
0.22 0.439338521400778
0.23 0.431439688715953
0.24 0.429416342412452
0.25 0.413968871595331
};
\end{axis}

\end{tikzpicture}}
	\caption{Funding probability of the ``least robust'' initially funded project.}\label{fig:ov_robust2}
\end{subfigure}\hfill
	\begin{subfigure}{0.23\textwidth}
		\resizebox{1.05\textwidth}{!}{\begin{tikzpicture}[every plot/.append style={line width=2.5pt}]

\definecolor{color0}{rgb}{1,0.549019607843137,0}
\definecolor{color1}{rgb}{0.133333333333333,0.545098039215686,0.133333333333333}
\definecolor{color2}{rgb}{0.117647058823529,0.564705882352941,1}

\begin{axis}[
legend columns=2, 
legend cell align={left},
legend style={
  fill opacity=0.8,
  draw opacity=1,
  draw=none,
  text opacity=1,
  at={(0.44,1.46)},
  line width=1.5pt,
  anchor=north,
   /tikz/column 2/.style={
  	column sep=10pt,
  }, font=\Large
},legend image post style={line width =4.5pt},
legend entries={\Greedy,
	full dataset,
	\GreedyCost,
	selected dataset,
	\Phragmen, {\phantom{a}},
	\MES},
tick align=outside,
tick pos=left,
x grid style={white!69.0196078431373!black},
xlabel={resampling probability},
xmin=0, xmax=0.25,
xtick style={color=black},
ytick={0.85,0.9,0.95,1},
yticklabels={85\%,90\%,95\%,100\%},
xtick={0,0.05,0.1,0.15,0.2,0.25},
xticklabels={0\%,5\%,10\%,15\%,20\%,25\%},
y grid style={white!69.0196078431373!black},
ylabel={fraction of still funded projects},
ymin=0.85, ymax=1,
ytick style={color=black},every tick label/.append style={font=\Large}, 
label style={font=\Large}
]
\addlegendimage{red!54.5098039215686!black}
\addlegendimage{gray}
\addlegendimage{color0}
\addlegendimage{gray,dotted}
\addlegendimage{color1}
\addlegendimage{white,dashed}
\addlegendimage{color2}
\addplot [semithick, red!54.5098039215686!black]
table {%
0 0.999999999999998
0.01 0.993633750086056
0.02 0.991860649159699
0.03 0.989198410611883
0.04 0.987938262349318
0.05 0.986493304053233
0.06 0.984834056981555
0.07 0.984132521664506
0.08 0.982984309713475
0.09 0.98197018180568
0.1 0.980385151020419
0.11 0.979433643400557
0.12 0.978762087186753
0.13 0.97774552515663
0.14 0.976546192275664
0.15 0.975054902174502
0.16 0.974262972865783
0.17 0.973210041853794
0.18 0.972185996029206
0.19 0.971602360840478
0.2 0.969957462825344
0.21 0.969304459643422
0.22 0.968560781198056
0.23 0.967573706010317
0.24 0.966139250873663
0.25 0.964700692030069
};
\addplot [semithick, red!54.5098039215686!black, dotted]
table {%
0 0.999999999999998
0.01 0.989028346261922
0.02 0.98605404081392
0.03 0.98248748178824
0.04 0.980559927717836
0.05 0.978362209303824
0.06 0.976177116211805
0.07 0.97493891580027
0.08 0.973814025617754
0.09 0.972030907874124
0.1 0.969774920718743
0.11 0.96842357990663
0.12 0.967544559312998
0.13 0.966929812452274
0.14 0.964616494009901
0.15 0.962327859509309
0.16 0.961012727859715
0.17 0.959990442802366
0.18 0.958871495492666
0.19 0.958141340179888
0.2 0.956078180171736
0.21 0.955309274529116
0.22 0.95380618999616
0.23 0.952835520989633
0.24 0.951818887939826
0.25 0.949829268236231
};

\addplot [semithick, color0]
table {%
0 0.999999999999999
0.01 0.997325605335923
0.02 0.995073623005613
0.03 0.992901441926279
0.04 0.990942394980821
0.05 0.989455795286948
0.06 0.987726183868573
0.07 0.98584595201453
0.08 0.98401040981185
0.09 0.9822010199297
0.1 0.980672491651608
0.11 0.979517384503893
0.12 0.977938466352424
0.13 0.976388093610388
0.14 0.974816008876621
0.15 0.973151990782729
0.16 0.971143263892777
0.17 0.969419697985879
0.18 0.96785126295145
0.19 0.966465323569976
0.2 0.965215719938371
0.21 0.963911362496691
0.22 0.962671952006416
0.23 0.962086512021013
0.24 0.961352929905322
0.25 0.959750228198922
};
\addplot [semithick, color0, dotted]
table {%
0 0.999999999999999
0.01 0.995401861690759
0.02 0.992001426391372
0.03 0.988149792293472
0.04 0.984947477397579
0.05 0.982358621914385
0.06 0.979357501606526
0.07 0.9759975299588
0.08 0.973183353489435
0.09 0.969875602232386
0.1 0.967380075848535
0.11 0.965257032987751
0.12 0.962438500086052
0.13 0.960312847660147
0.14 0.957069117016129
0.15 0.954044290471554
0.16 0.951315719949579
0.17 0.948114461373322
0.18 0.945304415243145
0.19 0.94311362498042
0.2 0.941157700978401
0.21 0.939177217627613
0.22 0.93806580106572
0.23 0.936501684339866
0.24 0.935825731814451
0.25 0.933686142406619
};

\addplot [semithick, color1]
table {%
0 0.999999999999998
0.01 0.997847461609692
0.02 0.995763755382695
0.03 0.99395143673974
0.04 0.991880451304748
0.05 0.989829299809288
0.06 0.987648541130269
0.07 0.985829608845749
0.08 0.983343624155841
0.09 0.981399205858696
0.1 0.980102202679975
0.11 0.978777688996768
0.12 0.97704011132988
0.13 0.97504001002369
0.14 0.973650543787869
0.15 0.971734253781197
0.16 0.969123214118745
0.17 0.967647303028519
0.18 0.966151764857471
0.19 0.964989189663881
0.2 0.963279020060392
0.21 0.962308458651421
0.22 0.961249943760633
0.23 0.960480871705918
0.24 0.959290875181554
0.25 0.957986952466654
};
\addplot [semithick, color1, dotted]
table {%
0 0.999999999999996
0.01 0.996577817148348
0.02 0.993147837131155
0.03 0.989960418548433
0.04 0.986637831925536
0.05 0.983159935320823
0.06 0.979390796038157
0.07 0.976142180249915
0.08 0.972001663133658
0.09 0.96870472823938
0.1 0.966651568450999
0.11 0.963948693480668
0.12 0.961006941686563
0.13 0.957676550107048
0.14 0.955246920081657
0.15 0.952379456404552
0.16 0.94792454495474
0.17 0.945163275016098
0.18 0.942595838908847
0.19 0.940704309630341
0.2 0.938011773468203
0.21 0.936426541254943
0.22 0.935164901465319
0.23 0.934093793757167
0.24 0.932312911710073
0.25 0.931260762084559
};

\addplot [semithick, color2]
table {%
0 1
0.01 0.992282896086307
0.02 0.989028737672515
0.03 0.98671503234718
0.04 0.984033808428136
0.05 0.981390869674576
0.06 0.979258811111678
0.07 0.976694838804856
0.08 0.974527683148489
0.09 0.972208377874559
0.1 0.970483409047499
0.11 0.96875491151082
0.12 0.967253784153923
0.13 0.965449144094362
0.14 0.964150562167052
0.15 0.962929783490465
0.16 0.961714275154378
0.17 0.959974040729836
0.18 0.959101286742173
0.19 0.957846717486794
0.2 0.956381986119084
0.21 0.955590988741744
0.22 0.953640180218076
0.23 0.952250850269926
0.24 0.950735375631652
0.25 0.949130059365602
};
\addplot [semithick, color2, dotted]
table {%
0 0.999999999999999
0.01 0.987044760010489
0.02 0.982132782726774
0.03 0.977832392161414
0.04 0.973562060510366
0.05 0.969054480766415
0.06 0.965680034133005
0.07 0.961088605753732
0.08 0.957586254661713
0.09 0.953908732200903
0.1 0.950389296677129
0.11 0.94753806606862
0.12 0.94544780504352
0.13 0.942578215512869
0.14 0.941019137294207
0.15 0.939102880053768
0.16 0.936431583398599
0.17 0.934135754556998
0.18 0.933168864967399
0.19 0.930929413112646
0.2 0.928896513575437
0.21 0.927702633561042
0.22 0.925416534114975
0.23 0.922807374879147
0.24 0.921209403100045
0.25 0.919056197241671
};
\end{axis}

\end{tikzpicture}}
		\caption{Fraction of initially funded projects that are still funded.}\label{fig:ov_robust3}
	\end{subfigure}\hfill
	\begin{subfigure}{0.23\textwidth}
		\resizebox{1.05\textwidth}{!}{	\begin{tikzpicture}[every plot/.append style={line width=2.5pt}]
	
	\definecolor{color0}{rgb}{1,0.549019607843137,0}
	\definecolor{color1}{rgb}{0.133333333333333,0.545098039215686,0.133333333333333}
	\definecolor{color2}{rgb}{0.117647058823529,0.564705882352941,1}
	
	\begin{axis}[
	legend columns=2, 
	legend cell align={left},
	legend style={
		fill opacity=0.8,
		draw opacity=1,
		draw=none,
		text opacity=1,
		at={(0.44,1.46)},
		line width=1.5pt,
		anchor=north,
		/tikz/column 2/.style={
			column sep=10pt,
		}, font=\Large
	},legend image post style={line width =4.5pt},
	legend entries={\Greedy,
		full dataset,
		\GreedyCost,
		selected dataset,
		\Phragmen, {\phantom{a}},
		\MES},
	tick align=outside,
	tick pos=left,
	x grid style={white!69.0196078431373!black},
	xlabel={resampling probability},
	xmin=0, xmax=0.25,
	xtick style={color=black},
ytick={0.85,0.9,0.95,1},
yticklabels={85\%,90\%,95\%,100\%},
	xtick={0,0.05,0.1,0.15,0.2,0.25},
	xticklabels={0\%,5\%,10\%,15\%,20\%,25\%},
	y grid style={white!69.0196078431373!black},
	ylabel={fraction of budget spend the same},
	ymin=0.85, ymax=1,
	ytick style={color=black},every tick label/.append style={font=\Large}, 
label style={font=\Large}
	]
	\addlegendimage{red!54.5098039215686!black}
	\addlegendimage{gray}
	\addlegendimage{color0}
	\addlegendimage{gray,dotted}
	\addlegendimage{color1}
	\addlegendimage{white,dashed}
	\addlegendimage{color2}
	\addplot [semithick, red!54.5098039215686!black]
	table {%
		0 0.999999999999996
		0.01 0.994629484828849
		0.02 0.992718237432664
		0.03 0.989922421250827
		0.04 0.988641986202698
		0.05 0.987364210365502
		0.06 0.985620919559924
		0.07 0.984626716890984
		0.08 0.983851249171741
		0.09 0.982309446734773
		0.1 0.980860453736605
		0.11 0.980254543943538
		0.12 0.979249767201843
		0.13 0.978495467492228
		0.14 0.976855328844152
		0.15 0.975546719447029
		0.16 0.974863134159883
		0.17 0.973367616336063
		0.18 0.972234300731548
		0.19 0.971394382648013
		0.2 0.970170588099014
		0.21 0.969746606161129
		0.22 0.969273232485804
		0.23 0.967582733776247
		0.24 0.966621306055326
		0.25 0.964534589066306
	};
	\addplot [semithick, red!54.5098039215686!black, dotted]
	table {%
		0 0.999999999999999
		0.01 0.990773108051721
		0.02 0.987674151735368
		0.03 0.983832126888199
		0.04 0.981931409769794
		0.05 0.980152136419559
		0.06 0.978075663944087
		0.07 0.976572409858973
		0.08 0.975940493554415
		0.09 0.973609405385157
		0.1 0.971437352193139
		0.11 0.970817432067488
		0.12 0.969433069460782
		0.13 0.969418633694839
		0.14 0.966268331007882
		0.15 0.964359577092387
		0.16 0.963408445570245
		0.17 0.961666533788334
		0.18 0.960923889514293
		0.19 0.959538863404856
		0.2 0.958378978687916
		0.21 0.957881237902814
		0.22 0.95665461707792
		0.23 0.954869602070178
		0.24 0.954571789205161
		0.25 0.951723239025048
	};
	
	\addplot [semithick, color0]
	table {%
		0 0.999999999999996
		0.01 0.995309787193525
		0.02 0.991773131196201
		0.03 0.988462381753933
		0.04 0.985352894239514
		0.05 0.983109469388635
		0.06 0.980647987628086
		0.07 0.977879661876134
		0.08 0.974905386312997
		0.09 0.9717338724209
		0.1 0.968969509921153
		0.11 0.966395274850769
		0.12 0.963493163189988
		0.13 0.960617384303179
		0.14 0.958049973447648
		0.15 0.955078624278398
		0.16 0.95181463729749
		0.17 0.948899284593768
		0.18 0.946081925959205
		0.19 0.94351926253928
		0.2 0.940990567389401
		0.21 0.938342241137214
		0.22 0.935921322077844
		0.23 0.934033284831888
		0.24 0.932620525088569
		0.25 0.929711226235506
	};
	\addplot [semithick, color0, dotted]
	table {%
		0 0.999999999999999
		0.01 0.991931516877496
		0.02 0.986262657007255
		0.03 0.980348939628191
		0.04 0.975195964349083
		0.05 0.971272885225437
		0.06 0.967011127679918
		0.07 0.962133974691986
		0.08 0.957243626608292
		0.09 0.951574681742552
		0.1 0.946877523822231
		0.11 0.942365085026809
		0.12 0.937009574436433
		0.13 0.932598033036857
		0.14 0.927612038745963
		0.15 0.922274374200135
		0.16 0.917386145114518
		0.17 0.912133148697226
		0.18 0.907066545375516
		0.19 0.902870649126167
		0.2 0.898806682146838
		0.21 0.894512303181558
		0.22 0.891427950063601
		0.23 0.887635996581776
		0.24 0.885830288726275
		0.25 0.881423911521071
	};
	
	\addplot [semithick, color1]
	table {%
		0 0.999999999999997
		0.01 0.99544159156577
		0.02 0.991357579346961
		0.03 0.988270042624906
		0.04 0.984662284352145
		0.05 0.981352240392958
		0.06 0.977702982297751
		0.07 0.974322194322863
		0.08 0.970219616181594
		0.09 0.967217911626461
		0.1 0.964399322456844
		0.11 0.961633098987751
		0.12 0.958512364712071
		0.13 0.954777201711197
		0.14 0.95198226333271
		0.15 0.948724400336821
		0.16 0.944496943818931
		0.17 0.941495486097518
		0.18 0.938265443469262
		0.19 0.936097822895268
		0.2 0.933010457666908
		0.21 0.930966791364606
		0.22 0.928219540285781
		0.23 0.926103385076896
		0.24 0.923614271152821
		0.25 0.920880993926958
	};
	\addplot [semithick, color1, dotted]
	table {%
		0 0.999999999999999
		0.01 0.992273332938998
		0.02 0.985292850252621
		0.03 0.979831319871617
		0.04 0.973811507162159
		0.05 0.968120762500692
		0.06 0.961709969847368
		0.07 0.955691005589034
		0.08 0.948681799900815
		0.09 0.943542916296192
		0.1 0.938792414108311
		0.11 0.933564801007861
		0.12 0.928395368889414
		0.13 0.92195895280782
		0.14 0.917084457852565
		0.15 0.911696434233188
		0.16 0.904547414645665
		0.17 0.899225597067278
		0.18 0.893684705616404
		0.19 0.890150427968415
		0.2 0.884878234698663
		0.21 0.881581636074786
		0.22 0.877669484511385
		0.23 0.874566739335204
		0.24 0.870709569813849
		0.25 0.867421728219562
	};
	
	\addplot [semithick, color2]
	table {%
		0 0.999999999999997
		0.01 0.98900726791416
		0.02 0.983820995510298
		0.03 0.979951403078769
		0.04 0.975859195919574
		0.05 0.971617275247184
		0.06 0.968210189943286
		0.07 0.963914583201517
		0.08 0.960625888918292
		0.09 0.956735082822816
		0.1 0.953473547174523
		0.11 0.95047454988819
		0.12 0.947528372144721
		0.13 0.944654161200288
		0.14 0.942548508815018
		0.15 0.939822755303433
		0.16 0.937879555356376
		0.17 0.935049622059006
		0.18 0.933162254337991
		0.19 0.931250045234623
		0.2 0.928931141384628
		0.21 0.92659119580522
		0.22 0.92366978148385
		0.23 0.921126400338291
		0.24 0.918831645552481
		0.25 0.916190198385148
	};
	\addplot [semithick, color2, dotted]
	table {%
		0 0.999999999999999
		0.01 0.981351707965072
		0.02 0.973236519744494
		0.03 0.966244823994657
		0.04 0.959558010117018
		0.05 0.952265289821012
		0.06 0.946768125351546
		0.07 0.93944161461137
		0.08 0.933917511492157
		0.09 0.927546441952615
		0.1 0.921340176659781
		0.11 0.916111434605112
		0.12 0.911735663588374
		0.13 0.90712272003606
		0.14 0.903975161089392
		0.15 0.899690634175096
		0.16 0.895667716839302
		0.17 0.891772710117932
		0.18 0.889047854798916
		0.19 0.885559518192342
		0.2 0.882073963079176
		0.21 0.878883359384819
		0.22 0.874737970741959
		0.23 0.870079138117058
		0.24 0.867405774186485
		0.25 0.8634411211058
	};
	\end{axis}
	
	\end{tikzpicture}}
		\caption{Fraction of the budget spent as in the initial outcome.}\label{fig:ov_robust4}
	\end{subfigure}
	\caption{Some average statistics reflecting the robustness of outcomes for different rules.}
	\label{fig:ov_robust}
\end{figure}
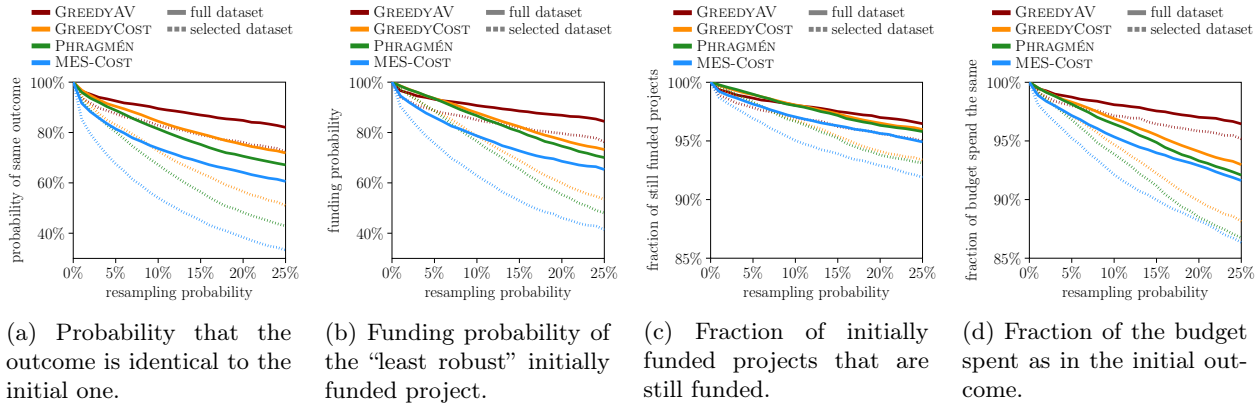

So far, we have only examined the $50\%$-winner threshold, which
indicates under what level of noise it is likely that the outcome
changes.  However, it remains unclear how drastic the change is (or,
what changes are possible under smaller noise levels).  Note that,
intuitively, as our rules work in a sequential fashion and thus one
change in their execution could potentially lead to selecting
completely different outcomes, it is unclear what to expect.  To shed
some light on this issue, in \Cref{fig:ov_robust} we take a more
nuanced look at the robustness of outcomes.\footnote{Note that in
  these plots we only consider averaged values. We did not include
  other statistical quantities for the sake of readability and their
  lack of relevance to the goal of our study.}  We examine here both
the \emph{full} dataset, as well as the smaller dataset of
\emph{selected} instances.  First, in \Cref{fig:ov_robust1} we show
the probability that the outcome remains unchanged, depending on the
resampling probability.  Our observation regarding the relation
between the different rules still holds here for both datasets and for
all considered values of the resampling probability.

Second, in \Cref{fig:ov_robust2} we show the funding probability of
the ``least robust'' initially funded project, i.e., the project that
has the lowest funding probability at resampling probability $25\%$.
The results here look relatively similar to the results from
\Cref{fig:ov_robust1}, which indicates that in case the outcome
changes, this might often be because the same funding
decision gets changed.  These results suggest that there often is one
funded project whose funding decision is closest to being overturned.
In fact, we will now show that for a majority of projects the decision
of whether or not they are selected is quite robust.

In \Cref{fig:ov_robust3}, we depict how the fraction of initially
funded projects which remain funded depends on the resampling
probability.  We see here that this fraction is quite high and in
particular above $95\%$ for all considered rules and resampling
probabilities on the full dataset.  It is thus in particular much
higher than the probability of the same outcome (see
\Cref{fig:ov_robust1}).  Even on the selected dataset, more than
$90\%$ of funding decisions do not get reverted.  This indicates that
even in the case that the outcome changes, most of the originally
funded projects remain funded.  Together with \Cref{fig:ov_robust2},
it also highlights that there is a clear gap in the robustness of
funding decisions of the ``least robust'' projects and the other ones.
This observation is reassuring for practitioners, as it means that
even in case that random changes can affect the outcome, these often
only regard a limited number of projects; however, it also motivates
one to identify which funding decisions are the non-robust ones in
practice, to understand them better.

Lastly, in \Cref{fig:ov_robust4} we show the average fraction 
of the budget which under a given noise level is spent on the same projects
as in the initial instance.
Comparing \Cref{fig:ov_robust4} to \Cref{fig:ov_robust3}, we see that (except for \Greedy) the fraction of still funded projects is generally higher than the fraction of budget spent in the same way. 
This indicates that funding decisions of expensive projects are more fragile to noise (we will explore this connection between the cost of a project and its robustness in more depth in \Cref{sub:properties}).

\subsection{(Types of) Non-Robust Instances} \label{sec:non-robust}
Next, we want to get a sense for how funding probabilities behave in non-robust instances. 
We will find that for all rules there exist examples where the funding probabilities of projects quickly change, which is often due to complex interactions between projects in the selection process that our approach identifies. 
To this end, in \Cref{fig:non-robust} for each rule we depict an instance where the funding decisions have been particularly non-robust, and where drastic changes in funding probabilities appeared already for a small noise level. 
We next analyze our rules one-by-one.

\begin{figure}[t!]
	\centering 
	\begin{subfigure}[b]{0.23\textwidth}
		\includegraphics[width=1\textwidth]{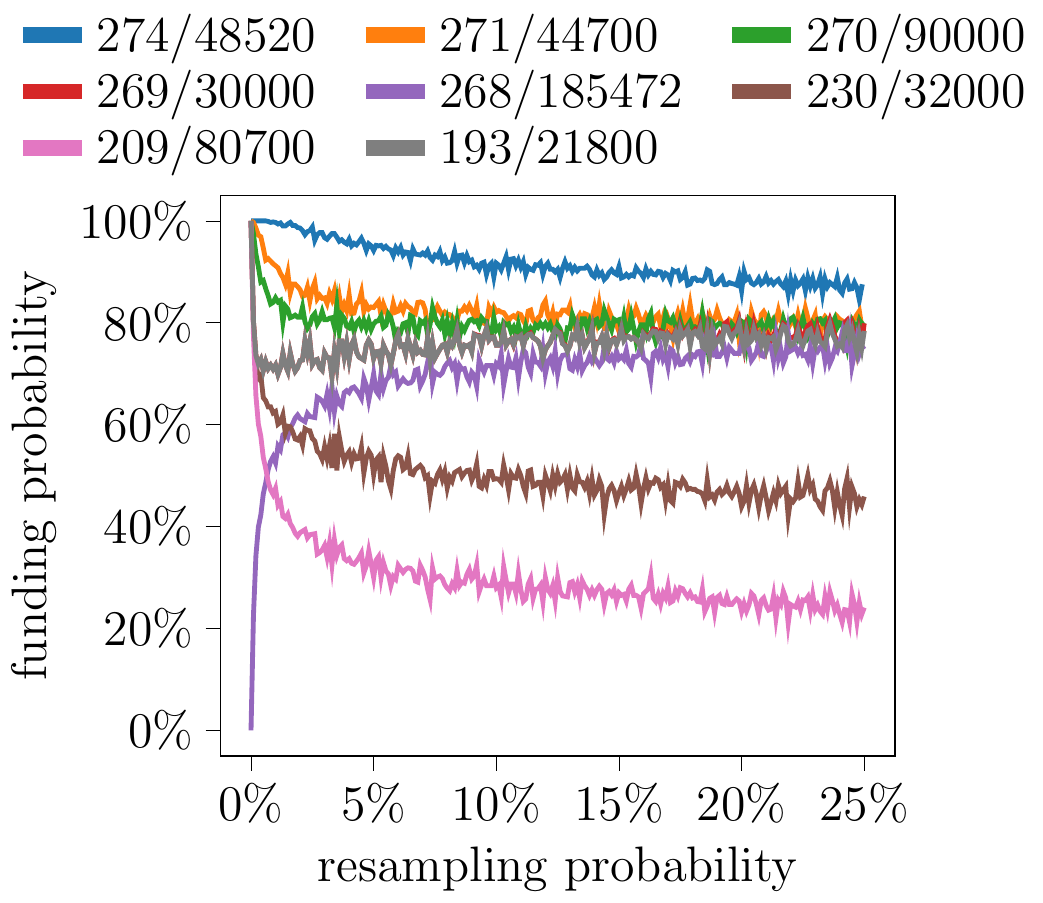}
		\caption{Wrzeciono Mlociny 2019 for \Greedy}\label{fig:non-robust1}
	\end{subfigure}\hfill
	\begin{subfigure}[b]{0.23\textwidth}
		\includegraphics[width=1.1\textwidth]{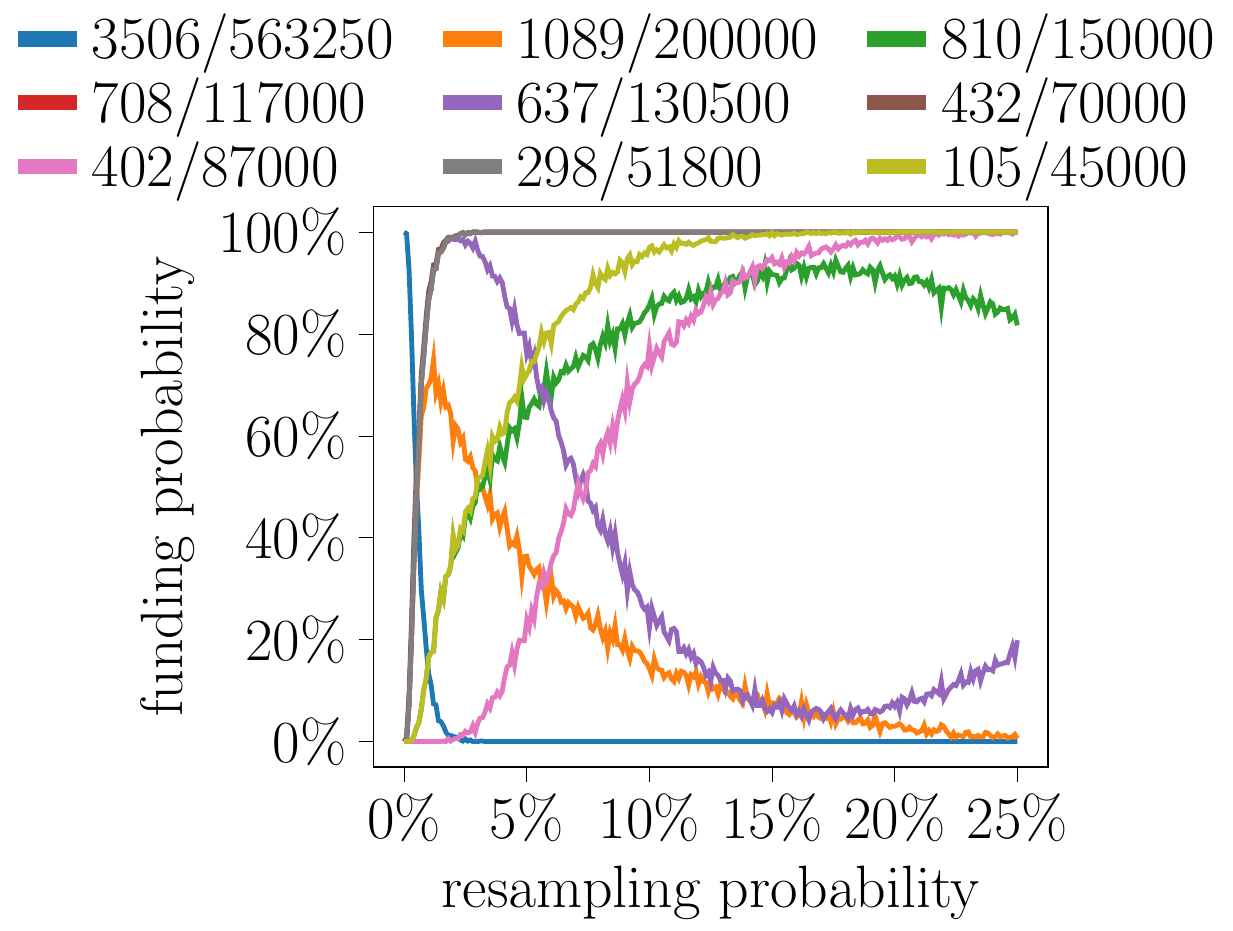}
		\caption{Wola 2021  for \GreedyCost}\label{fig:non-robust2}
	\end{subfigure} \hfill
	\begin{subfigure}[b]{0.23\textwidth}
		\includegraphics[width=1.1\textwidth]{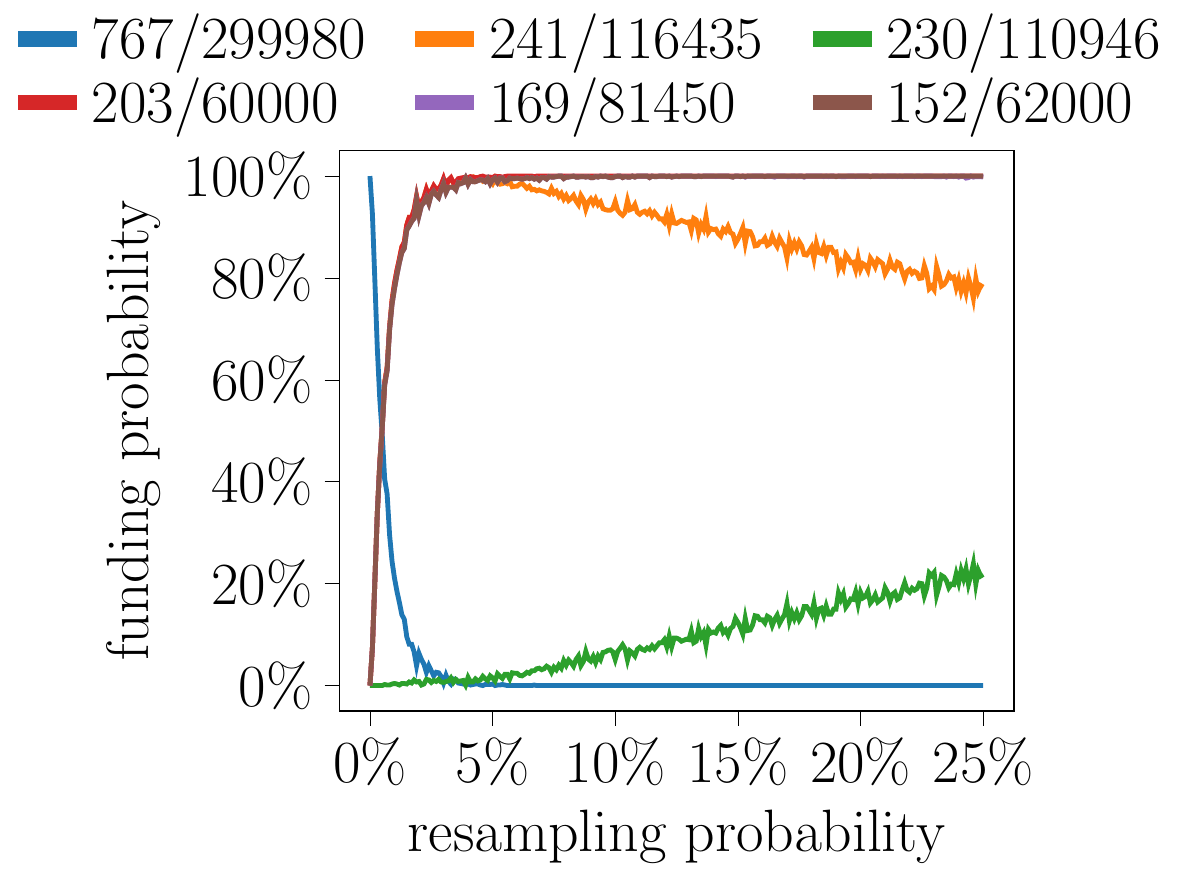}
		\caption{Wilanow Obszar I 2017  for \Phragmen}\label{fig:non-robust3}
	\end{subfigure}\hfill
	\begin{subfigure}[b]{0.23\textwidth}
		\includegraphics[width=1.1\textwidth]{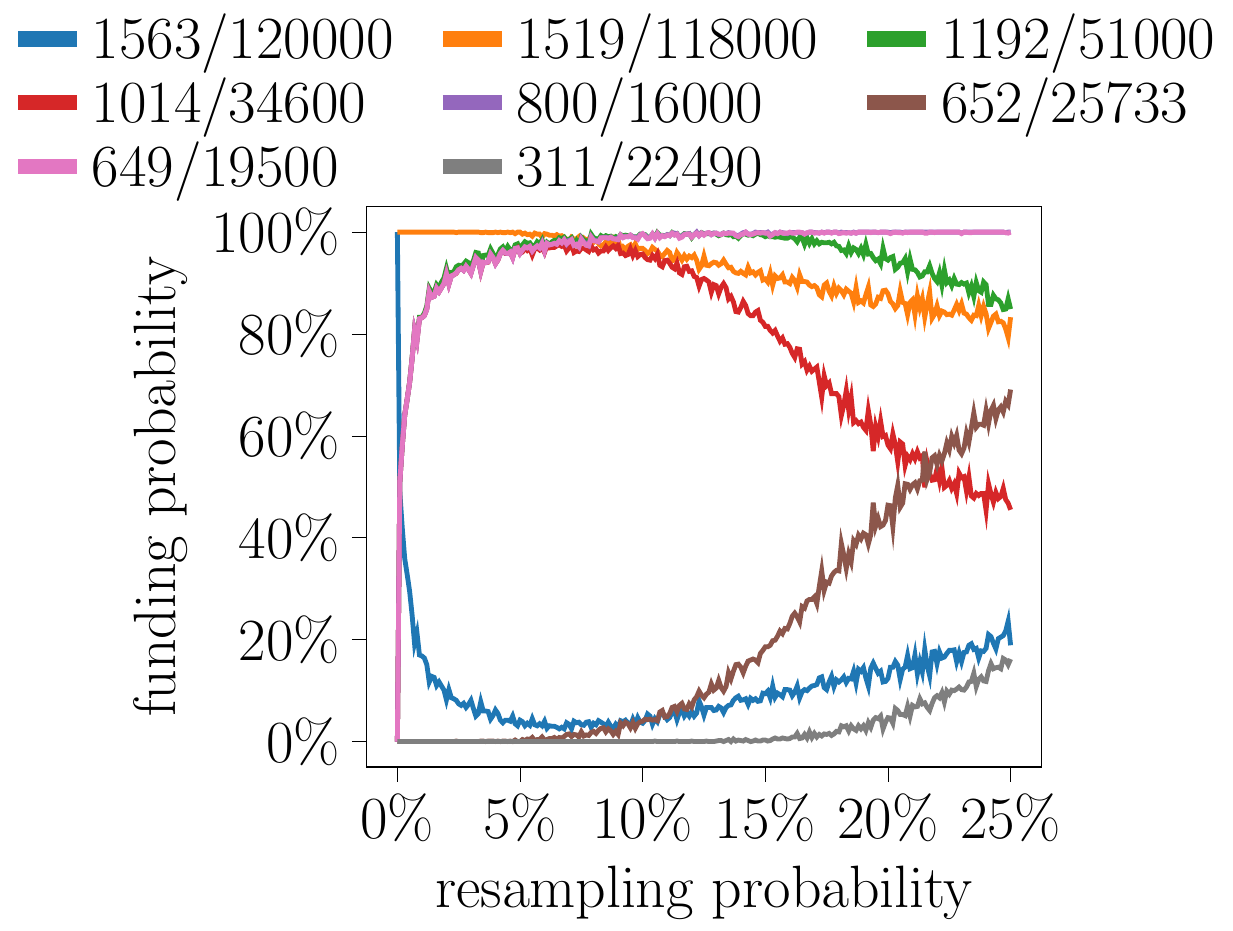}
		\caption{Bielany 2019 for \MES}\label{fig:non-robust4}
	\end{subfigure}
	\caption{Examples for non-robust funding decisions for each
          budgeting rule. For each plot, the caption includes the name
          of the presented instance from Pabulib and the rule used. In
          \Cref{fig:non-robust1}, the red and gray line overlap. In
          \Cref{fig:non-robust2}, the brown, gray and red line
          overlap. In \Cref{fig:non-robust3}, the brown, red and
          purple line overlap. In \Cref{fig:non-robust4}, the pink and
          purple line overlap. }\label{fig:non-robust}
\end{figure}

\subsubsection{\Greedy.} \Cref{fig:non-robust1} shows an example of a
non-robust instance for \Greedy.  In this instance, the outcome
quickly starts to change with non-negligible probability: The pink
project ($209$/$80700$) is the least robust initially funded one: Its
funding probability is, respectively, $77.4\%$ and $47.5\%$ for
resampling probabilities of $0.1\%$ and $1\%$.  In contrast, the
purple project ($268$/$185472$), which initially is not funded, has a
funding probability of, respectively, $22.4\%$ and $52.5\%$ for
resampling probabilities of $0.1\%$ and $1\%$.  Note that in this
instance a resampling probability of $0.1\%$ (respectively, $1\%$) means
that in expectation $10.2$ (respectively, $102$) approvals are flipped (this amounts
to a $0.024\%$ or $0.24\%$ fraction of all approvals in the
instance). Recalling that the performed changes are random and, thus,
do not affect the pink or purple projects with high probability, it is
quite remarkable that such non-robust instances appear in practice.

To understand this instance better, let us take a closer look at the
approval scores and costs of the projects.  What stands out here is
that the purple project has significantly higher cost than the others.
As a result of this, in the initial execution of the rule, the funds
are not sufficient to afford the purple project ($268$/$185472$), yet
the brown ($230$/$32000$), pink ($209$/$80700$), and gray ($193$/$21800$)
projects which all have fewer approvals than the purple one are
funded. Observe that for the outcome to change, and the purple
project ($268$/$185472$) to get selected, it is sufficient that it has
a higher number of approvals than the red ($269$/$30000$), green
($270$/$90000$), orange ($271$/$44700$), or blue project
($274$/$48520$).  In particular, overtaking the red project is
possible by deleting or adding two approvals\footnote{Note that the
  purple project appears behind the other ones in the tie-breaking
  order.}.  This is also reflected in the plot.  For a small
resampling probability, the purple project overtakes in some case one
or several among the red, green, orange or blue projects, which
results in the purple project replacing them in the committee.  As the
red project has the lowest approval score, it is most likely to be
swapped out and replaced by the purple project, and accordingly it has
the lowest funding probability among these projects.  Another tradeoff
that is visible in the plot is that in case the purple project gets
selected, the pink one is almost never funded, as funding the purple
project instead of one of the red, orange, green, or blue ones leads
to a cost increase, which results in that there are not enough funds
left once it is the pink project's turn.  This highlights that in some
situations whether a project gets funded depends much more on how the
approvals of other projects change (e.g., the red and purple project
in this case) than on the approvals of the project itself (e.g., the
pink project in this case).  In essence, this is the same effect that
we used in the proof of \Cref{thm:greedy-av-pb-npc}.

Lastly, let us remark that from the plot we can also see that the gray
project, which is the cheapest one, seems to be only funded in cases
where the red project is funded as well (their lines overlap).  This is because the red project
is the second-cheapest project and in case some other more expensive
project is funded (instead of the red one) there is not enough budget
left for the gray one.  Notably, the shown example is the instance
where funding probabilities change quickest under \Greedy. In
particular, there are no instances where the funding probability of
one project quickly changes by around $100\%$ as in the instances
shown for the other rules, which highlights the generally higher
robustness of outcomes produced by
\Greedy.

\begin{figure}[t!]
	\centering 
\begin{subfigure}[b]{0.32\textwidth}
\includegraphics[width=0.8\textwidth]{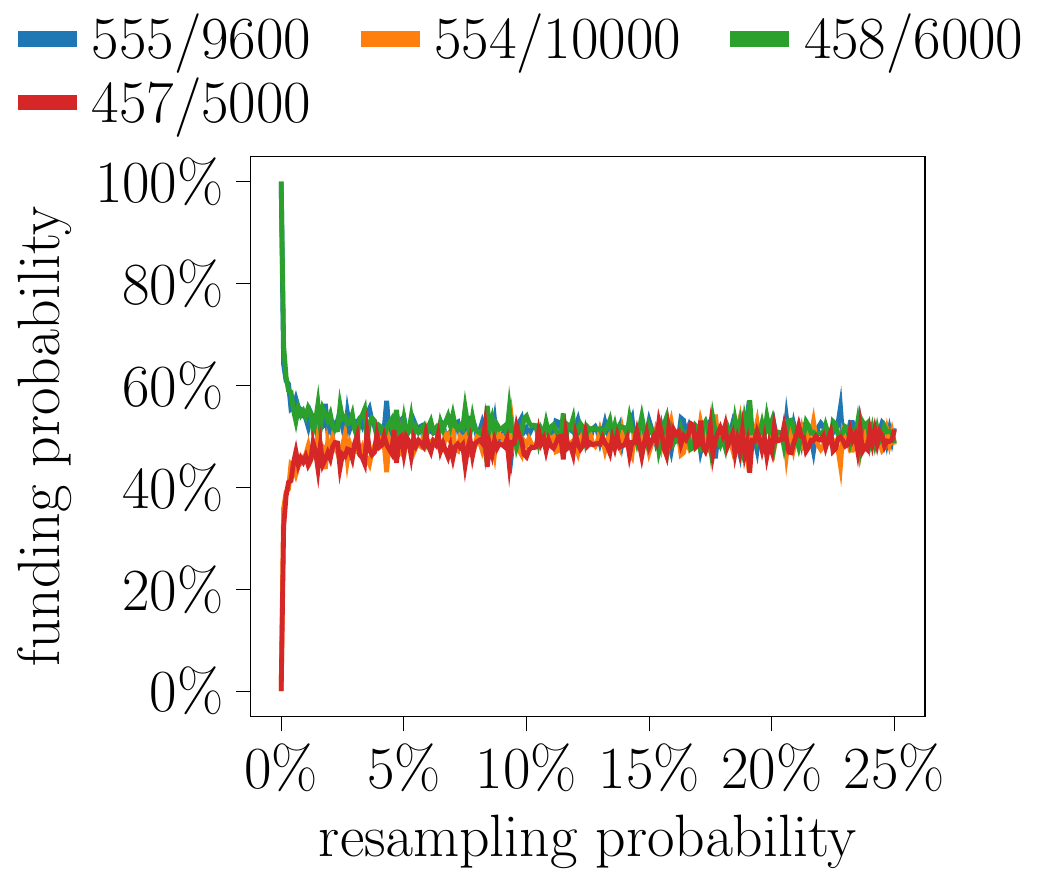}
\caption{Chojny Dabrowa 2020 (Lodz)}\label{fig:type1}
\end{subfigure}\qquad \qquad \qquad
\begin{subfigure}[b]{0.32\textwidth}
		\includegraphics[width=0.9\textwidth]{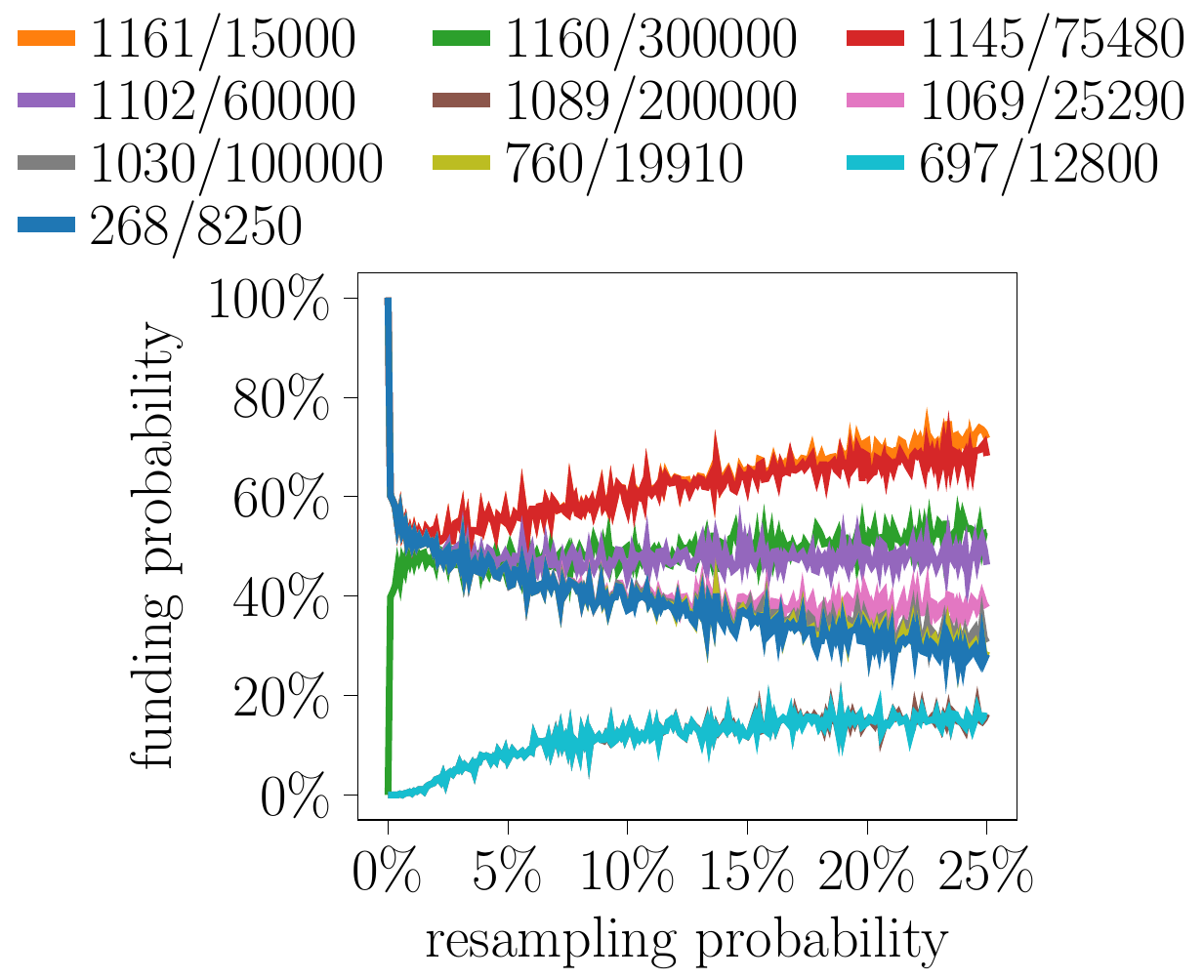}
		\caption{Wola 2021 (Warszawa)}\label{fig:type2}
	\end{subfigure}
	\caption{Different types of non-robust projects. Results for \Greedy on different instances. In \Cref{fig:type1}, the green and blue line and the red and orange line overlap. }\label{fig:type}
\end{figure} 

\paragraph{Different Types of Non-Robust Projects.} 
In  \Cref{fig:non-robust1}, we have seen an instance where the funding probability of an initially funded  project quickly dropped and continued to drop further when further increasing the resampling probability.
Interestingly, there are also initially funded projects with a quickly dropping funding probability that show a different behavior. 
Three different types of such projects seem to regularly occur:
\begin{enumerate*}[label=(\roman*)]
\item projects whose funding probability further decrease when adding more and more noise to the instance,\label{item:a}
\item projects whose funding probability stays around $50\%$ also when adding more and more noise to the instance, and\label{item:b}
\item projects whose funding probability increases again when adding more noise.\label{item:c}
\end{enumerate*}
Notably, the type of a project also influences the interpretation of
the non-robustness of its funding decision.  For type \ref{item:a} projects
one might raise justified concerns whether the initial funding
decision made on this project should not be overturned.  Type
\ref{item:b} projects question whether certain funding decisions should not have been viewed as ties instead.
For type \ref{item:c} projects, the situation is less clear and in some cases one might view the initial funding decision as justified despite the initial drop in the
funding probability.

We have already seen examples for projects of type \ref{item:a} in \Cref{fig:non-robust1}.  In
\Cref{fig:type1} we present a good example for projects of type
\ref{item:b}: The winning probabilities of the green ($458$/$6000$)
and blue projects ($555$/$9600$) quickly drop from $100\%$ to $50\%$,
and when increasing the resampling probability further, their funding
probability stays around $50\%$.  For the red ($457$/$5000$) and
orange projects ($554$/$10000$), we see the reverse of this behavior,
as their funding probability quickly goes up from $0\%$ to around
$50\%$, but then stays close to this value.  This indicates that
the initial decision to select the green and blue projects over the
the red and orange ones is somewhat arbitrary and one might rather see
these four projects as tied (and, perhaps, the city should choose
which ones to fund based on external arguments, such us the
compatibility between them, or simply fund them all).  Turning to
\Cref{fig:type2}, we see that non-robust projects of different types
can even occur in the same instance and may behave identically for
very low levels of noise.  In particular, the funding probability of
the red ($1145$/$75480$), blue ($268$/$8250$) and purple
($1102$/$60000$) initially funded projects all drop to around $50\%$
at a resampling probability of $1\%$.  However, when increasing the
noise level, for the blue project (type \ref{item:a}), the funding
probability drops further, for the purple project (type \ref{item:b})
it stays roughly constant, and for the red project (type \ref{item:c})
it increases again.

\subsubsection{\GreedyCost} \Cref{fig:non-robust2} shows an example of
a non-robust instance for \GreedyCost.  The situation here is in some
sense reversed compared to \Cref{fig:non-robust1}.  The initially
funded blue project ($3506$/$563250$) is quite expensive and in case
it is no longer funded, the saved money suffices to fund multiple
other projects. As a consequence, there are eight initially not-funded
projects which have a funding probability above $10\%$ at some point,
whereas there is only one initially funded project (the blue one) who
has a funding probability below $90\%$ at some point.  Looking at the
instance in more detail, the blue project has a approval-to-cost
ratio of $0.00622$, followed by the brown project ($432$/$70000$) with
a ratio of $0.00616$ and the red project ($708$/$117000$), with a
ratio of $0.00605$.  Changing the ordering of the blue project and at
least one of the other ones is sufficient to change the outcome, and
in fact already at resampling probability $0.5\%$ (which corresponds
to flipping $808$ of all the approvals in expectation, i.e., flipping
a fraction of $0.08\%$ of all the approvals) the outcome changes
in a majority of cases: At this point the initially funded blue
project and the initially not-funded brown ($432$/$70000$), gray
($298$/$51800$), purple ($637$/$130500$), red ($708$/$117000$), and
orange ($1089$/$200000$) projects all have a funding probability around
$50\%$.\footnote{Note that in contrast to the instance for \Greedy
  discussed above, in this instance at a resampling probability of
  $0.1\%$, the outcome changes only in $0.3\%$ of cases. Parts of the
  reason for this is that here at least five approvals are needed
  to change the outcome.}  Notably, increasing the resampling
probability further to $2.8\%$, the blue project has a funding
probability of zero, whereas the brown, gray, purple and red project
have a funding probability around $99\%$.  Note that the orange
project now behaves differently than these four projects, as it has only a funding probability of $58\%$ at this point.  In fact,
when increasing the resampling probability further, more initially not
funded projects start to have a non-negligible funding probability; at the same time, the funding probability of other initially not funded projects
already starts to decrease again.  For instance, there seems to be a
tradeoff between the purple and orange projects, on the one hand, and
the pink ($402$/$87000$) and dark ($810$/$150000$) and light green
($105$/$45000$) projects, on the other hand, for values of the
resampling probability larger than $1.5\%$.  The resulting
non-monotonic behavior of the funding probability of the orange and
purple projects is quite strong in this instance.  Generally, such a
pattern implies that only for a certain noise level the projects have a
high funding probability, and also appears regularly on other
instances.  Analyzing in which situations project's funding
probabilities behave monotonicly and in which cases they do not, as
well as possible implications of this phenomenon, is an interesting
question for future research.

\subsubsection{\Phragmen and \MES}
We conclude by briefly discussing some non-robust instances for our two proportional rules. 
\Cref{fig:non-robust3} shows an example of a non-robust instance for
\Phragmen.  In this instance, the funding probability of the initially
funded blue project ($767$/$299980$) starts to drop quickly.  For
small values of the resampling probability, in case the blue project
is no longer funded, always the orange ($241$/$116435$), brown
($152$/$62000$), red ($203$/$60000$), and purple ($169$/$81450$)
projects get funded instead, implying that their funding probabilities
are complements of one another.  Specifically, the funding probability
of the blue project drops to $92.6\%/49.3\%/10\%/0\%$ for a resampling
probability of $0.1\%/0.5\%/1.5\%/4.6\%$.

\Cref{fig:non-robust4} shows an example of a non-robust instance for
\MES. Here, as for \Phragmen, for small values of the resampling
probability, in case the blue project ($1563$/$120000$) is no longer
funded, always the pink ($649$/$19500$), green ($1192$/$51000$), red
($1014$/$34600$), and purple ($800$/$16000$) projects get funded.
However, compared to \Phragmen, the situation changes even quicker
here: The funding probability of the blue project drops to $49\%$,
respectively, $16.8\%$ for a resampling probability of $0.1\%$,
respectively, $1\%$.  Notably, in this case, a resampling probability
of $0.1\%$ means that we flip $55$ approvals in expectation, which
corresponds to only $0.026\%$ of all approvals.

\begin{figure}[t!]
	\centering 
	\begin{subfigure}{0.32\textwidth}
		\includegraphics[width=\textwidth]{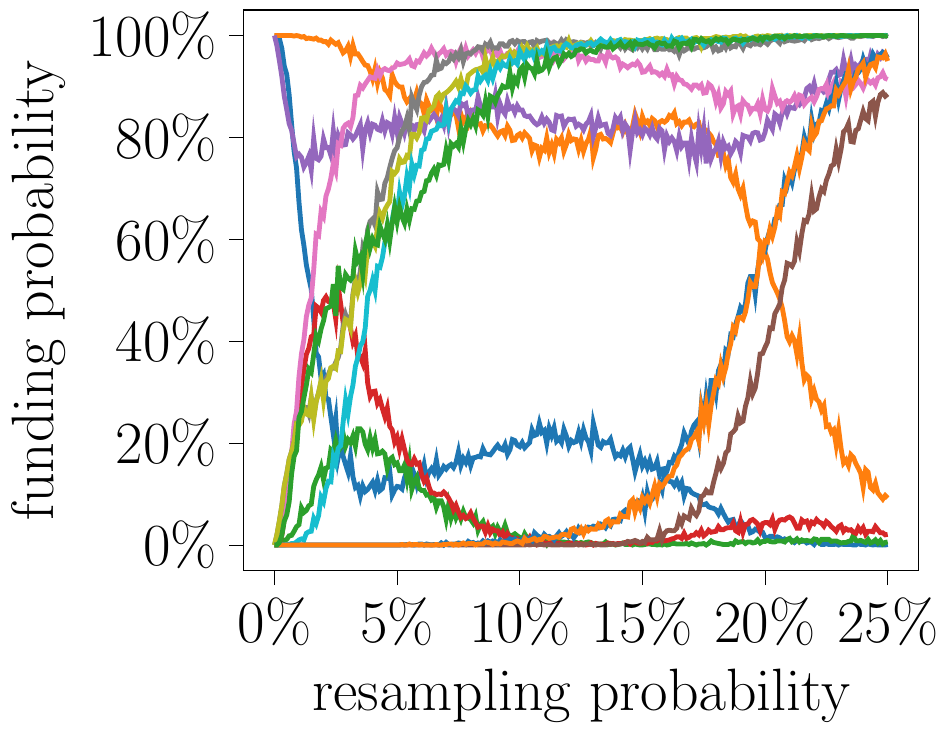}
		\caption{\Phragmen}
	\end{subfigure}\qquad \qquad \qquad \qquad 
 \begin{subfigure}{0.32\textwidth}
		\includegraphics[width=\textwidth]{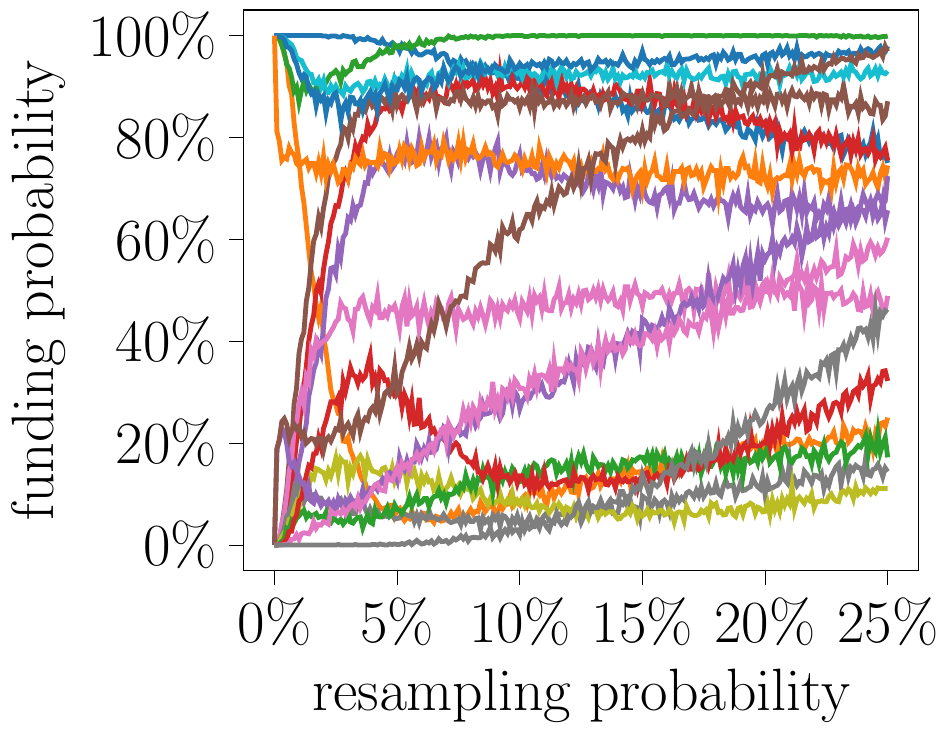}
		\caption{\MES}
	\end{subfigure}
	\caption{Praga Polnoc 2023 (Warszawa)}\label{fig:praga}
\end{figure} 

As for the other two rules, using our approach  we find tradeoffs between different sets of projects in both examples. 
Moreover, as for \GreedyCost, in both examples when increasing the resampling probability the funding probabilities of some initially not funded project already start to decrease again, indicating that new close decisions between different sets of projects emerge.  
However, there are also certain patterns that appear for the two proportional rules, i.e., \Phragmen and \MES, but not for the two greedy rules. 
To give an example of this, in \Cref{fig:praga}, we show the behavior of \Phragmen and \MES on the instance from Praga Polnoc in 2023. 
Overall, these examples exhibit a very chaotic behavior, highlighting the general non-robustness of \Phragmen and \MES on some instances.
In particular, there are many projects with a quickly changing funding probability. 
In contrast to most examples given above, the funding probabilities of these non-robust projects do not change simultaneously (which is caused by two projects often being funded in the same instances) indicating that a higher number of different outcomes are likely to appear here. 
Moreover, funding probabilities often behave in a non-monotonic fashion,
changing from increasing to decreasing (or the other way round) multiple times. 
Lastly, note that while both rules output non-robust outcomes on this instance, the precise picture and tradeoffs for \Phragmen and \MES are quite different here.
In fact, there are many instances where the two rules have a very different behavior in terms of the robustness of outcomes, which we will explore in more depth in \Cref{sec:correlation}.

\subsection{Additional Experiments}\label{sub:additional}
In \Cref{app:experiments}, we present further experiments, whose main findings we briefly survey here. 
In \Cref{sec:correlation}, we study the correlation between the robustness of outcomes produced by the different rules on the same instance.
We find no strong correlation between our rules with many instances where the outcome of one rule is very non-robust, whereas for other rules the robustness is high. 
In \Cref{sub:properties}, we examine how the robustness of outcomes relates to properties of the instance, yet find no connection to  properties such as the number of voters or projects. 
Moreover, we analyze how the properties of initially funded projects typically influence the robustness of their funding decision, observing rule-dependent differences.
Lastly, in \Cref{sub:completion}, we analyze different completion methods for \MES, finding that they behave quite differently in terms of robustness.

\section{Conclusions}
We have adopted the recently introduced approach for evaluating the
robustness of election results from the world of single-winner
voting~\cite{bau-hog:c:robust-winner,boe-bre-fal-nie:c:counting-bribery,boe-bre-fal-nie:c:robustness-single-winner}
to that of participatory budgeting. In our theoretical analysis, we
have explored the computational complexity of (\#)\AddRemProbName{}
problems.  As we proved most of these problems to be computationally
intractable even in very restricted cases, in our experiments we
turned to a sampling-based approach.

Our experimental findings illustrate several possible use cases of our
approach.  First, the $50\%$-winner threshold allows for a simple
quantification of the robustness of an outcome, which sheds some light
on how ``close'' the announced funding decisions are.  Our examples of
extremely non-robust real-world PB instances provide motivation to
carry out such an analysis in practice, to increase the transparency
of the process for voters, policy makers, and project proposers.  In
extreme cases non-robustness could lead either to audits or
reelections.  Second, by analyzing how quickly the funding
probabilities of different projects change, we can also quantify a
project's distance to receiving a different funding outcome on an
individual level. In particular, we identified three types of non-robust
projects (those that are selected/not-selected ``by luck,'' those that
are effectively tied, and those that act non-monotonically).
Third, our approach also allowed us to identify situations in which
rules made ``close'' decisions between different (sets) of projects,
as in this case project's funding probabilities changed simultaneous.

For future work, the most interesting theoretical open problem is the
study of the complexity of (\#)\AddRemProbName{} in case costs are
encoded in unary.  Moreover, from a practical perspective, it would be
interesting to explore different noise models taking into account
real voters' behavior.  Another research direction would be to
seek simpler robustness measures (both in terms of their computational
complexity and how easy it is to explain them and their results to the
public).

\section*{Acknowledgments}
NB was supported by the DFG project ComSoc-MPMS (NI 369/22).
This project has received funding from the European Research Council (ERC) under the European Union’s Horizon 2020 research and innovation programme (grant agreement No 101002854).
\begin{center}
  \includegraphics[width=3cm]{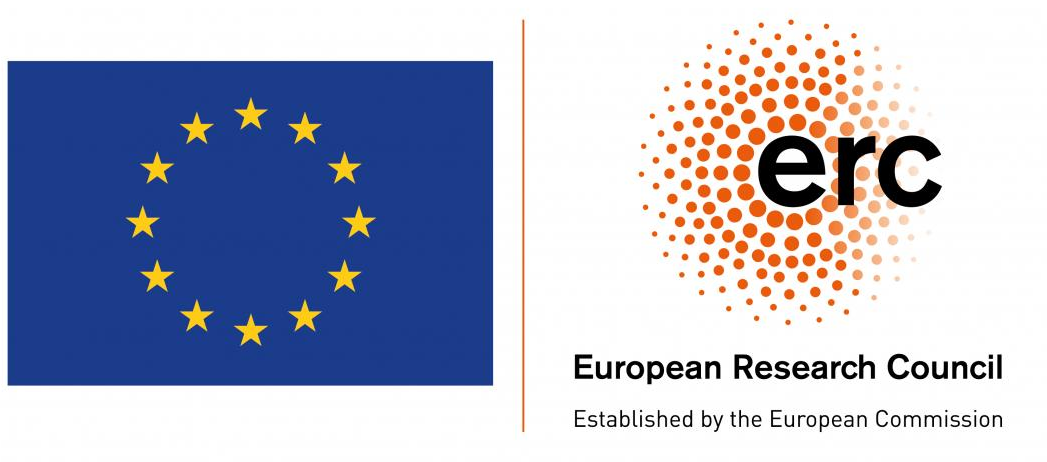}
\end{center}

\appendix

\newpage

\section{Missing Proofs from Section~\ref{sec:greedyav}}\label{app:greedy}

\thmfptmpspace*

\begin{proof}
  Let us consider a PB instance $E = (C,V,A,B,\cost)$, where
  $C = \{c_1, \ldots, c_m\}$ and $V = \{v_1, \ldots, v_n\}$. The
  preferred project is $p$ (so $p = c_\ell$ for some $\ell \in [m]$)
  and we should perform $r$ operations of adding or removing
  approval. Our algorithm considers all possible orderings of the
  projects from $C$ and for each of them does the following: First, it
  checks if $\greedyAVRule$ considered the projects in the given
  ordering, then would it include $p$ in the outcome. If so, then it
  counts the number of ways in which it is possible to flips approvals
  so that $\greedyAVRule$ indeed considers the projects in the
  considered order and adds it to the final answer (which initially is
  set to zero). After considering all the possible orderings, the
  algorithm outputs the final answer.  It is immediate to see that
  this strategy is correct and, in particular, there is no double
  counting.

  It remains to show that given an ordering of the projects, it is
  possible to compute in polynomial time the number of ways of
  flipping approvals that ensure that $\greedyAVRule$ considers the
  projects in the given order (formally, even an FPT algorithm would
  suffice, but we do not need this). W.l.o.g., we assume that the
  considered ordering of the projects is $c_1, c_2, \ldots, c_m$. Let
  $g$ be a function such that:
  \begin{enumerate}
  \item[] $f(c_i, s_i, r_i)$ is the number of ways in which it is
    possible to perform $r_i$ approval flips to projects
    $c_1, \ldots, c_i$ so that: (a) $c_i$ ends up with $s_i$
    approvals, and (b) if $\greedyAVRule$ is restricted to projects
    $c_1, \ldots, c_i$, then it considers them in the order
    $c_1, c_2, \ldots, c_i$.
  \end{enumerate}
  Our goal is to compute $\sum_{s_m=0}^n f(c_m, s_m, r)$. Next we show
  how to compute $f$ in polynomial time using dynamic programming. To
  do so, we will also need the following function:
  \begin{enumerate}
  \item[] $g(c_i,s_i,r_i)$ is the number of ways in which it is
    possible to perform $r_i$ approval flips for project $c_i$ so
    that it ends up with score $s_i$ (if doing so is impossible, then
    we have $g(c_i,s_i,r_i) = 0$).
  \end{enumerate}
  Computing the values of $g$ is straightforward (if $c_i$ originally
  has $a_i = |A(c_i)|$ approvals, then the value is
  $\sum_{x=\max(0,n-a_i)}^{\min(r_i,a_i)} {[a_i+x-(r_i-x) = 0] \cdot
    {n-a_i \choose x} \cdot {a_i \choose r_i-x}}$).  Next, we see that
  $f(c_1, s_1, r_1) = g(c_1, s_1, r_1)$. Finally, for each $i \geq 2$,
  we observe that $f(c_i, s_i, r_i)$ is equal to the following value
  (we take $t_i$ to either be $1$, if the internal tie-breaking would
  require $\greedyAVRule$ to choose $c_i$ before $c_{i-1}$ should they
  have the same number of approvals, or to be $0$, otherwise):
  \[
    \textstyle
     \sum_{s_{i-1}=s_i+t_i}^n {\sum_{x=0}^{r_i} f\big(c_{i-1},s_{i-1},r_i-x \big) \cdot g(c_i,s_i,x)}.
  \]
  Intuitively, this formula considers all approval scores that
  $c_{i-1}$ can end up with so that it still has more approvals than
  $c_i$ (or, more-or-equally-many approvals, if internal tie-breaking
  allows for this) so that $\greedyAVRule$ would consider $c_{i-1}$
  before $c_i$. Using this recursive formula and standard dynamic
  programming techniques, we can compute the values of $f$ in
  polynomial time.  \qed\end{proof}

    \cheapFirst*
    \begin{proof}[Continued]
  It is immediate that our algorithm works correctly for the case
  where $r < n$, so we focus on the case where $r \geq n$.  To this
  end, we will show that if there is a way to ensure that $p$ is
  selected, then it suffices to focus on solutions of a particular
  form.
  Assume that it is possible to ensure that $\greedyAVRule$ selects
  $p$ by flipping at most $r$ approvals, and let
  $E' = (C,V,A',B,\cost)$ be a PB instance obtained from $E$ in this
  way, for which $\greedyAVRule$ selects~$p$.  Without loss of
  generality, we assume that the sets of voters that approve $p$ are
  the same in both $E$ and $E'$ (if this were not the case, then we
  could guess how many approvals $p$ gets in $E'$ and view $E$ as the
  instance obtained after these approval were already added; further,
  it is clear that if there is a solution that removes $p$'s
  approvals, then omitting these removals does not preclude $p$ from
  winning). We write $\pre(E)$ to denote the set of projects that
  $\greedyAVRule$ considers prior to $p$ in instance $E$. We use
  analogous notation for $E'$ and other PB instances that include $p$.
  We will argue that it is
  possible to choose $E'$ in such a way that either $\pre(E) \subseteq \pre(E')$ or $p$ is approved by all the voters.
  If all the voters in $E'$ approve $p$, then the claim is clearly
  true, so we assume that this is not the case.

  Assume that it is not the case that $\pre(E) \subseteq \pre(E')$ and
  let $c$ be a project from $\pre(E)$ that does not belong to
  $\pre(E')$.  In other words, in $E'$ we remove sufficiently many
  approvals from $c$ so that $\greedyAVRule$ considers (and selects)
  $p$ prior to considering $c$. Let $E''$ be a PB instance identical
  to $E'$, except that instead of removing approvals from $c$, we add
  approvals to $p$. If $p$ ends up approved by all the voters, but we
  still have some $c$'s approval deletions left, we remove a single
  approval from $c$. Clearly, $\pre(E'') \subseteq \pre(E')$ and
  $c \notin \pre(E'')$. We claim that $p$ is still selected in
  $E''$. The reason is that in this instance $\greedyAVRule$ considers
  projects in the same order as in $E'$, except that it considers $p$
  at latest in the same iteration where it would have considered $c$
  in $E'$.  Since in $E'$ it would have selected $p$ later on, it must
  select $p$ in $E''$ at this iteration. Hence we can replace $E'$
  with $E''$.

  By repeatedly applying the above reasoning for each project
  that belongs to $\pre(E)$ but not to $\pre(E')$, we see that
  there is a PB instance $E^*$ such that:
  (a) it is possible to obtain $E^*$ from $E$ by flipping at most
  $r$ approvals, 
  (b) $\greedyAVRule$ selects $p$ in $E^*$,
  (c) either $p$ is approved by all the voters in $E^*$ or
  $\pre(E) \subseteq \pre(E^*)$.  We consider two cases:
  \begin{enumerate}
  \item If $p$ is approved by all the voters in $E^*$ then it
    certainly is also selected in the instance computed by our
    algorithm. This is so, because our algorithm ensures that in all
    the instances that can be obtained from $E$ by at most $r$
    approval flips and where $p$ is approved by all the voters, in the
    instance that it computes $\greedyAVRule$ uses up the smallest
    possible amount of budget prior to considering $p$.
  \item If $p$ is not approved by all the voters in $E^*$, but
    $\pre(E) \subseteq \pre(E^*)$, then $\greedyAVRule$ also selects
    $p$ in the instance computed by our algorithm. The reason is that
    in $E^*$, prior to considering $p$, $\greedyAVRule$ considers and
    selects all the projects from $\pre(E^*)$ (and, hence, all the
    projects from $\pre(E)$) that either are cheaper than $p$ or cost
    as much as $p$, but are preferred to $p$ by the tie-breaking order
    (this is so, because it finally also selects $p$). In the instance
    computed by our algorithm, every project that $\greedyAVRule$
    considers prior to $p$ belongs to $\pre(E)$ and is either cheaper
    than $p$ or costs the same, while being preferred by the tie-breaking
    order. Hence, after considering these projects, $\greedyAVRule$
    still has sufficient amount of budget left to select $p$.
  \end{enumerate}
  This shows that our algorithm  is indeed correct.

  To show $\wone$-hardness, it suffices to use the same construction
  as in the proof of \Cref{thm:greedy-av-pb-npc}, except for the following
  changes: (a) there are $6k+1$ voters, (b) $p$ is not approved by any
  voter, (c) each project $d_i$ is approved by $2k$ voters, (d) each
  project $y_i$ is approved by $4k$ voters, and (e) each project $x_i$
  is approved by $6k$ voters. One can verify that irrespective how we
  flip up to $k$ approvals, $\greedyAVRule$ will first consider all
  the $x_i$ projects, followed by all the $y_i$ ones, followed by all
  the $d_i$ ones, and eventually it will consider $p$. Further, by
  appropriate up-to-$k$ approval flips we can ensure that it first
  selects $k$ of the $x_i$ projects that correspond to the
  \textsc{Sized Subset Sum} solution. All in all, analogous arguments as
  in the proof of \Cref{thm:greedy-av-pb-npc} work in the
  current case.
  \qed\end{proof}

  \greedyavsharpone*
\begin{proof}
  Membership in $\sharpp$ follows by standard arguments.  We give a
  Turing reduction from \textsc{\#Sized Subset Sum}, where we are
  given a set $U = \{u_1, \ldots, u_n\}$ of positive integers, a
  targer integer $t$, and solution size $k$, and we ask for the number
  of size-$k$ subsets of $U$ that sum up to $t$.  This problem is
  $\sharpp$-complete and $\sharpwone$-hard for the parameterization
  by $k$.  We let $T = \sum_{i=1}u_i$ and, w.l.o.g., we assume that
  $t < T$.  Our reduction forms an instance $I$ of
  $\#\greedyAVAddRemPB$ with project set $C = \{x_1, \ldots, x_n,p\}$,
  where each project $x_i$ has cost $u_i$, $p$ has cost $2T-t$, and
  the budget is $2T$.  Further, there is just a single voter who
  initially approves project $p$ only, and we perform exactly $r = k$
  approval flips.

  We observe that $p$ costs more than every other project (indeed,
  more than all of them taken together). Hence, if we remove the
  approval from $p$, then in the resulting instance $\greedyAVRule$
  will choose all the projects except for $p$. Hence, to ensure that
  $\greedyAVRule$ selects $p$ by making exactly $k$ approval flips, we
  need to add approvals to $k$ projects whose total cost is at most
  $t$. In other words, the number of solutions for $I$ is the number
  of size-$k$ subsets of $U$ whose sum is at most $t$. We denote this
  number as $\#I$. Next, we form an instance $I'$, which is identical
  to $I$ except that we decrease the available budget to $2T-1$. By
  the same reasoning as before, the number of solutions for $I'$ is
  equal to the number of size-$k$ subsets of $U$ that sum up to at
  most $t-1$; we denote this value by $\#I'$.

  Our reduction outputs value $\#I - \#I'$, which---by the above
  discussion---is equal exactly to the number of size-$k$ subsets of
  $U$ that sum up to $t$. 
\qed\end{proof}

\greedyavunitx*

\begin{proof}
  Consider an instance of $\#\greedyAVAddRemPB$ with a PB instance
  $E=(C, V, A, B, \cost)$, preferred project $p \in C$, and required
  number of approval flips $r$. We rename the projects so that
  $C = \{p, c_1, \ldots, c_m\}$ and we let $n$ be the number of
  voters. Since all the projects have unit costs, w.l.o.g., we assume
  that $B < |C|$ (otherwise all the projects would always be selected,
  irrespective of the votes).  Our goal is to compute the number of
  ways in which we can ensure that $\greedyAVRule$ selects $p$,
  provided that we flip exactly $r$ approvals.  For each project $c$,
  we write $A(c)$ to refer to the set of voters that approve it.

  Given two nonnegative integers $\ell$ and $x$, and a project $c$, we
  let $h_c(\ell,x)$ be the number of ways in which it is possible to
  flip exactly $x$ approvals regarding $c$ so that, in the end, $c$ is
  approved by exactly $\ell$ voters (simple analysis shows that this
  value is computable in polynomial time).  Further, for each
  nonnegative integer $\ell$ and each nonnegative integer $r'$, we let
  $f(\ell,r')$ be the number of ways in which it is possible to
  perform exactly $r'$ approval flips (but excluding those that
  involve project $p$) so that $\greedyAVRule$ would select $p$,
  provided that~$p$ had exactly~$\ell$ approvals (if $r' < 0$, then we
  define $f(\ell,r')$ to be zero). One can verify that our algorithm
  should output the following value:
  \[
    \sum_{\ell \in [n] \cup \{0\}} \sum_{ x \in [r] \cup \{0\}}
    h_p(\ell,x) \cdot f(\ell, r-x).
  \]
  Hence, it suffices to argue that we can compute values $f(\ell,r')$
  in polynomial time. Throughout the rest of the proof, we describe
  how to do it using dynamic programming.

  For each nonnegative integer $\ell$, each nonnegative integer $r'$,
  each $i \in [m]$, and each $j \in [B]$, we define $g(\ell, r',i,j)$
  to be the number of ways in which it is possible to perform exactly
  $r'$ approval flips that regard projects form the set
  $C_i := \{c_1, \ldots, c_i\}$, so that the number of those projects
  from $C_i$ that have more than $\ell$ approvals or that have exactly
  $\ell$ approvals but precede $p$ in the tie-breaking order, is
  exactly $j$. We see that:
  \[
    f(\ell,r') = \sum_{b \in [B-1] \cup \{0\}} g(\ell, r', m, (B-1) - b)
  \]
  Indeed, the right-hand side of this equation simply gives the number
  of ways of performing $r'$ approval flips regarding the projects in
  $C \setminus \{p\}$, so that $\greedyAVRule$ considers  at most $B-1$ other
  projects, prior to considering $p$.  One can verify that for each $\ell'$, each $r'$, and each
  $j$, it is possible to compute $g(\ell,r',1,j)$, via simple case
  analysis. Next, we observe that for $i \geq 2$ the following holds (we let $t_i = 1$ if $c_i$ precedes $p$
  in the tie-breaking order and we let $t_i = 0$ otherwise):
  \begin{align*}
    g(\ell, r',i,j) &=\sum_{r'' \in [r'] \cup \{0\}} \sum_{\ell'' \in [\ell-t_i] \cup \{0\}} g(\ell, r'-r'', i-1,j) h_{c_i}(\ell'',r'') \\
                    &+\sum_{r'' \in [r'] \cup \{0\}} \sum_{\ell'' \geq \ell+(1-t_i)} g(\ell, r'-r'', i-1,j-1) h_{c_i}(\ell'',r'').
  \end{align*}
  The sums in the first row correspond to the situations where the
  $i$-th project ends up being considered after project~$p$, and the
  second row corresponds to the situations where it is considered
  prior to $p$. Using the above formula and standard dynamic
  programming techniques, we can compute $g(\ell, r',i,j)$ in
  polynomial time.    
  \qed\end{proof}

\section{Missing Proofs from Section~\ref{sec:beyond}}
\label{app:beyond}

\greedycostnpc*
\begin{proof}
  It is immediate that $\greedyCostAddRemPB{}$ belongs to~$\np$. To
  show~$\np$-hardness, we provide a polynomial-time many-one reduction
  from the \textsc{Sized Subset Sum} problem. Let~$I$ be an instance
  thereof, with the set~$U = \{u_1, \ldots, u_n\}$ of positive
  integers, a target integer~$t$, and a requested number~$k$ of
  elements. Without loss of generality, we assume that for
  each~$k' < k$, the sum of elements of every size-$k'$
  subset~$U' \subset U$ is smaller than~$t$, i.e.,\
  $\sum_{u \in U'} < t$. If this is not the case, then one can add a
  sufficiently big constant~$X$ (e.g.,\ $X = 1+\sum_{u \in U} u$) to
  every element of~$U$ and increase~$t$ by~$kX$, obtaining an
  equivalent instance that fulfills the assumption. Similarly, we
  assume that, for each~$u_i \in U$, it holds that~$u_i < t$, and
  that~$k < n$.
 
  We reduce~$I$ to an instance~$I'$ of~$\greedyCostAddRemPB{}$ with a
  PB instance~$E = (C, \{v\}, A, B, \cost)$ as follows. We compose~$C$
  of the preferred project~$p$, projects~$c_i$ corresponding to the
  elements $u_i \in U$, and $k+1$ \emph{dummy projects}~$d_1$ up
  to~$d_{k+1}$. Regarding the project costs: (a)~we
  set~$\cost(p) = t$; (b)~for each~$u_i \in U$, we
  set~$\cost(c_i) = u_i$; and (c)~we set $\cost(d_i) = t+1$ for each
  dummy project~$d_1$ to~$d_{k+1}$.  Instance~$E$ consists of a single
  voter~$v$ who approves all the candidates except for~$p$, i.e.,\
  $A(v) = C \setminus \{p\}$. Concluding the construction of~$I'$, we
  set the budget~$B = \sum_{u_i \in U} u_i$, the radius~$r = k$, and
  an arbitrary but fixed tie-breaking order~$\succ$ where~$p$ is the
  most-preferred candidate (recall that under \greedyCostRule,
  internal tie-breaking between projects that receive a single
  approval first follows the increasing order of their costs, and only
  for projects with equal costs uses $\succ$; on the other hand,
  internal tie-breaking for projects without any approvals always
  follows $\succ$).
 
  Before we proceed with a proof of correctness of the reduction,
  observe that prior to any flips, $\greedyCostRule{}$ selects a
  committee (of size~$n$) consisting solely of candidates~$c_1$
  to~$c_n$, which exhaust budget~$B$.  Hence, to ensure that $p$ is in
  at least one winning committee, one needs to perform some approval
  flips. Notably, since we are allowed to perform only at
  most~$k$~flips, after each possible alterations of~$E$, there always
  exists at least one approved dummy project, out of the dummy
  projects~$d_1$ to~$d_{k+1}$.

  We start by showing that a solution~$S \subseteq U$ to instance~$I$
  implies that we can flip at most~$k$ approvals to make~$p$ be
  selected.  To achieve the latter, we remove the approvals from a set
  of projects~$S'$ whose elements correspond to the elements of~$S$,
  that is, $S' = \{ c_i \mid u_i \in S\}$. After doing so,
  $\greedyCostRule{}$ performs as follows. First, it considers the
  (still approved) projects corresponding to the elements
  from~$U \setminus S$ because these are the cheapest approved
  projects. By our assumption on~$S$, the total cost of these projects
  is~$B - t$, so they are all selected to be funded. Then,
  $\greedyCostRule{}$ skips the dummy projects, since selecting any of
  them (each of cost~$t+1$) would exceed the budget. Now, only the
  unapproved projects remain. For all of them the ratio of their
  approval score to their cost is zero, so $\greedyCostRule{}$ resorts
  to the tie-breaking order~$\succ$, according to which $p$ should be
  considered first. Since~$\cost(p) = t$ and the remaining budget is
  exactly~$t$, $p$ is selected, which proves that a solution~$I$
  implies a solution to~$I'$.
 
  Next, we show that if one can get~$p$ to be selected by performing
  at most~$k$~approval flips, then there is a solution~$S \subseteq U$
  to~$I$.  Let $E'$ be the PB instance obtained from $E$ by performing
  up to $k$ approval flips, such that $\greedyCostRule$ selects $p$ in
  $E'$.  We write~$A'(v)$ to denote the set of projects that~$v$
  approves in~$E'$.
 
  We show that to obtain~$E'$ all of the $k$~approval flips have to be
  spent to delete approvals from (a subset of) the candidates~$c_1$
  to~$c_n$, i.e.,\, $|A'(v) \cap \{c_1, \ldots, c_n\}| = n -
  k$. Toward a contradiction, we assume that there are fewer
  than~$k$~approval flips spent on removing approvals from
  candidates~$c_1$ to~$c_n$ and simulate a run of~$\greedyCostRule$
  on~$E'$. The rule starts by considering the approved projects
  among~$c_1$ to~$c_n$. This holds because, by our assumption on~$I$,
  for each~$c_i$, $\cost(c_i) < \cost(p)$ and the dummy projects are
  even more expensive than~$p$. Furthermore, because (by definition)
  $B = \sum_{u_i \in U} \cost(c_i)$, all these approved projects
  among~$c_1$ to~$c_n$ are selected to the winning committee. Hence,
  applying our assumption on~$I$ that no group of~$k-1$ (or fewer)
  elements of~$U$ sum up to at least~$t$, the total cost of the
  selected projects at this stage is greater
  than~$(\sum_{u_i \in U} u_i) - t = B - t$. So, afterward, the
  remaining funds are insufficient for~$p$.  This yields a
  contradiction and confirms that to obtain~$E'$, we need to spend all
  $k$~approval flips to remove approvals from projects~$c_1$
  to~$c_n$. Notably, this observation implies that $p$ has no
  approvals in~$E'$. Consequently, let us consider a set~$S$ of the
  projects among~$c_1$ to~$c_n$ from which we remove the approvals,
  and let~$y = \sum_{u_i \in S} u_i$~be the total cost of these
  projects. Note that after~$\greedyCostRule$ analyzes the projects
  from~$c_1$ to~$c_n$ that remain approved in $E'$, the remaining budget
  is~$B - (B - y) = y$. We now analyze the following three cases
  depending on the value of~$y$:
 \begin{enumerate}
 \item $y < t$: then the remaining budget $y$ is smaller than~$t$,
   and~$p$ is too expensive to be selected;
 \item $y = t$: then the remaining budget is exactly~$t$,
   and~$\greedyCostRule$ skips all dummy projects, which are too
   expensive, and proceeds with considering~$p$, which it selects;
 \item $y > t$: then $y<2t$, which follows from the fact that, for
   each (arbitrarily selected) $s' \in S$,
   $y = (\sum_{s \in S} s - s') + s'$ and, by our assumptions, both
   the sum~$(\sum_{s \in S} s - s')$ equivalent to the sum of
   some~$k-1$~elements of~$U$ and the element~$s'$ have cost at
   most~$t$. What follows is that~$\greedyCostRule$ in the next step
   select exactly one dummy project, making the remaining budget too
   small to select~$p$.
 \end{enumerate}
 The above case distinction shows that $p$ is selected if and only one can make
 $k$~approval flips corresponding to a set of~$k$ elements from~$U$ which sum up
 to~$t$.

 The reduction clearly works in polynomial-time. Furthermore, the
 value of parameter~$k$ for instance~$I$ is used as the number of
 approval flips, which yields the $\wone$-hardness result. Since
 \#\textsc{Sized Subset Sum} is $\sharpp$-hard and the solutions
 to~$I'$ one-to-one correspond to solutions to~$I$, the claimed
 $\sharpp$-hardness follows as well.  \qed\end{proof}

\subsection{Proof of \Cref{thm:mes-phragmen}}
\label{app:mes-phragmen}

We prove \Cref{thm:mes-phragmen} by providing the following two lemmas.

\begin{lemma} \label{thm:phragmen-pb-np-complete-unit-costs}
    $\phragmenAddRemPB$ is $\np$-complete, even if all projects cost exactly $1$.
\end{lemma}
\begin{proof}
  Naturally, $\phragmenAddRemPB$ belongs to $\np$ class because we can
  guess a set of at most $r$ operations, perform them on a given
  instance, and check if a given project is included in the final
  committee. Please note that in our input we are given a default
  tie-breaking order, therefore the outcome is unique and our
  algorithm runs in polynomial time.

  To show that $\phragmenAddRemPB$ is $\np$-hard, we reduce from
  $\rxthreec$. In the $\rxthreec$ problem, we are given a universe
  $U = \{u_1, u_2, \ldots, u_{3n}\}$ and a collection of size-$3$ sets
  $S = \{S_1, S_2, \ldots, S_{3n}\}$ such that each element of $U$
  belongs to exactly $3$ sets from $S$. We ask if there are $n$ sets
  from $S$ that form an exact cover over $U$.

  Let us focus on the proof provided by Faliszewski et al. in
  \cite{fal-gaw-kus:c:greedy-robustness} for the theorem stating that
  \textit{\textsc{Phragmén-Add-robustness-radius} is $\np$-complete}.

  As the authors observed, with the budget $n$ for buying $n$
  approvals, we can only change the order of considering $S$-projects
  and, if we put $n$ sets forming an exact cover at the beginning of
  the order by adding approvals to the empty voters, successfully
  select $p$ to the committee. Further, the authors observe that
  \textit{\textsc{Phragmén-remove-robustness-radius} is
    $\np$-complete} for the same reason, but with a slight
  modification of the above construction -- budget is increased from
  $n$ to $2n$ and $n$ empty voters are replaced with $n$ voters, each
  of them approves a unique $S$-project. Summing up, even if we add
  $n$ approvals or remove $2n$ approvals from a single project, the
  only thing that can can change is the tie-breaking order of
  $S$-projects, so the whole analysis remains the same as in the
  proof.
    
  For this reason, if we take an instance of $\rxthreec$, construct a
  PB instance as in the cited theorem
  \textit{\textsc{Phragmén-Add-robustness-radius} is $\np$-complete},
  and set the preferred project to be $p$, the available budget to be
  $B=n$, and bribery budget to be $r=n$, then there exists an exact
  cover for the $\rxthreec$ instance if and only if we can perform up
  to $r$ operations of flipping approvals on the constructed PB
  instance such that the preferred project $p$ is winning.
\end{proof}

The proof for $\MESPB$ will require more explanation since, in contrast to $\phragmenRule$ rule, it slightly modifies the constructed instance.

\begin{lemma} \label{thm:mes-pb-np-complete-unit-costs}
    $\MESCostUtilAddRemPB$ and $\MESAprUtilAddRemPB$ are $\np$-complete, even if all projects cost exactly $1$.
\end{lemma}
\begin{proof}
  Since $\MESCostUtilRule$ and $\MESAprUtilRule$ work in the same way
  for unit costs (because dividing by cost does not change anything in
  this case), we will focus on $\MESAprUtilRule$. Further, since they
  become equivalent to MES approval voting rule for unit project
  costs, we will use a construction provided for multiwinner approval
  voting case and adapt it to the PB world.

    Similarly to $\phragmenAddRemPB$ proved above, $\MESAprUtilAddRemPB$
    belongs to $\np$ class because we can guess a set of at most $r$
    operations, perform them on a given instance, and check if a given
    project is included in the final committee. We note that since in
    our input we are given a default tie-breaking order, the outcome
    is unique and thus our algorithm runs in polynomial time.

    To show that $\MESAprUtilAddRemPB$ is $\np$-hard, we reduce from
    $\rxthreec$. In the $\rxthreec$ problem, we are given a universe
    $U = \{u_1, u_2, \ldots, u_{3n}\}$ and a collection of size-$3$
    sets $S = \{S_1, S_2, \ldots, S_{3n}\}$ such that each element of
    $U$ belongs to exactly $3$ sets from $S$. We ask if there are $n$
    sets from $S$ that form an exact cover over $U$.

    We form an instance in the same way as that provided by Janeczko
    and Faliszewski \cite{jan-fal:c:ties-multiwinner} in the theorem stating
    that \textit{\textsc{Unique-Committee} is coNP-complete for Phase
      1 of MEqS.}, but with one significant change: we replace each
    voter in this instance with exactly $T = 14400n^7$ copies of the
    same preferences. We set the preferred project to be $p$, the
    available budget to be $B' = k$, and the bribery budget to be
    $r = n$. In contrast to the former theorem, we assume that we have
    some fixed tie-breaking order in which $p$ is preferred over $d$ (the rest can be arbitrary). We ask if we can flip up to $r = 2n$
    approvals such that $p$ belongs to a winning outcome.

    Now let us analyze how $\MESAprUtilRule$ behaves on the particular
    instance.  Since we scaled the number of voters by $T$, each voter
    receives exactly $\frac{1}{T} = \frac{1}{14400n^7}$ units of the
    budget.  Therefore, even if we add $r=2n$ approvals to some project
    or remove $r=2n$ existing approvals, then the number of voters to
    participate in buying project may change by at most $r = 2n \ll T$
    and their total budget by at most
    $2n \cdot \frac{1}{T} = \frac{2n}{14400n^7} = \frac{2}{14400n^6}$.
    In our case, since each voter was replaced with $T \gg 2n$ copies,
    these changes are far too low to move projects between groups,
    they can only affect the order of considering them inside a group.
    
    More precisely, at the beginning we will select all projects
    from $C_B$ and then all projects from $C_U$. The below
    inequalities were estimated with the following observations: since
    we can add or remove up to $r = 2n$ approvals, we can either
    add/remove them from the considered group of voters, or add/remove
    them from some other group of voters (approving or disapproving
    some project(s)).  After selecting all projects from $C_B$,
    each voter from $B \cup V_{pd}$ will pay at least
    $\frac{1}{(144n^3)T+r}$ for each project, and thus will be be
    left with at most
    $\frac{1}{T} - (144n^3-12n^2) \cdot \frac{1}{(144n^3)T+r} =
    \frac{((144n^3)T+r) - (144n^3-12n^2)T}{T \cdot ((144n^3)T+r)} =
    \frac{12n^2T+r}{T \cdot ((144n^3)T+r)}$.  Analogously, after
    selecting all projects from $C_U$, each voter from
    $B_U \cup U' \cup U''$ will pay at least
    $\frac{1}{(54n^3+27n^2+3n)T+r}$ for each project, and thus will
    be be left with at most
    $\frac{1}{T} - (54n^3 + 24n^2 + \nicefrac{5n}{2}) \cdot
    \frac{1}{(54n^3+27n^2+3n)T + r} = \frac{((54n^3+27n^2+3n)T + r) -
      (54n^3 + 24n^2 + \nicefrac{5n}{2})T}{T \cdot ((54n^3+27n^2+3n)T
      + r)} = \frac{T(3n^2+\nicefrac{n}{2})+r}{T \cdot
      (54n^3+27n^2+3n)T + r)}$.  Since $T=14400n^7$ and $r=2n$, the
    total budget that voters from $V_{pd} \cup U'$ have after selecting all projects from $C_B$ and $C_U$ is at most
    $12nT \cdot \frac{12n^2T+r}{T \cdot ((144n^3)T+r)} + 3nT \cdot
    \frac{T(3n^2+\nicefrac{n}{2})+r}{T \cdot (54n^3+27n^2+3n)T + r)}
    \leq 12nT \cdot \frac{12n^2T+r}{T(144n^3T)} + 3nT \cdot
    \frac{T(3n^2+\nicefrac{n}{2})+r}{T^2(54n^3+27n^2+3n)} =
    \frac{12n^2T+r}{12n^2T} +
    \frac{T(3n^2+\nicefrac{n}{2})+r}{T(18n^2+9n+1)} = 1 +
    \frac{r}{12n^2T} + \frac{1}{6} -
    \frac{n+\nicefrac{1}{6}}{18n^2+9n+1} + \frac{r}{T(18n^2+9n+1)} < 1
    + \nicefrac{1}{12} + \nicefrac{1}{6} + \nicefrac{1}{12} = \nicefrac{16}{12} = \nicefrac{4}{3}$. For this
    reason, since
    $\nicefrac{4}{3} + r \cdot \frac{1}{T} = \nicefrac{4}{3} +
    \frac{2}{14400n^6} < 2$, we will never select both $p$ and $d$ to
    the committee.

    Further, we will be selecting sequentially some $S$-projects whose
    elements were not present in the previous iterations.  We claim
    that if this phase does not end up with selecting sets forming an
    exact cover over $U$, then $d$ must be selected to the committee.
    For the sake of the proof's length, we will not explain every
    inequality in details, especially if its intuition is clear from
    the context.
    
    Suppose that an element $u_i$ was not covered.  Then it means that
    $T$ voters from $U'$ resembling $u_i$ must have some left
    budget. More precisely, they couldn't have spent more than
    $\frac{1}{(54n^3+27n^2+3n)T-r}$ on each $B_U$-project so each of them must be left with at least
    $\frac{1}{T} - (54n^3 + 24n^2 + \nicefrac{5n}{2}) \cdot
    \frac{1}{(54n^3+27n^2+3n)T - r} = \frac{((54n^3+27n^2+3n)T - r) -
      (54n^3 + 24n^2 + \nicefrac{5n}{2})T}{T \cdot ((54n^3+27n^2+3n)T-r)} = \frac{T(3n^2+\nicefrac{n}{2})-r}{T \cdot
      ((54n^3+27n^2+3n)T - r)} > \frac{1}{90nT}$.  Let $\rho_d$ be the
    maximum amount of money a voter needs to pay to buy $d$.  In the
    worst case, we remove some existing approvals of voters $V_{pd}$
    towards $d$ and of voters $U'$ towards $d$. Then even if each voter from $U$ would contribute ``only'' $\frac{1}{90nT}$, each voter from
    $V_{pd}$ would need to pay for $d$ at most
    $\frac{1 - (T-r) \cdot \frac{1}{90nT}}{12nT-r} = \frac{1}{12nT-r} -
    \frac{T-r}{90nT(12nT-r)}$.  Thus, since $\frac{1}{12nT-r} - \frac{T-r}{90nT(12nT-r)} > \frac{1}{90nT}$,
    $\rho_d \leq \frac{1}{12nT-r} - \frac{T-r}{90nT(12nT-r)}$. Let us
    compute analogous $\rho_p$ value for project $p$. In the best
    case, we convince some dummy voters to vote for $p$, after these
    operations the value of $\rho_p$ would be at least
    $\frac{1}{12nT+r}$ (if all supporters of $p$ would contribute equally). It is not hard to see that
    $\rho_d \leq \frac{1}{12nT-r} - \frac{T-r}{90nT(12nT-r)} <
    \frac{1}{12nT+r} \leq \rho_p$, thus no matter
    of operations we would select $d$ rather than $p$.

    We showed above that if there is no exact cover over $U$, then we
    will always prefer $d$ to $p$. Moreover, the opposite direction
    also holds: if there is an exact cover, then we can perform up to
    $r = 2n$ operations such that $p$ is selected to the
    committee. Particularly, if ${i_1, i_2, \ldots, i_n}$ are the
    indices of sets forming an exact cover, we remove one approval
    from sets $S \ \{S_{i_1}, S_{i_2}, \ldots, S_{i_n}\}$ (preferably
    from separate $V_S$-voters, $3n-n=2n$ operations in total). This bribery makes the exact cover projects being considered at the beginning and would use the budget of $U'$-voters exhaustively. 
    For this reason, while considering projects from
    $\{p, d\}$, $d$ ties with $p$, so due to the tie-breaking order $p$ will be selected to the committee.
    
    Summing up, there exists an exact cover for the $\rxthreec$
    instance if and only if project $p$ can be winning after doing at
    most $r=2n$ approval flips in the constructed PB instance.
\end{proof}

\section{Additional Experiments}\label{app:experiments}

\subsection{Correlation between Rules}\label{sec:correlation}
We have already observed above that the used budgeting rule has a decisive influence on the robustness of outcomes and that \MES and \Phragmen tend to produce less robust results than \Greedy and \GreedyCost. In this section, we want to explore the relation between the different budgeting rules in more depth. 
In particular, we analyze whether instances that have a very robust outcome under one rule might have a non-robust outcome under a different rule.

\begin{figure}[t!]
    \centering 
\begin{subfigure}{0.32\textwidth}
\resizebox{\textwidth}{!}{
\begin{tikzpicture}

\definecolor{darkgray176}{RGB}{176,176,176}

\begin{axis}[
tick align=outside,
tick pos=left,
x grid style={darkgray176},
xlabel={\Greedy},
xmin=-0.00685, xmax=0.31985,
xtick style={color=black},
y grid style={darkgray176},
ylabel={\GreedyCost},
ymin=-0.0045, ymax=0.3145,
ytick style={color=black},
xtick={0,0.05,0.1,0.15,0.2,0.25,0.3},
xticklabels={0\%,5\%,10\%,15\%,20\%,25\%,$\ge\hspace*{-0.1cm}25\%$},
ytick={0,0.05,0.1,0.15,0.2,0.25,0.3},
yticklabels={0\%,5\%,10\%,15\%,20\%,25\%,$\ge\hspace*{-0.1cm}25\%$},
	ytick style={color=black},every tick label/.append style={font=\large}, 
label style={font=\large}
]
\addplot [draw=black, fill=black, mark=*, only marks]
table{%
x  y
0.3 0.16
0.22 0.3
0.01 0.09
0.299 0.16
0.09 0.3
0.3 0.17
0.11 0.25
0.3 0.1
0.08 0.16
0.3 0.09
0.16 0.3
0.3 0.03
0.3 0.07
0.08 0.03
0.3 0.01
0.07 0.3
0.03 0.3
0.299 0.01
0.11 0.01
0.3 0.19
0.069 0.3
0.1 0.11
0.3 0.12
0.3 0.15
0.299 0.17
0.23 0.02
0.3 0.2
0.09 0.13
0.3 0.02
0.3 0.13
0.3 0.24
0.3 0.22
0.23 0.3
0.299 0.07
0.06 0.3
0.3 0.21
0.3 0.06
0.299 0.02
0.03 0.12
0.04 0.01
0.18 0.1
0.21 0.22
0.02 0.19
0.3 0.04
0.219 0.3
0.299 0.12
0.301 0.01
0.11 0.05
0.229 0.3
0.13 0.3
0.3 0.05
0.24 0.21
0.2 0.13
0.19 0.08
0.299 0.04
0.24 0.14
0.14 0.09
0.299 0.05
0.299 0.1
0.3 0.14
0.301 0.05
0.299 0.03
0.06 0.11
0.301 0.02
0.06 0.1
0.14 0.15
0.14 0.3
0.3 0.11
0.12 0.3
0.299 0.14
0.301 0.03
0.3 0.08
0.059 0.3
0.01 0.3
0.139 0.3
0.298 0.01
0.22 0.08
0.301 0.14
0.24 0.11
0.009 0.3
0.029 0.3
0.301 0.17
0.159 0.3
0.302 0.01
0.011 0.3
0.299 0.21
0.02 0.3
0.301 0.1
0.298 0.02
0.299 0.13
0.298 0.05
0.299 0.11
0.301 0.11
0.009 0.09
0.299 0.08
0.141 0.3
0.302 0.02
0.3 0.23
0.298 0.03
0.299 0.19
0.299 0.06
0.19 0.3
0.03 0.01
0.301 0.13
0.298 0.1
0.189 0.3
0.298 0.17
0.302 0.05
0.008 0.3
0.301 0.16
0.298 0.14
0.299 0.15
0.3 0.18
0.299 0.2
0.301 0.21
0.297 0.02
0.138 0.3
0.18 0.3
0.089 0.3
0.299 0.09
0.16 0.22
0.303 0.02
0.15 0.3
0.129 0.3
0.019 0.3
0.08 0.09
0.302 0.14
0.297 0.01
0.302 0.03
0.3 0.25
0.21 0.17
0.08 0.18
0.11 0.3
0.02 0.05
0.298 0.16
0.301 0.2
0.296 0.02
0.297 0.03
0.301 0.09
0.301 0.06
0.04 0.11
0.179 0.3
0.301 0.04
0.297 0.05
0.299 0.23
0.301 0.07
0.302 0.17
0.298 0.09
0.304 0.02
0.07 0.02
0.191 0.3
0.298 0.07
0.188 0.3
0.15 0.19
0.298 0.21
0.119 0.3
0.295 0.02
0.13 0.2
0.19 0.25
0.302 0.1
0.305 0.02
0.24 0.3
0.299 0.25
0.301 0.12
0.061 0.3
0.303 0.03
0.181 0.3
0.012 0.3
0.058 0.3
};
\end{axis}

\end{tikzpicture}}
  \caption{\Greedy and \GreedyCost}
\end{subfigure}\hfill
\begin{subfigure}{0.32\textwidth}
\resizebox{\textwidth}{!}{
\begin{tikzpicture}

\definecolor{darkgray176}{RGB}{176,176,176}

\begin{axis}[
tick align=outside,
tick pos=left,
x grid style={darkgray176},
xlabel={\Greedy},
xmin=-0.0069, xmax=0.3209,
xtick style={color=black},
y grid style={darkgray176},
ylabel={\Phragmen},
ymin=-0.0045, ymax=0.3145,
ytick style={color=black},
xtick={0,0.05,0.1,0.15,0.2,0.25,0.3},
xticklabels={0\%,5\%,10\%,15\%,20\%,25\%,$\ge\hspace*{-0.1cm}25\%$},
ytick={0,0.05,0.1,0.15,0.2,0.25,0.3},
yticklabels={0\%,5\%,10\%,15\%,20\%,25\%,$\ge\hspace*{-0.1cm}25\%$},
	ytick style={color=black},every tick label/.append style={font=\large}, 
label style={font=\large}
]
\addplot [draw=black, fill=black, mark=*, only marks]
table{%
x  y
0.3 0.06
0.22 0.3
0.01 0.09
0.3 0.17
0.09 0.3
0.11 0.03
0.3 0.11
0.08 0.3
0.3 0.24
0.3 0.03
0.3 0.02
0.16 0.14
0.299 0.03
0.299 0.02
0.3 0.09
0.08 0.06
0.3 0.01
0.299 0.01
0.07 0.09
0.03 0.3
0.301 0.01
0.11 0.3
0.3 0.1
0.07 0.07
0.1 0.07
0.3 0.14
0.3 0.08
0.301 0.03
0.3 0.21
0.299 0.08
0.23 0.3
0.298 0.01
0.09 0.02
0.299 0.06
0.3 0.13
0.298 0.03
0.23 0.14
0.299 0.17
0.06 0.3
0.3 0.07
0.302 0.01
0.3 0.22
0.03 0.06
0.04 0.17
0.18 0.14
0.21 0.1
0.02 0.2
0.297 0.01
0.219 0.3
0.3 0.19
0.3 0.18
0.11 0.05
0.301 0.06
0.229 0.3
0.13 0.21
0.299 0.14
0.24 0.14
0.2 0.22
0.19 0.05
0.299 0.09
0.24 0.3
0.14 0.09
0.3 0.05
0.298 0.06
0.299 0.1
0.299 0.24
0.302 0.03
0.06 0.15
0.3 0.04
0.06 0.1
0.14 0.08
0.3 0.12
0.139 0.09
0.299 0.05
0.12 0.3
0.301 0.14
0.299 0.04
0.3 0.16
0.06 0.01
0.01 0.3
0.14 0.3
0.303 0.01
0.22 0.14
0.301 0.1
0.24 0.13
0.299 0.16
0.009 0.3
0.03 0.07
0.16 0.3
0.296 0.01
0.298 0.1
0.011 0.3
0.301 0.24
0.302 0.06
0.02 0.12
0.299 0.12
0.299 0.07
0.299 0.13
0.299 0.11
0.301 0.13
0.01 0.22
0.301 0.16
0.139 0.3
0.297 0.03
0.301 0.17
0.301 0.11
0.301 0.05
0.301 0.02
0.301 0.07
0.19 0.3
0.03 0.04
0.298 0.07
0.189 0.3
0.298 0.02
0.304 0.01
0.298 0.13
0.298 0.14
0.295 0.01
0.302 0.02
0.008 0.3
0.298 0.16
0.301 0.09
0.301 0.04
0.297 0.02
0.141 0.3
0.297 0.06
0.18 0.3
0.089 0.3
0.298 0.09
0.3 0.15
0.16 0.21
0.305 0.01
0.15 0.14
0.13 0.3
0.02 0.04
0.08 0.14
0.3 0.2
0.294 0.01
0.301 0.12
0.303 0.03
0.21 0.16
0.298 0.17
0.08 0.23
0.109 0.3
0.02 0.06
0.299 0.19
0.3 0.25
0.299 0.2
0.299 0.21
0.302 0.13
0.296 0.03
0.302 0.09
0.299 0.18
0.04 0.18
0.179 0.3
0.302 0.14
0.302 0.07
0.303 0.02
0.298 0.04
0.297 0.14
0.306 0.01
0.296 0.02
0.07 0.06
0.191 0.3
0.297 0.07
0.19 0.14
0.15 0.18
0.119 0.3
0.298 0.05
0.13 0.17
0.297 0.13
0.188 0.3
0.303 0.07
0.24 0.05
0.304 0.02
0.298 0.12
0.06 0.04
0.298 0.24
0.296 0.07
0.181 0.3
0.012 0.3
0.298 0.11
0.059 0.3
};
\end{axis}

\end{tikzpicture}}
\caption{\Greedy and \Phragmen}
\end{subfigure}\hfill
\begin{subfigure}{0.32\textwidth}
  \resizebox{\textwidth}{!}{
\begin{tikzpicture}

\definecolor{darkgray176}{RGB}{176,176,176}

\begin{axis}[
tick align=outside,
tick pos=left,
x grid style={darkgray176},
xlabel={\Phragmen},
xmin=-0.00665, xmax=0.31565,
xtick style={color=black},
y grid style={darkgray176},
ylabel={\GreedyCost},
ymin=-0.0045, ymax=0.3145,
ytick style={color=black},
xtick={0,0.05,0.1,0.15,0.2,0.25,0.3},
xticklabels={0\%,5\%,10\%,15\%,20\%,25\%,$\ge\hspace*{-0.1cm}25\%$},
ytick={0,0.05,0.1,0.15,0.2,0.25,0.3},
yticklabels={0\%,5\%,10\%,15\%,20\%,25\%,$\ge\hspace*{-0.1cm}25\%$},
	ytick style={color=black},every tick label/.append style={font=\large}, 
label style={font=\large}
]
\addplot [draw=black, fill=black, mark=*, only marks]
table{%
x  y
0.06 0.3
0.3 0.16
0.09 0.09
0.17 0.16
0.3 0.17
0.03 0.25
0.11 0.1
0.299 0.16
0.24 0.3
0.03 0.3
0.02 0.09
0.14 0.3
0.03 0.03
0.02 0.3
0.09 0.07
0.06 0.03
0.01 0.3
0.01 0.01
0.09 0.3
0.009 0.01
0.3 0.01
0.3 0.19
0.1 0.3
0.07 0.3
0.07 0.11
0.14 0.12
0.08 0.15
0.029 0.3
0.21 0.17
0.08 0.3
0.3 0.02
0.01 0.2
0.02 0.13
0.06 0.02
0.13 0.13
0.3 0.24
0.03 0.22
0.139 0.3
0.17 0.07
0.069 0.3
0.3 0.21
0.01 0.06
0.22 0.02
0.06 0.12
0.17 0.01
0.14 0.1
0.1 0.22
0.2 0.19
0.009 0.3
0.3 0.04
0.19 0.12
0.18 0.01
0.05 0.05
0.059 0.3
0.21 0.3
0.14 0.05
0.14 0.21
0.22 0.13
0.05 0.08
0.09 0.04
0.3 0.14
0.089 0.09
0.049 0.05
0.06 0.1
0.099 0.3
0.24 0.14
0.3 0.05
0.029 0.03
0.15 0.11
0.04 0.02
0.1 0.1
0.079 0.15
0.12 0.3
0.089 0.3
0.05 0.11
0.14 0.14
0.04 0.03
0.16 0.08
0.011 0.3
0.011 0.01
0.14 0.08
0.1 0.14
0.13 0.11
0.16 0.3
0.071 0.3
0.299 0.17
0.008 0.01
0.101 0.3
0.239 0.3
0.06 0.21
0.119 0.3
0.12 0.1
0.07 0.02
0.129 0.13
0.299 0.05
0.11 0.11
0.129 0.11
0.22 0.09
0.159 0.08
0.03 0.02
0.17 0.3
0.11 0.23
0.05 0.03
0.02 0.19
0.07 0.06
0.04 0.01
0.3 0.13
0.07 0.1
0.019 0.3
0.008 0.3
0.13 0.17
0.139 0.05
0.012 0.3
0.021 0.3
0.16 0.16
0.09 0.14
0.3 0.15
0.3 0.18
0.3 0.2
0.04 0.3
0.299 0.21
0.02 0.02
0.061 0.3
0.091 0.09
0.15 0.3
0.21 0.22
0.01 0.02
0.141 0.3
0.039 0.3
0.14 0.09
0.2 0.14
0.012 0.01
0.12 0.03
0.029 0.25
0.16 0.17
0.169 0.3
0.23 0.18
0.06 0.05
0.19 0.16
0.25 0.3
0.2 0.2
0.209 0.3
0.13 0.02
0.031 0.03
0.3 0.09
0.09 0.06
0.18 0.3
0.18 0.11
0.14 0.04
0.07 0.05
0.3 0.23
0.02 0.07
0.041 0.3
0.14 0.17
0.01 0.09
0.019 0.02
0.059 0.02
0.07 0.07
0.138 0.3
0.18 0.19
0.301 0.21
0.05 0.02
0.17 0.2
0.13 0.3
0.3 0.25
0.069 0.1
0.299 0.02
0.05 0.3
0.299 0.25
0.018 0.3
0.12 0.12
0.038 0.3
0.241 0.3
0.068 0.3
0.3 0.03
0.11 0.3
};
\end{axis}

\end{tikzpicture}}
  \caption{\Phragmen and \GreedyCost}
\end{subfigure}\hfill
\begin{subfigure}{0.32\textwidth}
 \resizebox{\textwidth}{!}{
\begin{tikzpicture}

\definecolor{darkgray176}{RGB}{176,176,176}

\begin{axis}[
tick align=outside,
tick pos=left,
x grid style={darkgray176},
xlabel={\MES},
xmin=-0.0078, xmax=0.3178,
xtick style={color=black},
y grid style={darkgray176},
ylabel={\Phragmen},
ymin=-0.0045, ymax=0.3145,
ytick style={color=black},
xtick={0,0.05,0.1,0.15,0.2,0.25,0.3},
xticklabels={0\%,5\%,10\%,15\%,20\%,25\%,$\ge\hspace*{-0.1cm}25\%$},
ytick={0,0.05,0.1,0.15,0.2,0.25,0.3},
yticklabels={0\%,5\%,10\%,15\%,20\%,25\%,$\ge\hspace*{-0.1cm}25\%$},
	ytick style={color=black},every tick label/.append style={font=\large}, 
label style={font=\large}
]
\addplot [draw=black, fill=black, mark=*, only marks]
table{%
x  y
0.3 0.06
0.03 0.3
0.07 0.3
0.11 0.3
0.109 0.3
0.09 0.3
0.3 0.09
0.19 0.17
0.08 0.03
0.01 0.11
0.18 0.3
0.3 0.24
0.05 0.03
0.089 0.3
0.02 0.02
0.14 0.3
0.3 0.14
0.079 0.03
0.3 0.02
0.15 0.09
0.08 0.06
0.3 0.01
0.05 0.3
0.02 0.01
0.07 0.09
0.139 0.3
0.15 0.01
0.179 0.3
0.141 0.3
0.19 0.3
0.08 0.1
0.05 0.07
0.04 0.07
0.2 0.14
0.12 0.08
0.3 0.03
0.03 0.21
0.119 0.08
0.02 0.3
0.049 0.3
0.05 0.01
0.07 0.02
0.299 0.06
0.21 0.13
0.06 0.03
0.299 0.14
0.15 0.3
0.3 0.17
0.069 0.3
0.17 0.07
0.01 0.01
0.01 0.3
0.3 0.22
0.03 0.06
0.02 0.17
0.04 0.14
0.07 0.1
0.02 0.2
0.299 0.01
0.21 0.3
0.2 0.19
0.3 0.18
0.15 0.05
0.301 0.06
0.138 0.3
0.04 0.21
0.01 0.14
0.301 0.14
0.03 0.22
0.071 0.3
0.01 0.05
0.01 0.09
0.142 0.3
0.019 0.3
0.299 0.09
0.03 0.05
0.22 0.3
0.091 0.3
0.298 0.06
0.3 0.1
0.08 0.3
0.299 0.24
0.16 0.3
0.06 0.3
0.03 0.03
0.17 0.15
0.03 0.04
0.068 0.3
0.03 0.1
0.3 0.08
0.189 0.3
0.149 0.3
0.079 0.3
0.3 0.12
0.301 0.09
0.009 0.05
0.059 0.3
0.081 0.3
0.23 0.14
0.09 0.04
0.051 0.3
0.111 0.3
0.3 0.16
0.08 0.01
0.072 0.3
0.009 0.01
0.298 0.14
0.079 0.1
0.17 0.3
0.13 0.13
0.22 0.16
0.3 0.07
0.067 0.3
0.073 0.3
0.079 0.01
0.05 0.1
0.12 0.3
0.029 0.3
0.137 0.3
0.06 0.24
0.302 0.06
0.151 0.3
0.04 0.12
0.021 0.3
0.299 0.12
0.01 0.07
0.07 0.13
0.061 0.3
0.08 0.11
0.018 0.3
0.2 0.13
0.05 0.22
0.009 0.3
0.022 0.3
0.03 0.16
0.169 0.3
0.058 0.3
0.011 0.3
0.01 0.03
0.299 0.17
0.13 0.3
0.06 0.11
0.031 0.3
0.02 0.05
0.299 0.02
0.11 0.07
0.01 0.04
0.299 0.07
0.069 0.02
0.301 0.01
0.3 0.13
0.25 0.14
0.049 0.01
0.301 0.02
0.24 0.3
0.25 0.3
0.008 0.3
0.299 0.16
0.191 0.3
0.08 0.09
0.062 0.3
0.017 0.3
0.15 0.04
0.019 0.02
0.11 0.06
0.1 0.3
0.023 0.3
0.181 0.3
0.298 0.09
0.3 0.15
0.029 0.21
0.011 0.01
0.302 0.14
0.066 0.3
0.028 0.3
0.3 0.04
0.297 0.14
0.019 0.2
0.008 0.01
0.01 0.12
0.159 0.3
0.23 0.03
0.301 0.16
0.21 0.17
0.3 0.23
0.119 0.3
0.297 0.06
0.22 0.19
0.01 0.25
0.3 0.2
0.3 0.21
0.14 0.13
0.1 0.03
0.048 0.3
0.12 0.09
0.21 0.18
0.239 0.3
0.01 0.18
0.17 0.14
0.08 0.07
0.099 0.3
0.021 0.02
0.299 0.04
0.07 0.14
0.019 0.01
0.01 0.02
0.01 0.06
0.07 0.07
0.303 0.14
0.07 0.18
0.029 0.05
0.11 0.17
0.249 0.3
0.06 0.13
0.219 0.3
0.012 0.3
0.016 0.3
0.03 0.07
0.032 0.3
0.04 0.05
0.171 0.3
0.143 0.3
0.11 0.02
0.007 0.3
0.18 0.12
0.101 0.3
0.188 0.3
0.04 0.04
0.16 0.24
0.009 0.07
0.027 0.3
0.12 0.11
0.033 0.3
};
\end{axis}

\end{tikzpicture}}
\caption{\MES and \Phragmen}
\end{subfigure}\hfill
\begin{subfigure}{0.32\textwidth}
\resizebox{\textwidth}{!}{
\begin{tikzpicture}

\definecolor{darkgray176}{RGB}{176,176,176}

\begin{axis}[
tick align=outside,
tick pos=left,
x grid style={darkgray176},
xlabel={\MES},
xmin=-0.00775, xmax=0.31675,
xtick style={color=black},
y grid style={darkgray176},
ylabel={\GreedyCost},
ymin=-0.0045, ymax=0.3145,
ytick style={color=black},
xtick={0,0.05,0.1,0.15,0.2,0.25,0.3},
xticklabels={0\%,5\%,10\%,15\%,20\%,25\%,$\ge\hspace*{-0.1cm}25\%$},
ytick={0,0.05,0.1,0.15,0.2,0.25,0.3},
yticklabels={0\%,5\%,10\%,15\%,20\%,25\%,$\ge\hspace*{-0.1cm}25\%$},
	ytick style={color=black},every tick label/.append style={font=\large}, 
label style={font=\large}
]
\addplot [draw=black, fill=black, mark=*, only marks]
table{%
x  y
0.03 0.3
0.07 0.16
0.11 0.3
0.109 0.3
0.09 0.3
0.3 0.09
0.19 0.16
0.3 0.17
0.08 0.25
0.01 0.1
0.18 0.16
0.05 0.3
0.089 0.3
0.02 0.09
0.14 0.3
0.08 0.03
0.15 0.07
0.079 0.03
0.049 0.3
0.02 0.01
0.07 0.3
0.139 0.3
0.15 0.01
0.3 0.01
0.18 0.3
0.141 0.3
0.19 0.19
0.08 0.3
0.051 0.3
0.04 0.11
0.2 0.12
0.12 0.15
0.03 0.17
0.12 0.3
0.02 0.02
0.048 0.3
0.05 0.2
0.07 0.13
0.3 0.02
0.21 0.13
0.3 0.24
0.06 0.22
0.15 0.3
0.3 0.07
0.069 0.3
0.17 0.3
0.3 0.21
0.01 0.06
0.01 0.3
0.299 0.02
0.03 0.12
0.019 0.01
0.04 0.1
0.07 0.22
0.02 0.19
0.3 0.04
0.21 0.3
0.199 0.12
0.299 0.01
0.15 0.05
0.138 0.3
0.04 0.3
0.01 0.05
0.299 0.21
0.03 0.13
0.071 0.3
0.01 0.08
0.01 0.04
0.14 0.14
0.02 0.3
0.299 0.09
0.03 0.05
0.22 0.3
0.091 0.3
0.3 0.1
0.079 0.3
0.3 0.14
0.16 0.05
0.06 0.3
0.03 0.03
0.17 0.11
0.03 0.02
0.068 0.3
0.03 0.1
0.3 0.15
0.19 0.3
0.149 0.3
0.081 0.3
0.01 0.11
0.059 0.3
0.078 0.3
0.23 0.14
0.09 0.03
0.052 0.3
0.111 0.3
0.3 0.08
0.082 0.3
0.072 0.3
0.01 0.01
0.299 0.08
0.08 0.14
0.169 0.3
0.13 0.11
0.219 0.3
0.07 0.17
0.067 0.3
0.08 0.01
0.047 0.3
0.119 0.3
0.029 0.3
0.142 0.3
0.061 0.3
0.301 0.21
0.151 0.3
0.039 0.3
0.019 0.3
0.299 0.1
0.01 0.02
0.069 0.13
0.06 0.05
0.08 0.11
0.021 0.3
0.2 0.11
0.05 0.09
0.009 0.3
0.018 0.3
0.03 0.08
0.171 0.3
0.058 0.3
0.011 0.3
0.009 0.02
0.13 0.3
0.06 0.23
0.031 0.3
0.02 0.03
0.3 0.19
0.11 0.06
0.009 0.01
0.3 0.13
0.301 0.1
0.073 0.3
0.299 0.17
0.25 0.05
0.053 0.3
0.24 0.3
0.25 0.3
0.008 0.3
0.3 0.16
0.189 0.3
0.079 0.14
0.06 0.15
0.02 0.18
0.3 0.2
0.148 0.3
0.298 0.21
0.019 0.02
0.108 0.3
0.1 0.3
0.022 0.3
0.179 0.3
0.301 0.09
0.03 0.22
0.011 0.02
0.066 0.3
0.028 0.3
0.298 0.09
0.02 0.14
0.011 0.01
0.01 0.03
0.16 0.3
0.23 0.25
0.301 0.17
0.209 0.3
0.3 0.18
0.121 0.3
0.3 0.05
0.22 0.16
0.012 0.3
0.299 0.2
0.14 0.02
0.1 0.03
0.049 0.09
0.12 0.06
0.211 0.3
0.239 0.3
0.009 0.11
0.17 0.04
0.08 0.05
0.1 0.23
0.02 0.07
0.069 0.17
0.019 0.09
0.008 0.02
0.012 0.02
0.07 0.07
0.07 0.19
0.302 0.21
0.029 0.02
0.11 0.2
0.249 0.3
0.062 0.3
0.221 0.3
0.01 0.25
0.017 0.3
0.029 0.1
0.031 0.02
0.041 0.3
0.168 0.3
0.14 0.25
0.112 0.3
0.007 0.3
0.18 0.12
0.099 0.3
0.191 0.3
0.038 0.3
0.159 0.3
0.013 0.3
0.3 0.03
0.032 0.3
0.118 0.3
0.027 0.3
};
\end{axis}

\end{tikzpicture}}
\caption{\MES and \GreedyCost}
\end{subfigure}\hfill
\begin{subfigure}{0.32\textwidth}
\resizebox{\textwidth}{!}{
\begin{tikzpicture}

\definecolor{darkgray176}{RGB}{176,176,176}

\begin{axis}[
tick align=outside,
tick pos=left,
x grid style={darkgray176},
xlabel={\MES},
xmin=-0.01205, xmax=0.27505,
xtick style={color=black},
y grid style={darkgray176},
ylabel={\Greedy},
ymin=-0.0025, ymax=0.2725,
ytick style={color=black},
xtick={0,0.05,0.1,0.15,0.2,0.25,0.3},
xticklabels={0\%,5\%,10\%,15\%,20\%,25\%,$\ge\hspace*{-0.1cm}25\%$},
ytick={0,0.05,0.1,0.15,0.2,0.25,0.3},
yticklabels={0\%,5\%,10\%,15\%,20\%,25\%,$\ge\hspace*{-0.1cm}25\%$},
	ytick style={color=black},every tick label/.append style={font=\large}, 
label style={font=\large}
]
\addplot [draw=black, fill=black, mark=*, only marks]
table{%
x  y
0.03 0.26
0.07 0.26
0.11 0.26
0.109 0.26
0.09 0.22
0.26 0.01
0.19 0.26
0.26 0.09
0.08 0.11
0.01 0.26
0.18 0.08
0.05 0.26
0.09 0.26
0.02 0.26
0.14 0.26
0.26 0.16
0.08 0.26
0.15 0.26
0.08 0.08
0.049 0.26
0.019 0.26
0.07 0.07
0.14 0.03
0.149 0.26
0.26 0.11
0.18 0.26
0.139 0.26
0.189 0.26
0.079 0.26
0.05 0.07
0.04 0.1
0.2 0.26
0.12 0.26
0.029 0.26
0.119 0.26
0.02 0.23
0.051 0.26
0.048 0.26
0.07 0.09
0.21 0.26
0.06 0.26
0.26 0.23
0.151 0.26
0.07 0.06
0.17 0.26
0.009 0.26
0.011 0.26
0.03 0.03
0.02 0.04
0.04 0.18
0.07 0.21
0.02 0.02
0.21 0.22
0.199 0.26
0.15 0.11
0.14 0.23
0.04 0.13
0.008 0.26
0.26 0.24
0.03 0.2
0.069 0.26
0.01 0.19
0.012 0.26
0.14 0.24
0.021 0.26
0.26 0.14
0.031 0.26
0.22 0.26
0.089 0.26
0.081 0.26
0.16 0.26
0.059 0.26
0.028 0.26
0.17 0.06
0.032 0.26
0.071 0.26
0.03 0.06
0.259 0.14
0.191 0.26
0.148 0.26
0.078 0.26
0.261 0.14
0.007 0.26
0.061 0.26
0.082 0.26
0.26 0.12
0.23 0.26
0.091 0.26
0.052 0.26
0.111 0.26
0.08 0.06
0.07 0.01
0.258 0.14
0.013 0.26
0.26 0.22
0.077 0.26
0.169 0.26
0.13 0.24
0.219 0.26
0.259 0.01
0.26 0.03
0.068 0.26
0.259 0.16
0.072 0.26
0.083 0.26
0.047 0.26
0.12 0.01
0.027 0.26
0.141 0.26
0.058 0.26
0.152 0.26
0.04 0.02
0.018 0.26
0.006 0.26
0.067 0.26
0.062 0.26
0.076 0.26
0.022 0.26
0.201 0.26
0.05 0.01
0.014 0.26
0.017 0.26
0.033 0.26
0.171 0.26
0.057 0.26
0.01 0.14
0.005 0.26
0.13 0.26
0.063 0.26
0.026 0.26
0.023 0.26
0.108 0.26
0.26 0.19
0.01 0.03
0.259 0.19
0.073 0.26
0.25 0.26
0.053 0.26
0.24 0.26
0.25 0.01
0.015 0.26
0.188 0.26
0.084 0.26
0.056 0.26
0.016 0.26
0.147 0.26
0.024 0.26
0.262 0.14
0.112 0.26
0.1 0.26
0.26 0.18
0.015 0.26
0.18 0.09
0.03 0.16
0.004 0.26
0.26 0.15
0.07 0.13
0.034 0.26
0.26 0.02
0.26 0.08
0.025 0.26
0.016 0.26
0.003 0.26
0.159 0.26
0.229 0.26
0.26 0.21
0.209 0.26
0.259 0.08
0.259 0.11
0.121 0.26
0.259 0.02
0.221 0.26
0.017 0.26
0.138 0.26
0.099 0.26
0.046 0.26
0.118 0.26
0.211 0.26
0.239 0.26
0.01 0.04
0.259 0.18
0.168 0.26
0.075 0.26
0.101 0.26
0.014 0.26
0.066 0.26
0.026 0.26
0.002 0.26
0.01 0.07
0.261 0.19
0.074 0.26
0.258 0.19
0.07 0.15
0.259 0.12
0.025 0.26
0.11 0.13
0.249 0.26
0.064 0.26
0.218 0.26
0.009 0.19
0.013 0.26
0.035 0.26
0.024 0.26
0.04 0.24
0.172 0.26
0.142 0.26
0.107 0.26
0.018 0.26
0.179 0.26
0.098 0.26
0.192 0.26
0.04 0.06
0.161 0.26
0.000999999999999999 0.26
0.03 0.18
0.261 0.01
0.122 0.26
0.029 0.06
};
\end{axis}

\end{tikzpicture}}
\caption{\MES and \Greedy}
\end{subfigure}
\caption{Correlation between $50\%$-winner thresholds for pairs of budgeting rules.} 
\label{fig:ov_corr}
\end{figure}
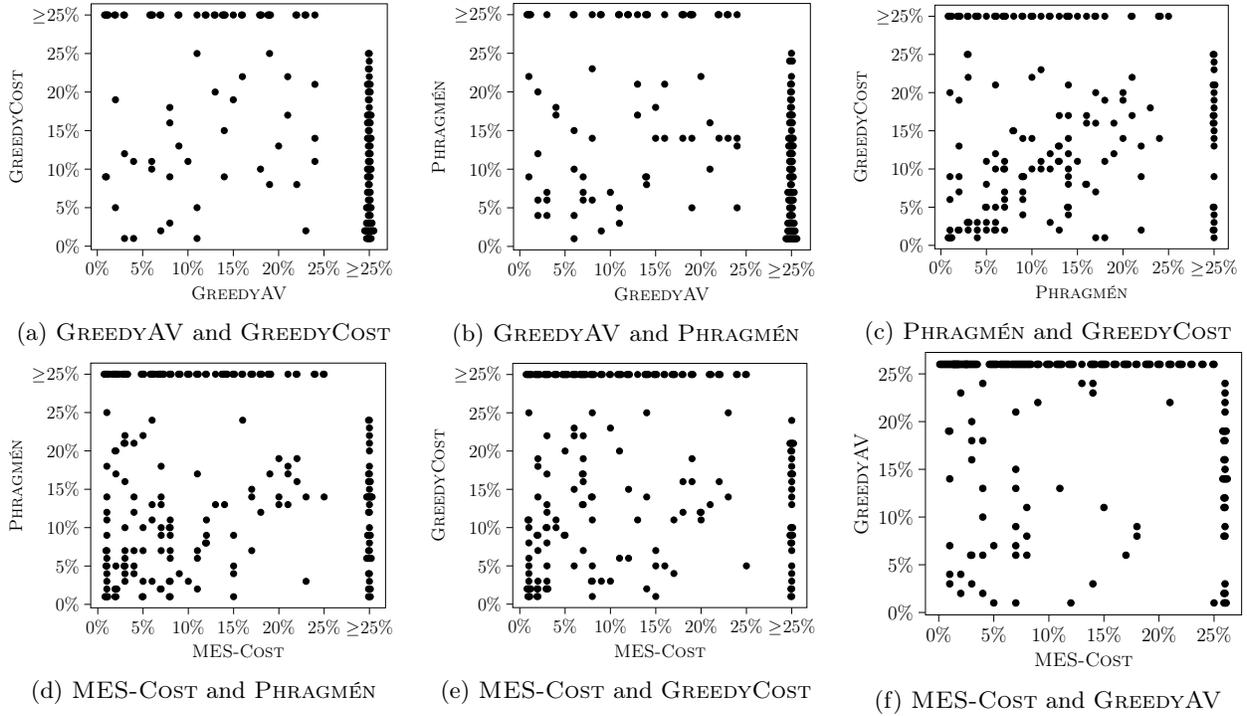 
\subsubsection{Correlation Plots.}
In \Cref{fig:ov_corr}, we show scatterplots depicting the correlation
between the robustness of outcomes for different budgeting rules.  In
these plots, each point represents one instance and its $x$- and
$y$-coordinates depict the instance's $50\%$-winner threshold under
the respective rule.  Points are slightly perturbed if they were to
overlap.  Note that as we only examine resampling probabilities up
until $25\%$, we group instances with a larger threshold together.
Instances with a $50\%$-winner threshold above $25\%$ for both rules
are omitted.  What we see in \Cref{fig:ov_corr} is that there is only
a loose correlation between the robustness of outcomes for all pairs
of rules. In particular, for all pairs there are instances which are
very non-robust under one rule (with a $50\%$-winner threshold of only
$1\%$), yet very robust under the other rule (with a $50\%$-winner
threshold above $25\%$).  If we restrict our attention to instances
with a $50\%$-winner threshold smaller or equal to $25\%$ for both
instances, some (weak) correlation can be observed.  In particular,
for \Greedy and \GreedyCost and for \GreedyCost and \Phragmen some
correlation is visible.  In contrast, for instance for \MES and
\Greedy we find no (positive) correlation.

\begin{figure}[t!]
	\centering 
	\begin{subfigure}{0.245\textwidth}
		\includegraphics[width=\textwidth]{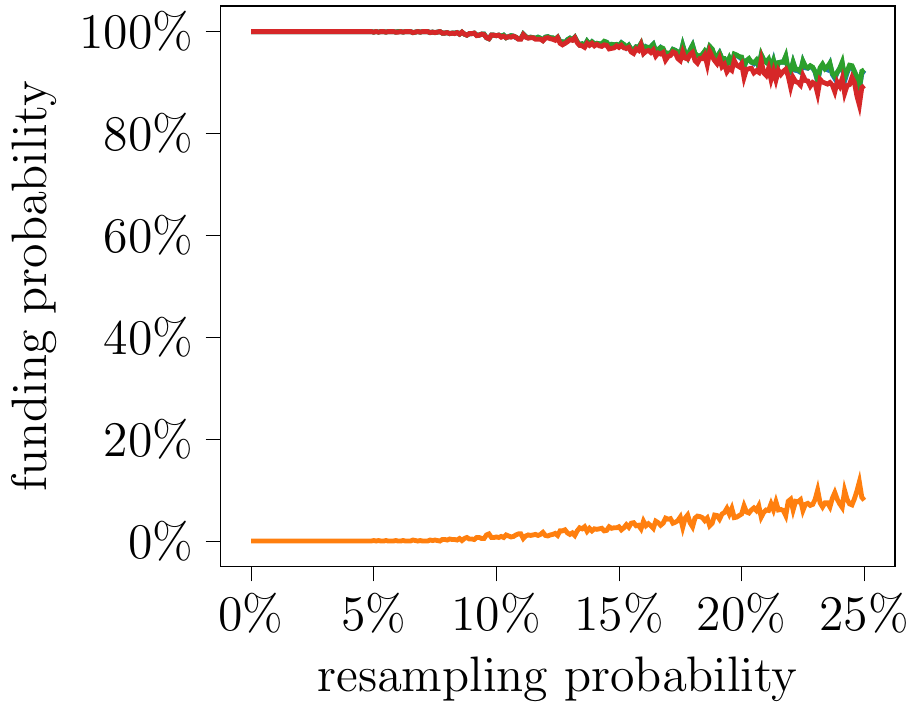}
		\caption{\Greedy}
	\end{subfigure}\hfill
\begin{subfigure}{0.245\textwidth}
\includegraphics[width=\textwidth]{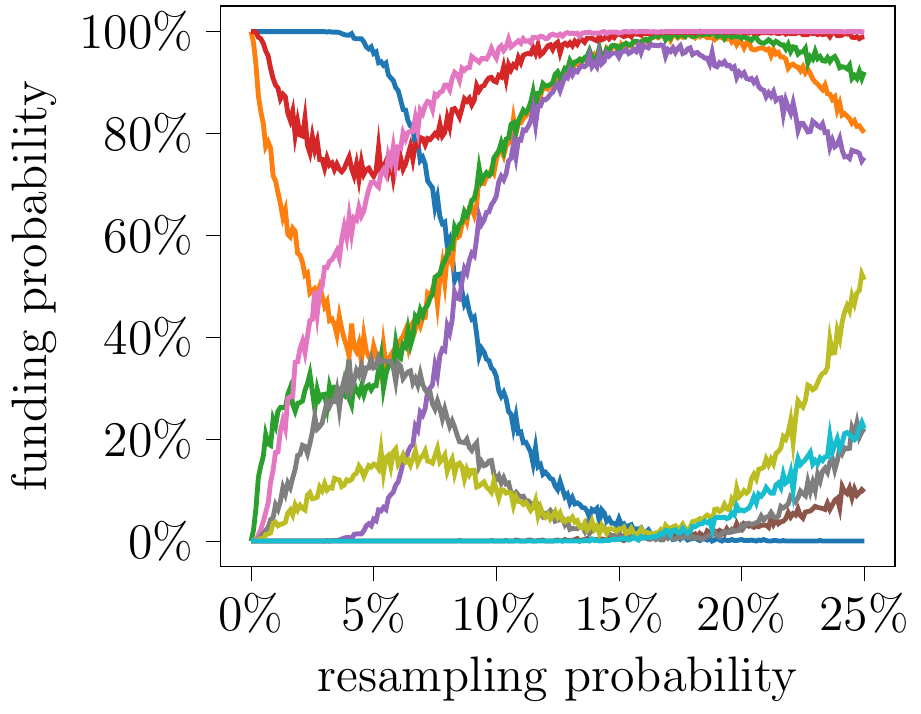}
\caption{\GreedyCost}
\end{subfigure}\hfill%
\begin{subfigure}{0.245\textwidth}
	\includegraphics[width=\textwidth]{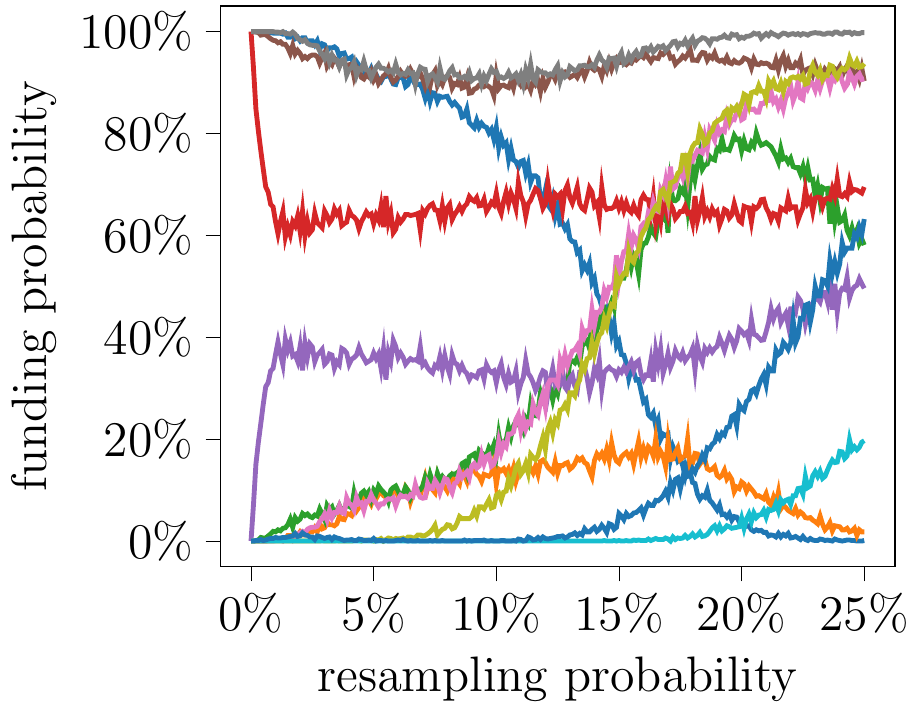}
	\caption{\Phragmen}
\end{subfigure}\hfill
\begin{subfigure}{0.245\textwidth}
	\includegraphics[width=\textwidth]{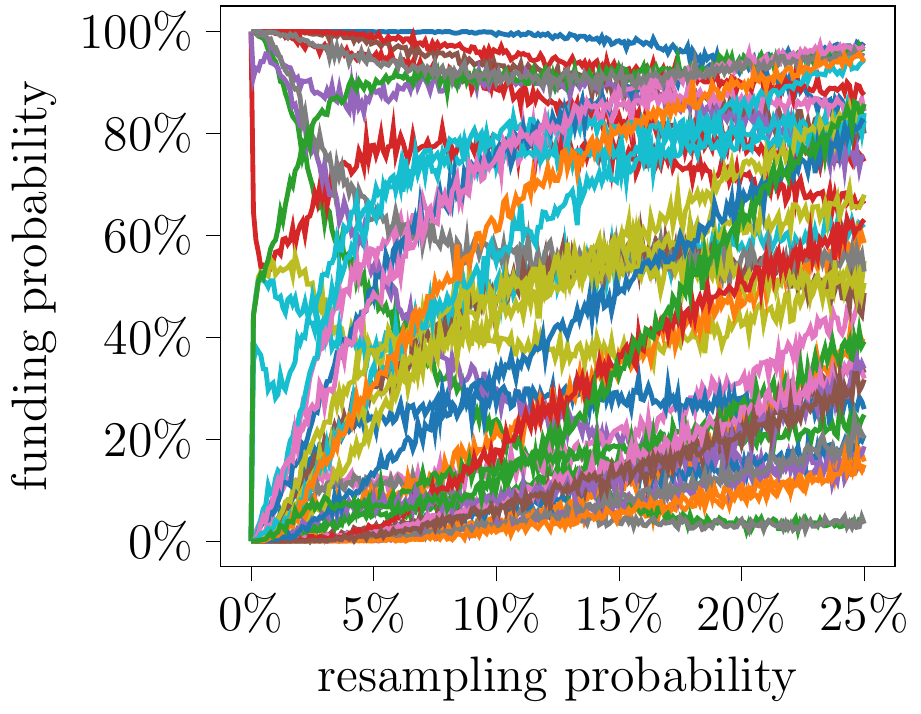}
	\caption{\MES}
\end{subfigure} 
	\caption{Robustness of outcomes on instance from Wawer 2020 (Warszawa) for different budgeting rules.} \label{plot:Wawer2020}
\end{figure}
\subsubsection{Instance-Based View.}
To get a better understanding of how different rules can behave on one instance, we give some concrete examples in the following. 
In \Cref{plot:Wawer2020}, we take a closer look at the outcomes produced by the different budgeting rules on the participatory budgeting instance from  Wawer 2020 in Warszawa. 
We see a drastic difference between the rules here. 
On the one extreme, the outcome produced by the \Greedy rule is obviously very robust. 
On the other extreme, we have \MES for which the behavior is quite chaotic with the funding probabilities of numerous projects substantially changing even for a small resampling probability.
\Phragmen and \GreedyCost are between the two extremes, yet both behave quite differently: 
For \Phragmen, we see that it was quite close whether the red or the purple project was funded in the initial outcome, as their winning probabilities substantially change already for a small resampling probability. 
All other funding decisions are more robust and only at a resampling probability of $15\%$ does the funding probability of an initially funded projects drops below $50\%$. 
In contrast, for \GreedyCost the funding probabilities of more projects are already subjected to substantial change for a small resampling probability. 
However, for \GreedyCost there is no clear trade-off between projects as visible for \GreedyCost.

\begin{figure}[t!]
	\centering 
	\begin{subfigure}{0.3\textwidth}
 \centering
		\includegraphics[width=0.8\textwidth]{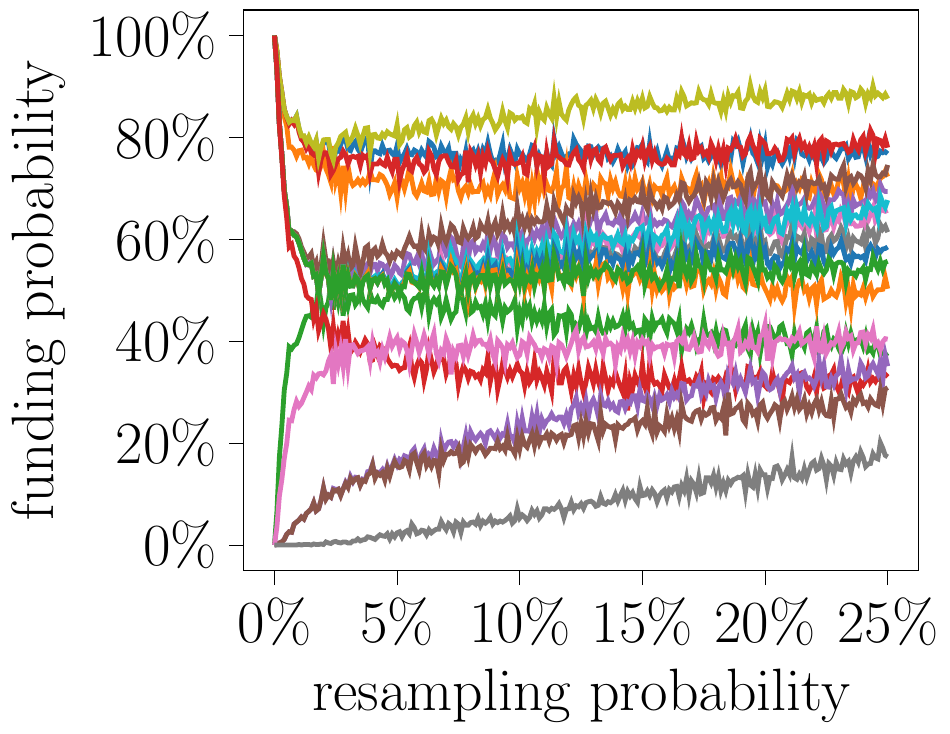}
		\caption{Sielce 2019 (Warszawa) for \Greedy}\label{fig:diffbeh1}
	\end{subfigure}	\hfill
	\begin{subfigure}{0.3\textwidth}
 \centering
		\includegraphics[width=0.8\textwidth]{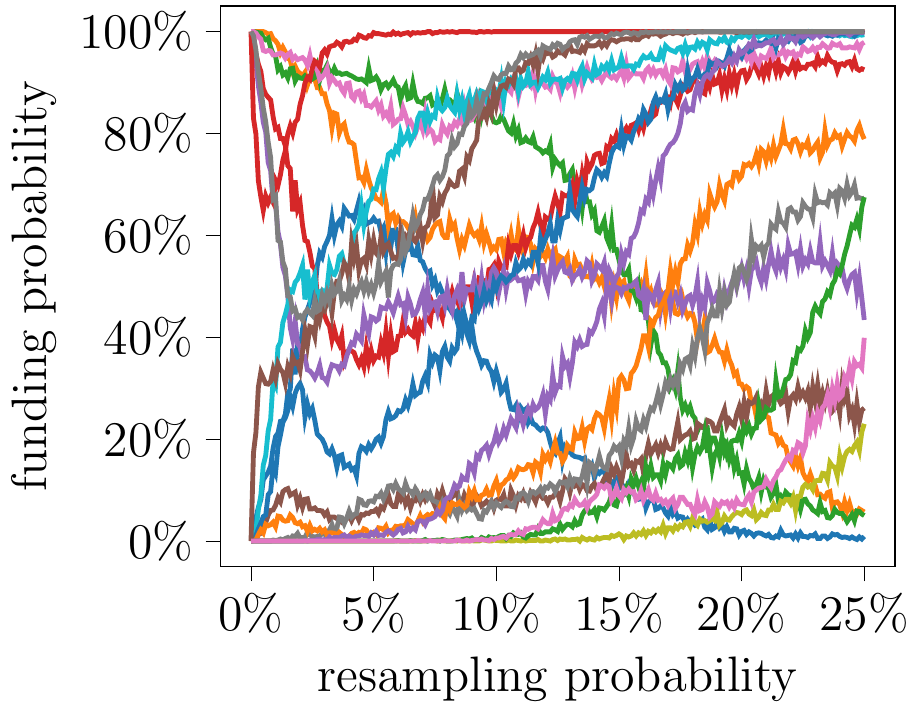}
		\caption{Ochota 2023 (Warszawa) for \Phragmen }\label{fig:diffbeh2}
	\end{subfigure} \hfill
	\begin{subfigure}{0.3\textwidth}
 \centering
		\includegraphics[width=0.8\textwidth]{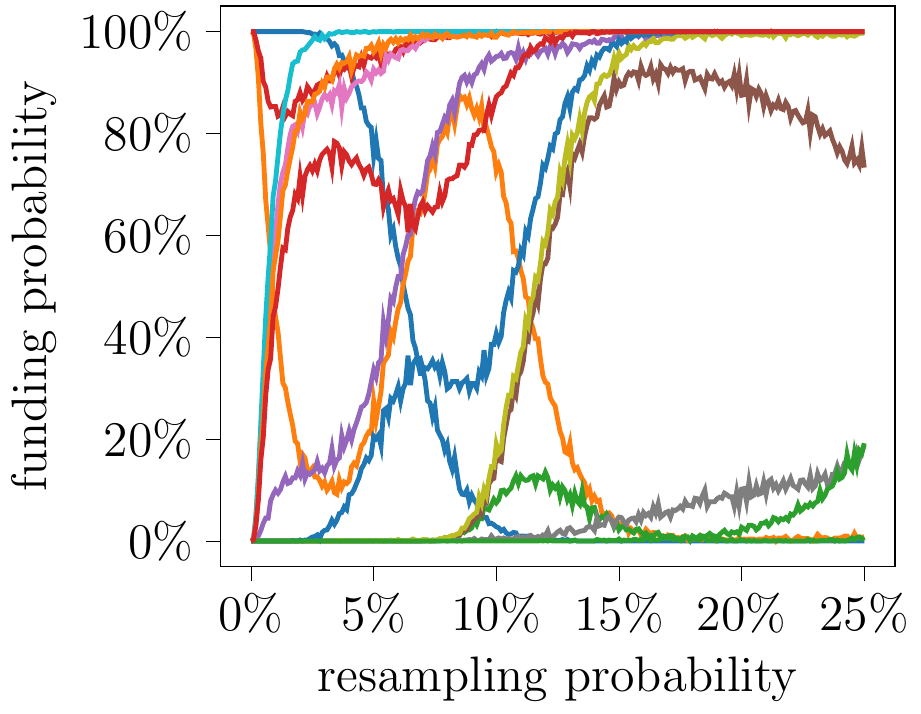}
		\caption{Ochota 2023 (Warszawa) for \GreedyCost}\label{fig:diffbeh3}
	\end{subfigure}
	\caption{Exemplaric instances where different rules produce different outcomes. On Sielce 2019, all rules except \Greedy have a $50\%$-winner threshold greater than $25\%$ and on Ochota 2023 \MES and \Greedy have a 50\%-winner threshold   greater than  $25\%$. }\label{fig:diffbeh} 
\end{figure}

The behavior of our rules in \Cref{plot:Wawer2020} is somewhat in line
with our general observations from above, e.g., in that \Greedy
produces the most, and \MES the least robust results.  However, there
are also examples differing from this trend.  For instance, in
\Cref{fig:diffbeh1} we show the results of the \Greedy rule on the
instance from Sielce in 2019.  On this instance, the funding decisions
for many of the projects were quite non-robust with drastic changes
happening already for small values of the resampling probability.  In
contrast, the outcome returned by the other three rules is very
robust: Within the considered range of the resampling probability
between $0\%$ and $25\%$ for \GreedyCost/\Phragmen/\MES the maximum
change of the funding probability of one project is $1\%/5\%/20\%$.
The reason for the non-robustness of the outcome of the \Greedy rule
on this instance is that there were multiple projects with roughly the
same number of approvals, i.e., $9$ projects with between $290$ and
$303$ approvals, which leads to many possibilities for how small
changes can alter the outcome.  Part of the reason why the outcome for
the other three rules is more robust is that these projects with
roughly the same number of approvals have significantly varying costs,
which, for instance, implies that the ``gaps'' between projects in the
ordering for \GreedyCost are quite large.  This highlights that while
\Greedy generally produces the most robust results, there are also
clear exceptions.  Moreover, there are also examples where \Phragmen
and \GreedyCost lead to substantially less robust outcomes than the
other rules: In \Cref{fig:diffbeh2,fig:diffbeh3} we see the behavior
of \Phragmen and \GreedyCost on the instance from Ochota 2023, where
the funding probabilities of multiple projects drastically change in
case of few random changes.  In contrast, on this instance \Greedy and
\MES produce much more robust results, as both have a $50\%$-winner
threshold larger than $25\%$.

\begin{figure}[t!]
	\centering 
	\begin{subfigure}{0.32\textwidth}
		\resizebox{\textwidth}{!}{\begin{tikzpicture}

\definecolor{darkgray176}{RGB}{176,176,176}

\begin{axis}[
xmode=log,
tick align=outside,
tick pos=left,
x grid style={darkgray176},
xmin=-4231.55, xmax=93878.55,
xtick style={color=black},
y grid style={darkgray176},
ymin=-0.002, ymax=0.262,
ytick style={color=black},
ytick={0,0.05,0.1,0.15,0.2,0.25},
yticklabels={0\%,5\%,10\%,15\%,20\%,25\%},
	ytick style={color=black},every tick label/.append style={font=\Large}, 
label style={font=\Large},
xlabel={number of voters},
ylabel={50\%-winner threshold}
]
\addplot [draw=black, fill=black, mark=*, only marks]
table{%
x  y
1147 0.03
580 0.07
1119 0.11
743 0.11
1551 0.09
1909 0.19
791 0.08
62529 0.01
1470 0.18
777 0.05
397 0.09
5045 0.02
510 0.14
10568 0.08
5783 0.15
2358 0.08
826 0.05
2361 0.02
1515 0.07
972 0.14
4392 0.15
851 0.18
1063 0.14
7493 0.19
2119 0.08
2827 0.05
1172 0.04
1302 0.2
6260 0.12
9173 0.03
1357 0.12
2680 0.02
11402 0.05
5765 0.05
1572 0.07
4377 0.21
7683 0.06
2149 0.15
752 0.07
8171 0.17
3990 0.01
6134 0.01
10756 0.03
8003 0.02
2468 0.04
7405 0.07
4519 0.02
1205 0.21
1147 0.2
3628 0.15
2220 0.14
2701 0.04
1288 0.01
4197 0.03
2967 0.07
4255 0.01
6603 0.01
3197 0.14
1487 0.02
3007 0.03
1506 0.22
1134 0.09
2784 0.08
861 0.16
48839 0.06
5180 0.03
5867 0.17
11063 0.03
4313 0.07
7477 0.03
2751 0.19
494 0.15
6760 0.08
2082 0.01
586 0.06
1746 0.08
1055 0.23
40092 0.09
1733 0.05
545 0.11
1601 0.08
2093 0.07
9375 0.01
4930 0.08
89419 0.17
1811 0.13
1517 0.22
67103 0.07
674 0.07
10105 0.08
1006 0.05
3371 0.12
3157 0.03
53801 0.14
1448 0.06
1107 0.15
2333 0.04
585 0.02
4641 0.01
1518 0.07
1078 0.06
1189 0.08
5201 0.02
4400 0.2
5220 0.05
2006 0.01
1472 0.02
80643 0.03
2777 0.17
2260 0.06
1774 0.01
5452 0.01
2242 0.13
4940 0.06
1684 0.03
5495 0.02
88745 0.11
8647 0.01
3112 0.07
1622 0.25
1779 0.05
1588 0.24
1560 0.25
5091 0.01
856 0.19
5349 0.08
1438 0.06
2413 0.02
3507 0.15
5438 0.02
1158 0.11
976 0.1
1735 0.02
2374 0.18
512 0.03
6672 0.01
2438 0.07
2106 0.03
689 0.02
2157 0.01
4956 0.01
1326 0.16
1717 0.23
1069 0.21
5318 0.12
1975 0.22
1578 0.01
4091 0.14
1275 0.1
1751 0.05
5520 0.12
1667 0.21
3045 0.24
2614 0.01
3819 0.17
5624 0.08
11332 0.1
4697 0.02
3779 0.07
8172 0.02
2818 0.01
4662 0.01
228 0.07
1391 0.07
2948 0.03
3995 0.11
1485 0.25
1534 0.06
793 0.22
3005 0.01
1131 0.02
2600 0.03
4849 0.03
3495 0.04
457 0.17
6828 0.14
3883 0.11
1175 0.01
5142 0.18
12546 0.1
713 0.19
1320 0.04
2207 0.16
2136 0.01
4262 0.03
2736 0.12
1356 0.03
};
\end{axis}

\end{tikzpicture}}
		\caption{Number of votes}
	\end{subfigure}\hfill
	\begin{subfigure}{0.32\textwidth}
		\resizebox{\textwidth}{!}{\begin{tikzpicture}

\definecolor{darkgray176}{RGB}{176,176,176}

\begin{axis}[
tick align=outside,
tick pos=left,
x grid style={darkgray176},
xmin=-5.1, xmax=151.1,
xtick style={color=black},
y grid style={darkgray176},
ymin=-0.002, ymax=0.262,
ytick style={color=black},
ytick={0,0.05,0.1,0.15,0.2,0.25},
yticklabels={0\%,5\%,10\%,15\%,20\%,25\%},
	ytick style={color=black},every tick label/.append style={font=\Large}, 
label style={font=\Large},
xlabel={number of projects},
ylabel={50\%-winner threshold}
]
\addplot [draw=black, fill=black, mark=*, only marks]
table{%
x  y
15 0.03
19 0.07
23 0.11
12 0.11
20 0.09
4 0.19
26 0.08
50 0.01
23 0.18
21 0.05
8 0.09
111 0.02
17 0.14
106 0.08
17 0.15
35 0.08
15 0.05
9 0.02
41 0.07
24 0.14
81 0.15
15 0.18
11 0.14
46 0.19
37 0.08
34 0.05
38 0.04
11 0.2
64 0.12
72 0.03
26 0.12
39 0.02
14 0.05
57 0.05
29 0.07
53 0.21
58 0.06
24 0.15
14 0.07
43 0.17
42 0.01
34 0.01
113 0.03
108 0.02
54 0.04
79 0.07
72 0.02
15 0.21
14 0.2
75 0.15
43 0.14
29 0.04
66 0.01
33 0.03
20 0.07
65 0.01
109 0.01
38 0.14
40 0.02
59 0.03
17 0.22
19.9985 0.09
31 0.08
21 0.16
18 0.06
83 0.03
48 0.17
110 0.03
22 0.07
44 0.03
21 0.19
11 0.15
67 0.08
33.9985 0.01
5 0.06
14 0.08
5 0.23
11 0.09
5 0.05
8 0.11
39 0.08
18 0.07
12 0.01
54 0.08
76 0.17
53 0.13
31 0.22
52 0.07
4 0.07
121 0.08
25 0.05
40 0.12
32 0.03
39 0.14
22 0.06
21 0.15
41 0.04
17 0.02
55 0.01
27 0.07
29 0.06
11 0.08
62 0.02
15 0.2
48 0.05
18 0.01
21 0.02
144 0.03
26 0.17
19 0.06
44 0.01
137 0.01
32 0.13
87 0.06
17 0.03
91 0.02
123 0.11
107 0.01
28 0.07
17 0.25
31 0.05
8 0.24
27 0.25
79 0.01
12 0.19
24 0.08
24 0.06
16.9985 0.02
44 0.15
88 0.02
13 0.11
15 0.1
18 0.02
45 0.18
19 0.03
106.9985 0.01
35 0.07
31.9985 0.03
6 0.02
14 0.01
98 0.01
20.9985 0.16
29 0.23
18 0.21
25 0.12
15 0.22
25 0.01
21 0.14
6 0.1
27 0.05
12 0.12
27 0.21
5 0.24
90 0.01
32 0.17
12 0.08
12 0.1
67 0.02
51 0.07
132 0.02
91 0.01
100 0.01
2 0.07
27.9985 0.07
68 0.03
41 0.11
15 0.25
26 0.06
16.9985 0.22
24 0.01
30 0.02
43 0.03
63 0.03
48 0.04
17 0.17
23.9985 0.14
47 0.11
17 0.01
12 0.18
8 0.1
15 0.19
44 0.04
34 0.16
23 0.01
23 0.03
39 0.12
22 0.03
};
\end{axis}

\end{tikzpicture}}
		\caption{Number of projects}
	\end{subfigure}\hfill
	\begin{subfigure}{0.32\textwidth}
		\resizebox{\textwidth}{!}{\begin{tikzpicture}

\begin{axis}[
xmode=log,
tick align=outside,
tick pos=left,
x grid style={white!69.0196078431373!black},
xmin=0, xmax=16794146.85,
xtick style={color=black},
y grid style={white!69.0196078431373!black},
ymin=-0.002, ymax=0.262,
ytick style={color=black},
ytick={0,0.05,0.1,0.15,0.2,0.25},
yticklabels={0\%,5\%,10\%,15\%,20\%,25\%},
	ytick style={color=black},every tick label/.append style={font=\Large}, 
label style={font=\Large},
xlabel={budget},
ylabel={50\%-winner threshold}
]
\addplot [draw=black, fill=black, mark=*, only marks]
table{%
x  y
505000 0.07
3906234 0.03
550000 0.2
143007 0.07
750000 0.18
1330900 0.06
1005065 0.01
1070000 0.01
831114 0.13
3087978 0.08
694000 0.01
5663326 0.08
411000 0.02
3989345 0.17
4398029 0.07
1700000 0.07
3200000 0.05
972679 0.05
16000000 0.17
935751 0.04
900000 0.08
1300000 0.25
361677 0.1
6067849 0.12
117063 0.11
831114 0.02
2030000 0.1
1994673 0.01
310000 0.1
898387 0.16
1200000 0.22
750000 0.12
600000 0.05
770039 0.06
450087 0.01
1590777 0.04
736557 0.12
735616 0.02
600000 0.03
750000 0.24
791850 0.08
1323000 0.24
1500000 0.01
1719224 0.14
4321791 0.02
1535183 0.17
550560 0.19
2807253 0.02
459900 0.21
721000 0.25
881160 0.12
725000 0.03
462911 0.14
152813 0.17
336250 0.18
1130000 0.04
1415000 0.01
2160895 0.03
2493341 0.01
1700000 0.02
136877 0.02
4753148 0.03
2493340.9985 0.01
700000 0.15
1000000 0.15
521100 0.05
2432952 0.01
310000 0.19
625320 0.15
200000 0.07
521100 0.01
2700000 0.02
1058150 0.08
1846775 0.07
750000 0.14
2000000 0.06
4865905 0.01
850000 0.06
710000 0.23
5614506 0.01
300000 0.23
500000 0.01
935751 0.01
9000000 0.11
870000 0.02
1100000 0.21
750000 0.08
1900000 0.19
475000 0.11
400000 0.14
827525 0.05
827838 0.09
753000 0.07
487500 0.02
625320 0.22
1310000 0.05
363734 0.05
750000 0.2
527500 0.07
9000000 0.03
1800000 0.07
300000 0.01
1662227 0.02
1000000 0.09
583109 0.14
1412893 0.18
1011308 0.04
969245 0.07
4000000 0.14
2160894.9985 0.03
1000000 0.07
1516962 0.08
900000 0.17
1100000 0.11
300000 0.08
216829 0.19
1000000 0.03
487500 0.16
4986682 0.08
5900907 0.08
4321790.9985 0.02
475000 0.05
4000000 0.01
124802 0.09
1700000 0.25
2629401 0.01
4854279 0.03
4072457 0.02
800000 0.21
4654236 0.01
936760 0.11
999999.9985 0.07
4500000 0.07
831114 0.07
730000 0.01
1039053 0.12
122468 0.14
776314 0.03
1000000 0.04
804828 0.03
150000 0.07
3823123 0.14
999999.9985 0.03
1045538 0.21
2245802 0.15
1500000 0.05
487500 0.07
3033924 0.01
4000000 0.09
700000 0.22
900000 0.22
3337317 0.15
1246670 0.03
1000000 0.05
327724 0.19
749999.9985 0.08
1037382 0.17
850340 0.08
400000 0.15
1026000 0.06
550000 0.08
4072456.9985 0.02
1200000 0.01
300000 0.02
4585180 0.06
3906233.9985 0.03
616000 0.03
740793 0.18
250000 0.06
1096900 0.16
2427140 0.02
1091081 0.01
1117635 0.12
5258802 0.01
309914 0.06
763500 0.11
4652017 0.11
310000 0.2
1011464 0.06
1300000 0.03
827838 0.13
1300000 0.02
394076 0.04
1227000 0.03
567000 0.03
1700000 0.1
625319.9985 0.15
};
\end{axis}

\end{tikzpicture}}
		\caption{Budget}
	\end{subfigure}\hfill
	\begin{subfigure}{0.32\textwidth}
		\resizebox{\textwidth}{!}{\begin{tikzpicture}
\begin{axis}[
tick align=outside,
tick pos=left,
x grid style={white!69.0196078431373!black},
xmin=0.39538582259534, xmax=13.6638977254979,
xtick style={color=black},
y grid style={white!69.0196078431373!black},
ymin=-0.002, ymax=0.262,
ytick style={color=black},
ytick={0,0.05,0.1,0.15,0.2,0.25},
yticklabels={0\%,5\%,10\%,15\%,20\%,25\%},
	ytick style={color=black},every tick label/.append style={font=\Large}, 
label style={font=\Large},
xlabel={average vote length},
ylabel={50\%-winner threshold}
]
\addplot [draw=black, fill=black, mark=*, only marks]
table{%
x  y
5.41709511568123 0.07
10.5315917924614 0.03
3.27027027027027 0.2
2.74137931034483 0.07
1 0.18
4.58831710709318 0.06
3.91048252911814 0.01
9.19500480307397 0.01
8.6322473771397 0.13
10.7239350912779 0.08
5.82074408117249 0.01
10.3467455621302 0.08
1.51959361393324 0.02
10.1307312084541 0.17
9.99527346387576 0.07
5.03014143287735 0.07
6.36270598438855 0.05
8.34665723381677 0.05
1 0.17
8.88787878787879 0.04
6.14538558786346 0.08
4.9 0.25
5.01331967213115 0.1
11.5223642172524 0.12
1.83302752293578 0.11
7.67854741089442 0.02
1.52080344332855 0.1
9.03243243243243 0.01
1 0.1
6.47621205256004 0.16
5.08167330677291 0.22
1 0.12
3.56885456885457 0.05
5.36504424778761 0.06
5.35832794656324 0.01
9.69570502431118 0.04
7.25350036845984 0.12
7.77173913043478 0.02
4.1328125 0.03
1 0.24
7.62460961898813 0.08
2.94647355163728 0.24
7.76215538847118 0.01
9.50900900900901 0.14
10.4653695545766 0.02
6.99106596499816 0.17
6.9058524173028 0.19
8.89772051536174 0.02
3.39569691300281 0.21
5.23551171393342 0.25
6.66630116959064 0.12
6.18346437931856 0.03
3.12699905926623 0.14
5.65864332603939 0.17
5.47003525264395 0.18
8.55831173639393 0.04
6.60172807303554 0.01
10.9005653475224 0.03
9.27498927498928 0.01
5.46380597014925 0.02
4.97777777777778 0.02
9.79624986373051 0.03
8.71955245781365 0.01
1 0.15
4.49744067007911 0.15
7.17818999437886 0.05
11.3618974751339 0.01
1 0.19
4.73139435414885 0.15
1 0.07
3.080259222333 0.01
11.6764083830033 0.02
6.01005747126437 0.08
6.83824586628325 0.07
1 0.14
13.0607835480932 0.06
10.6351658337119 0.01
5.0599739243807 0.06
6.07105416423995 0.23
11.58603117506 0.01
1 0.23
1 0.01
10.1110248447205 0.01
1 0.11
1 0.02
7.9370125974805 0.21
1 0.08
7.61110369678366 0.19
5.4763181411975 0.11
6.46913580246914 0.14
5.39065606361829 0.05
4.38297872340426 0.09
2.58274351196495 0.07
6.87267904509284 0.02
4.36065573770492 0.22
1.26916330468339 0.05
3.24455205811138 0.05
1 0.2
5.04521276595745 0.07
1 0.03
7.89706271500397 0.07
0.9985 0.01
10.3519267617628 0.02
6.24691358024691 0.09
5.11341970178441 0.14
8.67270429654591 0.18
9.30972696245734 0.04
5.9866220735786 0.07
1.87643352354045 0.14
10.2615332428765 0.03
7.27973748974569 0.07
9.58821034775233 0.08
5.08641005498822 0.17
4.01615074024226 0.11
0.9985 0.08
3.61851332398317 0.19
6.95312005068103 0.03
4.64808362369338 0.16
11.3603166749134 0.08
11.1444928084784 0.08
8.4850680994627 0.02
2.74610502019619 0.05
1.79054518703322 0.01
2.49370277078086 0.09
5.78855218855219 0.25
11.5092264017033 0.01
10.7969111969112 0.03
9.93375796178344 0.02
4.18174273858921 0.21
10.6720249797618 0.01
5.19574468085106 0.11
5.93575063613232 0.07
1.77628421978153 0.07
8.03960396039604 0.07
5.05640050697085 0.01
3.59007145543437 0.12
3.45294117647059 0.14
5.46199524940618 0.03
7.90570081440206 0.04
5.46153846153846 0.03
0.9985 0.07
7.69033468877072 0.14
6.41616828212002 0.03
5.51336531871145 0.21
10.2841786108049 0.15
8.07766990291262 0.05
4.52635046113307 0.07
8.28403064230996 0.01
1 0.09
6.67633487145682 0.22
1 0.22
11.549635701275 0.15
8.33576923076923 0.03
7.5051724137931 0.05
3.9053738317757 0.19
1.0015 0.08
6.88728844076341 0.17
8.93440302029259 0.08
4.28744939271255 0.15
7.71270718232044 0.06
3.16151202749141 0.08
8.77951452739978 0.02
7.01617021276596 0.01
2.38541234977207 0.02
9.43157894736842 0.06
9.16706954258088 0.03
5.00784655623365 0.03
4.83537414965986 0.18
1 0.06
3.68325791855204 0.16
9.4080548793981 0.02
4.24812734082397 0.01
4.79916938593889 0.12
11.3999192897498 0.01
4.99628942486085 0.06
2.77029360967185 0.11
9.02858614473345 0.11
0.9985 0.2
2.67235494880546 0.06
5.39601769911504 0.03
6.56868867082962 0.13
3.5342939481268 0.02
5.71530758226037 0.04
5.14618162364585 0.03
2.384561238855 0.03
1.31238969290505 0.1
3.8970189701897 0.15
};
\end{axis}

\end{tikzpicture}}
		\caption{Average number of approved projects per voter}
	\end{subfigure}\qquad \qquad \qquad 
	\begin{subfigure}{0.32\textwidth}
		\resizebox{\textwidth}{!}{\begin{tikzpicture}

\begin{axis}[
xmode=log,
tick align=outside,
tick pos=left,
x grid style={white!69.0196078431373!black},
xmin=-161.745454545455, xmax=3825.98787878788,
xtick style={color=black},
y grid style={white!69.0196078431373!black},
ymin=-0.002, ymax=0.262,
ytick style={color=black},
ytick={0,0.05,0.1,0.15,0.2,0.25},
yticklabels={0\%,5\%,10\%,15\%,20\%,25\%},
	ytick style={color=black},every tick label/.append style={font=\Large}, 
label style={font=\Large},
xlabel={ratio},
ylabel={50\%-winner threshold}
]
\addplot [draw=black, fill=black, mark=*, only marks]
table{%
x  y
111.142857142857 0.07
100.572727272727 0.03
81.9285714285714 0.2
30.5263157894737 0.07
428.5 0.18
59.9166666666667 0.06
125.208333333333 0.01
61.2352941176471 0.01
34.1698113207547 0.13
91.2962962962963 0.08
40.3181818181818 0.01
100.89552238806 0.08
114.833333333333 0.02
122.229166666667 0.17
93.7341772151899 0.07
196.045454545455 0.07
101.140350877193 0.05
83.1470588235294 0.05
1176.56578947368 0.17
30 0.04
30.4230769230769 0.08
57.7777777777778 0.25
65.0666666666667 0.1
97.8125 0.12
68.125 0.11
37.175 0.02
1568.25 0.1
65.4615384615385 0.01
212.5 0.1
64.9117647058823 0.16
88.5882352941177 0.22
460 0.12
37 0.05
118.947368421053 0.06
84.3818181818182 0.01
45.7037037037037 0.04
52.1923076923077 0.12
70.0952380952381 0.02
26.9473684210526 0.03
609 0.24
41.0512820512821 0.08
198.5 0.24
95 0.01
51.6279069767442 0.14
61.9090909090909 0.02
190.023255813953 0.17
131 0.19
45.4504504504505 0.02
59.3888888888889 0.21
95.4117647058823 0.25
70.1538461538462 0.12
127.181818181818 0.03
96.6363636363636 0.14
26.8823529411765 0.17
56.7333333333333 0.18
93.1379310344828 0.04
180.411764705882 0.01
50.9661016949153 0.03
46.62 0.01
68.7179487179487 0.02
34.4117647058824 0.02
127.402777777778 0.03
39.7956204379562 0.01
340.176470588235 0.15
89.5416666666667 0.15
57.3870967741936 0.05
29.0444444444444 0.01
477.25 0.19
79.7045454545455 0.15
114 0.07
111.444444444444 0.01
83.8870967741936 0.02
89.8064516129032 0.08
49.6785714285714 0.07
284.5 0.14
132.465517241379 0.06
60.5779816513761 0.01
59 0.06
59.2068965517241 0.23
62.3551401869159 0.01
211 0.23
781.25 0.01
19.5151515151515 0.01
721.50406504065 0.11
262.333333333333 0.02
61.7407407407407 0.21
222.875 0.08
162.891304347826 0.19
48.6521739130435 0.11
40.5 0.14
40.24 0.05
77.55 0.09
148.35 0.07
37.7 0.02
46.6470588235294 0.22
814.428571428571 0.05
55.0666666666667 0.05
293.333333333333 0.2
53.7142857142857 0.07
560.020833333333 0.03
74.0980392156863 0.07
154.071428571429 0.01
70.1044776119403 0.02
56.7 0.09
194.809523809524 0.14
52.7555555555556 0.18
30.8421052631579 0.04
116.277777777778 0.07
1379.51282051282 0.14
43.3529411764706 0.03
69.6571428571429 0.07
67.3714285714286 0.08
119.34375 0.17
61.9166666666667 0.11
108.090909090909 0.08
47.5333333333333 0.19
98.65625 0.03
41 0.16
83.5123966942149 0.08
99.6981132075472 0.08
74.1018518518518 0.02
346.6 0.05
1250.58 0.01
49.625 0.09
99 0.25
30.967032967033 0.01
62.4096385542169 0.03
60.3846153846154 0.02
80.3333333333333 0.21
80.8130841121495 0.01
97.4390243902439 0.11
54.2068965517241 0.07
1290.44230769231 0.07
36.9512195121951 0.07
63.12 0.01
212.72 0.12
30 0.14
99.0588235294118 0.03
56.9024390243902 0.04
65.8125 0.03
168.5 0.07
84.1315789473684 0.14
76.968253968254 0.03
82.5849056603774 0.21
48.3733333333333 0.15
64.8518518518518 0.05
56.2222222222222 0.07
64.4430379746835 0.01
3644.72727272727 0.09
48.9354838709677 0.22
131.666666666667 0.22
54.2222222222222 0.15
60.4651162790698 0.03
108.75 0.05
71.3333333333333 0.19
468.666666666667 0.08
106.807692307692 0.17
57.2702702702703 0.08
44.9090909090909 0.15
65.8181818181818 0.06
124.714285714286 0.08
61.7954545454545 0.02
69.1176470588235 0.01
141.941176470588 0.02
56.7816091954023 0.06
95.1858407079646 0.03
76.4666666666667 0.03
63.9130434782609 0.18
2713.27777777778 0.06
63.1428571428571 0.16
62.7638888888889 0.02
92.8695652173913 0.01
84.275 0.12
50.5714285714286 0.01
37.1724137931034 0.06
89.0769230769231 0.11
82.6170212765958 0.11
118.363636363636 0.2
117.2 0.06
61.6363636363636 0.03
70.0625 0.13
96.3888888888889 0.02
72.8125 0.04
169.931818181818 0.03
185.304347826087 0.03
944.333333333333 0.1
52.7142857142857 0.15
};
\end{axis}

\end{tikzpicture}}
		\caption{Ratio between the number of voters and projects}
	\end{subfigure}
	\caption{Correlation between 50\%-winner threshold and different quantities for \MES (for all instances with a $50\%$-winner threshold smaller or equal to $25\%$).}
	\label{fig:properties_instances}
\end{figure}
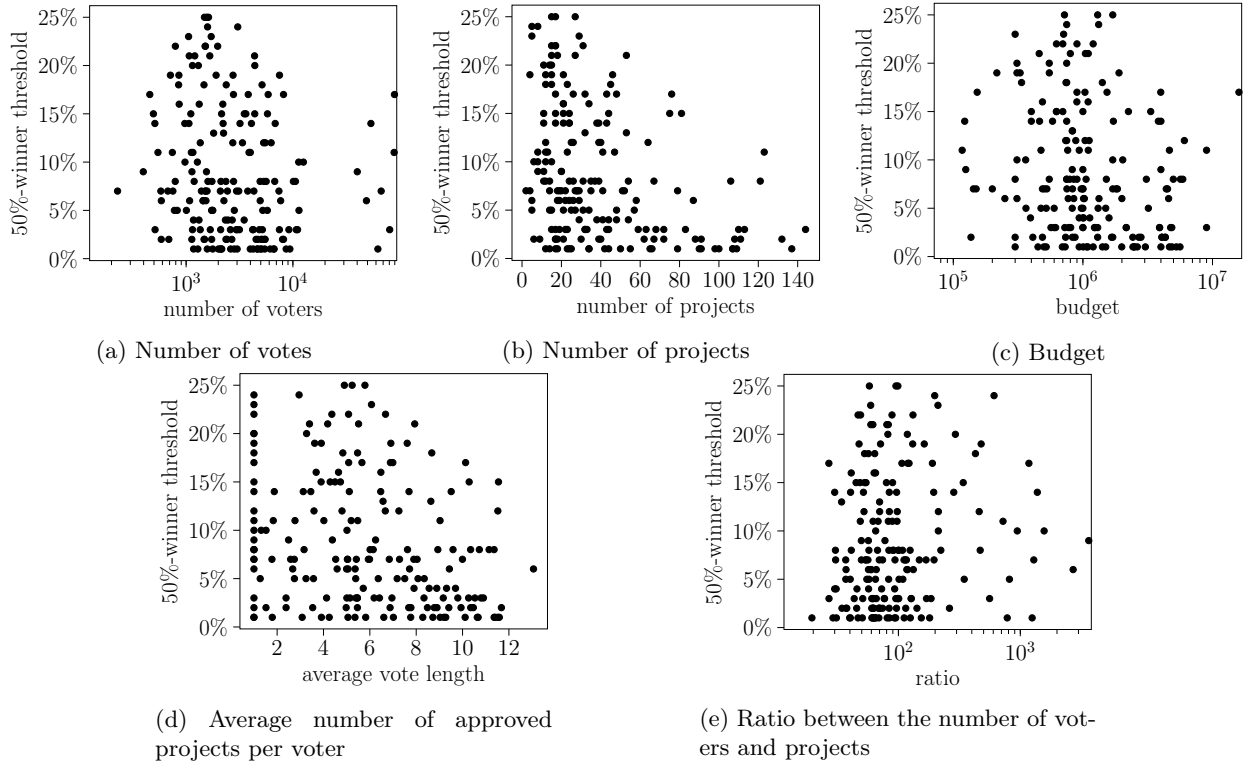

\subsection{Properties of (Non)-robust Outcomes} \label{sub:properties}

Previously, we have observed that in some real-world participatory budgeting instances funding decisions are very robust, whereas in others they might quickly be reverted if some random noise is added. 
Similarly, within one instance, some initially funded projects remain to be funded even if a substantial level of noise is introduced, whereas for others their funding probability quickly change. 
This raises the question of what are properties of instances or projects that contribute to them being (non)-robust.
We conduct a first investigation of this issue in the following, where we examine instances first and afterwards turn to projects.

\subsubsection{Instances.}
To identify properties of instances with non-robust outcomes, we again measure the robustness of an instance via the $50\%$-winner threshold.
In \Cref{fig:properties_instances}, we show the correlation between the robustness of an instance and five fundamental properties of instances: 
\begin{enumerate*}[label=(\roman*)]
\item the number of voters, 
\item the number of projects,
\item the total budget,
\item the average number of approved projects per voter, and 
\item the ratio between the number of voters and projects. 
\end{enumerate*}
In these scatterplots, each instance with a $50\%$-winner threshold smaller or equal to $25\%$ is represented by a point whose $x$-coordinate specifies the value of the properties and whose $y$-coordinate gives the instance's $50\%$-winner threshold. 
We exemplarily focus on \MES here; the general picture for all other rules is similar, yet there are naturally fewer points in the plots for the other rules, as fewer instances have a small $50\%$-winner threshold.
For all five examined properties, we see no positive correlation between it and the robustness of the instance (in fact the Pearson correlation coefficient for the depicted properties ranges between $-0.2$ and $0.2$). 
In particular, instances with non-robust outcomes do not typically fall into a specific range of the value of one of these properties. 
This observation that none of the five properties seems to have a clear influence on the robustness of outcomes might be viewed as partly surprising, as, e.g., one might expect that---in case of many projects or a small ratio of voters to projects---the difference between the approval score of projects becomes smaller, implying that decisions might be less robust. 
However, it seems that what makes outcomes of instances more or less robust is more subtle than one of these five quantities. 
This indicates that it might actually be necessary to run simulations as described in this paper to detect in which instances outcomes are robust to random noise and in which they are not (instead of resorting to measuring some simpler property of an instance).

\begin{figure}[t!]
	\centering 
	\begin{subfigure}{0.32\textwidth}
		\resizebox{\textwidth}{!}{\begin{tikzpicture}[every plot/.append style={line width=2.5pt}]

\definecolor{color0}{rgb}{1,0.549019607843137,0}
\definecolor{color1}{rgb}{0.133333333333333,0.545098039215686,0.133333333333333}
\definecolor{color2}{rgb}{0.117647058823529,0.564705882352941,1}

\begin{axis}[
legend columns=2, 
legend cell align={left},
legend style={
  fill opacity=0.8,
  draw opacity=1,
  draw=none,
  text opacity=1,
  at={(0.5,1.45)},
  line width=1.5pt,
  anchor=north,
   /tikz/column 2/.style={
  	column sep=10pt,
  }, font=\Large
},
legend image post style={line width =4.5pt},
legend entries={\Greedy,
	highest,
	\GreedyCost,
	lowest, 
	\Phragmen, {\phantom{a}},
	\MES},
tick align=outside,
tick pos=left,
x grid style={white!69.0196078431373!black},
xlabel={resampling probability},
xmin=0, xmax=0.25,
xtick style={color=black},
xtick={0,0.05,0.1,0.15,0.2,0.25},
xticklabels={0\%,5\%,10\%,15\%,20\%,25\%},
ytick={0.9,0.92,0.94,0.96,0.98,1},
yticklabels={90\%,92\%,94\%,96\%,98\%,100\%},
y grid style={white!69.0196078431373!black},
ylabel={funding probability},
ymin=0.89, ymax=1,
ytick style={color=black},
ytick style={color=black},every tick label/.append style={font=\Large}, 
label style={font=\Large}
]
\addlegendimage{red!54.5098039215686!black}
\addlegendimage{gray,dotted}
\addlegendimage{color0}
\addlegendimage{gray}
\addlegendimage{color1}
\addlegendimage{white,dashed}
\addlegendimage{color2}
\addplot [semithick, red!54.5098039215686!black]
table {%
0 0.999999999999995
0.01 0.976130434782605
0.02 0.969260869565213
0.03 0.962195652173909
0.04 0.957869565217388
0.05 0.952347826086953
0.06 0.948391304347822
0.07 0.946521739130431
0.08 0.943782608695649
0.09 0.940739130434779
0.1 0.936021739130432
0.11 0.932717391304345
0.12 0.931630434782605
0.13 0.927630434782605
0.14 0.924956521739127
0.15 0.920130434782605
0.16 0.919499999999997
0.17 0.916434782608692
0.18 0.912673913043476
0.19 0.911543478260867
0.2 0.909608695652171
0.21 0.905086956521736
0.22 0.904913043478258
0.23 0.901934782608693
0.24 0.896934782608693
0.25 0.894239130434781
};
\addplot [semithick, red!54.5098039215686!black, dotted]
table {%
0 0.999999999999995
0.01 0.999891304347821
0.02 0.999847826086952
0.03 0.999304347826082
0.04 0.999108695652169
0.05 0.999173913043473
0.06 0.998652173913039
0.07 0.998652173913039
0.08 0.998304347826082
0.09 0.998086956521734
0.1 0.99802173913043
0.11 0.997956521739126
0.12 0.997956521739126
0.13 0.997695652173908
0.14 0.997326086956517
0.15 0.997043478260865
0.16 0.996695652173908
0.17 0.99652173913043
0.18 0.99647826086956
0.19 0.996369565217387
0.2 0.995326086956517
0.21 0.995739130434778
0.22 0.995391304347821
0.23 0.99552173913043
0.24 0.994760869565213
0.25 0.994413043478256
};

\addplot [semithick, color0]
table {%
0 0.999999999999995
0.01 0.995130434782604
0.02 0.993413043478256
0.03 0.99197826086956
0.04 0.989760869565212
0.05 0.987739130434778
0.06 0.985130434782604
0.07 0.983086956521734
0.08 0.981782608695647
0.09 0.980478260869561
0.1 0.979326086956517
0.11 0.979326086956517
0.12 0.978369565217386
0.13 0.977673913043474
0.14 0.976913043478256
0.15 0.975739130434778
0.16 0.974260869565213
0.17 0.97297826086956
0.18 0.971347826086952
0.19 0.970608695652169
0.2 0.969369565217387
0.21 0.968478260869561
0.22 0.967413043478256
0.23 0.966869565217387
0.24 0.965978260869561
0.25 0.964630434782604
};
\addplot [semithick, color0, dotted]
table {%
0 0.999999999999995
0.01 0.997913043478256
0.02 0.995282608695647
0.03 0.992804347826082
0.04 0.990239130434778
0.05 0.987760869565213
0.06 0.986478260869561
0.07 0.983739130434778
0.08 0.9805652173913
0.09 0.976739130434778
0.1 0.974043478260865
0.11 0.971326086956517
0.12 0.968086956521734
0.13 0.964999999999996
0.14 0.962913043478256
0.15 0.959543478260865
0.16 0.956826086956517
0.17 0.955152173913039
0.18 0.9525652173913
0.19 0.950978260869561
0.2 0.94952173913043
0.21 0.947086956521735
0.22 0.945456521739126
0.23 0.9445652173913
0.24 0.943282608695648
0.25 0.941478260869561
};

\addplot [semithick, color1]
table {%
0 0.999999999999995
0.01 0.996869565217387
0.02 0.993782608695647
0.03 0.992173913043473
0.04 0.990630434782604
0.05 0.988434782608691
0.06 0.987347826086952
0.07 0.985782608695647
0.08 0.982739130434778
0.09 0.980826086956517
0.1 0.979869565217387
0.11 0.979173913043474
0.12 0.977760869565213
0.13 0.976304347826082
0.14 0.975499999999995
0.15 0.973978260869561
0.16 0.971847826086952
0.17 0.9700652173913
0.18 0.969804347826082
0.19 0.968891304347821
0.2 0.967739130434778
0.21 0.966543478260865
0.22 0.966499999999995
0.23 0.966434782608691
0.24 0.965434782608691
0.25 0.964413043478256
};
\addplot [semithick, color1, dotted]
table {%
0 0.999999999999995
0.01 0.998956521739126
0.02 0.996804347826082
0.03 0.995630434782604
0.04 0.99397826086956
0.05 0.992173913043473
0.06 0.990391304347821
0.07 0.988478260869561
0.08 0.985260869565213
0.09 0.9830652173913
0.1 0.981608695652169
0.11 0.979347826086952
0.12 0.9775652173913
0.13 0.975152173913039
0.14 0.973869565217387
0.15 0.971608695652169
0.16 0.968695652173909
0.17 0.966956521739126
0.18 0.964239130434778
0.19 0.962239130434778
0.2 0.959239130434778
0.21 0.957891304347822
0.22 0.954499999999995
0.23 0.953086956521735
0.24 0.950826086956517
0.25 0.949239130434778
};

\addplot [semithick, color2]
table {%
0 0.999999999999995
0.01 0.978239130434778
0.02 0.971978260869561
0.03 0.967847826086952
0.04 0.963369565217387
0.05 0.959869565217387
0.06 0.957326086956517
0.07 0.953673913043474
0.08 0.951869565217387
0.09 0.949586956521735
0.1 0.948152173913039
0.11 0.945695652173909
0.12 0.944086956521735
0.13 0.942217391304344
0.14 0.940369565217388
0.15 0.938999999999996
0.16 0.935065217391301
0.17 0.934282608695648
0.18 0.9335652173913
0.19 0.932326086956518
0.2 0.929086956521735
0.21 0.92915217391304
0.22 0.926369565217388
0.23 0.927347826086953
0.24 0.924739130434779
0.25 0.924456521739127
};
\addplot [semithick, color2, dotted]
table {%
0 0.999999999999995
0.01 0.997304347826082
0.02 0.995673913043473
0.03 0.994065217391299
0.04 0.993217391304343
0.05 0.991543478260865
0.06 0.990999999999995
0.07 0.990086956521734
0.08 0.989043478260865
0.09 0.988282608695647
0.1 0.98752173913043
0.11 0.986760869565212
0.12 0.986304347826082
0.13 0.985108695652169
0.14 0.984934782608691
0.15 0.983891304347821
0.16 0.984217391304343
0.17 0.982760869565213
0.18 0.982391304347822
0.19 0.981543478260865
0.2 0.980913043478256
0.21 0.980326086956517
0.22 0.978826086956517
0.23 0.977739130434778
0.24 0.975891304347821
0.25 0.975217391304343
};
\end{axis}

\end{tikzpicture}}
		\caption{Number of approvals}\label{fig:properties_projects1}
	\end{subfigure}\hfill
	\begin{subfigure}{0.32\textwidth}
		\resizebox{\textwidth}{!}{\begin{tikzpicture}[every plot/.append style={line width=2.5pt}]

\definecolor{color0}{rgb}{1,0.549019607843137,0}
\definecolor{color1}{rgb}{0.133333333333333,0.545098039215686,0.133333333333333}
\definecolor{color2}{rgb}{0.117647058823529,0.564705882352941,1}

\begin{axis}[
legend columns=2, 
legend cell align={left},
legend style={
  fill opacity=0.8,
  draw opacity=1,
  draw=none,
  text opacity=1,
  at={(0.5,1.45)},
  line width=1.5pt,
  anchor=north,
   /tikz/column 2/.style={
  	column sep=10pt,
}, font=\Large
},
legend image post style={line width =4.5pt},
legend entries={\Greedy,
highest,
\GreedyCost,
lowest, 
\Phragmen, {\phantom{a}},
\MES},
tick align=outside,
tick pos=left,
x grid style={white!69.0196078431373!black},
xlabel={resampling probability},
xmin=0, xmax=0.25,
xtick style={color=black},
xtick={0,0.05,0.1,0.15,0.2,0.25},
xticklabels={0\%,5\%,10\%,15\%,20\%,25\%},
ytick={0.8,0.85,0.9,0.95,1},
yticklabels={80\%,85\%,90\%,95\%,100\%},
y grid style={white!69.0196078431373!black},
ylabel={funding probability},
ymin=0.77, ymax=1,
ytick style={color=black},every tick label/.append style={font=\Large}, 
label style={font=\Large}
]
\addlegendimage{red!54.5098039215686!black}
\addlegendimage{gray,dotted}
\addlegendimage{color0}
\addlegendimage{gray}
\addlegendimage{color1}
\addlegendimage{white,dashed}
\addlegendimage{color2}
\addplot [semithick, red!54.5098039215686!black]
table {%
0 0.999999999999995
0.01 0.995760869565213
0.02 0.993826086956517
0.03 0.991195652173908
0.04 0.989391304347822
0.05 0.987391304347822
0.06 0.986239130434778
0.07 0.985847826086952
0.08 0.984695652173909
0.09 0.983673913043474
0.1 0.982586956521735
0.11 0.981434782608691
0.12 0.981217391304343
0.13 0.979543478260865
0.14 0.978543478260865
0.15 0.976673913043474
0.16 0.977804347826083
0.17 0.975499999999996
0.18 0.975130434782604
0.19 0.975869565217387
0.2 0.973869565217387
0.21 0.97252173913043
0.22 0.972304347826083
0.23 0.970239130434779
0.24 0.969391304347822
0.25 0.968999999999996
};
\addplot [semithick, red!54.5098039215686!black, dotted]
table {%
0 0.999999999999995
0.01 0.996673913043474
0.02 0.994586956521734
0.03 0.992173913043474
0.04 0.990934782608691
0.05 0.989804347826083
0.06 0.987891304347822
0.07 0.98652173913043
0.08 0.986456521739126
0.09 0.984586956521735
0.1 0.983130434782604
0.11 0.983130434782605
0.12 0.981934782608691
0.13 0.981478260869561
0.14 0.979173913043474
0.15 0.977543478260866
0.16 0.976956521739126
0.17 0.975782608695648
0.18 0.973847826086953
0.19 0.973717391304344
0.2 0.971543478260866
0.21 0.971804347826083
0.22 0.970826086956518
0.23 0.968717391304344
0.24 0.968804347826083
0.25 0.96610869565217
};

\addplot [semithick, color0]
table {%
0 0.999999999999995
0.01 0.999760869565212
0.02 0.998782608695647
0.03 0.997956521739126
0.04 0.996586956521734
0.05 0.995086956521734
0.06 0.993847826086952
0.07 0.992565217391299
0.08 0.99097826086956
0.09 0.990065217391299
0.1 0.989173913043474
0.11 0.988739130434778
0.12 0.987652173913039
0.13 0.986543478260865
0.14 0.985347826086952
0.15 0.984130434782604
0.16 0.982543478260865
0.17 0.981326086956517
0.18 0.980086956521734
0.19 0.979217391304343
0.2 0.9785652173913
0.21 0.977478260869561
0.22 0.976260869565213
0.23 0.975999999999995
0.24 0.975499999999995
0.25 0.974239130434778
};
\addplot [semithick, color0, dotted]
table {%
0 0.999999999999995
0.01 0.991282608695647
0.02 0.983108695652169
0.03 0.975543478260865
0.04 0.966173913043474
0.05 0.958999999999996
0.06 0.952413043478257
0.07 0.946217391304344
0.08 0.936978260869561
0.09 0.927260869565214
0.1 0.920369565217387
0.11 0.91265217391304
0.12 0.902826086956518
0.13 0.893499999999996
0.14 0.885956521739127
0.15 0.877282608695649
0.16 0.869304347826084
0.17 0.860521739130431
0.18 0.85265217391304
0.19 0.845782608695649
0.2 0.837978260869562
0.21 0.829260869565215
0.22 0.820913043478258
0.23 0.813282608695649
0.24 0.809326086956519
0.25 0.800891304347823
};

\addplot [semithick, color1]
table {%
0 0.999999999999995
0.01 0.999673913043473
0.02 0.998739130434778
0.03 0.997999999999995
0.04 0.997086956521734
0.05 0.995934782608691
0.06 0.994152173913039
0.07 0.993217391304343
0.08 0.991347826086952
0.09 0.990217391304343
0.1 0.989608695652169
0.11 0.989173913043473
0.12 0.988152173913039
0.13 0.986934782608691
0.14 0.986043478260865
0.15 0.9845652173913
0.16 0.982434782608691
0.17 0.981260869565213
0.18 0.980456521739126
0.19 0.97952173913043
0.2 0.978043478260865
0.21 0.977499999999995
0.22 0.976826086956517
0.23 0.976652173913039
0.24 0.975760869565213
0.25 0.974826086956517
};
\addplot [semithick, color1, dotted]
table {%
0 0.999999999999995
0.01 0.991695652173908
0.02 0.982652173913039
0.03 0.976304347826082
0.04 0.967652173913039
0.05 0.9605652173913
0.06 0.9505652173913
0.07 0.940739130434779
0.08 0.929869565217387
0.09 0.922826086956518
0.1 0.914130434782605
0.11 0.906173913043475
0.12 0.898130434782605
0.13 0.888043478260866
0.14 0.879478260869562
0.15 0.869804347826084
0.16 0.85919565217391
0.17 0.849282608695649
0.18 0.838543478260867
0.19 0.830521739130432
0.2 0.82065217391304
0.21 0.813304347826084
0.22 0.802565217391301
0.23 0.794630434782606
0.24 0.787130434782606
0.25 0.779173913043475
};

\addplot [semithick, color2]
table {%
0 0.999999999999995
0.01 0.997565217391299
0.02 0.996760869565213
0.03 0.995891304347821
0.04 0.995152173913039
0.05 0.993956521739126
0.06 0.993413043478256
0.07 0.992956521739126
0.08 0.991695652173908
0.09 0.991369565217386
0.1 0.990826086956517
0.11 0.990043478260865
0.12 0.989413043478256
0.13 0.988630434782604
0.14 0.987369565217387
0.15 0.98797826086956
0.16 0.987043478260865
0.17 0.98602173913043
0.18 0.985543478260865
0.19 0.9845652173913
0.2 0.983695652173908
0.21 0.983891304347821
0.22 0.982130434782604
0.23 0.981304347826082
0.24 0.979934782608691
0.25 0.979391304347822
};
\addplot [semithick, color2, dotted]
table {%
0 0.999999999999995
0.01 0.988195652173908
0.02 0.982652173913039
0.03 0.97902173913043
0.04 0.974369565217387
0.05 0.969108695652169
0.06 0.965086956521735
0.07 0.958695652173909
0.08 0.955217391304344
0.09 0.951347826086953
0.1 0.947304347826082
0.11 0.943782608695648
0.12 0.939434782608691
0.13 0.936347826086952
0.14 0.933130434782605
0.15 0.929956521739127
0.16 0.92765217391304
0.17 0.923826086956517
0.18 0.921260869565214
0.19 0.918260869565213
0.2 0.914826086956518
0.21 0.910217391304344
0.22 0.906434782608692
0.23 0.902608695652171
0.24 0.89910869565217
0.25 0.896108695652171
};
\end{axis}

\end{tikzpicture}}
		\caption{Project cost}\label{fig:properties_projects2}
	\end{subfigure}\hfill
	\begin{subfigure}{0.32\textwidth}
		\resizebox{\textwidth}{!}{\begin{tikzpicture}[every plot/.append style={line width=2.5pt}]

\definecolor{color0}{rgb}{1,0.549019607843137,0}
\definecolor{color1}{rgb}{0.133333333333333,0.545098039215686,0.133333333333333}
\definecolor{color2}{rgb}{0.117647058823529,0.564705882352941,1}

\begin{axis}[
legend columns=2, 
legend cell align={left},
legend style={
  fill opacity=0.8,
  draw opacity=1,
  draw=none,
  text opacity=1,
  at={(0.5,1.45)},
  line width=1.5pt,
  anchor=north,
   /tikz/column 2/.style={
  	column sep=10pt,
  }, font=\Large
},
legend image post style={line width =4.5pt},
legend entries={\Greedy,
	highest,
	\GreedyCost,
	lowest, 
	\Phragmen, {\phantom{a}},
	\MES},
tick align=outside,
tick pos=left,
x grid style={white!69.0196078431373!black},
xlabel={resampling probability},
xmin=0, xmax=0.25,
xtick style={color=black},
xtick={0,0.05,0.1,0.15,0.2,0.25},
xticklabels={0\%,5\%,10\%,15\%,20\%,25\%},
ytick={0.8,0.85,0.9,0.95,1},
yticklabels={80\%,85\%,90\%,95\%,100\%},
y grid style={white!69.0196078431373!black},
ylabel={funding probability},
ymin=0.77, ymax=1,
ytick style={color=black},
ytick style={color=black},every tick label/.append style={font=\Large}, 
label style={font=\Large}
]
\addlegendimage{red!54.5098039215686!black}
\addlegendimage{gray,dotted}
\addlegendimage{color0}
\addlegendimage{gray}
\addlegendimage{color1}
\addlegendimage{white,dashed}
\addlegendimage{color2}
\addplot [semithick, red!54.5098039215686!black]
table {%
0 0.999999999999995
0.01 0.993826086956517
0.02 0.99052173913043
0.03 0.986673913043474
0.04 0.984434782608692
0.05 0.982760869565213
0.06 0.980999999999996
0.07 0.978739130434779
0.08 0.977478260869561
0.09 0.975347826086952
0.1 0.972869565217387
0.11 0.97260869565217
0.12 0.9705652173913
0.13 0.97019565217391
0.14 0.967369565217387
0.15 0.965260869565214
0.16 0.964978260869562
0.17 0.961586956521736
0.18 0.959543478260866
0.19 0.958956521739127
0.2 0.958478260869562
0.21 0.955521739130431
0.22 0.956347826086953
0.23 0.95315217391304
0.24 0.952782608695649
0.25 0.949326086956518
};
\addplot [semithick, red!54.5098039215686!black, dotted]
table {%
0 0.999999999999995
0.01 0.999108695652169
0.02 0.998630434782604
0.03 0.997282608695647
0.04 0.997173913043474
0.05 0.996108695652169
0.06 0.994913043478256
0.07 0.995304347826082
0.08 0.994608695652169
0.09 0.993782608695648
0.1 0.993173913043473
0.11 0.993152173913039
0.12 0.992999999999995
0.13 0.991673913043474
0.14 0.991586956521735
0.15 0.990760869565213
0.16 0.990739130434778
0.17 0.990282608695648
0.18 0.989652173913039
0.19 0.989717391304343
0.2 0.988369565217387
0.21 0.988043478260865
0.22 0.987999999999996
0.23 0.986913043478257
0.24 0.985934782608691
0.25 0.9865652173913
};

\addplot [semithick, color0]
table {%
0 0.999999999999995
0.01 0.976630434782604
0.02 0.9610652173913
0.03 0.947695652173909
0.04 0.93702173913043
0.05 0.931413043478257
0.06 0.924826086956518
0.07 0.917891304347822
0.08 0.913282608695649
0.09 0.905065217391301
0.1 0.89965217391304
0.11 0.893978260869562
0.12 0.888065217391301
0.13 0.883565217391301
0.14 0.879347826086953
0.15 0.874565217391301
0.16 0.870891304347823
0.17 0.866630434782605
0.18 0.862695652173909
0.19 0.86019565217391
0.2 0.85665217391304
0.21 0.854782608695649
0.22 0.854478260869562
0.23 0.853717391304345
0.24 0.853217391304344
0.25 0.85115217391304
};
\addplot [semithick, color0, dotted]
table {%
0 0.999999999999995
0.01 0.999760869565212
0.02 0.998782608695647
0.03 0.99802173913043
0.04 0.99697826086956
0.05 0.996369565217387
0.06 0.9955652173913
0.07 0.994586956521734
0.08 0.993695652173908
0.09 0.992847826086952
0.1 0.992260869565213
0.11 0.9920652173913
0.12 0.991434782608691
0.13 0.99047826086956
0.14 0.989717391304343
0.15 0.988652173913039
0.16 0.987434782608691
0.17 0.986652173913039
0.18 0.985934782608691
0.19 0.985217391304343
0.2 0.984804347826082
0.21 0.984239130434778
0.22 0.983586956521734
0.23 0.983826086956517
0.24 0.9835652173913
0.25 0.982782608695647
};

\addplot [semithick, color1]
table {%
0 0.999999999999995
0.01 0.981934782608691
0.02 0.968304347826082
0.03 0.957999999999996
0.04 0.947021739130431
0.05 0.937304347826083
0.06 0.928173913043475
0.07 0.916891304347822
0.08 0.907717391304344
0.09 0.899347826086953
0.1 0.891086956521736
0.11 0.885130434782605
0.12 0.878739130434779
0.13 0.869826086956518
0.14 0.862434782608693
0.15 0.856130434782606
0.16 0.847499999999997
0.17 0.841282608695649
0.18 0.836347826086954
0.19 0.833543478260867
0.2 0.831021739130432
0.21 0.827847826086953
0.22 0.825217391304345
0.23 0.822847826086954
0.24 0.818391304347823
0.25 0.814717391304345
};
\addplot [semithick, color1, dotted]
table {%
0 0.999999999999995
0.01 0.999673913043473
0.02 0.998739130434778
0.03 0.997999999999995
0.04 0.997239130434778
0.05 0.996413043478256
0.06 0.995391304347821
0.07 0.994804347826082
0.08 0.993652173913039
0.09 0.992695652173908
0.1 0.99252173913043
0.11 0.992173913043473
0.12 0.991434782608691
0.13 0.990434782608691
0.14 0.989869565217387
0.15 0.988586956521734
0.16 0.987152173913039
0.17 0.986826086956517
0.18 0.986043478260865
0.19 0.985369565217387
0.2 0.984413043478256
0.21 0.984152173913039
0.22 0.983804347826082
0.23 0.983760869565213
0.24 0.983347826086952
0.25 0.982717391304343
};

\addplot [semithick, color2]
table {%
0 0.999999999999995
0.01 0.971956521739126
0.02 0.959586956521735
0.03 0.951630434782604
0.04 0.94260869565217
0.05 0.935913043478257
0.06 0.92810869565217
0.07 0.919978260869562
0.08 0.913434782608693
0.09 0.905347826086953
0.1 0.89815217391304
0.11 0.892717391304344
0.12 0.885195652173909
0.13 0.879347826086953
0.14 0.87415217391304
0.15 0.867260869565214
0.16 0.862152173913041
0.17 0.855021739130431
0.18 0.852086956521736
0.19 0.847260869565214
0.2 0.841086956521736
0.21 0.835847826086953
0.22 0.829804347826084
0.23 0.824717391304345
0.24 0.820086956521736
0.25 0.816347826086954
};
\addplot [semithick, color2, dotted]
table {%
0 0.999999999999995
0.01 0.999804347826082
0.02 0.999239130434778
0.03 0.998760869565212
0.04 0.998347826086952
0.05 0.997565217391299
0.06 0.997260869565213
0.07 0.996804347826082
0.08 0.996217391304343
0.09 0.995847826086952
0.1 0.99552173913043
0.11 0.994869565217386
0.12 0.9945652173913
0.13 0.993847826086952
0.14 0.993413043478256
0.15 0.993326086956517
0.16 0.99297826086956
0.17 0.992391304347821
0.18 0.9920652173913
0.19 0.991717391304343
0.2 0.991260869565213
0.21 0.99152173913043
0.22 0.990217391304343
0.23 0.989652173913039
0.24 0.988739130434778
0.25 0.987956521739126
};
\end{axis}

\end{tikzpicture}}
		\caption{Ratio of the number of approvals and project cost} \label{fig:properties_projects3}
	\end{subfigure}
	\caption{Average funding probability of initially funded project with highest/lowest property value for three different properties on the full dataset.}
	\label{fig:properties_projects}
\end{figure}
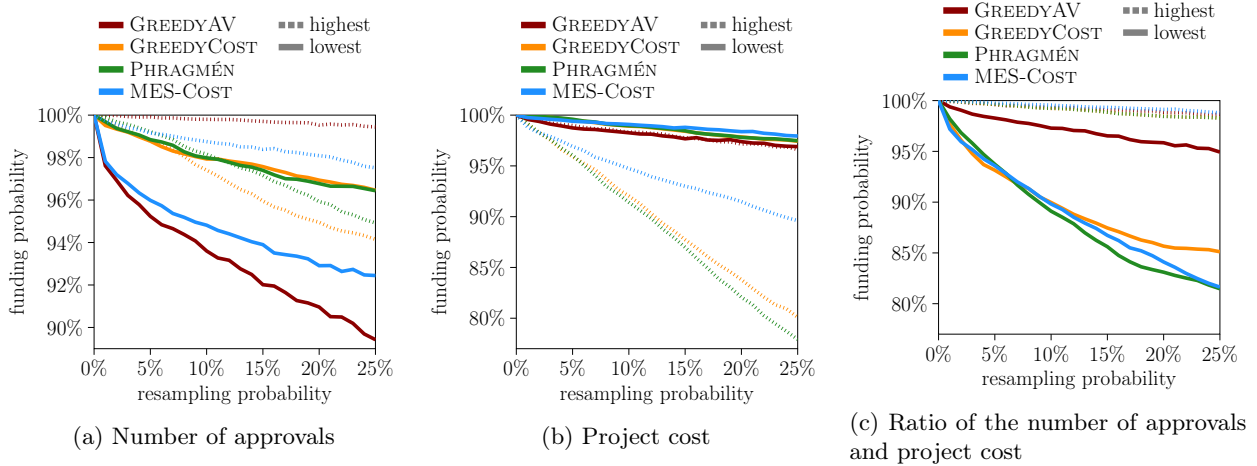 

\subsubsection{Projects.}
We now shift our focus to single projects and analyze the effect that different properties of projects have on whether an initially funded project continues to be funded. 
In particular, we consider \begin{enumerate*}[label=(\roman*)]
\item the number of approvals of a project, 
\item the project's cost, and
\item the ratio between the number of approvals and cost.
\end{enumerate*}
To keep the setup simple, we always focus on the two initially funded projects with lowest or highest property value and compare how their funding probability changes when increasing the resampling probability. 
\Cref{fig:properties_instances} shows the results of this experiment, where we depict the funding probability of the initially funded project with the lowest property value (solid) and highest property value (dashed) for the different rules. We now consider the three different properties one by one.

First, regarding the number of approvals of a project (see
\Cref{fig:properties_projects1}), we see a different behavior of the
rules.  For \Greedy, the impact of the number of approvals is
strongest: Recalling that \Greedy always funds the project which can
still be funded and has the highest number of approvals it is quite
intuitive that the project with the highest number of approvals is
also funded in almost all cases even if random noise is introduced.
In contrast, the initially funded project with the lowest number of
approvals stops to be founded, e.g., in around $10\%$ of cases for a
resampling probability of $25\%$.  In general, it is to be expected
that the initially funded project with the highest number of approvals
is more robust than the one with the lowest approval score, and, in
fact, the gap between them might be even seen as surprisingly small.
Similarly, for \MES, for all considered values of the resampling
probability, there is a clear gap between the funding probability of
the most and least approved initially funded project.  In contrast,
for \GreedyCost and \Phragmen, the situation is the other way around:
While for a small resampling probability the most and least approved
initially funded project still have a somewhat similar funding
probability for these two rules, the project with the lowest number of
approvals has a higher funding probability as the resampling
probability increases.  While this might look very counterintuitive at
first, as it is always advantageous for a project to have more
approvals, this observation can be explained by the fact that whether
a project gets funded also depends decisively on its cost and that on
our data cost is generally positively correlated with the number of
approvals.  Accordingly, the projects with the highest number of
approvals often tend to be the most expensive ones.  Whether and how
this influences its funding probability now depends on the used
budgeting rule, where our experiments show that for \MES and \Greedy a
high number of approvals seems to have a ``stronger'' influence, whereas
for \GreedyCost and \Phragmen the higher cost seems to be dominant.
The difference between \MES and \Phragmen is quite intriguing here, as
both try to achieve somewhat similar proportionality objectives.  An
explanation for why \MES tends to favor funding expensive projects is
that in \MES (in contrast to \MESAprUtilRule~and \Phragmen) voters
derive a higher utility for more expensive projects, which implies
that---in case supporters of a project have sufficient virtual money
available to them---more expensive projects are $q$-affordable for
lower values of $q$.  Especially as for \MES voters start with their
full virtual budget (and do not slowly acquire it over time), this
leads to expensive projects being funded at the beginning of the
execution of the rule.

Second, regarding project's costs (see
\Cref{fig:properties_projects2}), the picture is in some sense
reversed.  For \Greedy the robustness of projects seems to be mostly
independent of its cost, as the most and least expensive initially
funded projects have roughly the same winning probability for all
considered values of the resampling probability.  As \Greedy funds
projects in the order of the number of their approvals, it is quite
intuitive that the cost of projects only has a secondary influence on
the funding decisions; nevertheless, it is still remarkable that there
seems to be no influence at all here, given that whether a project
gets funded once it is ``its turn'' depends on whether the remaining
budget is sufficient to cover its cost.  For \MES, \Phragmen, and
\GreedyCost the cheapest initially funded project has a higher
robustness than the most expensive one, yet the contrast between the
two is less strong for \MES.  The reasons for the distinct behavior of
\MES might be similar as the ones described above for the number of
approvals.

Third, regarding the ratio between the number of approvals and the cost of
a project (see \Cref{fig:properties_projects3}), for all rules the
initially funded project with the highest number of approvals per cost
has a substantially higher funding probability than the initially
funded project with the lowest one.  This is quite intuitive, as this
measure both benefits projects with a high number of approvals (which
is of high relevance for \Greedy and \MES; see
\Cref{fig:properties_projects1}) and with a low cost (which is of high
relevance for \Phragmen and \GreedyCost; see
\Cref{fig:properties_projects2}).  Given that \GreedyCost funds
projects in order of the approval-cost ratio, it is remarkable that
for \Phragmen and \MES the influence of the approval-cost ratio on
project's robustness seems stronger than for \GreedyCost, as the
funding probability of the initially funded project with the lowest
ratio is smaller for these two rules.  A possible explanation for this
might be that \Phragmen and \MES lead in general to less robust
funding decisions.  All in all, the ratio between the number of
approvals and the cost of a project seems to be the best indicator to
judge how likely projects will continue to be funded if votes are
perturbed with random noise.

\subsection{Taking a Closer Look at Completion Methods for \MES} \label{sub:completion}

As already discussed in the main body, outcomes produced by \MES may not be exhaustive, i.e., the remaining budget might be sufficient to fund additional projects.
In this section, we compare the robustness of outcomes for different completion methods for \MES.
We analyze the following completion methods  (as implemented in the Pabutools library \cite{pabutools}), which are employed in case \MES on the initial instance returns a non-exhaustive outcome:
\begin{description}
\item[Add1+\Greedy] The initial endowments of the voters are increased
  (by increasing the total budget by $1\%$).  We continue to do so
  until the cost of the outcome output by the rule exceeds the
  original budget.  In case the returned outcome is not exhaustive, we
  apply \Greedy to spend the remaining budget. This is the method
  analyzed in the main body and recently used in Wieliczka, Poland~\cite{mes}.
\item[Add1] Like Add1+\Greedy but without the final \Greedy phase.
\item[Epsilon] 
The idea of the epsilon method is to modify the given instance such that voters also derive some utility for projects they do not approve, as in this case the outcome returned by \MES is guaranteed to be exhaustive \cite{pet-pie-sko:c:pb-mes}. Specifically, in the definition of \MES we define $u_j(c)$ to be $\cost(c)$ in case $v_j$ approves $c$ and $\epsilon\cdot \cost(c)$ in case $v_j$ does not approve $c$, where we pick the smallest possible value of $\epsilon$ such that the produced outcome is exhaustive. 
\item[\Greedy] The outcome as computed by \MES on the initial instance is completed by running \Greedy with the remaining budget and projects.
\item[None] Returns the outcome as computed by \MES on the initial instance.
\end{description}

\subsubsection{Correlation between Rules.}
We start by looking how the robustness of outcomes produced by the the different completion methods of \MES relate to each other. 
For this, in \Cref{fig:corr}, we show scatterplots for some pairs of completion methods (similar to \Cref{fig:ov_corr}).
In these plots, each point represents one instance and its $x$- and $y$-coordinates depict the instance's $50\%$-winner threshold under the respective rule.
Points are slightly perturbed if they were to overlap. 
Note that as we only examine resampling probabilities up until $25\%$, we group instances with a larger threshold together. 
Instances with a $50\%$-winner threshold above $25\%$ for both rules are omitted. 

\begin{figure}[t!]
	\centering 
	\begin{subfigure}[t]{0.245\textwidth}
		\resizebox{\textwidth}{!}{
\begin{tikzpicture}

\begin{axis}[
tick align=outside,
tick pos=left,
x grid style={white!69.0196078431373!black},
xlabel={Add1+\Greedy},
xmin=-0.01505, xmax=0.31605,
xtick style={color=black},
y grid style={white!69.0196078431373!black},
ylabel={Add1},
ymin=-0.0045, ymax=0.3145,
ytick style={color=black},
xtick={0,0.05,0.1,0.15,0.2,0.25,0.3},
xticklabels={0\%,5\%,10\%,15\%,20\%,25\%,$\ge\hspace*{-0.1cm}25\%$},
ytick={0,0.05,0.1,0.15,0.2,0.25,0.3},
yticklabels={0\%,5\%,10\%,15\%,20\%,25\%,$\ge\hspace*{-0.1cm}25\%$},
	ytick style={color=black},every tick label/.append style={font=\large}, 
label style={font=\large}
]
\addplot [draw=black, fill=black, mark=*, only marks]
table{%
x  y
0.07 0.07
0.03 0.03
0.2 0.08
0.07 0.06
0.18 0.19
0.06 0.05
0.01 0.01
0.009 0.01
0.3 0.08
0.13 0.12
0.08 0.01
0.011 0.01
0.08 0.08
0.02 0.02
0.17 0.17
0.3 0.09
0.07 0.08
0.3 0.01
0.069 0.06
0.05 0.05
0.3 0.05
0.05 0.06
0.17 0.16
0.04 0.04
0.08 0.05
0.3 0.12
0.25 0.25
0.1 0.1
0.12 0.13
0.3 0.06
0.3 0.17
0.11 0.1
0.019 0.02
0.099 0.1
0.008 0.01
0.1 0.08
0.3 0.04
0.16 0.17
0.22 0.23
0.299 0.04
0.12 0.12
0.049 0.05
0.06 0.07
0.012 0.01
0.039 0.04
0.119 0.12
0.021 0.02
0.3 0.18
0.03 0.02
0.24 0.24
0.299 0.08
0.08 0.02
0.24 0.23
0.01 0.02
0.14 0.11
0.02 0.01
0.169 0.16
0.19 0.14
0.019 0.01
0.21 0.11
0.25 0.18
0.12 0.1
0.029 0.03
0.14 0.07
0.17 0.07
0.18 0.18
0.04 0.03
0.007 0.01
0.029 0.02
0.301 0.08
0.013 0.01
0.018 0.02
0.3 0.23
0.3 0.22
0.022 0.02
0.3 0.25
0.031 0.03
0.006 0.01
0.299 0.23
0.15 0.15
0.15 0.14
0.051 0.05
0.014 0.01
0.19 0.03
0.15 0.07
0.069 0.08
0.3 0.21
0.005 0.01
0.02 0.03
0.079 0.02
0.3 0.2
0.071 0.06
0.14 0.15
0.06 0.06
0.009 0.02
0.059 0.06
0.299 0.12
0.299 0.06
0.23 0.3
0.015 0.01
0.3 0.15
0.23 0.23
0.01 0.3
0.004 0.01
0.11 0.01
0.021 0.01
0.21 0.22
0.08 0.3
0.19 0.18
0.11 0.14
0.139 0.15
0.05 0.01
0.09 0.1
0.068 0.06
0.017 0.02
0.22 0.21
0.048 0.05
0.05 0.09
0.2 0.2
0.07 0.04
0.028 0.03
0.072 0.06
0.016 0.01
0.023 0.02
0.09 0.06
0.14 0.13
0.179 0.19
0.04 0.02
0.069 0.07
0.299 0.01
0.3 0.14
0.14 0.14
0.299 0.18
0.03 0.01
0.07 0.03
0.079 0.08
0.171 0.16
0.11 0.11
0.3 0.11
0.08 0.14
0.189 0.18
0.032 0.03
0.301 0.01
0.299 0.2
0.16 0.2
0.08 0.07
0.079 0.05
0.301 0.04
0.016 0.02
0.052 0.05
0.003 0.01
0.3 0.07
0.09 0.05
0.249 0.25
0.3 0.03
0.017 0.01
0.299 0.21
0.031 0.02
0.024 0.02
0.21 0.23
0.002 0.01
0.109 0.11
0.071 0.07
0.068 0.07
0.07 0.09
0.018 0.01
0.121 0.12
0.139 0.13
0.3 0.02
0.027 0.03
0.039 0.03
0.033 0.03
0.072 0.07
0.14 0.09
0.028 0.02
0.21 0.2
0.15 0.08
0.301 0.21
0.05 0.17
0.071 0.08
0.000999999999999999 0.01
0.09 0.09
0.22 0.22
0.299 0.05
0.22 0.12
0.149 0.07
0.029 0.01
0.047 0.05
0.19 0.19
0.079 0.07
0.299 0.02
0.169 0.17
0.081 0.08
0.15 0.18
0.059 0.07
0.08 0.12
0.015 0.02
0.019 0.01
0.025 0.02
0.061 0.06
0.032 0.02
0.026 0.03
0.18 0.13
0.058 0.06
0.16 0.13
0.014 0.02
0.011 0.02
0.118 0.12
0 0.01
0.061 0.07
0.111 0.11
0.108 0.11
0.199 0.2
0.059 0.05
0.299 0.03
0.034 0.03
0.13 0.06
0.026 0.02
0.041 0.03
0.025 0.03
0.301 0.03
0.035 0.03
0.299 0.17
0.301 0.02
0.101 0.1
0.149 0.15
};
\end{axis}

\end{tikzpicture}}
		\caption{Add1+\Greedy vs Add1}\label{fig:corr1}
	\end{subfigure}\hfill
	\begin{subfigure}[t]{0.245\textwidth}
		\resizebox{\textwidth}{!}{
\begin{tikzpicture}

\begin{axis}[
tick align=outside,
tick pos=left,
x grid style={white!69.0196078431373!black},
xlabel={None},
xmin=-0.0129, xmax=0.3149,
xtick style={color=black},
y grid style={white!69.0196078431373!black},
ylabel={\Greedy},
ymin=-0.0045, ymax=0.3145,
ytick style={color=black},
xtick={0,0.05,0.1,0.15,0.2,0.25,0.3},
xticklabels={0\%,5\%,10\%,15\%,20\%,25\%,$\ge\hspace*{-0.1cm}25\%$},
ytick={0,0.05,0.1,0.15,0.2,0.25,0.3},
yticklabels={0\%,5\%,10\%,15\%,20\%,25\%,$\ge\hspace*{-0.1cm}25\%$},
	ytick style={color=black},every tick label/.append style={font=\large}, 
label style={font=\large}
]
\addplot [draw=black, fill=black, mark=*, only marks]
table{%
x  y
0.13 0.3
0.03 0.03
0.3 0.24
0.12 0.12
0.16 0.3
0.04 0.09
0.05 0.07
0.04 0.05
0.04 0.15
0.11 0.09
0.03 0.1
0.06 0.06
0.2 0.18
0.01 0.08
0.03 0.02
0.09 0.09
0.02 0.02
0.3 0.25
0.029 0.02
0.1 0.1
0.03 0.04
0.01 0.01
0.02 0.3
0.19 0.19
0.08 0.07
0.009 0.01
0.11 0.19
0.02 0.03
0.029 0.03
0.07 0.07
0.19 0.3
0.14 0.13
0.04 0.04
0.01 0.3
0.06 0.07
0.12 0.3
0.01 0.02
0.069 0.07
0.139 0.13
0.189 0.3
0.15 0.3
0.24 0.23
0.1 0.11
0.17 0.3
0.01 0.04
0.05 0.04
0.08 0.08
0.3 0.16
0.019 0.02
0.011 0.01
0.14 0.14
0.149 0.3
0.021 0.02
0.3 0.08
0.13 0.13
0.05 0.05
0.009 0.02
0.299 0.24
0.049 0.05
0.09 0.1
0.1 0.3
0.008 0.01
0.04 0.03
0.051 0.05
0.06 0.3
0.039 0.04
0.17 0.15
0.159 0.3
0.099 0.11
0.199 0.18
0.17 0.12
0.079 0.08
0.07 0.03
0.18 0.18
0.059 0.06
0.041 0.04
0.018 0.02
0.012 0.01
0.071 0.07
0.06 0.05
0.19 0.05
0.019 0.3
0.007 0.01
0.239 0.23
0.15 0.14
0.013 0.01
0.21 0.22
0.08 0.09
0.25 0.3
0.189 0.19
0.031 0.03
0.099 0.1
0.191 0.3
0.006 0.01
0.07 0.05
0.14 0.3
0.24 0.24
0.191 0.19
0.1 0.09
0.014 0.01
0.18 0.17
0.01 0.03
0.22 0.22
0.07 0.19
0.022 0.02
0.021 0.3
0.188 0.19
0.061 0.06
0.18 0.3
0.07 0.08
0.089 0.1
0.18 0.13
0.14 0.15
0.16 0.23
0.079 0.09
0.005 0.01
0.04 0.3
0.039 0.15
0.07 0.3
0.141 0.13
0.13 0.14
0.151 0.3
0.15 0.15
0.079 0.07
0.091 0.1
0.12 0.11
0.179 0.3
0.05 0.06
0.028 0.03
0.3 0.06
0.068 0.07
0.04 0.01
0.09 0.3
0.031 0.02
0.16 0.16
0.069 0.3
0.015 0.01
0.201 0.18
0.17 0.17
0.16 0.12
0.3 0.1
0.299 0.08
0.05 0.3
0.181 0.3
0.09 0.08
0.058 0.06
0.24 0.3
0.02 0.04
0.004 0.01
0.11 0.3
0.05 0.02
0.048 0.05
0.299 0.16
0.09 0.11
0.2 0.23
0.08 0.1
0.028 0.02
0.139 0.3
0.052 0.05
0.05 0.01
0.109 0.3
0.049 0.3
0.089 0.3
0.017 0.02
0.17 0.18
0.2 0.13
0.01 0.16
0.16 0.17
0.032 0.03
0.3 0.13
0.081 0.09
0.2 0.2
0.05 0.09
0.101 0.11
0.081 0.08
0.023 0.02
0.069 0.03
0.032 0.02
0.11 0.1
0.029 0.04
0.039 0.03
0.078 0.09
0.059 0.07
0.19 0.18
0.099 0.3
0.041 0.03
0.3 0.14
0.17 0.19
0.241 0.23
0.016 0.01
0.009 0.3
0.13 0.16
0.038 0.04
0.05 0.22
0.003 0.01
0.047 0.05
0.049 0.06
0.2 0.3
0.027 0.03
0.3 0.09
0.059 0.3
0.051 0.3
0.019 0.03
0.017 0.01
0.14 0.03
0.3 0.12
0.192 0.19
0.018 0.3
0.18 0.15
0.199 0.3
0.033 0.03
0.08 0.01
0.089 0.08
0.002 0.01
0.026 0.03
0.2 0.21
0.034 0.03
0.061 0.3
0.018 0.01
0.12 0.13
0.101 0.3
0.011 0.02
0.129 0.13
0.021 0.03
0.131 0.13
0.071 0.3
0.04 0.25
0.018 0.03
0.048 0.3
};
\end{axis}

\end{tikzpicture}}
		\caption{None vs Greedy}\label{fig:corr2}
	\end{subfigure}\hfill
	\begin{subfigure}[t]{0.245\textwidth}
		\resizebox{\textwidth}{!}{
\begin{tikzpicture}

\begin{axis}[
tick align=outside,
tick pos=left,
x grid style={white!69.0196078431373!black},
xlabel={Add1+\Greedy},
xmin=-0.00675, xmax=0.31775,
xtick style={color=black},
y grid style={white!69.0196078431373!black},
ylabel={\Greedy},
ymin=-0.0045, ymax=0.3145,
ytick style={color=black},
xtick={0,0.05,0.1,0.15,0.2,0.25,0.3},
xticklabels={0\%,5\%,10\%,15\%,20\%,25\%,$\ge\hspace*{-0.1cm}25\%$},
ytick={0,0.05,0.1,0.15,0.2,0.25,0.3},
yticklabels={0\%,5\%,10\%,15\%,20\%,25\%,$\ge\hspace*{-0.1cm}25\%$},
	ytick style={color=black},every tick label/.append style={font=\large}, 
label style={font=\large}
]
\addplot [draw=black, fill=black, mark=*, only marks]
table{%
x  y
0.07 0.3
0.03 0.03
0.3 0.24
0.3 0.12
0.2 0.3
0.069 0.3
0.18 0.3
0.06 0.09
0.01 0.07
0.01 0.05
0.3 0.15
0.13 0.09
0.08 0.1
0.01 0.06
0.3 0.18
0.08 0.08
0.02 0.3
0.17 0.02
0.3 0.09
0.07 0.02
0.3 0.25
0.069 0.02
0.05 0.1
0.3 0.04
0.3 0.01
0.05 0.19
0.17 0.3
0.04 0.07
0.08 0.01
0.3 0.19
0.25 0.3
0.1 0.3
0.12 0.03
0.3 0.03
0.11 0.3
0.02 0.07
0.3 0.13
0.299 0.04
0.099 0.3
0.009 0.07
0.101 0.3
0.16 0.02
0.22 0.07
0.299 0.13
0.12 0.3
0.05 0.3
0.06 0.23
0.01 0.11
0.04 0.04
0.12 0.04
0.019 0.3
0.3 0.08
0.03 0.16
0.3 0.02
0.24 0.3
0.079 0.01
0.24 0.14
0.01 0.02
0.299 0.08
0.301 0.13
0.14 0.05
0.02 0.02
0.299 0.24
0.17 0.05
0.19 0.1
0.02 0.01
0.21 0.03
0.25 0.05
0.119 0.04
0.03 0.15
0.14 0.3
0.169 0.3
0.179 0.3
0.04 0.11
0.01 0.18
0.03 0.12
0.301 0.08
0.01 0.03
0.02 0.18
0.02 0.06
0.03 0.04
0.299 0.02
0.01 0.01
0.3 0.07
0.15 0.3
0.15 0.05
0.05 0.05
0.009 0.01
0.19 0.3
0.149 0.3
0.3 0.23
0.07 0.14
0.299 0.01
0.3 0.22
0.299 0.09
0.01 0.3
0.02 0.19
0.08 0.03
0.071 0.3
0.139 0.3
0.06 0.1
0.011 0.01
0.06 0.05
0.301 0.24
0.299 0.19
0.301 0.09
0.23 0.3
0.008 0.01
0.3 0.17
0.229 0.3
0.009 0.3
0.009 0.03
0.11 0.22
0.019 0.19
0.21 0.02
0.08 0.3
0.19 0.19
0.109 0.3
0.141 0.3
0.3 0.06
0.05 0.08
0.3 0.1
0.09 0.13
0.068 0.3
0.02 0.15
0.22 0.23
0.298 0.09
0.049 0.3
0.051 0.3
0.199 0.3
0.301 0.01
0.072 0.3
0.03 0.3
0.299 0.15
0.07 0.13
0.3 0.14
0.01 0.15
0.019 0.07
0.299 0.1
0.09 0.11
0.138 0.3
0.18 0.06
0.04 0.03
0.299 0.06
0.07 0.07
0.298 0.01
0.142 0.3
0.03 0.02
0.07 0.16
0.081 0.01
0.17 0.18
0.299 0.17
0.111 0.3
0.299 0.12
0.079 0.3
0.301 0.1
0.189 0.3
0.03 0.08
0.298 0.08
0.301 0.06
0.16 0.3
0.08 0.04
0.078 0.01
0.019 0.02
0.049 0.05
0.01 0.16
0.3 0.11
0.09 0.23
0.249 0.3
0.298 0.1
0.009 0.02
0.03 0.05
0.019 0.01
0.21 0.3
0.011 0.02
0.11 0.18
0.298 0.13
0.069 0.16
0.07 0.17
0.07 0.03
0.01 0.13
0.12 0.09
0.3 0.2
0.302 0.09
0.14 0.11
0.302 0.08
0.029 0.02
0.039 0.03
0.031 0.02
0.067 0.3
0.14 0.1
0.301 0.04
0.029 0.03
0.21 0.09
0.15 0.07
0.299 0.18
0.048 0.3
0.073 0.3
0.011 0.03
0.09 0.3
0.22 0.14
0.301 0.19
0.219 0.23
0.15 0.01
0.029 0.16
0.052 0.3
0.191 0.3
0.081 0.3
0.298 0.04
0.17 0.22
0.302 0.01
0.08 0.05
0.151 0.3
0.06 0.06
0.078 0.3
0.02 0.03
0.01 0.09
0.021 0.3
0.06 0.03
0.03 0.01
0.031 0.03
0.301 0.12
0.181 0.3
0.06 0.19
0.16 0.15
0.019 0.03
0.297 0.01
0.011 0.3
0.12 0.08
0.012 0.01
0.06 0.3
0.299 0.03
0.11 0.21
0.11 0.03
0.201 0.3
0.059 0.3
0.303 0.01
0.302 0.13
0.029 0.3
0.13 0.02
0.02 0.13
0.041 0.03
0.03 0.13
0.031 0.3
0.299 0.25
0.098 0.3
0.15 0.03
};
\end{axis}

\end{tikzpicture}}
		\caption{Add1+\Greedy vs Greedy}\label{fig:corr3}
	\end{subfigure}\hfill
\begin{subfigure}[t]{0.245\textwidth}
\resizebox{\textwidth}{!}{
\begin{tikzpicture}

\begin{axis}[
tick align=outside,
tick pos=left,
x grid style={white!69.0196078431373!black},
xlabel={Add1+\Greedy},
xmin=-0.0098, xmax=0.3158,
xtick style={color=black},
y grid style={white!69.0196078431373!black},
ylabel={Epsilon},
ymin=-0.0045, ymax=0.3145,
ytick style={color=black},
xtick={0,0.05,0.1,0.15,0.2,0.25,0.3},
xticklabels={0\%,5\%,10\%,15\%,20\%,25\%,$\ge\hspace*{-0.1cm}25\%$},
ytick={0,0.05,0.1,0.15,0.2,0.25,0.3},
yticklabels={0\%,5\%,10\%,15\%,20\%,25\%,$\ge\hspace*{-0.1cm}25\%$},
	ytick style={color=black},every tick label/.append style={font=\large}, 
label style={font=\large}
]
\addplot [draw=black, fill=black, mark=*, only marks]
table{%
x  y
0.07 0.3
0.03 0.04
0.2 0.3
0.3 0.25
0.069 0.3
0.18 0.3
0.06 0.3
0.01 0.3
0.009 0.3
0.13 0.3
0.08 0.1
0.01 0.14
0.08 0.3
0.02 0.3
0.17 0.2
0.071 0.3
0.07 0.25
0.05 0.13
0.3 0.04
0.05 0.3
0.17 0.01
0.04 0.23
0.08 0.07
0.25 0.23
0.1 0.3
0.12 0.09
0.3 0.02
0.11 0.3
0.019 0.3
0.3 0.15
0.3 0.14
0.099 0.3
0.011 0.3
0.101 0.3
0.16 0.09
0.22 0.3
0.12 0.02
0.3 0.2
0.3 0.11
0.049 0.3
0.06 0.24
0.008 0.3
0.04 0.04
0.12 0.3
0.021 0.3
0.03 0.3
0.3 0.07
0.24 0.3
0.079 0.3
0.24 0.14
0.01 0.23
0.3 0.19
0.14 0.3
0.02 0.07
0.17 0.3
0.299 0.25
0.19 0.1
0.3 0.21
0.02 0.22
0.21 0.04
0.25 0.17
0.12 0.11
0.03 0.01
0.139 0.3
0.169 0.3
0.179 0.3
0.04 0.09
0.012 0.3
0.029 0.3
0.3 0.03
0.009 0.14
0.02 0.01
0.018 0.3
0.031 0.3
0.299 0.07
0.01 0.02
0.3 0.09
0.15 0.3
0.149 0.3
0.051 0.3
0.01 0.03
0.19 0.3
0.151 0.3
0.068 0.3
0.299 0.02
0.299 0.09
0.007 0.3
0.022 0.3
0.081 0.3
0.072 0.3
0.141 0.3
0.06 0.1
0.01 0.1
0.06 0.07
0.3 0.08
0.3 0.23
0.23 0.3
0.01 0.22
0.229 0.3
0.013 0.3
0.01 0.05
0.11 0.07
0.02 0.18
0.21 0.11
0.078 0.3
0.189 0.3
0.109 0.3
0.138 0.3
0.05 0.1
0.299 0.11
0.09 0.3
0.067 0.3
0.017 0.3
0.219 0.3
0.3 0.01
0.05 0.18
0.048 0.3
0.299 0.15
0.199 0.3
0.3 0.12
0.073 0.3
0.028 0.3
0.301 0.15
0.3 0.1
0.066 0.3
0.006 0.3
0.02 0.02
0.089 0.3
0.142 0.3
0.18 0.07
0.04 0.03
0.299 0.2
0.07 0.09
0.137 0.3
0.029 0.04
0.074 0.3
0.301 0.09
0.082 0.3
0.171 0.3
0.111 0.3
0.077 0.3
0.191 0.3
0.032 0.3
0.298 0.15
0.16 0.3
0.083 0.3
0.08 0.04
0.299 0.1
0.02 0.12
0.05 0.05
0.01 0.08
0.09 0.23
0.25 0.3
0.01 0.06
0.3 0.05
0.03 0.25
0.301 0.02
0.02 0.21
0.298 0.02
0.21 0.3
0.009 0.03
0.108 0.3
0.3 0.24
0.065 0.3
0.075 0.3
0.064 0.3
0.014 0.3
0.119 0.3
0.3 0.22
0.143 0.3
0.027 0.3
0.039 0.03
0.033 0.3
0.076 0.3
0.136 0.3
0.299 0.04
0.03 0.05
0.21 0.06
0.15 0.25
0.299 0.14
0.05 0.08
0.063 0.3
0.011 0.03
0.091 0.3
0.221 0.3
0.301 0.25
0.218 0.3
0.15 0.01
0.026 0.3
0.05 0.25
0.188 0.3
0.076 0.3
0.168 0.3
0.299 0.03
0.08 0.16
0.148 0.3
0.059 0.07
0.084 0.3
0.019 0.01
0.005 0.3
0.023 0.3
0.06 0.03
0.03 0.02
0.03 0.11
0.181 0.3
0.059 0.3
0.298 0.09
0.159 0.3
0.02 0.04
0.299 0.01
0.015 0.3
0.12 0.08
0.009 0.02
0.06 0.06
0.112 0.3
0.107 0.3
0.201 0.3
0.059 0.06
0.299 0.19
0.034 0.3
0.129 0.3
0.016 0.3
0.301 0.2
0.04 0.06
0.025 0.3
0.035 0.3
0.098 0.3
0.152 0.3
};
\end{axis}

\end{tikzpicture}}
\caption{Add1+\Greedy vs Epsilon}\label{fig:corr4}
\end{subfigure}
	\caption{Correlation between $50\%$-winner thresholds for some pairs of completion methods for \MES.}
	\label{fig:corr}
\end{figure}
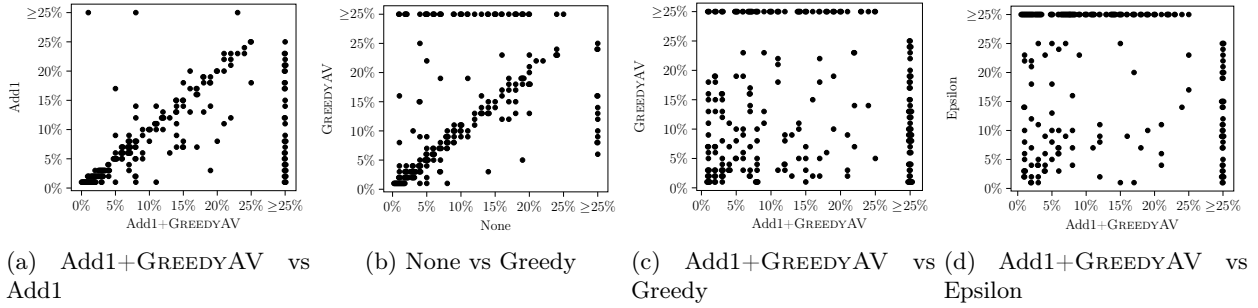

Examining the plots, what we see here is that there is a clear
correlation between $50\%$-winner thresholds for the completion
methods Add1 and Add1+\Greedy (see \Cref{fig:corr1}).  The fact that
outcomes produced by these completion methods are closely connected to
each other is also quite intuitive, as both of them produce identical
outcomes in many cases, as the final \Greedy phase takes only place on
some instances.  However, also note that the correlation is not
perfect, as there are still instances which have a $50\%$-winner
threshold above $25\%$ for Add1+\Greedy but a $50\%$-winner threshold
of $1\%$ for Add1 (see \Cref{fig:Rejon} for an example).  The
completion methods None and \Greedy (see \Cref{fig:corr2}) exhibit a
similarly strong correlation.  Here, the reason for their correlation
is because \Greedy is a very robust method. Thus, in case the outcome
produced by \MES on the initial instance does not change, which is the
one returned by ``None,'' \Greedy is also likely to spend the remaining
budget on the same projects.  In contrast, for all other pairs of
rules the correlation is less strong: In \Cref{fig:corr3}, we show the
correlation between Add1+\Greedy and \Greedy (the plots for Add1 and
\Greedy, Add1+\Greedy and None, and Add1 and None look similar).
Here, no correlation is visible, and thus the two groups of completion
methods seem not to be correlated in terms of unrobust outcomes.  The
many instances exhibiting a different behavior under the two
completion methods also highlight the high impact that the completion
method has on the produced outcome.  Lastly, in \Cref{fig:corr4}, we
include the correlation between Add1+\Greedy and Epsilon (the plots
for all other rules compared to Epsilon are similar). Here, again no
clear correlation is visible, indicating that the Epsilon method acts
in some sense differently than the other four.

To sum up, there seem to be three groups of completion methods which show a somewhat similar behavior on an instance-level with the first one containing Add1+\Greedy and Add1, the second one containing Epsilon, and the third one containing \Greedy and None.
Overall, the influence of the completion method on the robustness of outcomes (at least in terms of the $50\%$-winner threshold) seems to be surprisingly strong. 
We will investigate this influence in more detail in the next part.

\begin{table}[t!]
	\centering
	\begin{tabular}{c|c|c|c|c|c}
		completion method & \# instances $\leq 25\%$ & \# instances $\leq 10\%$ & \# instances $\leq 5\%$ & mean (for $\leq 25\%$) & median (for $\leq 25\%$)   \\\hline
		Add1+\Greedy & 187  & 121 & 75  & 0.09 & 0.07 \\
		Add1 & 223  & 149 & 97 & 0.08 & 0.07\\
		Epsilon & 114 & 68 & 36 & 0.11 & 0.09 \\
		\Greedy & 185 & 119 & 75 & 0.09 & 0.08 \\
		None & 216 & 142 & 91 & 0.09 & 0.07 \\
	\end{tabular}
	\caption{Statistical quantities regarding the $50\%$-winner threshold of \MES with different completion methods. The first three columns contain the number of instances with a 50\%-winner threshold smaller equal $25\%/10\%/5\%$. The last two columns the mean/median threshold among all instances with a threshold of smaller equal $25\%$. } \label{ta:MES-rob}
\end{table} 
\subsubsection{An Aggregate View on Outcome's Robustness.}
In \Cref{ta:MES-rob} (analogous to \Cref{ta:50-winner}), we provide
some quantities regarding the $50\%$-winner threshold of \MES with
different completion methods on our full dataset. There is a clear
ordering of the completion methods in terms of the robustness of
produced outcomes for all five considered quantities: Add1 produces
the least robust result, than None, than Add1+\Greedy and \Greedy,
which perform roughly similarly, and Epsilon produces the most robust
results.  The difference between Epsilon and Add1 is quite remarkable
here, as Add1 produces roughly twice as many instances with a
$50\%$-winner threshold below $25\%/10\%/5\%$, again highlighting the
impact of the chosen completion method.

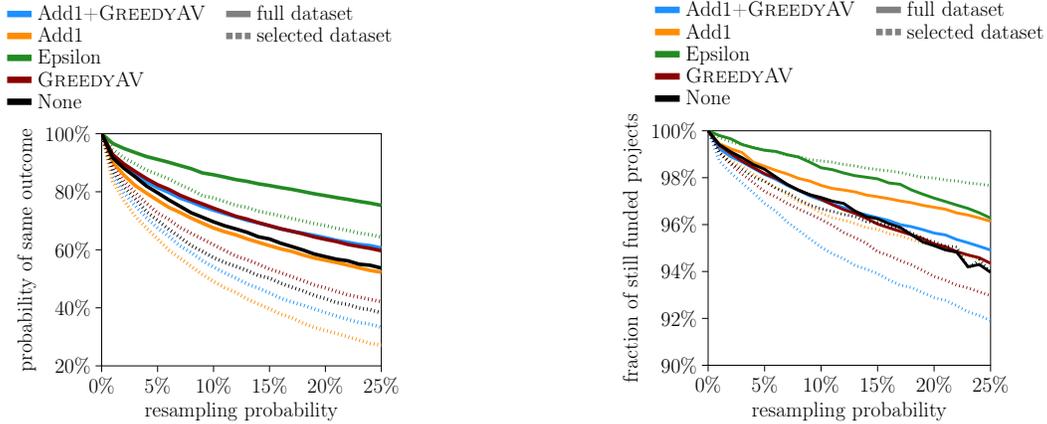
\begin{figure}[t!]
	\centering 
	\begin{subfigure}[b]{0.47\textwidth}
        \centering
		\resizebox{0.7\textwidth}{!}{\begin{tikzpicture}[every plot/.append style={line width=2.5pt}]

\definecolor{color0}{rgb}{1,0.549019607843137,0}
\definecolor{color1}{rgb}{0.133333333333333,0.545098039215686,0.133333333333333}
\definecolor{color2}{rgb}{0.117647058823529,0.564705882352941,1}
\definecolor{color3}{rgb}{0,0,0}
\begin{axis}[
legend columns=2, 
legend cell align={left},
legend style={
  fill opacity=0.8,
  draw opacity=1,
  draw=none,
  text opacity=1,
  at={(0.35,1.58)},
  line width=1.5pt,
  anchor=north,
   /tikz/column 2/.style={
  	column sep=10pt,
  }, font=\Large
},
ytick={0.2,0.4,0.6,0.8,1},
yticklabels={20\%,40\%,60\%,80\%,100\%},
legend image post style={line width =4.5pt},
legend entries={Add1+\Greedy,
	full dataset,
	Add1,
	selected dataset,
	Epsilon, {\phantom{a}},
	\Greedy, {\phantom{a}},
None},
tick align=outside,
tick pos=left,
x grid style={white!69.0196078431373!black},
xlabel={resampling probability},
xmin=0, xmax=0.25,
xtick style={color=black},
xtick={0,0.05,0.1,0.15,0.2,0.25},
xticklabels={0\%,5\%,10\%,15\%,20\%,25\%},
y grid style={white!69.0196078431373!black},
ylabel={probability of same outcome},
ymin=0.2, ymax=1,
ytick style={color=black},every tick label/.append style={font=\Large}, 
label style={font=\Large}
]
\addlegendimage{color2}
\addlegendimage{gray}
\addlegendimage{color0}
\addlegendimage{gray,dotted}
\addlegendimage{color1}
\addlegendimage{white,dashed}
\addlegendimage{red!54.5098039215686!black}
\addlegendimage{white,dashed}
\addlegendimage{color3}
\addplot [semithick, color2]
table {%
0 0.999999999999995
0.01 0.916391304347823
0.02 0.883065217391301
0.03 0.858086956521736
0.04 0.834456521739128
0.05 0.81423913043478
0.06 0.797717391304345
0.07 0.778260869565215
0.08 0.764760869565216
0.09 0.749021739130433
0.1 0.736543478260867
0.11 0.724652173913042
0.12 0.713586956521737
0.13 0.701521739130433
0.14 0.692978260869564
0.15 0.682521739130433
0.16 0.672043478260868
0.17 0.66386956521739
0.18 0.657086956521738
0.19 0.649717391304346
0.2 0.641521739130433
0.21 0.632543478260868
0.22 0.625673913043477
0.23 0.618695652173912
0.24 0.614304347826086
0.25 0.605913043478259
};
\addplot [semithick, color2, dotted]
table {%
0 0.999999999999998
0.01 0.852607003891049
0.02 0.795719844357976
0.03 0.75124513618677
0.04 0.710661478599221
0.05 0.675214007782101
0.06 0.646731517509727
0.07 0.613307392996109
0.08 0.590194552529183
0.09 0.562879377431906
0.1 0.540544747081712
0.11 0.520311284046693
0.12 0.502295719844358
0.13 0.482529182879377
0.14 0.467898832684825
0.15 0.451206225680934
0.16 0.432918287937743
0.17 0.42
0.18 0.408949416342413
0.19 0.395992217898833
0.2 0.383813229571984
0.21 0.371556420233463
0.22 0.360272373540856
0.23 0.349416342412451
0.24 0.343774319066148
0.25 0.332879377431907
};
\addplot [semithick, color0]
table {%
0 0.999999999999995
0.01 0.900869565217388
0.02 0.856869565217389
0.03 0.821391304347824
0.04 0.795304347826084
0.05 0.770369565217389
0.06 0.746217391304346
0.07 0.728630434782607
0.08 0.709108695652173
0.09 0.691782608695651
0.1 0.675217391304347
0.11 0.662086956521738
0.12 0.651347826086955
0.13 0.637543478260868
0.14 0.625782608695651
0.15 0.614434782608694
0.16 0.603152173913042
0.17 0.594521739130434
0.18 0.582086956521738
0.19 0.571717391304347
0.2 0.563847826086956
0.21 0.556086956521738
0.22 0.547108695652173
0.23 0.537826086956521
0.24 0.528891304347825
0.25 0.522717391304347
};
\addplot [semithick, color0, dotted]
table {%
0 0.999999999999998
0.01 0.834863813229571
0.02 0.770233463035019
0.03 0.715797665369649
0.04 0.674785992217898
0.05 0.637237354085602
0.06 0.599338521400778
0.07 0.570972762645914
0.08 0.542918287937743
0.09 0.514980544747081
0.1 0.490778210116731
0.11 0.469533073929961
0.12 0.452762645914397
0.13 0.430583657587548
0.14 0.413268482490272
0.15 0.396731517509728
0.16 0.3784046692607
0.17 0.365836575875486
0.18 0.34704280155642
0.19 0.332295719844358
0.2 0.321789883268483
0.21 0.311750972762646
0.22 0.30011673151751
0.23 0.288910505836576
0.24 0.277315175097276
0.25 0.27136186770428
};
\addplot [semithick, color1]
table {%
0 0.999999999999995
0.01 0.966152173913039
0.02 0.948956521739127
0.03 0.934891304347823
0.04 0.921347826086953
0.05 0.91160869565217
0.06 0.901499999999996
0.07 0.890326086956519
0.08 0.879521739130432
0.09 0.864934782608692
0.1 0.858586956521736
0.11 0.850499999999997
0.12 0.841782608695649
0.13 0.835804347826084
0.14 0.827413043478258
0.15 0.821434782608693
0.16 0.814391304347823
0.17 0.808304347826084
0.18 0.80073913043478
0.19 0.792999999999997
0.2 0.786608695652171
0.21 0.779826086956519
0.22 0.773717391304345
0.23 0.766804347826085
0.24 0.760826086956519
0.25 0.752978260869563
};
\addplot [semithick, color1, dotted]
table {%
0 0.999999999999998
0.01 0.945058365758753
0.02 0.916653696498053
0.03 0.895953307392995
0.04 0.874824902723734
0.05 0.859727626459143
0.06 0.844163424124513
0.07 0.824319066147859
0.08 0.807431906614785
0.09 0.786459143968871
0.1 0.778677042801555
0.11 0.766147859922178
0.12 0.753229571984435
0.13 0.746031128404669
0.14 0.732918287937743
0.15 0.725214007782101
0.16 0.71669260700389
0.17 0.707743190661478
0.18 0.699299610894941
0.19 0.69
0.2 0.683268482490272
0.21 0.674124513618676
0.22 0.6668093385214
0.23 0.657509727626459
0.24 0.652256809338521
0.25 0.64396887159533
};
\addplot [semithick,  red!54.5098039215686!black]
table {%
0 0.999999999999995
0.01 0.927630434782606
0.02 0.895347826086953
0.03 0.869804347826084
0.04 0.845652173913041
0.05 0.824152173913041
0.06 0.808173913043476
0.07 0.788152173913041
0.08 0.772369565217389
0.09 0.758260869565215
0.1 0.743478260869563
0.11 0.729434782608693
0.12 0.71528260869565
0.13 0.705608695652172
0.14 0.693413043478259
0.15 0.683152173913041
0.16 0.672630434782607
0.17 0.663934782608694
0.18 0.654847826086955
0.19 0.644347826086955
0.2 0.635978260869564
0.21 0.62832608695652
0.22 0.618413043478259
0.23 0.611434782608695
0.24 0.604956521739129
0.25 0.596760869565216
};
\addplot [semithick,  red!54.5098039215686!black, dotted]
table {%
0 0.999999999999998
0.01 0.886498054474707
0.02 0.836186770428015
0.03 0.797392996108949
0.04 0.760233463035019
0.05 0.7268093385214
0.06 0.704435797665369
0.07 0.6768093385214
0.08 0.656614785992218
0.09 0.637159533073929
0.1 0.617470817120622
0.11 0.597587548638132
0.12 0.5784046692607
0.13 0.56568093385214
0.14 0.550428015564202
0.15 0.533540856031128
0.16 0.52124513618677
0.17 0.50704280155642
0.18 0.493501945525292
0.19 0.478988326848249
0.2 0.46863813229572
0.21 0.456264591439689
0.22 0.444630350194553
0.23 0.436809338521401
0.24 0.428910505836576
0.25 0.421128404669261
};
\addplot [semithick,color3]
table {%
0 0.999999999999995
0.01 0.916478260869562
0.02 0.880652173913041
0.03 0.851456521739128
0.04 0.819869565217389
0.05 0.796391304347824
0.06 0.772043478260867
0.07 0.749891304347824
0.08 0.727673913043477
0.09 0.712173913043477
0.1 0.696347826086955
0.11 0.682652173913043
0.12 0.671782608695651
0.13 0.660173913043477
0.14 0.644934782608694
0.15 0.637304347826086
0.16 0.622782608695651
0.17 0.611152173913042
0.18 0.599173913043477
0.19 0.585521739130434
0.2 0.576391304347825
0.21 0.567434782608695
0.22 0.562086956521738
0.23 0.550456521739129
0.24 0.546717391304347
0.25 0.536956521739129
};
\addplot [semithick,color3, dotted]
table {%
0 0.999999999999998
0.01 0.871361867704279
0.02 0.816342412451361
0.03 0.775875486381323
0.04 0.731789883268482
0.05 0.70136186770428
0.06 0.670544747081711
0.07 0.642918287937743
0.08 0.614474708171206
0.09 0.594552529182879
0.1 0.572996108949416
0.11 0.556342412451362
0.12 0.54284046692607
0.13 0.52852140077821
0.14 0.51124513618677
0.15 0.502568093385214
0.16 0.484747081712062
0.17 0.47124513618677
0.18 0.458210116731518
0.19 0.442918287937743
0.2 0.431712062256809
0.21 0.417704280155642
0.22 0.411750972762646
0.23 0.399338521400778
0.24 0.393346303501946
0.25 0.382334630350195
};

\end{axis}

\end{tikzpicture}}
		\caption{Probability that outcome is identical to initial one. }\label{fig:ov_robust_MES1}
	\end{subfigure}\hfill
	\begin{subfigure}[b]{0.51\textwidth}
 \centering
		\resizebox{0.69\textwidth}{!}{\begin{tikzpicture}[every plot/.append style={line width=2pt}]

\definecolor{color0}{rgb}{1,0.549019607843137,0}
\definecolor{color1}{rgb}{0.133333333333333,0.545098039215686,0.133333333333333}
\definecolor{color2}{rgb}{0.117647058823529,0.564705882352941,1}
\definecolor{color3}{rgb}{0,0,0}
\begin{axis}[
legend columns=2, 
legend cell align={left},
legend style={
  fill opacity=0.8,
  draw opacity=1,
  draw=none,
  text opacity=1,
  at={(0.5,1.58)},
  line width=1.5pt,
  anchor=north,
   /tikz/column 2/.style={
  	column sep=10pt,
  }, font=\Large
},
legend image post style={line width =4.5pt},
legend entries={Add1+\Greedy,
	full dataset,
	Add1,
	selected dataset,
	Epsilon, {\phantom{a}},
	\Greedy, {\phantom{a}},
None},
tick align=outside,
tick pos=left,
x grid style={white!69.0196078431373!black},
xlabel={resampling probability},
xmin=0, xmax=0.25,
xtick style={color=black},
ytick={0.9,0.92,0.94,0.96,0.98,1},
yticklabels={90\%,92\%,94\%,96\%,98\%,100\%},
xtick={0,0.05,0.1,0.15,0.2,0.25},
xticklabels={0\%,5\%,10\%,15\%,20\%,25\%},
y grid style={white!69.0196078431373!black},
ylabel={fraction of still funded projects},
ymin=0.9, ymax=1,
ytick style={color=black},every tick label/.append style={font=\Large}, 
label style={font=\Large}
]
\addlegendimage{color2}
\addlegendimage{gray}
\addlegendimage{color0}
\addlegendimage{gray,dotted}
\addlegendimage{color1}
\addlegendimage{white,dashed}
\addlegendimage{red!54.5098039215686!black}
\addlegendimage{white,dashed}
\addlegendimage{color3}
\addplot [semithick, color2]
table {%
0 1
0.01 0.992282896086307
0.02 0.989028737672515
0.03 0.98671503234718
0.04 0.984033808428136
0.05 0.981390869674576
0.06 0.979258811111678
0.07 0.976694838804856
0.08 0.974527683148489
0.09 0.972208377874559
0.1 0.970483409047499
0.11 0.96875491151082
0.12 0.967253784153923
0.13 0.965449144094362
0.14 0.964150562167052
0.15 0.962929783490465
0.16 0.961714275154378
0.17 0.959974040729836
0.18 0.959101286742173
0.19 0.957846717486794
0.2 0.956381986119084
0.21 0.955590988741744
0.22 0.953640180218076
0.23 0.952250850269926
0.24 0.950735375631652
0.25 0.949130059365602
};
\addplot [semithick, color2, dotted]
table {%
0 0.999999999999999
0.01 0.98704476001049
0.02 0.982132782726774
0.03 0.977832392161415
0.04 0.973562060510366
0.05 0.969054480766415
0.06 0.965680034133005
0.07 0.961088605753731
0.08 0.957586254661713
0.09 0.953908732200904
0.1 0.950389296677129
0.11 0.94753806606862
0.12 0.945447805043519
0.13 0.942578215512868
0.14 0.941019137294207
0.15 0.939102880053769
0.16 0.936431583398599
0.17 0.934135754556998
0.18 0.933168864967398
0.19 0.930929413112645
0.2 0.928896513575436
0.21 0.927702633561042
0.22 0.925416534114975
0.23 0.922807374879146
0.24 0.921209403100044
0.25 0.919056197241671
};
\addplot [semithick, color0]
table {%
0 0.999999999999999
0.01 0.994460750490904
0.02 0.992325945560167
0.03 0.990748157912742
0.04 0.986629772397104
0.05 0.984856624316135
0.06 0.983099700034361
0.07 0.981592405208674
0.08 0.979698234416713
0.09 0.978555025604619
0.1 0.976704666640885
0.11 0.975409084507369
0.12 0.9747502097614
0.13 0.973934123719151
0.14 0.972617948031076
0.15 0.971863561147137
0.16 0.97106092596996
0.17 0.97008325891787
0.18 0.9693369418968
0.19 0.96821279591299
0.2 0.967384568682953
0.21 0.966559722501065
0.22 0.964992124188327
0.23 0.964071554448223
0.24 0.962821387263398
0.25 0.961408671569021
};
\addplot [semithick, color0, dotted]
table {%
0 0.999999999999999
0.01 0.990976952316162
0.02 0.987475927816306
0.03 0.98483483190697
0.04 0.981519827261937
0.05 0.978650782049786
0.06 0.975611311196734
0.07 0.972822049062361
0.08 0.969765853420426
0.09 0.967800835665453
0.1 0.964953270105629
0.11 0.9629349762593
0.12 0.96188505913363
0.13 0.960491538085325
0.14 0.958661818565593
0.15 0.95800511088399
0.16 0.95613936812889
0.17 0.955314791404404
0.18 0.953622587168873
0.19 0.952162830839701
0.2 0.95130017508888
0.21 0.95029446897417
0.22 0.947694944428499
0.23 0.946673884003849
0.24 0.945272599965767
0.25 0.943321285366999
};
\addplot [semithick, color1]
table {%
0 0.999999999999997
0.01 0.997978795279782
0.02 0.996601383337123
0.03 0.994167724787557
0.04 0.992770750803618
0.05 0.99170725881167
0.06 0.991193162857181
0.07 0.989809361824997
0.08 0.988795574278162
0.09 0.986614746662239
0.1 0.98424128532848
0.11 0.98343259580552
0.12 0.98200325269173
0.13 0.980851515579194
0.14 0.980204381019554
0.15 0.979504828274696
0.16 0.977762400202545
0.17 0.97718776917087
0.18 0.974776211987157
0.19 0.972975086053849
0.2 0.971292233778609
0.21 0.969836838149028
0.22 0.968324737346583
0.23 0.966893303555199
0.24 0.965015216446831
0.25 0.962778756505699
};
\addplot [semithick, color1, dotted]
table {%
0 0.999999999999996
0.01 0.997173980040385
0.02 0.995259047276801
0.03 0.994073103507436
0.04 0.99296052747254
0.05 0.991806686586946
0.06 0.991008892532801
0.07 0.989660877661869
0.08 0.988941781236605
0.09 0.987667823906232
0.1 0.987173284190401
0.11 0.986617645569976
0.12 0.98564901075836
0.13 0.985038688624744
0.14 0.984068321763453
0.15 0.983474450406774
0.16 0.982897310143502
0.17 0.982246541138133
0.18 0.981081759188024
0.19 0.980138945036791
0.2 0.979871518796025
0.21 0.978995493839646
0.22 0.978615801619321
0.23 0.977753691239617
0.24 0.977292547683721
0.25 0.976613466099885
};
\addplot [semithick,  red!54.5098039215686!black]
table {%
0 0.999999999999995
0.01 0.993741644543919
0.02 0.990325221276458
0.03 0.987091307105943
0.04 0.984404512073783
0.05 0.981698088543942
0.06 0.979804762767941
0.07 0.977120975750909
0.08 0.97474530261309
0.09 0.972728745940319
0.1 0.970844204059814
0.11 0.968466074799297
0.12 0.96638399905534
0.13 0.964765749968863
0.14 0.962691556094561
0.15 0.960227251308584
0.16 0.958644296842369
0.17 0.95711738038825
0.18 0.95561543019963
0.19 0.954015724667768
0.2 0.95173791035548
0.21 0.95053954589571
0.22 0.948481661649067
0.23 0.947048922532613
0.24 0.945869189585268
0.25 0.94347169841984
};
\addplot [semithick,  red!54.5098039215686!black, dotted]
table {%
0 0.999999999999993
0.01 0.990472536099833
0.02 0.98568821823414
0.03 0.981515734781917
0.04 0.977980671789743
0.05 0.974088135986291
0.06 0.971788236882893
0.07 0.969050431504394
0.08 0.966682332387726
0.09 0.964478435205485
0.1 0.962208991621967
0.11 0.959592775899521
0.12 0.956615020022379
0.13 0.954412815311357
0.14 0.952089220119958
0.15 0.948430495278319
0.16 0.946840451871742
0.17 0.94461770799657
0.18 0.94281550043694
0.19 0.940419008566227
0.2 0.937889977186338
0.21 0.93627383685349
0.22 0.934164998121214
0.23 0.932988379168843
0.24 0.931284482195253
0.25 0.929707561634564
};
\addplot [semithick,color3]
table {%
0 0.999999999999997
0.01 0.994073303165585
0.02 0.991058765981031
0.03 0.988562981755258
0.04 0.985622735206442
0.05 0.983591692565064
0.06 0.980337339047133
0.07 0.977130154749697
0.08 0.97467751307708
0.09 0.972745204976513
0.1 0.971547619033216
0.11 0.969988054500379
0.12 0.96902014009431
0.13 0.965786025005302
0.14 0.963426763645821
0.15 0.962473856142346
0.16 0.960421184286942
0.17 0.958778809978567
0.18 0.955690315842517
0.19 0.952583251955346
0.2 0.951005303570546
0.21 0.949200359399693
0.22 0.948259083095298
0.23 0.941999960113666
0.24 0.942981890924808
0.25 0.939644750201663
};
\addplot [semithick,color3, dotted]
table {%
0 0.999999999999996
0.01 0.990803932591675
0.02 0.98626484019769
0.03 0.983470230530923
0.04 0.980253089270622
0.05 0.978462326152066
0.06 0.975494006118102
0.07 0.973671132290183
0.08 0.970694313567343
0.09 0.968339783199389
0.1 0.966703604819059
0.11 0.965632990167622
0.12 0.964285915048747
0.13 0.963409625533494
0.14 0.961287231149446
0.15 0.960838935842106
0.16 0.959523979529797
0.17 0.95845374750089
0.18 0.956475051687853
0.19 0.95416227761176
0.2 0.95264699858864
0.21 0.950872108908108
0.22 0.950433999062545
0.23 0.941576659087528
0.24 0.94478735005544
0.25 0.939830946229131
};

\end{axis}

\end{tikzpicture}}
		\caption{Fraction of initially funded projects that are still~funded.}\label{fig:ov_robust_MES2}
	\end{subfigure}
	\caption{Some statistics regarding the $50\%$-winner threshold for \MES with different completion methods. }
	\label{fig:ov_robust_MES}
\end{figure}
In \Cref{fig:ov_robust_MES} (similar to \Cref{fig:ov_robust}), we take a more nuanced look at the robustness of outcomes. 
First, in \Cref{fig:ov_robust_MES1}, we show the probability that the outcome stays the same as the initial one depending on the resampling probability. 
The picture here is very similar to our conclusions from \Cref{ta:MES-rob}. The only completion methods that behave a bit differently are None and Add1+\Greedy. Compared to the other ones, None seems to produce quite non-robust results on the full dataset, yet more robust ones on the selected dataset. 
This implies that None is comparably non-robust on instances which are part of the full but not part of the selected dataset, i.e., instances which had very robust outcomes under all rules considered in the main body. 
The reason for why None produced non-robust outcomes might be that outcomes returned by None are often non-exhaustive, and whether additional projects get funded during the first phase of \MES might be irrelevant for the outcomes of other rules. 
For Add1+\Greedy we see that it is of comparable robustness to \Greedy on the full dataset, yet less robust on the selected one. 

Second, in \Cref{fig:ov_robust_MES2} we show the fraction of initially funded projects which are still funded for different values of the resampling probability.
So here we are not so much interested in whether the outcome changes but more by how much it changes. 
Interestingly, the relations between the different rules are different here than described before. 
The only similarity is that Epsilon still is the most robust one. 
However Add1, which we observed to be the least robust rule before here appears as the second most robust one, implying that for Add1 the outcome often changes but in case it changes it does not do so too drastically. 
None comes afterwards and the two least robust methods in terms of the fraction of overturned funding decisions are Add1+\Greedy and \Greedy. 
Given that both of them include a \Greedy part, the reason for this could be that if \Greedy is started with a different initial budget or a different set of available project (which happens in case the output of the first phase of \MES changes), then other funding decisions are expected to be overturned as well, as for \Greedy one change can easily trigger lengthy subsequent changes.

\subsubsection{An Instance-Based View on Different Completion Methods.}
We conclude by presenting the behavior of the different completion methods on some exemplary instances in \Cref{fig:MES1,fig:MES2,fig:MES3,fig:MES4,fig:MES5}.\footnote{Recall that an empty plot means that for this instance and rule, there is no initially funded project whose funding probability dropped below $90\%$ and no initially not funded project whose funding probability exceeded $10\%$.} 
These instances should illustrate that Add1+\Greedy and Add1, and \Greedy and None, behave often somewhat similarly, but that all other pairs of rules often behaves fundamentally different from each other. 
For instance in \Cref{fig:MES1,fig:MES2}, we see that \Greedy and None produce very non-robust outcomes, whereas the outcomes for Add1 and Add1+\Greedy are more robust.
Epsilon produces a semi-robust outcome in \Cref{fig:MES1} , and a very robust one in \Cref{fig:MES2}.
In \Cref{fig:MES3} we show that it can also be the other way around, as here Add1 and Add1+\Greedy produce very non-robust outcomes, whereas Epsilon, \Greedy, and None are very robust. 
Lastly, in \Cref{fig:MES4} and \Cref{fig:MES5} we show that although we have previously observed that \Greedy and None, and Add1 and Add1+\Greedy behave quite similarly on many instances, there are also examples where this is not the case. 
Notably, in both cases, this is because None and Add1 return non-exhaustive outcomes on the respective instances, which cannot happen for Add1+\Greedy and \Greedy.

\begin{figure}[t!]
	\centering 
	\begin{subfigure}{0.195\textwidth}
		\includegraphics[width=\textwidth]{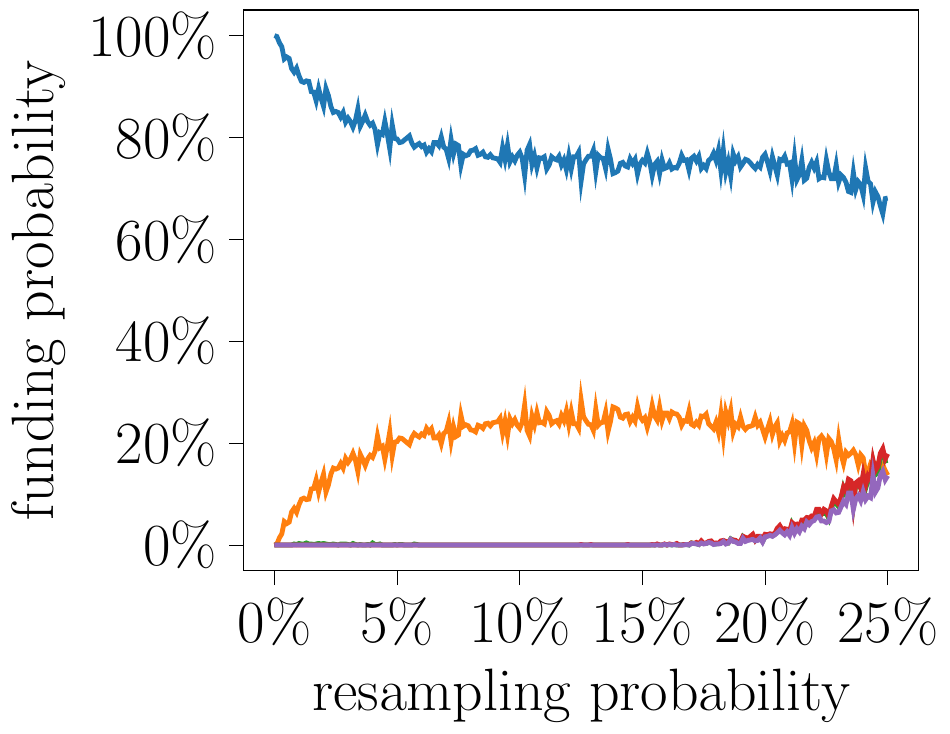}
		\caption{Add1+\Greedy}
	\end{subfigure}\hfill
	\begin{subfigure}{0.195\textwidth}
		\includegraphics[width=\textwidth]{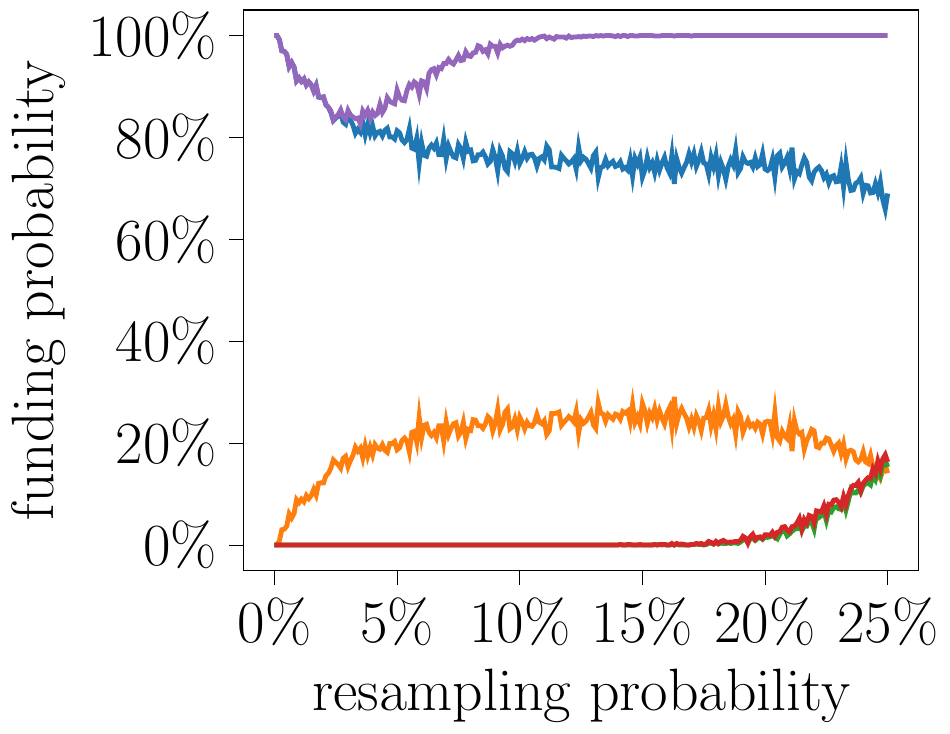}
		\caption{Add1}
	\end{subfigure}\hfill%
	\begin{subfigure}{0.195\textwidth}
		\includegraphics[width=\textwidth]{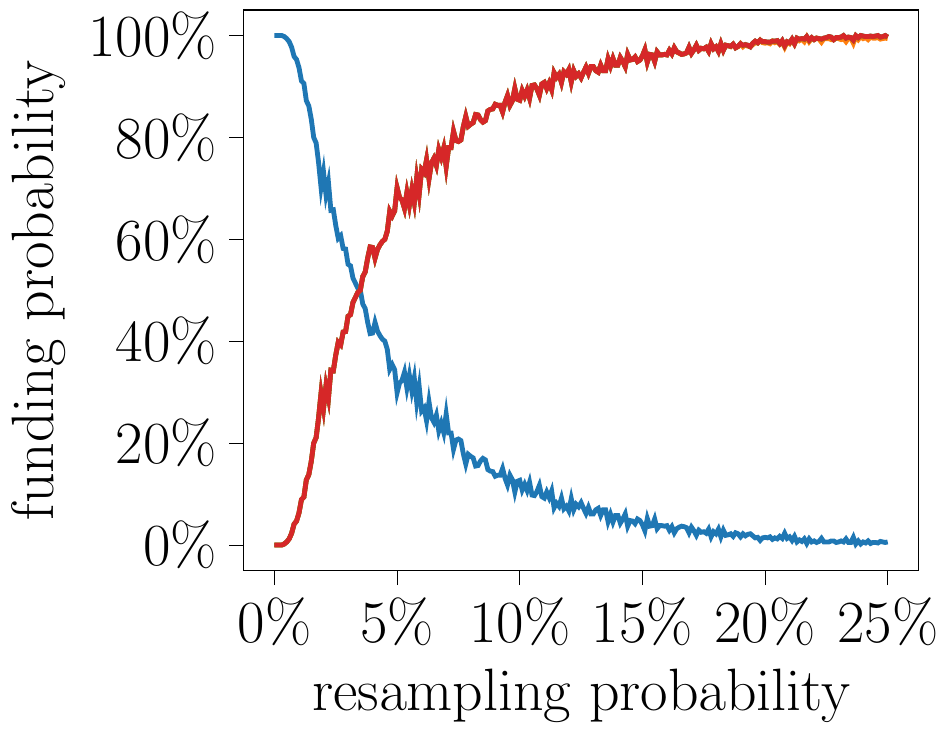}
		\caption{Epsilon}
	\end{subfigure}\hfill
	\begin{subfigure}{0.195\textwidth}
		\includegraphics[width=\textwidth]{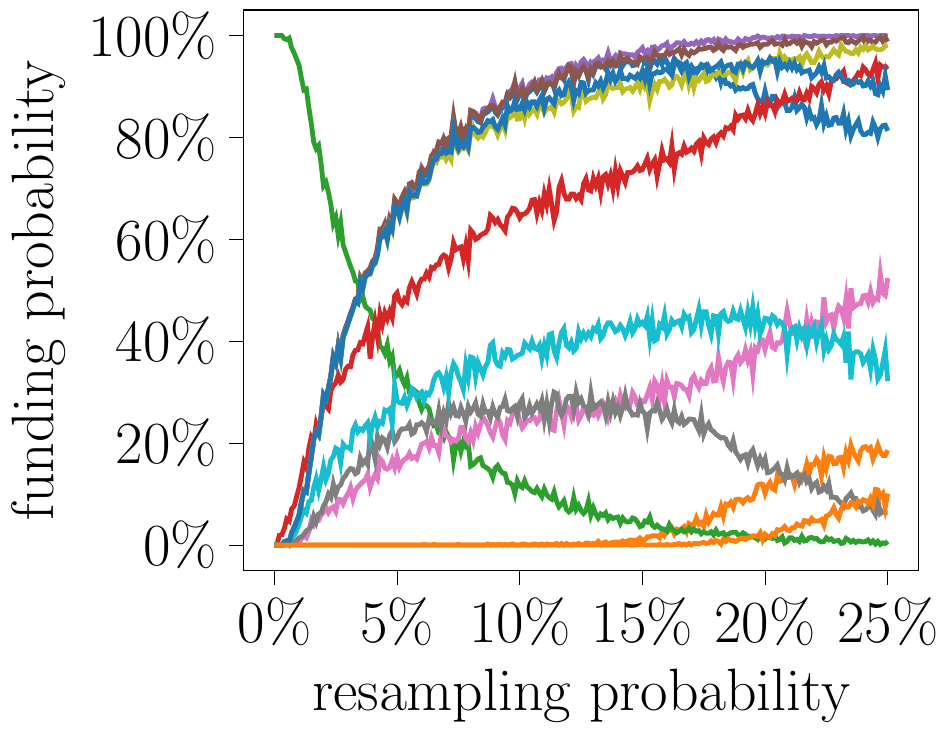}
		\caption{\Greedy}
	\end{subfigure}\hfill
	\begin{subfigure}{0.195\textwidth}
	\includegraphics[width=\textwidth]{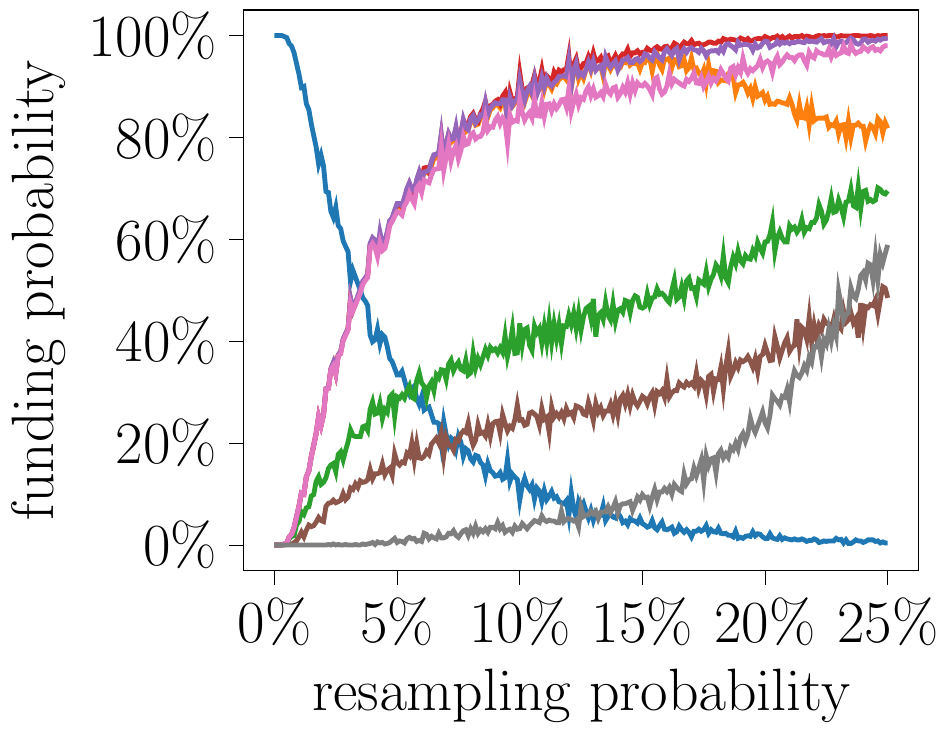}
	\caption{None}
	\end{subfigure}
	\caption{Bemowo 2022 (Warszawa) for \MES with different completion methods.} \label{fig:MES1}
\end{figure}

\begin{figure}[t!]
	\centering 
	\begin{subfigure}{0.195\textwidth}
		\includegraphics[width=\textwidth]{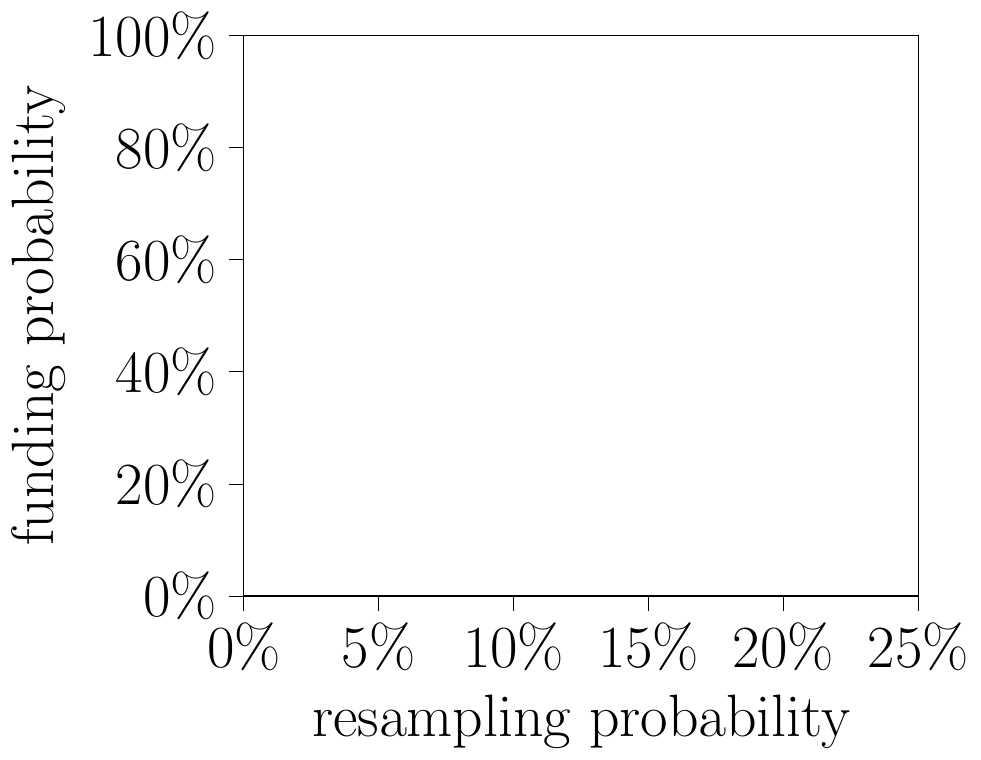}
		\caption{Add1+\Greedy}
	\end{subfigure}\hfill
	\begin{subfigure}{0.195\textwidth}
		\includegraphics[width=\textwidth]{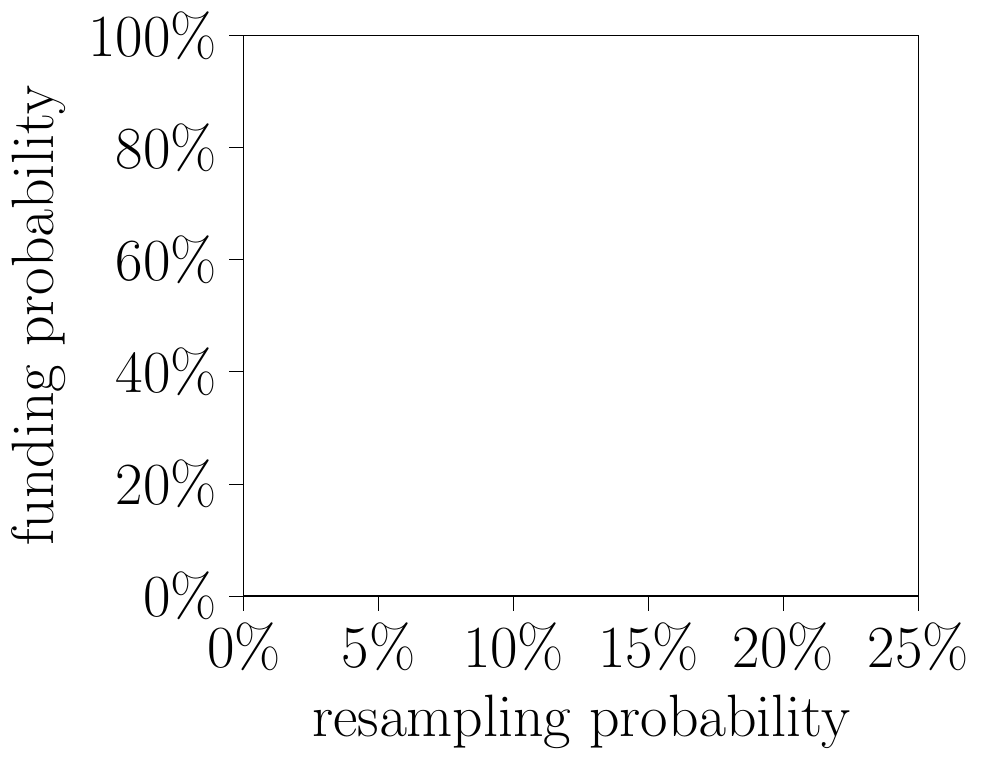}
		\caption{Add1}
	\end{subfigure}\hfill%
	\begin{subfigure}{0.195\textwidth}
		\includegraphics[width=\textwidth]{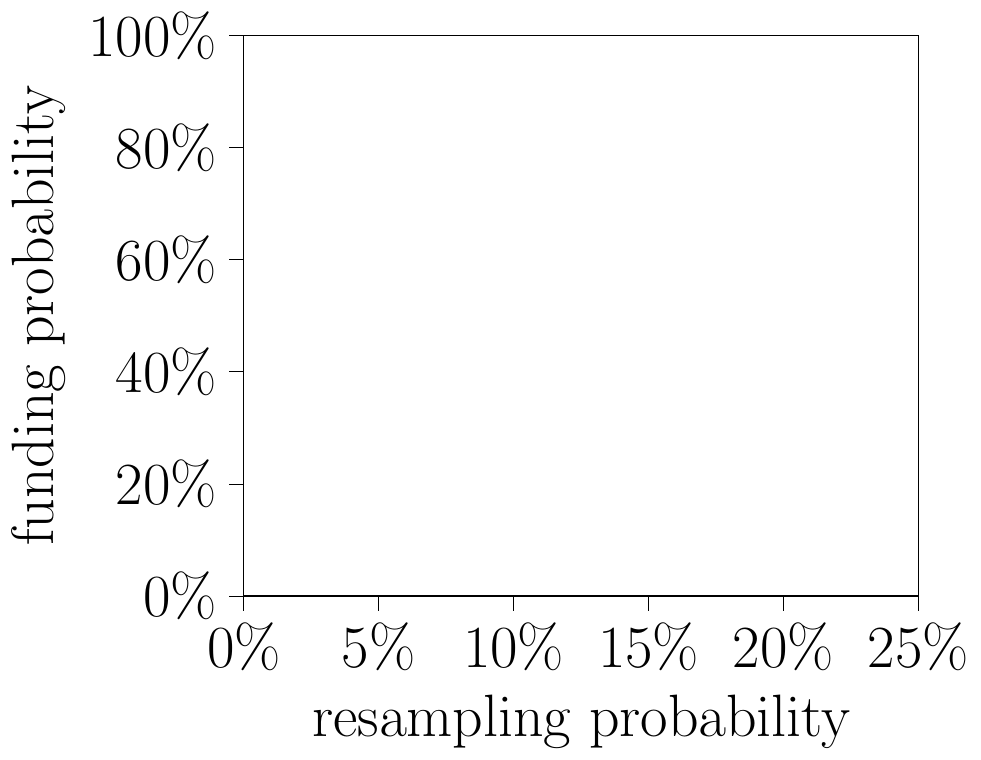}
		\caption{Epsilon}
	\end{subfigure}\hfill
	\begin{subfigure}{0.195\textwidth}
		\includegraphics[width=\textwidth]{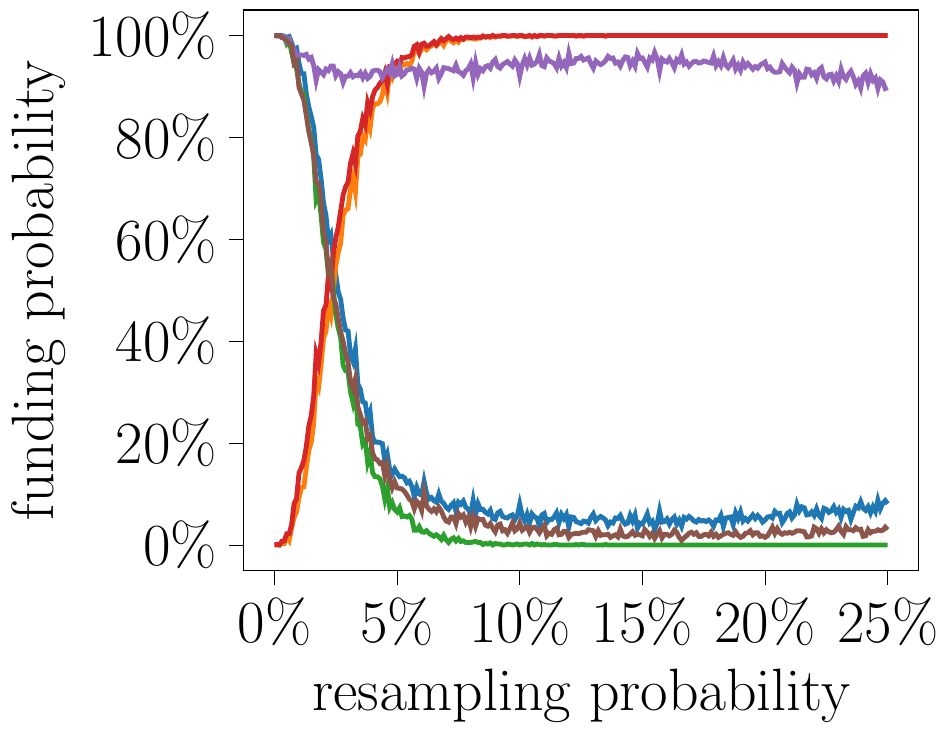}
		\caption{\Greedy}
	\end{subfigure}\hfill
	\begin{subfigure}{0.195\textwidth}
		\includegraphics[width=\textwidth]{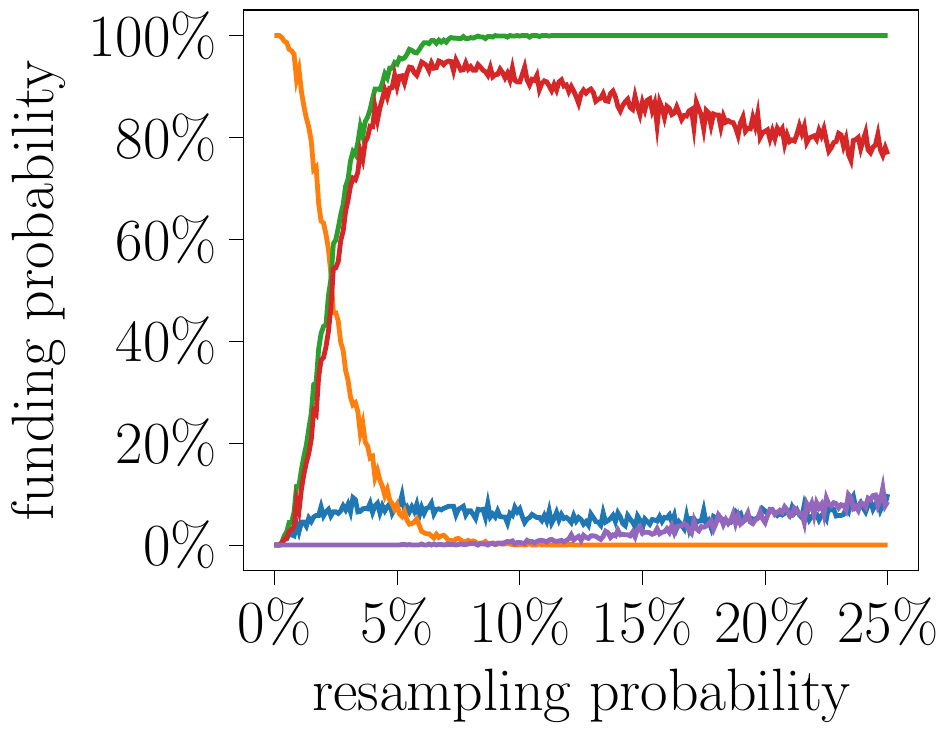}
		\caption{None}
	\end{subfigure}
	\caption{Chomiczowka 2018 (Warszawa) for \MES with different completion methods.}\label{fig:MES2}
\end{figure}

\begin{figure}[t!]
	\centering 
	\begin{subfigure}{0.195\textwidth}
		\includegraphics[width=\textwidth]{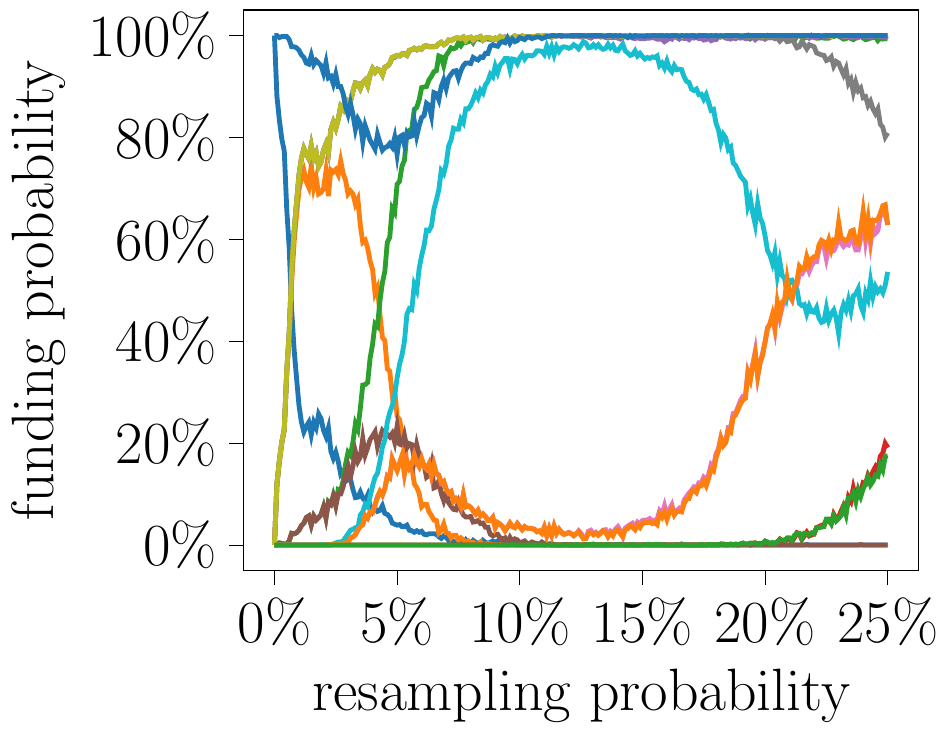}
		\caption{Add1+\Greedy}
	\end{subfigure}\hfill
	\begin{subfigure}{0.195\textwidth}
		\includegraphics[width=\textwidth]{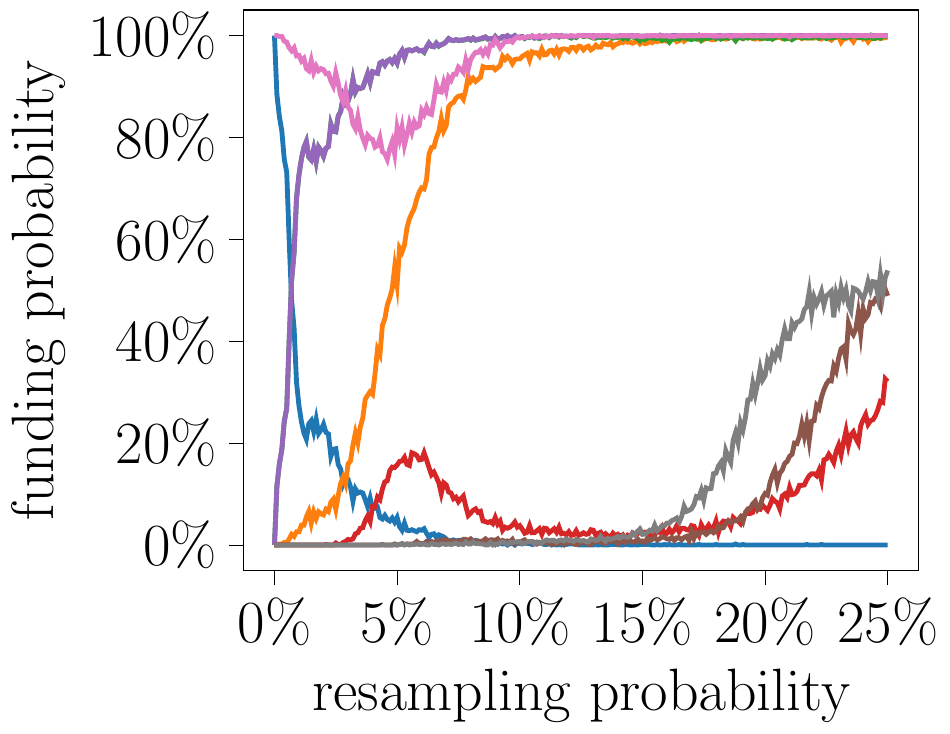}
		\caption{Add1}
	\end{subfigure}\hfill%
	\begin{subfigure}{0.195\textwidth}
		\includegraphics[width=\textwidth]{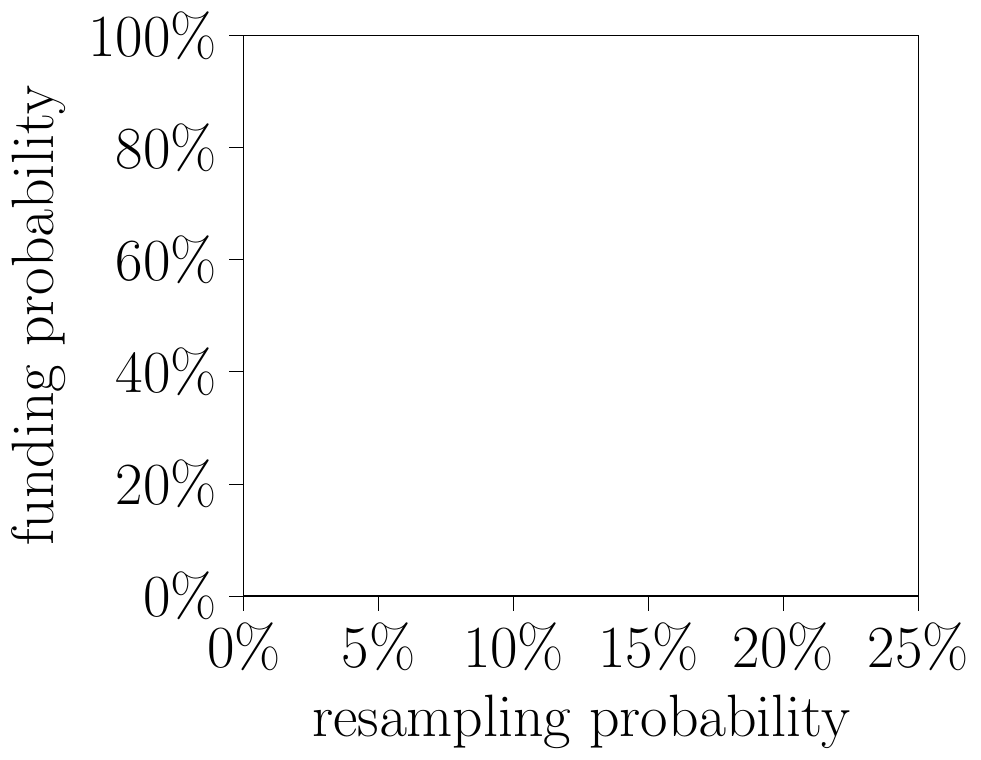}
		\caption{Epsilon}
	\end{subfigure}\hfill
	\begin{subfigure}{0.195\textwidth}
		\includegraphics[width=\textwidth]{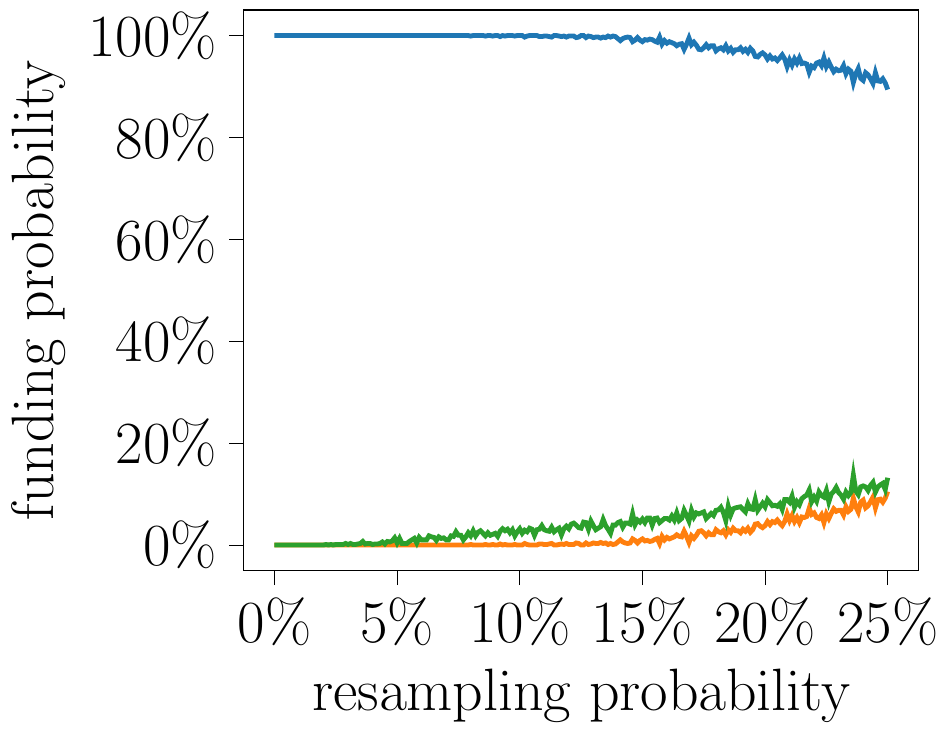}
		\caption{\Greedy}
	\end{subfigure}\hfill
	\begin{subfigure}{0.195\textwidth}
		\includegraphics[width=\textwidth]{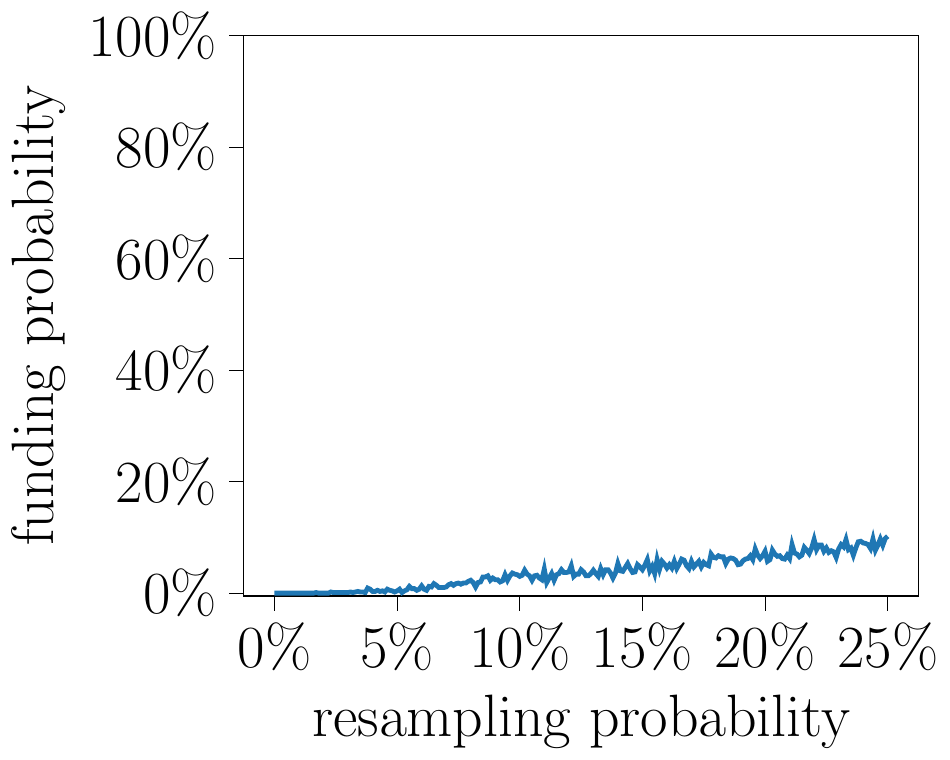}
		\caption{None}
	\end{subfigure}
	\caption{Brodno Podgrodzie 2018 (Warszawa) for \MES with different completion methods.}\label{fig:MES3}
\end{figure}

\begin{figure}[t!]
	\centering 
	\begin{subfigure}{0.195\textwidth}
		\includegraphics[width=\textwidth]{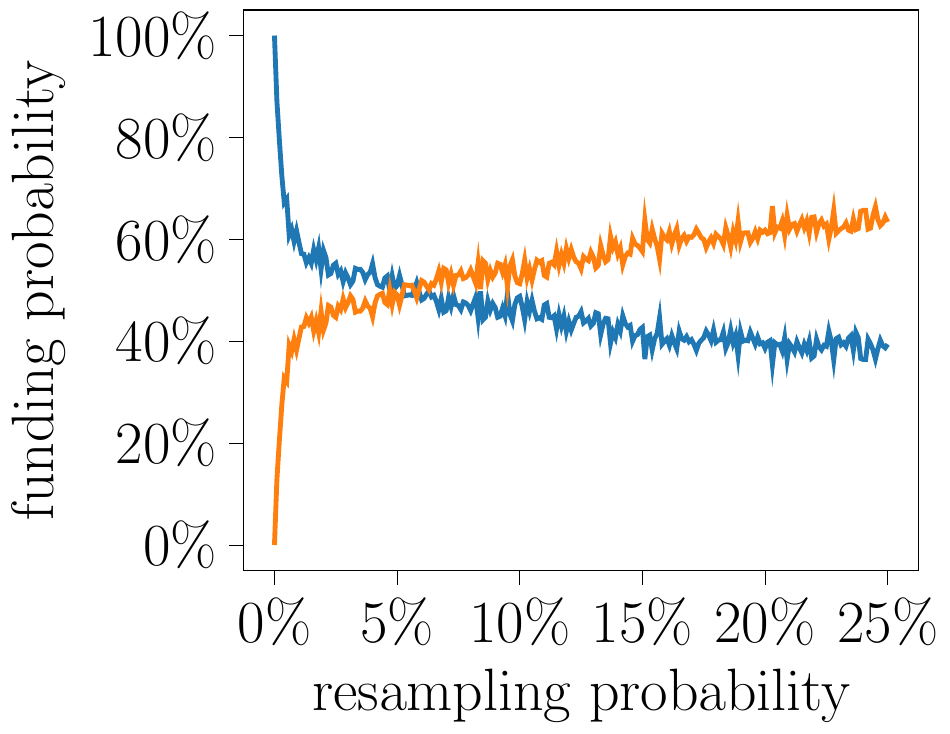}
		\caption{Add1+\Greedy}
	\end{subfigure}\hfill
	\begin{subfigure}{0.195\textwidth}
		\includegraphics[width=\textwidth]{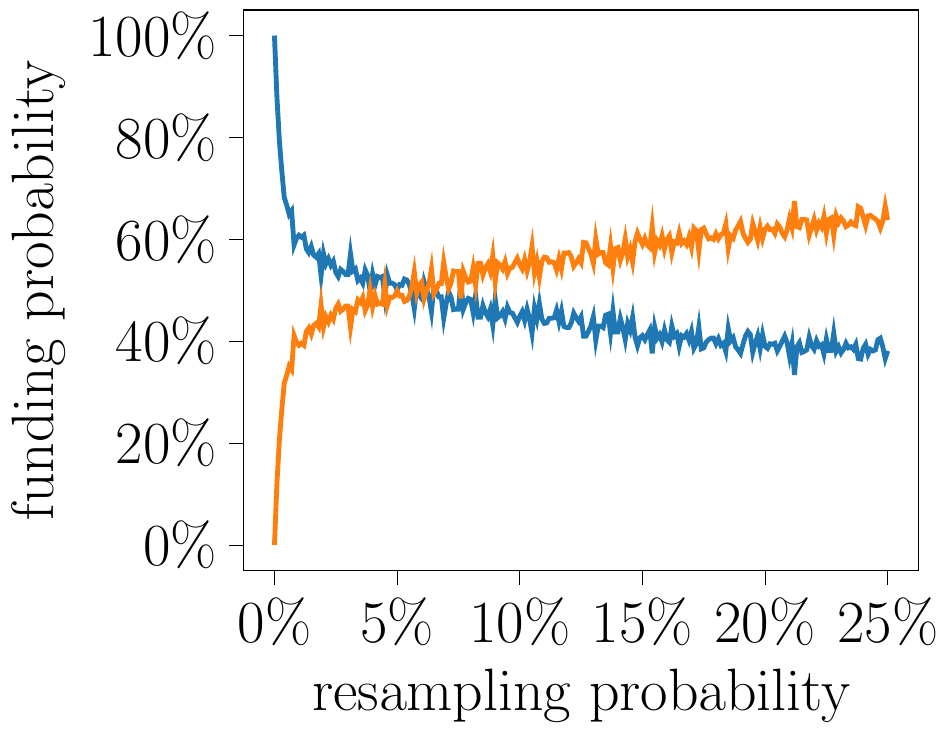}
		\caption{Add1}
	\end{subfigure}\hfill%
	\begin{subfigure}{0.195\textwidth}
		\includegraphics[width=\textwidth]{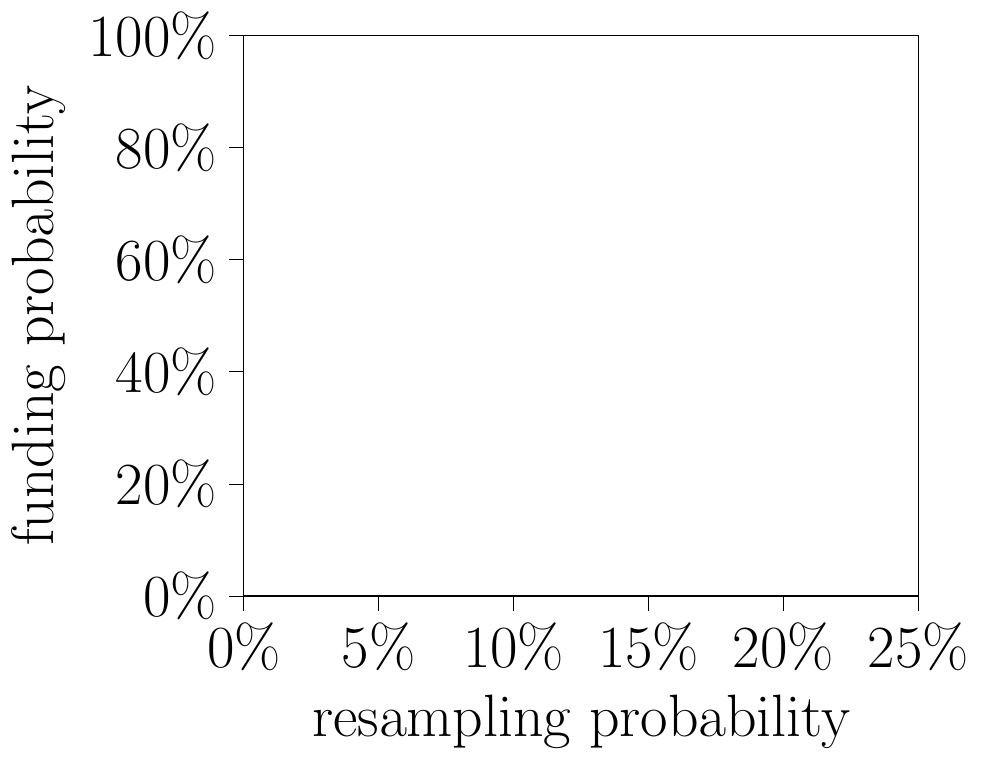}
		\caption{Epsilon}
	\end{subfigure}\hfill
	\begin{subfigure}{0.195\textwidth}
		\includegraphics[width=\textwidth]{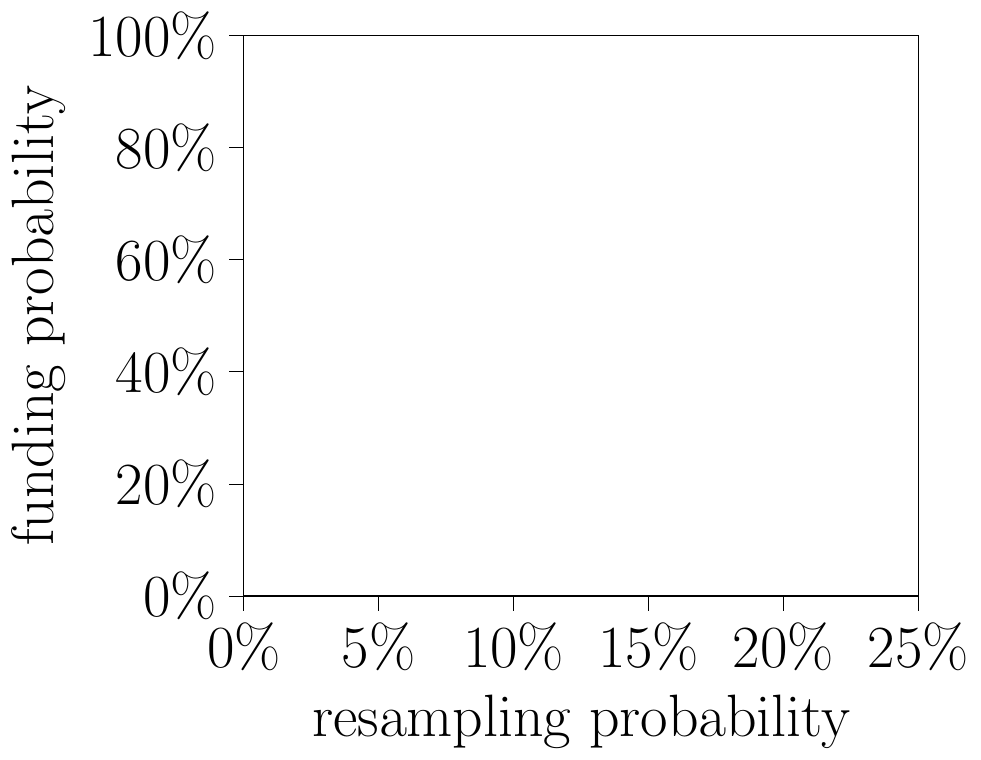}
		\caption{\Greedy}
	\end{subfigure}\hfill
	\begin{subfigure}{0.195\textwidth}
		\includegraphics[width=\textwidth]{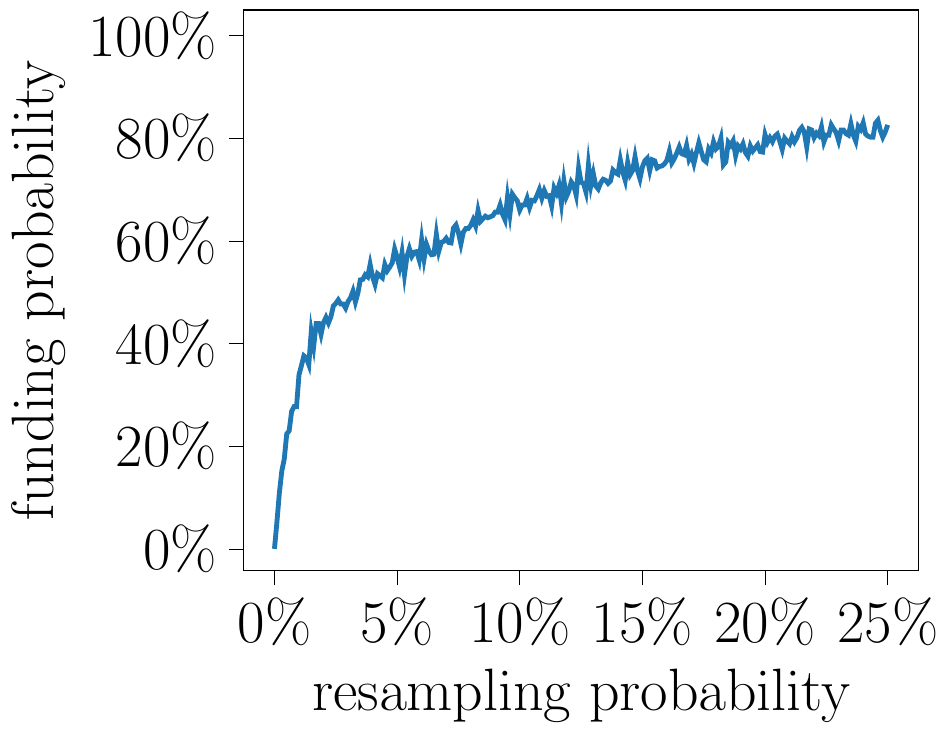}
		\caption{None}
	\end{subfigure}
	\caption{Obzar 2019 (Warszawa) for \MES with different completion methods.}\label{fig:MES4}
\end{figure}

\begin{figure}[t!]
	\centering 
	\begin{subfigure}{0.195\textwidth}
		\includegraphics[width=\textwidth]{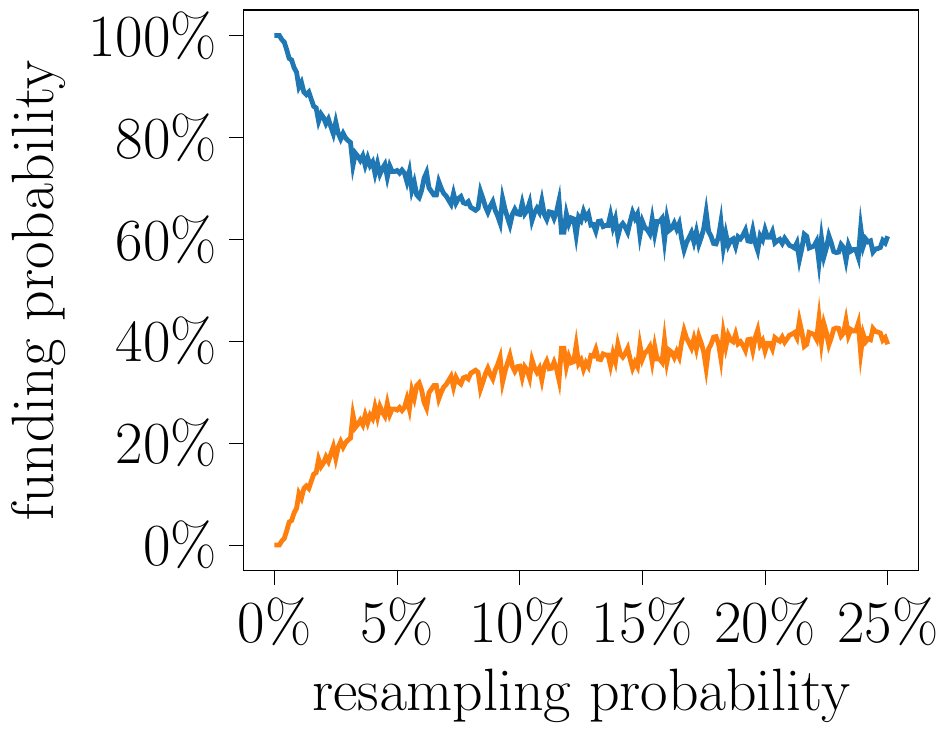}
		\caption{Add1+\Greedy}
	\end{subfigure}\hfill
	\begin{subfigure}{0.195\textwidth}
		\includegraphics[width=\textwidth]{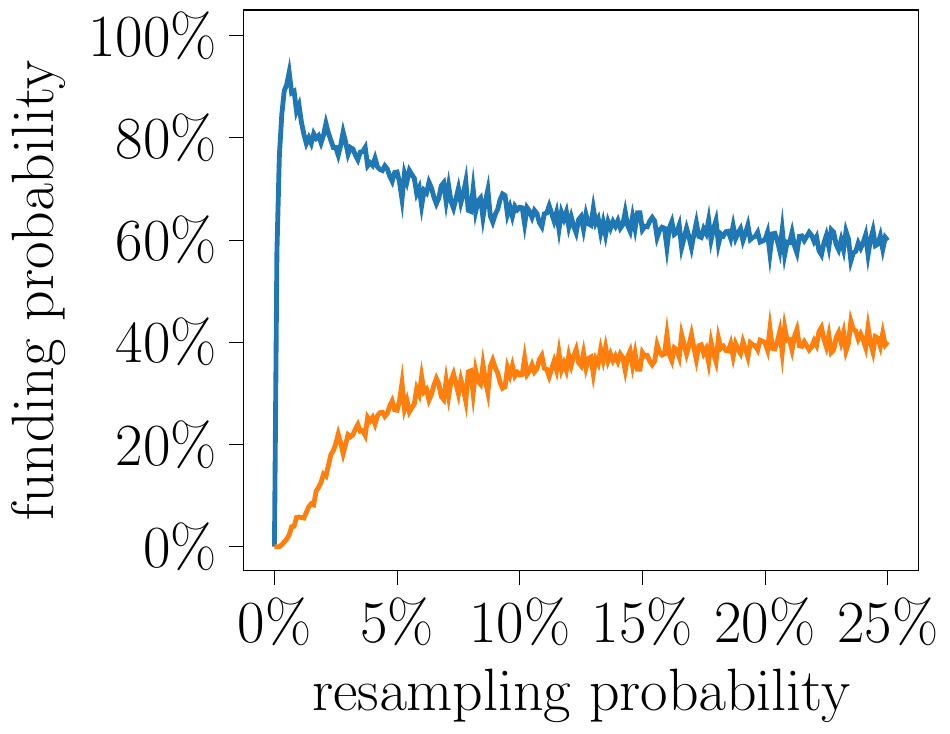}
		\caption{Add1}
	\end{subfigure}\hfill%
	\begin{subfigure}{0.195\textwidth}
		\includegraphics[width=\textwidth]{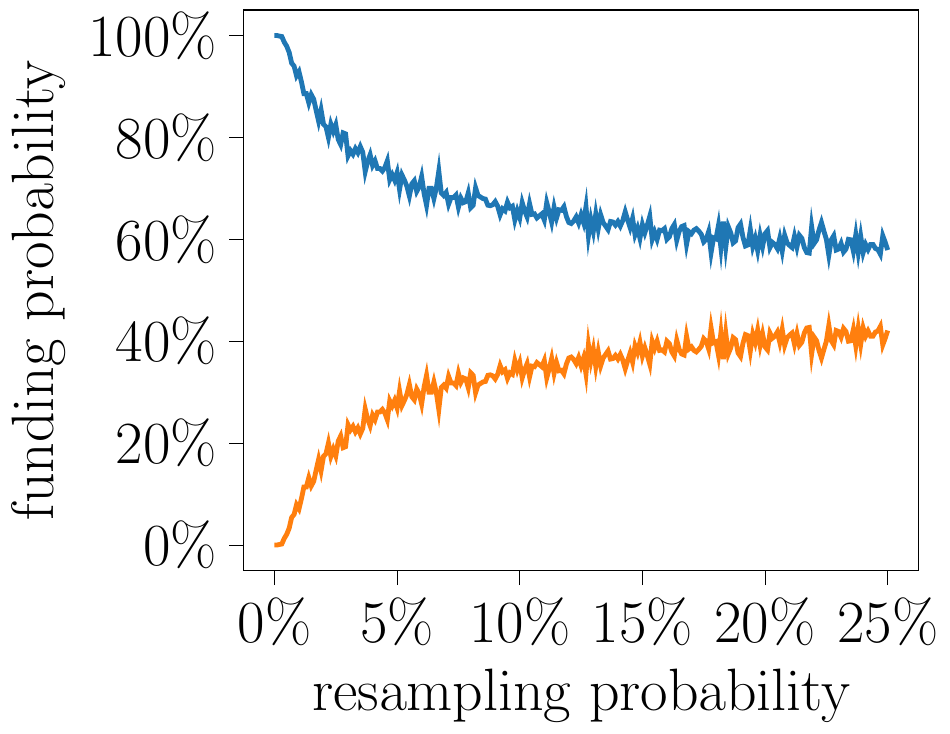}
		\caption{Epsilon}
	\end{subfigure}\hfill
	\begin{subfigure}{0.195\textwidth}
		\includegraphics[width=\textwidth]{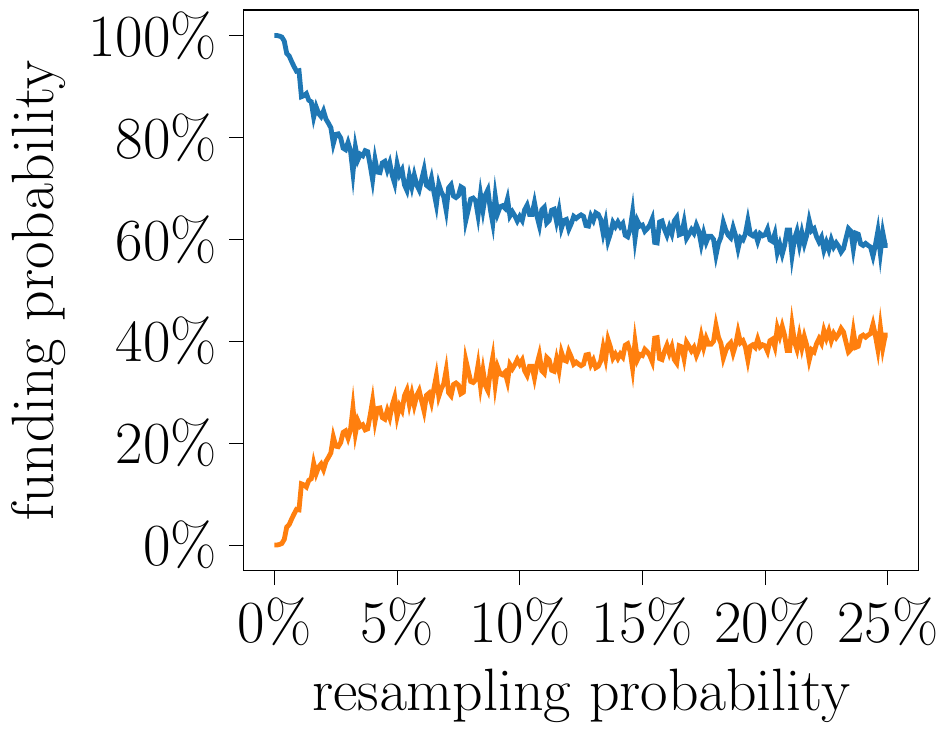}
		\caption{\Greedy}
	\end{subfigure}\hfill
	\begin{subfigure}{0.195\textwidth}
		\includegraphics[width=\textwidth]{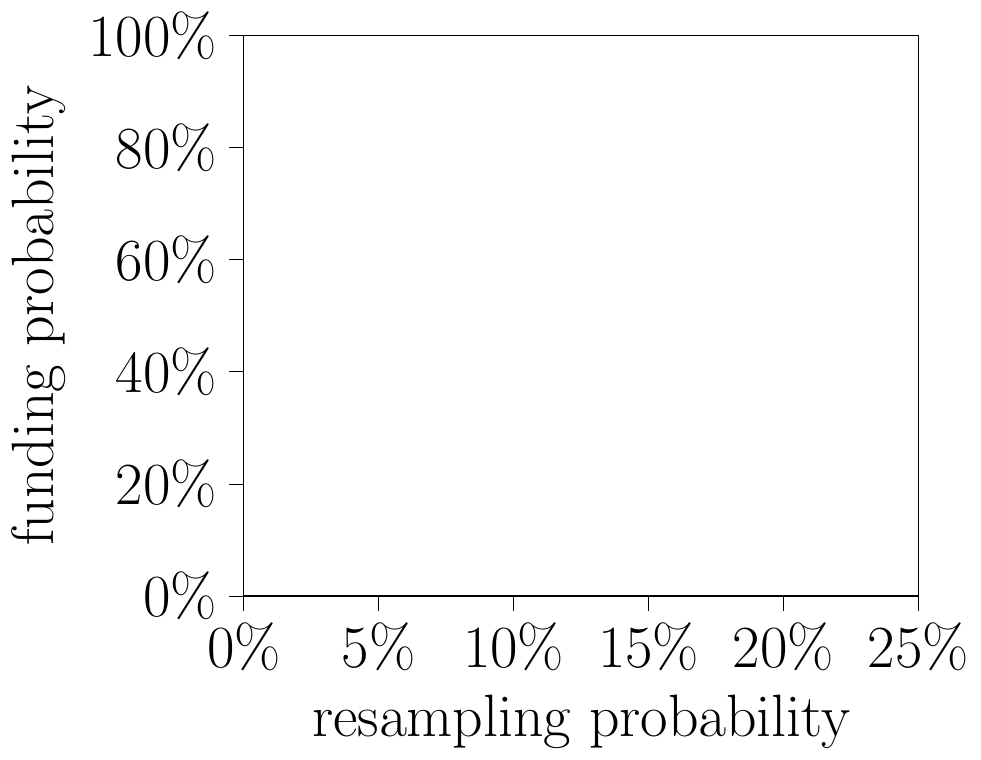}
		\caption{None}
	\end{subfigure}
	\caption{Rejon 13 2018 (Wroclaw) for \MES with different completion methods.} \label{fig:Rejon}\label{fig:MES5}
\end{figure}


\begin{thebibliography}{10}
	\providecommand{\url}[1]{\texttt{#1}}
	\providecommand{\urlprefix}{URL }
	\providecommand{\doi}[1]{https://doi.org/#1}
	
	\bibitem{bau-hog:c:robust-winner}
	Baumeister, D., Hogrebe, T.: On the complexity of predicting election outcomes
	and estimating their robustness. In: Proceedings of EUMAS-2021. pp. 228--244
	(2021)
	
	\bibitem{boe-bre-fal-nie:c:counting-bribery}
	Boehmer, N., Bredereck, R., Faliszewski, P., Niedermeier, R.: Winner robustness
	via swap- and shift-bribery: Parameterized counting complexity and
	experiments. In: Proceedings of IJCAI-2021. pp. 52--58 (2021)
	
	\bibitem{boe-bre-fal-nie:c:robustness-single-winner}
	Boehmer, N., Bredereck, R., Faliszewski, P., Niedermeier, R.: A quantitative
	and qualitative analysis of the robustness of (real-world) election winners.
	In: Proceedings of EAAMO-2022. pp. 7:1--7:10 (2022)
	
	\bibitem{bre-fal-kac-nie-sko-tal:j:multiwinner-robustness}
	Bredereck, R., Faliszewski, P., Kaczmarczyk, A., Niedermeier, R., Skowron, P.,
	Talmon, N.: Robustness among multiwinner voting rules. Artificial
	Intelligence  \textbf{290},  103403 (2021)
	
	\bibitem{bri-fre-jan-lac:c:phragmen}
	Brill, M., Freeman, R., Janson, S., Lackner, M.: Phragm{\'{e}}n's voting
	methods and justified representation. In: Proceedings of AAAI-2017. pp.
	406--413 (2017)
	
	\bibitem{car:c:stv-margin-of-victory}
	Cary, D.: Estimating the margin of victory for instant-runoff voting (2011),
	presented at 2011 Electronic Voting Technology Workshop/Workshop on
	Trushworthy Elections
	
	\bibitem{dow-fel:c:fpt-intractability}
	Downey, R., Fellows, M.: Fixed-parameter intractability. In: Proceedings of the
	7th Structure in Complexity Theory Conference. pp. 36--49 (1992)
	
	\bibitem{fal-gaw-kus:c:greedy-robustness}
	Faliszewski, P., Gawron, G., Kusek, B.: Robustness of greedy approval rules.
	In: Proceedings of EUMAS-2022. pp. 116--133 (2022)
	
	\bibitem{fal-hem-hem:j:bribery}
	Faliszewski, P., Hemaspaandra, E., Hemaspaandra, L.: How hard is bribery in
	elections? Journal of Artificial Intelligence Research  \textbf{35},
	485--532 (2009)
	
	\bibitem{fal-rot:b:control-bribery}
	Faliszewski, P., Rothe, J.: Control and bribery in voting. In: Brandt, F.,
	Conitzer, V., Endriss, U., Lang, J., Procaccia, A.D. (eds.) Handbook of
	Computational Social Choice, chap.~7. Cambridge University Press (2015)
	
	\bibitem{fal-sko-tal:c:bribery-measure-success}
	Faliszewski, P., Skowron, P., Talmon, N.: Bribery as a measure of candidate
	success: Complexity results for approval-based multiwinner rules. In:
	Proceedings of AAMAS-2017. pp. 6--14 (2017)
	
	\bibitem{flu-gro:j:parameterized-counting}
	Flum, J., Grohe, M.: The parameterized complexity of counting problems. SIAM
	Journal on Computing  \textbf{33}(4),  892--922 (2004)
	
	\bibitem{gar-joh1979:b:comp-and-intract}
	Garey, M.R., Johnson, D.S.: Computers and Intractability: {A} Guide to the
	Theory of NP-Completeness. W. H. Freeman (1979)
	
	\bibitem{gaw-fal:c:robustness-approval}
	Gawron, G., Faliszewski, P.: Robustness of approval-based multiwinner voting
	rules. In: Proceedings of ADT-2019. pp. 17--31 (2019)
	
	\bibitem{jan-fal:c:ties-multiwinner}
	Janeczko, L., Faliszewski, P.: Ties in multiwinner approval voting. In:
	Proceedings of IJCAI-2023 (2023), to appear. Available at
	arxiv.org/abs/2305.01769{.}
	
	\bibitem{kno-kou-mni:j:fpt-m-bribery}
	Knop, D., Kouteck{\'{y}}, M., Mnich, M.: Voting and bribing in
	single-exponential time. {ACM} Transactions on Economics and Computation
	\textbf{8}(3),  12:1--12:28 (2020)
	
	\bibitem{lac-sko:b:approval-survey}
	Lackner, M., Skowron, P.: Multi-Winner Voting with Approval Preferences.
	Springer Briefs in Intelligent Systems, Springer (2023)
	
	\bibitem{los-chr-gro:c:phragmen-pb}
	Los, M., Christoff, Z., Grossi, D.: Proportional budget allocations: Towards a
	systematization. In: Proceedings of IJCAI-2022. pp. 398--404 (2022)
	
	\bibitem{mag-riv-she-wag:c:stv-bribery}
	Magrino, T., Rivest, R., Shen, E., Wagner, D.: Computing the margin of victory
	in {IRV} elections (2011), presented at 2011 Electronic Voting Technology
	Workshop/Workshop on Trushworthy Elections
	
	\bibitem{mcc:j:parameterized-counting}
	McCartin, C.: Parameterized counting problems. Annals of Pure and Applied Logic
	\textbf{138}(1-3),  147--182 (2006)
	
	\bibitem{pet-pie-sko:c:pb-mes}
	Peters, D., Pierczynski, G., Skowron, P.: Proportional participatory budgeting
	with additive utilities. In: Proceedings of NeurIPS-2021. pp. 12726--12737
	(2021)
	
	\bibitem{pet-sko:c:welfarism-mes}
	Peters, D., Skowron, P.: Proportionality and the limits of welfarism. In:
	Proceedings of EC-2020. pp. 793--794 (2020)
	
	\bibitem{mes}
	Peters, D., Skowron, P.: Method of equal shares. \url{https://equalshares.net}
	(2023)
	
	\bibitem{pabutools}
	Pierczynski, G.: Pabutools (2023), \url{pypi.org/project/pabutools/}
	
	\bibitem{rey-mal:t:pb-survey}
	Rey, S., Maly, J.: The (computational) social choice take on indivisible
	participatory budgeting. Tech. Rep. arXiv.2303.00621, arXiv (2023)
	
	\bibitem{shi-yu-elk:c:robustness}
	Shiryaev, D., Yu, L., Elkind, E.: On elections with robust winners. In:
	Proceedings of AAMAS-2013. pp. 415--422 (2013)
	
	\bibitem{pabulib}
	Stolicki, D., Szufa, S., Talmon, N.: Pabulib: A participatory budgeting
	library. Tech. Rep. arXiv:2012.06539, arXiv (2020)
	
	\bibitem{DBLP:conf/ijcai/SzufaFJLSST22}
	Szufa, S., Faliszewski, P., Janeczko, L., Lackner, M., Slinko, A., Sornat, K.,
	Talmon, N.: How to sample approval elections? In: Proceedings of IJCAI-2022.
	pp. 496--502 (2022)
	
	\bibitem{xia:margin-of-victory}
	Xia, L.: Computing the margin of victory for various voting rules. In:
	Proceedings of EC-2012. pp. 982--999 (2012)
	
\end{thebibliography}
\end{document}